\documentclass[acmsmall,screen,nonacm]{acmart}

\setcopyright{acmcopyright}
\copyrightyear{2018}
\acmYear{2018}
\acmDOI{XXXXXXX.XXXXXXX}

\acmJournal{JACM}
\acmVolume{37}
\acmNumber{4}
\acmArticle{111}
\acmMonth{8}

\acmPrice{15.00}
\acmISBN{978-1-4503-XXXX-X/18/06}

\usepackage{mathpartir}
\usepackage{amsthm}
\usepackage{stmaryrd}
\usepackage{hyperref}
\usepackage{turnstile}
\usepackage{xspace}
\usepackage{listings}
\usepackage{enumitem}

\usepackage{formal-grammar}
\usepackage[page,header]{appendix}

\theoremstyle{definition}
\newtheorem{lemma}{Lemma}
\newtheorem{definition}{Definition}

\newtheorem{theorem}{Theorem}

\newcommand{\supplement}[2]{#2}

\newcommand{\gvc}{\textrm{GVC0}\xspace}
\newcommand{\gco}{\textrm{Gradual C0}\xspace}
\newcommand{\co}{\textrm{C\textsubscript{0}}\xspace}
\newcommand{\gviper}{\textrm{Gradual Viper}\xspace}
\newcommand{\svl}{\texorpdfstring{\textrm{SVL\textsubscript{C0}}}{SVL\_C0}\xspace}
\newcommand{\gvl}{\texorpdfstring{\textrm{GVL\textsubscript{C0}}}{GVL\_C0}\xspace}

\newcommand{\gvlrp}{\textrm{GVL\textsubscript{RP}}\xspace}

\newcommand{\ttt}[1]{\texttt{#1}}
\newcommand{\multiple}[1]{\overline{#1}}
\newcommand{\disableTttResize}[0]{\renewcommand{\small}[1]{##1}}

\newcommand{\dom}{\operatorname{dom}}
\newcommand{\notimplies}{\Longarrownot\Longrightarrow}

\newcommand{\pair}[2]{\langle #1, \, #2 \rangle}
\newcommand{\triple}[3]{\langle #1,\, #2, \, #3 \rangle}
\newcommand{\quadruple}[4]{\langle #1,\, #2,\, #3,\, #4 \rangle}
\newcommand{\quintuple}[5]{\langle #1,\, #2,\, #3,\, #4,\, #5 \rangle}

\newcommand{\gprogram}{\mathrm{program}}
\newcommand{\gpredicate}{\mathcal{P}}
\newcommand{\gmethod}{\mathcal{M}}
\newcommand{\gstruct}{\mathcal{S}}
\newcommand{\gform}{\tilde{\phi}}
\newcommand{\gstatement}{\mathit{s}}
\newcommand{\gcontract}{\mathrm{\Phi}}
\newcommand{\gtype}{\mathit{T}}

\newcommand{\gexpression}{\mathit{e}}
\newcommand{\gvar}{\mathit{x}}

\newcommand{\gpreciseform}{\phi}

\newcommand{\kensures}{\ttt{ensures}}
\newcommand{\krequires}{\ttt{requires}}
\newcommand{\kalloc}{\ttt{alloc}}
\newcommand{\knull}{\ttt{null}}
\newcommand{\kacc}{\ttt{acc}}
\newcommand{\kskip}{\ttt{skip}}
\newcommand{\kif}{\ttt{if}}
\newcommand{\kthen}{\ttt{then}}
\newcommand{\kelse}{\ttt{else}}
\newcommand{\kassert}{\ttt{assert}}
\newcommand{\kstruct}{\ttt{struct}}
\newcommand{\kwhile}{\ttt{while}}
\newcommand{\kinvariant}{\ttt{invariant}}
\newcommand{\kdo}{\ttt{do}}
\newcommand{\kfold}{\ttt{fold}}
\newcommand{\kunfold}{\ttt{unfold}}
\newcommand{\kint}{\ttt{int}}
\newcommand{\kchar}{\ttt{char}}
\newcommand{\kbool}{\ttt{bool}}
\newcommand{\ktrue}{\ttt{true}}
\newcommand{\kfalse}{\ttt{false}}
\newcommand{\kresult}{\ttt{result}}
\newcommand{\keq}{\mathbin{\ttt{==}}}
\newcommand{\kneq}{\mathrel{\ttt{!=}}}
\newcommand{\kadd}{\mathbin{\ttt{+}}}
\newcommand{\ksub}{\mathbin{\ttt{-}}}
\newcommand{\kdiv}{\mathbin{\ttt{/}}}
\newcommand{\kmul}{\mathbin{\ttt{*}}}
\newcommand{\kand}{\mathbin{\ttt{\&\&}}}
\newcommand{\kor}{\mathbin{\ttt{||}}}
\newcommand{\kneg}{\mathop{\ttt{!}}}
\newcommand{\kassign}{\mathbin{\ttt{=}}}

\newcommand{\sblock}[1]{\{ ~#1~ \}}
\newcommand{\salloc}[1]{\kalloc(#1)}
\newcommand{\sseq}[2]{#1 \ttt{;} ~#2}
\newcommand{\sif}[3]{\kif ~#1~ \kthen ~#2~ \kelse ~#3}
\newcommand{\sassert}[1]{\kassert ~#1}
\newcommand{\swhile}[3]{\kwhile ~#1~ \kinvariant ~#2~ \kdo ~#3}
\newcommand{\sfold}[1]{\kfold ~#1}
\newcommand{\sunfold}[1]{\kunfold ~#1}
\newcommand{\simprecise}[1]{\ttt{?} * #1}
\newcommand{\smethdef}[5]{#1 ~ #2(#3) ~#4~ #5}

\newcommand{\Literal}{\textsc{Literal}}
\newcommand{\Value}{\textsc{Value}}
\newcommand{\Field}{\textsc{Field}}
\newcommand{\Formula}{\textsc{Formula}}
\newcommand{\GFormula}{\tilde{\textsc{F}}\textsc{ormula}}
\newcommand{\Location}{\textsc{Location}}
\newcommand{\Var}{\textsc{Var}}
\newcommand{\Expr}{\textsc{Expr}}
\newcommand{\Stmt}{\textsc{Stmt}}
\newcommand{\Method}{\textsc{Method}}
\newcommand{\Predicate}{\textsc{Predicate}}
\newcommand{\Struct}{\textsc{Struct}}
\newcommand{\Type}{\textsc{Type}}

\newcommand{\Perm}{\textsc{Perm}}

\newcommand{\SExpr}{\textsc{SExpr}}
\newcommand{\SValue}{\textsc{SValue}}
\newcommand{\SField}{\textsc{SField}}
\newcommand{\SPredicate}{\textsc{SPredicate}}
\newcommand{\SCheck}{\textsc{SCheck}}
\newcommand{\SPerm}{\textsc{SPerm}}
\newcommand{\SState}{\textsc{SState}}

\newcommand{\eval}[4]{\pair{#1}{#2} \vdash #3 \Downarrow #4}

\newcommand{\assertion}[4]{\triple{#1}{#2}{#3} \vDash #4}

\newcommand{\foot}[4]{\lfloor #4 \rfloor_{\triple{#1}{#2}{#3}}}

\newcommand{\efoot}[3]{\llfloor #3 \rrfloor_{\pair{#1}{#2}}}

\newcommand{\vfoot}[3]{#1 \llparenthesis #3 \rrparenthesis_{#2}}

\newcommand{\frm}[4]{\triple{#1}{#2}{#3} \vdash_{\mathrm{frm}} #4}
\newcommand{\ifrm}[4]{\triple{#1}{#2}{#3} \vdash_{\mathrm{frmI}} #4}
\newcommand{\efrm}[4]{\triple{#1}{#2}{#3} \vdash_{\mathrm{frmE}} #4}

\newcommand{\dexec}[5]{\pair{#1}{#2},\, #3 \to \pair{#4}{#5}}

\newcommand{\dtrans}[4]{#1 \vdash #3,\, #2 \to #4}

\newcommand{\simheap}[4]{\pair{#3}{#4} \sdtstile{#1}{} #2}

\newcommand{\simenv}[3]{#3 \sdtstile{#1}{} #2}

\newcommand{\simstate}[5]{\triple{#3}{#4}{#5} \sdtstile{#1}{} #2}

\newcommand{\rtassert}[4]{\pair{#2}{#3} \vdash_{#1} #4}

\newcommand{\seval}[5]{#1 \vdash #2 \Downarrow #3 \dashv #4,\, #5}

\newcommand{\spceval}[4]{#1 \vdash #2 \downarrow #3 \dashv #4}

\newcommand{\sproduce}[3]{#1 \vdash #2 \lhd #3}

\newcommand{\sconsume}[6]{#1,\, #2 \vdash #3 \rhd #4,\, \allowbreak #5,\, \allowbreak #6}

\newcommand{\scons}[4]{#1 \vdash #2 \rhd #3, ~#4}

\newcommand{\sexec}[4]{#1 \vdash #2 \to #3 \dashv #4}

\newcommand{\sguard}[4]{#1 \rightharpoonup #2 \dashv #3,\, #4}

\newcommand{\strans}[3]{#1 \vdash #2 \to #3}

\newcommand{\vstate}{\Sigma}

\newcommand{\existential}[2]{\exists\, #1 : #2}
\newcommand{\nexistential}[2]{\nexists\, #1 : #2}
\newcommand{\universal}[2]{\forall\, #1 : #2}

\newcommand{\heap}{H}
\newcommand{\env}{\rho}
\newcommand{\perms}{\alpha}
\newcommand{\xperms}{\hat{\perms}}
\newcommand{\stack}{\mathcal{S}}
\newcommand{\tlist}{\multiple{t}}
\newcommand{\prog}{\Pi}
\newcommand{\dstate}{\Gamma}

\newcommand{\initsym}{\mathsf{init}}
\newcommand{\finalsym}{\mathsf{final}}
\newcommand{\nilsym}{\mathsf{nil}}

\newcommand{\ffresh}{\operatorname{fresh}}
\newcommand{\fpre}{\operatorname{pre}}
\newcommand{\fpost}{\operatorname{post}}
\newcommand{\fbody}{\operatorname{body}}
\newcommand{\fparams}{\operatorname{params}}
\newcommand{\fpred}{\operatorname{predicate}}
\newcommand{\fpredparams}{\operatorname{predicate\_params}}
\newcommand{\fstruct}{\operatorname{struct}}
\newcommand{\fdefault}{\operatorname{default}}
\newcommand{\frem}{\operatorname{rem}}
\newcommand{\fremf}{\operatorname{rem_f}}
\newcommand{\fremfp}{\operatorname{rem_{fp}}}
\newcommand{\falias}{\operatorname{alias}}
\newcommand{\fmodified}{\operatorname{modified}}
\newcommand{\fsep}{\operatorname{sep}}
\newcommand{\fsat}{\operatorname{sat}}
\newcommand{\powerset}[1]{\mathcal{P}(#1)}
\newcommand{\set}[1]{\{ #1 \}}
\newcommand{\iffdef}{\stackrel{\text{def}}{\iff}}
\newcommand{\pfunc}{\rightharpoonup}

\newcommand{\sstate}{\sigma}
\newcommand{\ssempty}{\sstate_{\mathrm{empty}}}
\newcommand{\sheap}{\mathsf{H}}
\newcommand{\oheap}{\mathcal{H}}
\newcommand{\scheck}{\mathcal{R}}
\newcommand{\senv}{\gamma}
\newcommand{\pc}{g}
\newcommand{\imp}{\iota} 
\newcommand{\sperms}{\Theta}
\newcommand{\sperm}{\theta}

\newcommand{\semanticsargs}[1][]{#1\end{mathpar}}
\newcommand{\semantics}[3][]{\begin{mathpar}\inferrule[#1]{#2}{#3}\semanticsargs}

\newcounter{defpartnum}[definition]
\newcounter{defsubpartnum}[defpartnum]
\newenvironment{defparts}
{
    \setcounter{defpartnum}{0}
    \setcounter{defsubpartnum}{0}
    \newcommand{\defpart}{\par\indent\refstepcounter{defpartnum}\textit{Part \thedefpartnum.}~}
    
}
{
    \par
}
\renewcommand*\thedefpartnum{\arabic{definition}.\arabic{defpartnum}}

\newcounter{casenum}
\newcounter{subcasenum}[casenum]
\newenvironment{enumcases}
{
    \setcounter{casenum}{0}
    \setcounter{subcasenum}{0}
    \newcommand{\case}{\par\indent\refstepcounter{casenum}\textbf{Case \thecasenum.}~}
    \let\SavedTheHcasenum=\theHcasenum
    \def\theHcasenum{\thelemma.\SavedTheHcasenum}
    \newcommand{\subcase}{\par\indent\refstepcounter{subcasenum}\textit{Case \thecasenum(\thesubcasenum).}~}
}
{
    \par
}
\renewcommand*\thecasenum{\arabic{casenum}}
\renewcommand*\thesubcasenum{\alph{subcasenum}}

\hypersetup{
    colorlinks=true,
    linkcolor = blue
}

\newcommand{\refrule}[1]{\hyperlink{#1}{\TirName {#1}}}

\definecolor{light-gray}{gray}{0.87}
\newcommand{\lightgray}[1]{\colorbox{light-gray}{#1}}

\definecolor{light-purple}{RGB}{229,204,255}
\newcommand{\lightpurple}[1]{\colorbox{light-purple}{#1}}

\definecolor{light-yellow}{RGB}{255,228,181}
\newcommand{\lightyellow}[1]{\colorbox{light-yellow}{#1}}

\definecolor{light-blue}{RGB}{189,215,238}
\newcommand{\lightblue}[1]{\colorbox{light-blue}{#1}}

\definecolor{light-green}{RGB}{152,251,152}

\definecolor{light-red}{RGB}{255,204,204}

\definecolor{rred}{RGB}{255,153,153}

\definecolor{light-pink}{RGB}{255,204,255}

\definecolor{light-orange}{RGB}{255,204,153}

\definecolor{neon-blue}{RGB}{153,255,255}

\definecolor{neon-yellow}{RGB}{255,255,153}

\setlist[itemize]{leftmargin=14pt}

\begin{document}

\addtocontents{toc}{\protect\setcounter{tocdepth}{0}}

\title{Sound Gradual Verification with Symbolic Execution}
\titlenote{This version contains supplementary material for the paper of the same name accepted for publication by Principles of Programming Languages (2024).}

\author{Conrad Zimmerman}
\email{conrad\_zimmerman@brown.edu}
\affiliation{%
  \institution{Brown University}
  \country{USA}
}

\author{Jenna DiVincenzo}
\affiliation{%
  \institution{Purdue University}%
  \country{USA}%
}
\email{jennad@purdue.edu}

\author{Jonathan Aldrich}
\affiliation{%
  \institution{Carnegie Mellon University}%
  \country{USA}%
}
 \email{jonathan.aldrich@cs.cmu.edu}

\renewcommand{\shortauthors}{C. Zimmerman, J. DiVincenzo, and J. Aldrich}

\begin{abstract}
Gradual verification, which supports explicitly partial specifications
and verifies them with a combination of static and dynamic checks, makes
verification more incremental and provides earlier feedback to developers.
While an abstract, weakest precondition-based approach to gradual verification
was previously proven sound, the approach did not provide sufficient guidance
for implementation and optimization of the required run-time checks.  More
recently, gradual verification was implemented using symbolic execution
techniques, but the soundness of the approach (as with related static
checkers based on implicit dynamic frames) was an open question.  This paper
puts practical gradual verification on a sound footing with a formalization
of symbolic execution, optimized run-time check generation, and run time
execution.  We prove our approach is sound; our proof also covers a core
subset of the Viper tool, for which we are aware of no previous soundness result.
Our formalization enabled us to find a soundness bug in an implemented gradual
verification tool and describe the fix necessary to make it sound.
\end{abstract}




\received{20 February 2007}
\received[revised]{12 March 2009}
\received[accepted]{5 June 2009}

\maketitle

\section{Introduction}\label{sec:introduction}

Static verification technology based on Hoare-logic-styled pre- and postconditions \cite{hoare1969axiomatic} has come a long way in the last few decades. Such tools can now support the modular verification of data structures that manipulate the heap \cite{reynolds2002separation,smans2012implicit} and are recursive \cite{parkinson2005separation}. However, verification is expensive, requiring many auxiliary specifications such as loop invariants and lemmas, and often costing an order of magnitude more human effort than development alone. In response, \citet{bader2018gradual} introduced the idea of gradual verification, which supports the incremental specification and verification of code by seamlessly combining static and dynamic verification. A developer can now write partial, \emph{imprecise} specifications---formulas such as $\simprecise{x.f \keq 2}$---backed by run-time checking. During static verification, imprecise specifications are strengthened in support of proof goals when it is necessary and non-contradictory to do so. Then, corresponding dynamic checks are inserted to ensure soundness. As a result, gradual verification allows users to specify and verify only the properties and components of their system that they care about, and incrementally increase the scope of verification as necessary.

Based on this early idea, \citet{wise2020gradual} and \citet{divincenzo2022gradual} extended gradual verification to support recursive heap data structures.
\citet{wise2020gradual} presented the first theory of gradual verification for \emph{implicit dynamic frames} (IDF) \citep{smans2012implicit}, a variant of \emph{separation logic} \citep{reynolds2002separation}, and \emph{abstract predicates} \citep{parkinson2005separation}. Their design, corresponding theory, and proofs rely heavily on the backward-reasoning technique called \emph{weakest liberal preconditions} (WLP), and on infinite sets that are not easy to approximate in finite form. Additionally, \citet{wise2020gradual}'s design checks all proof obligations at run-time, even when some obligations have been discharged statically. Therefore, it remained unclear how to implement gradual verification and whether gradually-verified programs could achieve good performance. Fortunately, in follow-up work, \citet{divincenzo2022gradual} implemented and empirically evaluated \gco, the first gradual verifier that can be used on real programs. \gco is based on symbolic execution, a forward-reasoning technique which is routinely used in static verifiers such as Viper \citep{viper16}, and optimizes run-time checks with statically available information to improve run-time performance. \citet{divincenzo2022gradual} showed that this improvement over prior work yields significant performance boosts.

Technically, \gco is built on top of Viper \citep{viper16}, which is a static verification infrastructure and tool that facilitates the development of program verifiers supporting IDF and recursive abstract predicates. Viper also uses symbolic execution at its core. Besides \gco, an array of widely-used verifiers have been built on top of Viper, including Prusti \citep{astrauskas2022prusti} for Rust, Nagini \citep{eilers2018nagini} for Python, and VerCors \citep{blom2017vercors} for Java. 
However, despite its prominence, Viper has not been proven sound; nor have, to our knowledge, other symbolic execution-based methods for verifying IDF logics. Thanks to the complexities of symbolic execution and Viper's support for practical but advanced verification features, \citet{schwerhoff16silicon}'s specification of Viper is full of implementation details that make it difficult to formally state and prove soundness. Since \gco is built on Viper, this problem carries over to \gco's specification in \citet{divincenzo2022gradual} and is made worse by the combination of static and dynamic checking. Thus \citet{divincenzo2022gradual} does not contain a proof of soundness for \gco. Furthermore, since \gco uses symbolic execution instead of WLP and optimizes run-time checks, \citet{wise2020gradual}'s proof is also not applicable. This is problematic, because the intricate interactions of static and dynamic checking in gradual verification can easily lead to subtle soundness bugs in gradual verifiers like \gco, as we will show in \S\ref{sec:unsoundness}.

Therefore, this paper presents a formal statement and proof of soundness for \gco and its underlying core subset of Viper. We formalize \gco's symbolic execution algorithm in sets of inference rules, rather than the CPS-style specification in \citet{divincenzo2022gradual} and \citet{schwerhoff16silicon}, to enable abstractions that improve the readability of the design and make it easier to state and prove soundness. The level of abstraction we use is far closer to the implementation of \gco than \citet{wise2020gradual}'s formal system, but slightly more abstract than \citet{divincenzo2022gradual}'s CPS-style specification, which is littered with implementation details. Reaching the right level of abstraction for our goals took some trial and error. We reflect on this process, including our missteps, in this paper as well. Our approach is inspired by the formal system for a basic type checker combined with symbolic execution in \citet{khoo10mix}. However, we separate the rules into several types of judgements to reflect the architecture of \citet{divincenzo2022gradual} and \citet{schwerhoff16silicon} and deal with the complexities of IDF and gradual verification. Given an initial symbolic state, the rules compute a next possible state (of which many may exist), and a set of run-time checks required for this transition when optimism is relied upon. That is, our rules are non-deterministic, but only in regards to the multiple execution paths explored by symbolic execution at program points like if statements, while loops, and logical conditionals.

Furthermore, we clearly separate the cases required to support imprecise specifications from those dealing with the underlying verification algorithm supporting only complete static specifications. Therefore, our formal system is a conservative extension of a core calculus of Viper; and so, by formalizing \gco and proving it sound, we have also formalized the core of Viper and proved it sound. To make it easier for readers of this paper to take advantage of our formal statement and proof of soundness for Viper for their own uses, we present first a core language, which we call \svl, along with verification rules modeling \gco's underlying static verification algorithm. We then define \gvl, which extends \svl to include gradual specifications and corresponds to the full language used by \gco. We also formally define static verification for \gvl, modeling the verification algorithm of \gco. We hope this separation provides a solid foundation for future proof endeavors of other static verifiers based on symbolic execution.


In order to fully define the behavior of \gvl and its subset \svl, we specify its dynamic semantics, which combines the semantics of \co \citep{arnold2010c0} with the dynamic semantics of \gvlrp, the language used to define the theory of gradual verification with recursive predicates in \citet{wise2020gradual}.
The \co programming language is a core, safe variant of the C language introduced for education \cite{arnold2010c0} and is also supported by \gco. \co allows specification of the pre- and post-conditions of methods, but does not include constructs necessary for static verification using IDF.
Thus we add the dynamic semantics from \citet{wise2020gradual} for IDF specifications, recursive predicates, and imprecise specifications. These semantics assert the validity of every specification at run-time, ensuring both memory safety and functional correctness of programs.
Thus these semantics provide a foundation against which we can establish the soundness of \gco's symbolic execution algorithm. That is, we prove that when all run-time checks produced by the symbolic execution algorithm are satisfied, then the program is guaranteed to dynamically execute successfully. A tricky part of this proof is defining a \emph{valuation function} \cite{khoo10mix}, which is a partial function mapping symbolic values from symbolic execution to their concrete values for a specific execution trace from program execution. This function is used to state the correspondence between symbolic and concrete execution states. While we start with \citet{khoo10mix}'s simplistic valuation function, we end up with one that is far more complex as it additionally connects \emph{isorecursive} symbolic predicates from static verification with their \emph{equirecursive} counterparts in dynamic verification and handles global invariants such as separation and access permissions from IDF.
This proof technique allows our formal system and reasoning to match the implementation more closely than other techniques such as the evidence calculus used in \citet{garcia2016}. This enables us to explore future developments using either the implementation or formalization, and easily update the other to ensure we remain both implementable and sound.

Finally, we present and discuss a soundness bug we found in \gco during our proof work and have since communicated to \citet{divincenzo2022gradual}. The bug is a specific interaction caused by reducing run-time checks using statically available information in isorecursive predicates, and then checking the remaining run-time checks using equirecursive predicates. This bug could not have arisen in \citet{wise2020gradual}'s work as their gradual verification approach checks all proof obligations at run time. We explore several options for addressing this soundness bug, explain our chosen method in detail, and discuss an implementation fix. Despite \citet{divincenzo2022gradual}'s thorough empirical evaluation and testing of \gco, this bug was never discovered in their testing. This is likely due to the subtle, intricate interactions between verification technologies in gradual verification that are hard to test. This demonstrates the value of formally proving soundness in the case of gradual verification, and we hope this paper serves as a basis for similar future work.

\noindent To summarize, this paper makes the following contributions:
\begin{itemize}[topsep=2pt]
    \item Formalization and proof of soundness for \gco, the first gradual verifier for recursive heap data structures that is based on symbolic execution \cite{divincenzo2022gradual}. The level of abstraction chosen for this proof work improves the readability of \gco's design and makes adapting this work to prove other symbolic execution-based gradual verifiers sound much easier.
    \item Formalization of a core subset of the Viper static verifier, which is based on symbolic execution and supports IDF. This work provides the first solid foundation for proof work on static verifiers that use symbolic execution and IDF.
    \item A reflection on the trial and error of picking the right level of abstraction for our proof work in this paper.
    \item Demonstration of a soundness bug we found in \gco during our work and have since communicated to \citet{divincenzo2022gradual}. We also provide several options for addressing this bug and advise on how to implement one of our solutions.
\end{itemize}

\section{\svl}\label{sec:precise}

We first introduce \svl and a corresponding static verification algorithm. Since it does not include imprecise specifications, \svl can be verified by existing static verification tools such as Viper \citep{viper16}.
The verification algorithm corresponds to the core algorithm of Viper, which is the foundation for static verification in \gco.
We illustrate how our formalism and soundness result can be applied to Viper. In later sections we extend \svl's verification algorithm to support the verification of gradually-specified \gvl programs.


\subsection{Definition}\label{sec:precise-defn}

We define an abstract syntax for \svl in Figure \ref{fig:precise-grammar}. Its form is similar to the language of Viper, which is intended for use as a generic backend for multiple frontend languages; however, we use the syntax of \co.



Programs consist of struct, predicate, and method\footnote{To distinguish them from pure functions (which are used in the specification language of similar verification tools) we use \emph{method} to refer to any potentially impure function.} definitions, and an entry statement. Struct definitions contain a list of fields, predicate definitions contain a parameter list and a formula (the \textit{predicate body}), and method definitions contain a parameter list, a return type, a pre-condition (denoted by $\krequires$), a post-condition (denoted by $\kensures$), and a statement (the \textit{method body}). The entry statement represents the body of the \ttt{main} method in traditional C programs. Statements in \svl follow C conventions, except for \ttt{while}, \ttt{alloc}, and \ttt{return}. All \ttt{while} statements specify a formula called a \textit{loop invariant}, which states the properties preserved by the loop during execution. An \ttt{alloc} statement allocates new memory on the heap, initializes it with a default value, and updates the variable on the left-hand side to contain a reference to the newly allocated value. This matches \co semantics, except \co returns a \textit{pointer}, not a \textit{reference}, and thus the type of the variable is written differently. We omit \ttt{return} statements; instead, the method body must assign to a special \ttt{result} variable, whose value is then returned after executing the method body. This reflects the behavior of \gviper which also does not have a \ttt{return} statement. Additionally, we simplify several statements to make formal definitions and proofs easier. For example, assignment only occurs to a variable or a field of a variable; statements such as \ttt{x.y.z = 1} are not permitted. We also omit \ttt{void} method calls, since these do not differ meaningfully from calls to value-returning methods.

Like \gco \citep{divincenzo2022gradual}, \svl does not support arrays. Verifying non-trivial properties of programs that use arrays would require significant extensions to existing gradual verification theory -- for example, quantified formulas. These extensions are left to future work. However, we can verify recursive data structures such as linked lists with abstract predicates. Note that Viper does support quantified formulas and arrays, thus further work is necessary to formally prove soundness of these capabilities.

We make several simplifying assumptions for \svl programs. All variables are initialized before they are used, and every execution path for a method body assigns the \ttt{result} variable at its end. Every program is well-typed; that is, expressions used in \ttt{if} conditions or as boolean operands will evaluate to \ttt{bool} values, all arguments passed to method parameters will match the defined parameter type, and the value assign to \ttt{result} has type equal to the method's return type. Finally, all specifications (predicate bodies, loop invariants, and method pre- and post-conditions) are \textit{self-framed}, which is a special well-formedness condition from IDF that we define later.

\begin{figure}[t]
  {\footnotesize\disableTttResize
  \begin{minipage}[t]{.35\linewidth}
    \begin{flalign*}
      x &\in \Var       &\textit{Variable names} \\
      f &\in \Field     &\textit{Field names}     \\
      p &\in \Predicate &\textit{Predicate names} \\
      m &\in \Method    &\textit{Method names}     \\
      S &\in \Struct    &\textit{Struct names} \\
      n &\in \mathbb{Z} &\textit{Integers}
    \end{flalign*}
    \begin{flalign*}
      \Pi         &::= \multiple{\mathcal{S}} ~ \multiple{\mathcal{P}} ~ \multiple{\mathcal{M}} ~ s &\\
      \mathcal{S} &::= \kstruct ~S~ \sblock{ \multiple{T~f} } &\\
      \mathcal{P} &::= p(\multiple{T~x}) = \phi &\\
      \mathcal{M} &::= T~m(\multiple{T~x}) ~ \Phi ~ \sblock{s} &
    \end{flalign*}
  \end{minipage}
  \hspace{.1\linewidth}
  \begin{minipage}[t]{.45\linewidth}
    \begin{flalign*}
      \Phi        ::=~& \krequires~\phi~\kensures~\phi \\
      T           ::=~& S \mid \kint \mid \kbool \mid \kchar \\
      s           ::=~& \sseq{s}{s} \mid \kskip \mid x \kassign e \mid x = \kalloc(S) \mid \\
                      & x = m(\multiple{e}) \mid \sassert{\phi} \mid \sfold{p(\multiple{e})} \mid \\
                      & \sunfold{p(\multiple{e})} \mid \sif{e}{s}{s} \mid \\
                      & \swhile{e}{\phi}{s} \\
      e           ::=~& l \mid x \mid e.f \mid e \oplus e \mid e \kor e \mid e \kand e \mid \kneg e \\
      l           ::=~& n \mid \knull \mid \ktrue \mid \kfalse \\
      \oplus      ::=~& \ttt{+} \mid \ttt{-} \mid \ttt{/} \mid \ttt{*} \mid \ttt{==} \mid \ttt{!=} \mid \ttt{<=} \mid \ttt{>=} \mid \ttt{<} \mid \ttt{>} \\
      \phi        ::=~& \phi * \phi \mid p(\multiple{e}) \mid e \mid \kacc(e.f) \mid \\
                      & \sif{e}{\phi}{\phi}
    \end{flalign*}
  \end{minipage}}
  \caption{Abstract syntax for \svl}
  \label{fig:precise-grammar}
\end{figure}

\subsubsection{Formulas}\label{sec:precise-formulas}

Formulas (specifications) in \svl are written in the logic of IDF \citep{smans2012implicit} and recursive predicates \citep{parkinson2005separation}. Thus formulas may contain expressions as well as abstract predicates and accessibility predicates from IDF; formulas may be joined by the separating conjunction $*$ \cite{smans2012implicit}. An accessibility predicate $\kacc(e.f)$ requires access to the heap location $e.f$. A predicate instance $p(\overline{e})$ applies the boolean predicate $p$ to the arguments $\overline{e}$. An expression $e$ requires that $e$ evaluates to $\ktrue$. A separating conjunction, as in $\phi_1 * \phi_2$, acts like a logical AND for $\phi_1$ and $\phi_2$, but also requires the heap locations specified by predicates and accessibility predicates in $\phi_1$ to be disjoint from those specified in $\phi_2$, e.g. $\kacc(x.f) * \kacc(y.f)$ implies $x \kneq y$. A conditional formula $\sif{e}{\phi_1}{\phi_2}$ denotes the validity of $\phi_1$ when $e$ evaluates to $\ktrue$; otherwise it denotes the validity of $\phi_2$.

Formulas in IDF, and thus in \svl, must be \textit{self-framed} \cite{smans2012implicit}, which requires permissions for all heap locations used in a formula to also be in that formula. For example, \ttt{x.value == 0} is not self-framed since it references the heap location \ttt{x.value}, but does not assert accessibility of the field \ttt{x.value}. However, \ttt{acc(x.value)} $*$ \ttt{x.value == 0} is self-framed. We specify rules for framing and self-framing in \S\ref{sec:dynamic-formulas}.

Static verification of predicates is done \textit{isorecursively} \cite{summers2013formal}, thus predicate instances must be explicitly folded before they can be asserted. Similarly, predicate bodies must be explicitly unfolded before asserting the implications of a predicate. This enables static verification of recursive predicates and simplifies reasoning about the verifier's behavior.

\subsection{Representation}\label{sec:precise-representation}

In this section, we formally define the data structures used during static verification of \svl programs.

\begin{itemize}
\item A \textit{symbolic value} $\nu \in \textsc{SValue}$ is an abstract value representing an unknown value, such as an integer or object reference. We leave the concrete type of $\textsc{SValue}$ undefined, but assume that an infinite number of distinct new values can be produced by a $\ffresh$ function. 
\item A \textit{symbolic expression} $t \in \SExpr$ is a symbolic or literal value, or is composed of other symbolic expressions and operators.
{\footnotesize
$$t \quad::=\quad \nu \quad\mid\quad l \quad\mid\quad \kneg t \quad\mid\quad t_1 \kand t_2 \quad\mid\quad t_1 \kor t_2 \quad\mid\quad t_1 \oplus t_2$$}

\item A \textit{path condition} $\pc \in \SExpr$ is a symbolic expression composed of conjuncts identifying a particular execution path. Conjuncts are added at every conditional branch during symbolic execution.

\item A \textit{field chunk} $\triple{f}{t}{t'} \in \SField$ represents, in the symbolic heap, the field $f$ of an object reference $t$ containing a value $t'$. A heap chunk is roughly approximate to the \textit{points to} construct in separation logic \cite{reynolds2002separation}. A \textit{predicate chunk} $\pair{p}{\multiple{t}} \in \SPredicate$ represents an isorecursive instance of a predicate $p$ with arguments $\multiple{t}$. Together, field chunks and predicate chunks are called \textit{heap chunks}.

\item A \textit{symbolic heap} $\sheap \in \powerset{\SField \cup \SPredicate}$ is a finite set of heap chunks. All heap chunks that it contains must represent distinct locations in the heap at run time.

\item A \textit{symbolic state} $\sstate \in \SState$ is a tuple containing a path condition (referenced by $\pc(\sstate)$), a symbolic heap (referenced by $\sheap(\sstate)$), and a symbolic environment (referenced by $\senv(\sstate)$). A symbolic state stores all values for a particular point during symbolic execution. The symbol $\ssempty$ represents an empty symbolic state, i.e. $\pc(\ssempty) = \ktrue$ and $\sheap(\ssempty) = \senv(\ssempty) = \emptyset$.

\item A \textit{verification state} $\vstate$ represents a particular point during static verification. It is either a special symbol or a triple $\triple{\sstate}{s}{\gform}$ consisting of a symbolic state $\sstate$, a statement $s$ that remains to be executed, and a formula $\gform$ that must be asserted after executing $s$. $\sstate(\vstate)$, $s(\vstate)$, and $\gform(\vstate)$ are used to reference a specific component of $\vstate$ when $\vstate$ is not a symbol.
{\footnotesize
$$\vstate \quad::=\quad \initsym \quad\mid\quad \finalsym \quad\mid\quad \triple{\sstate}{s}{\gform}$$}

\item A \textit{valuation} $V : \SValue \to \Value$ is a mapping from symbolic values to concrete values (defined in \S\ref{sec:dynamic-representation}). Valuations are implicitly extended to be defined for all $\SExpr$, following the structure of symbolic expressions.

A symbolic expression $t$ \textit{implies} the symbolic expression $t'$ (written $t \implies t'$) if, for all valuations $V$, $V(t) = \ktrue \implies V(t') = \ktrue$. For example, $t_1 \kand t_2 \implies t_2$. A symbolic expression $t$ is \textit{satisfiable}, denoted $\mathrm{sat}(t)$, if $V(t) = \ktrue$ for some valuation $V$.
\end{itemize}

\subsection{Evaluating expressions}\label{sec:precise-eval}

Symbolic execution evaluates an expression $e$ to a symbolic value $t$ using the symbolic state $\sstate$, and is denoted by the judgement $\sstate \vdash e \Downarrow t \dashv \sstate'$. 
It also yields a new symbolic state $\sstate'$ which may contain a more specific path condition if this particular evaluation short-circuits a boolean operator. 
Selected formal rules for symbolic evaluation are given in Figure \ref{fig:symbolic-evaluation-rules}. Literals are evaluated to themselves and variables are evaluated to the corresponding value in the symbolic store. Some operators, such as negation and arithmetic operators, are directly translated into a symbolic expression using the respective operator. In contrast, boolean operators are short-circuiting: when evaluating $e_1 \kand e_2$, if $e_1$ evaluates to $\kfalse$, then $e_2$ is never evaluated (in this case, $e_1 \keq \kfalse$ is added to the path condition). We define two non-deterministic rules for each binary boolean operator---\textsc{SEvalAndA} represents the short-circuiting case just described, while 
\textsc{SEvalAndB} represents the non-short-circuiting case where $e_1$ is true, so $e_2$ must also be evaluated to determine the result. 
Finally, field references are evaluated to the symbolic value contained in their corresponding field chunk in the symbolic heap. Note, a heap chunk for the field reference must be in the heap, otherwise evaluation fails (and ultimately static verification as well), thus the field reference must be framed by the current state. 

We also define a judgment of the form $\sstate \vdash e \downarrow t$ which symbolically evaluates an expression $e$ to a symbolic expression $t$ without short-circuiting. Thus the judgment is deterministic and does not update the path condition. Instead, logical operators such as $\kand$ are encoded directly in the symbolic expression (compare \textsc{SEvalPCAnd} with \textsc{SEvalAndA}/\textsc{SEvalAndB} in Figure \ref{fig:symbolic-evaluation-rules}). This results in a less specific path condition, but reduces the number of execution paths during symbolic execution. This matches the evaluation method described in \citet{divincenzo2022gradual} for evaluation in formulas, while the former style is used for evaluation in imperative code. 

\begin{figure}
  {\footnotesize\disableTttResize
  \begin{mathpar}
    \inferrule[SEvalLiteral]
      { }
      {\sstate \vdash l \Downarrow l \dashv \sstate} \and
    \inferrule[SEvalVar]
      { }
      {\sstate \vdash x \Downarrow \senv(\sstate)(x) \dashv \sstate } \and
    \inferrule[SEvalNeg]
      { \sstate \vdash e \Downarrow t \dashv \sstate' }
      { \sstate \vdash \kneg e \Downarrow \kneg t \dashv \sstate' } \and
    \inferrule[SEvalAndA]
      { \sstate \vdash e_1 \Downarrow t_1 \dashv \sstate' }
      { \sstate \vdash e_1 \kand e_2 \Downarrow t_1 \dashv \sstate'[\pc = \pc(\sstate') \kand \kneg t_1] } \and
    \inferrule[SEvalAndB]
      { \sstate \vdash e_1 \Downarrow t_1 \dashv \sstate' \\\\
        \sstate'[\pc = \pc(\sstate') \kand \kneg t_1] \vdash e_2 \Downarrow t_2 \dashv \sstate'' }
      { \sstate \vdash e_1 \kand e_2 \Downarrow t_2 \dashv \sstate'' } \and
    \inferrule[SEvalField]
      { \sstate \vdash e \Downarrow t_e \dashv \sstate' \\
        \pc(\sstate') \implies t_e \keq t_e' \\\\
        \triple{t_e'}{f}{t} \in \sheap(\sstate') }
      { \sstate \vdash e.f \Downarrow t \dashv \sstate' } \and
    \inferrule[SEvalPCAnd]
      { \sstate \vdash e_1 \downarrow t_1 \\
        \sstate \vdash e_2 \downarrow t_2 }
      { \sstate \vdash e_1 \kand e_2 \downarrow t_1 \kand t_2 \dashv \sstate'' }
  \end{mathpar}}
  \caption{Selected symbolic evaluation rules}
  \label{fig:symbolic-evaluation-rules}
\end{figure}

\subsection{Consuming formulas}\label{sec:precise-consume}

Given a symbolic state $\sstate$ and formula $\phi$, \textit{consuming} a formula $\phi$ first asserts that $\phi$ is established by $\sstate$, and second removes the heap chunks in $\sstate$ corresponding to permissions (predicates and accessibility predicates) in $\phi$. The judgment $\sstate \vdash \phi \rhd \sstate'$ denotes consumption; i.e., $\phi$ is consumed from $\sstate$, resulting in the new symbolic state $\sstate'$. See Figure \ref{fig:symbolic-consume-rules} for selected rules.


Consuming an accessibility predicate such as $\kacc(e.f)$ first asserts the predicate has a corresponding field chunk in the heap, and second removes the chunk from the heap (\textsc{SConsumeAcc}). 
Consuming a predicate similarly looks for and removes the corresponding predicate chunk from the heap (\textsc{SConsumePredicate}). If any of the chunks are missing from the heap, then verification fails. 
Expressions must evaluate to $\ktrue$ in the current symbolic execution path. That is, the current path condition must imply the symbolic value of the expression (\textsc{SConsumeValue}). As mentioned previously and seen in the aforementioned rule, expressions in formulas are evaluated with the deterministic evaluation judgment (i.e., not the short-circuiting one), which matches the behavior described in \citet{divincenzo2022gradual} and reduces the number of branches generated during symbolic execution. This differs from Viper, which uses a single, short-circuiting \ttt{eval} algorithm everywhere, including in \ttt{consume}. A separating conjunction, such as $\phi_1 * \phi_2$, is consumed left-to-right, i.e. $\phi_1$ is consumed and then $\phi_2$ is consumed (\textsc{SConsumeConjunction}). This enforces the separation of permissions between the two conjuncts -- heap chunks necessary to satisfy the permissions asserted in $\phi_1$ will be removed before consuming $\phi_2$, so, if they overlap, consumption of $\phi_2$ will fail. 
Finally, we define consumption of logical conditionals, like $\sif{e}{\phi_1}{\phi_2}$, in two non-deterministic rules. In \textsc{SConsumeConditionalA}, $e$ is assumed to be true in the path condition and $\phi_1$ is consumed. Likewise, in \textsc{SConsumeConditionalB}, $e$ is assumed to be false in the path condition and $\phi_2$ is consumed. 

Note, as we saw in \S\ref{sec:precise-eval}, evaluation of a field access in an expression requires the state to contain a heap chunk for the field. But consume removes heap chunks from the state in a left-to-right manner thanks to rules \textsc{SConsumeAcc} and \textsc{SConsumeConjunction}.
For example, we may want to consume the formula $\kacc(e.f) * e.f \keq 0$. First, a heap chunk for $\kacc(e.f)$ is found and removed from the heap. Then, the resulting state is used to frame and evaluate $e.f \keq 0$ in the next consume step. However, the heap chunk for $e.f$ was removed from the state so evaluation fails when it shouldn't since the original state contained the heap chunk. 
To solve this issue, we define consume using an underlying judgment, denoted $\sstate, \sstate_E \vdash \phi \rhd \sstate'$, which asserts and removes permissions from $\sstate$ while evaluating expressions with the unchanging reference state $\sstate_E$. The state $\sstate_E$ is the symbolic state before consumption. The rule \textsc{SConsume} defines the top-level consume judgment using this new underlying judgment.


Our consume judgment represents the core functionality of \citet{divincenzo2022gradual} and \citet{schwerhoff16silicon}'s \ttt{consume} algorithms. 
We, of course, ignore unnecessary implementation details like snapshots, which preserve certain portions of the state that are removed during consume. 

\begin{figure}[t]
  {\footnotesize\disableTttResize
  \begin{mathpar}
    \inferrule[SConsume]
      { \sstate, \sstate \vdash \phi \rhd \sstate' }
      { \sstate \vdash \phi \rhd \sstate' } \and
    \inferrule[SConsumeValue]
      { \sstate_E \vdash e \downarrow t \\\\
        \pc(\sstate) \implies t }
      { \sstate, \sstate_E \vdash e \rhd \sstate } \and
    \inferrule[SConsumeAcc]
      { \sstate_E \vdash e \downarrow t_e \\\\
        \pc(\sstate) \implies t_e \keq t_e' \\\\
        \sheap(\sstate) = \set{ \triple{f}{t_e'}{t} }\uplus \sheap' }
      { \sstate, \sstate_E \vdash \kacc(e.f) \rhd \sstate[\sheap = \sheap'] } \and
    \inferrule[SConsumePredicate]
      { \multiple{\sstate_E \vdash e \downarrow t} \\
        \multiple{\pc(\sstate) \implies t \keq t'} \\\\
        \sheap(\sstate) = \set{\pair{p}{\multiple{t'}}}\uplus \sheap' }
      { \sstate, \sstate_E \vdash \sstate[\sheap = \sheap'] } \and
    \inferrule[SConsumeConjunction]
      { \sstate, \sstate_E \vdash \phi_1 \rhd \sstate' \\\\
        \sstate', \sstate_E[\pc = \pc(\sstate')] \vdash \phi_2 \rhd \sstate'' }
      { \sstate, \sstate_E \vdash \phi_1 * \phi_2 \rhd \sstate'' } \and
    \inferrule[SConsumeConditionalA]
      { \sstate_E \vdash e \downarrow t \\
        \pc' = \pc(\sstate) \kand t \\\\
        \sstate[\pc = \pc'], \sstate_E[\pc = \pc'] \vdash \phi_1 \rhd \sstate' }
      { \sstate, \sstate_E \vdash \sif{e}{\phi_1}{\phi_2} \rhd \sstate' } \and
    \inferrule[SConsumeConditionalB]
      { \sstate_E \vdash e \downarrow t \\
        \pc' = \pc(\sstate) \kand \neg t \\\\
        \sstate[\pc = \pc'], \sstate_E[\pc = \pc'] \vdash \phi_2 \rhd \sstate' }
      { \sstate, \sstate_E \vdash \sif{e}{\phi_1}{\phi_2} \rhd \sstate' }
  \end{mathpar}}
  \caption{Selected consume rules}
  \label{fig:symbolic-consume-rules}
\end{figure}

\subsection{Producing formulas}\label{sec:precise-produce}

Given an initial state $\sstate$ and formula $\phi$, \emph{producing} $\phi$ adds the information in $\phi$ into the symbolic state $\sstate$, resulting in a new state $\sstate'$. The judgment for $\sstate \vdash \phi \lhd \sstate'$ denotes production; i.e., $\phi$ is produced into the state $\sstate$, resulting in $\sstate'$. In particular, produce adds heap chunks representing predicates in $\phi$ to the symbolic heap and symbolic expressions representing constraints from boolean expressions in $\phi$ to the path condition in a left-to-right manner. Note, each symbolic heap chunk represents a distinct region of memory at run-time, an invariant that we later prove. Thus overlapping heap chunks may only occur in symbolic states which represent an unreachable dynamic state and can safely be ignored.
%
%
When producing formulas, we use deterministic symbolic evaluation for expressions, but we introduce separate execution paths for conditionals (similar to \S\ref{sec:precise-consume}).

Formal rules are given in \supplement{the supplement \citep{supplement}}{\S\ref{sec:produce-rules}}.
These rules capture the functionality of the \ttt{produce} algorithm specified in \citet{divincenzo2022gradual} and \citet{schwerhoff16silicon}. As noted in the previous section, \citet{schwerhoff16silicon} uses short-circuiting evaluation in all places, while we use deterministic evaluation.


\subsection{Executing statements}\label{sec:precise-exec}

Now that we have formally defined symbolic execution of expressions and formulas, we can put the pieces together to define symbolic execution of program statements.

We represent the symbolic execution of program statements as small-step execution rules denoted by the judgment $\sstate \vdash s \to s' \dashv \sstate'$, where the initial statement $s$ is symbolically executed with the initial state $\sstate$, resulting in the state $\sstate'$, and then transitions to the next statement $s'$ with the new state $\sstate'$. Selected formal rules are shown in Figure \ref{fig:symbolic-exec-rules}. Executing a variable assignment updates the symbolic store (\textsc{SExecAssign}); while executing a field assignment first consumes $\kacc(x.f)$, 
and then adds a new heap chunk for $x.f$ to the heap that contains $x.f$'s new symbolic value after the write (\textsc{SExecAssignField}).
An $\kalloc(S)$ statement adds a heap chunk for each field in $S$ to the symbolic heap. The new object reference is a fresh value but the new field chunks are each initialized with default values, which reflects the behavior of \co (\textsc{SExecAlloc}). 
Execution rules for $\kif$ statements are non-deterministic: given a statement $\sif{e}{s_1}{s_2}$, \textsc{SExecIfA} adds $e$ to the path condition and continues execution with $s_1$, while \textsc{SExecIfB} adds $\kneg e$ to the path condition and continues execution with $s_2$.

Symbolic execution of method calls is modular; i.e., the behavior of the method call is represented by the method's pre- and post-conditions (\textsc{SExecCall}). 
First, the method's arguments are evaluated to symbolic values. Then the pre-condition is consumed using a special environment containing the argument values. A $\ffresh$ symbolic value is added to represent the return value of the method, and then the post-condition of the method is produced. 
The special environment is then replaced by the original environment, with the addition of the result's symbolic value. Loops (i.e $\kwhile$ statements) are executed similarly: the loop invariant is consumed, variables modified by the loop body are set to fresh values in the symbolic store, the loop invariant is produced, and the negated loop condition is added to the path condition (\textsc{SExecWhile}). 
Execution of the $\kfold$ and $\kunfold$ statements is also similar to loops and method calls: 
$\kfold$ consumes the predicate body and adds a representative predicate chunk to the symbolic heap, while $\kunfold$ consumes the predicate instance (thus removing the predicate chunk from the heap) and produces the predicate body.

\begin{figure}
  {\footnotesize\disableTttResize
  \begin{mathpar}
    \inferrule[SExecAssign]
      { \sstate \vdash e \Downarrow t \dashv \sstate' \\
        \senv' = \senv(\sstate)[x \mapsto t] }
      { \sstate \vdash \sseq{x \kassign e}{s} \to s \dashv \sstate'[\senv = \senv'] } \and
    \inferrule[SExecAssignField]
      { \sstate \vdash e \Downarrow t \dashv \sstate' \\
        \sstate' \vdash \kacc(x.f) \rhd \sstate'' \\\\
        \sheap' = \sheap(\sstate'') \cup \set{ \triple{f}{\senv(\sstate'')(x)}{t} } }
      { \sstate \vdash \sseq{x.f \kassign e}{s} \to s \dashv \sstate''[\sheap = \sheap']} \and
    \inferrule[SExecAlloc]
      { t = \ffresh \\
        \multiple{T~f} = \fstruct(S) \\\\
        \sheap' = \sheap(\sstate) \cup \set{ \multiple{\triple{f}{t}{\fdefault(T)}} } }
      { \sstate \vdash \sseq{x \kassign \kalloc(S)}{s} \to s \dashv \sstate[\sheap = \sheap'] } \and
    \inferrule[SExecCall]
      { \multiple{\sstate \vdash e \Downarrow t \dashv \sstate'} \\
        \multiple{x} = \fparams(m) \\\\
        \sstate'[\senv = [\multiple{x \mapsto t}]] \vdash \fpre(m) \rhd \sstate'' \\\\
        t' = \ffresh \\
        \senv' = \senv(\sstate')[y \mapsto t'] \\\\
        \sstate''[\senv = [\multiple{x \mapsto t}, \kresult \mapsto t']] \vdash \fpost(m) \lhd \sstate''' }
      { \sstate \vdash \sseq{y \kassign m(\multiple{e})}{s} \to s \dashv \sstate'''[\senv = \senv'] } \and
    \inferrule[SExecIfA]
      { \sstate \vdash e \Downarrow t \dashv \sstate' \\
        \sstate'' = \sstate'[\pc = \pc(\sstate') \kand t] }
      { \sstate \vdash \sseq{\sif{e}{s_1}{s_2}}{s} \to \sseq{s_1}{s} \dashv \sstate'' } \and
    \inferrule[SExecIfB]
      { \sstate \vdash e \Downarrow t \dashv \sstate' \\
      \sstate'' = \sstate'[\pc = \pc(\sstate') \kand \kneg t] }
      { \sstate \vdash \sseq{\sif{e}{s_1}{s_2}}{s} \to \sseq{s_2}{s} \dashv \sstate'' } \and
    \inferrule[SExecWhile]
      { \sstate \vdash \phi \rhd \sstate' \\
        \multiple{x} = \fmodified(s) \\
        \sstate'[\senv = \senv(\sstate')[\multiple{x \mapsto \ffresh}]] \vdash \phi \lhd \sstate'' \\
        \sstate'' \vdash e \downarrow t }
      { \sstate \vdash \sseq{\swhile{e}{\phi}{s}}{s'} \to s' \dashv \sstate''[\pc = \pc(\sstate'') \kand \kneg t] }
  \end{mathpar}}
  \caption{Selected symbolic execution rules}
  \label{fig:symbolic-exec-rules}
\end{figure}

\subsection{Modularly verifying programs}
We now define verification of entire programs.
We start by defining what a program $\prog$ is; it is a quadruple $\quadruple{s}{M}{P}{S}$ where $s$ is the entry statement of the program, $M$ is the set of method names, $P$ is the set of predicate names, and $S$ is the set of struct names in the program. 
Then, we define the judgment $\strans{\prog}{\vstate}{\vstate'}$ that specifies all possible symbolic execution steps that occur during verification of $\prog$. Selected rules are given in Figure \ref{fig:verify-rules}.

A verification state $\vstate$ is \textit{reachable} from program $\prog$ if $\vstate = \initsym$ or $\strans{\prog}{\vstate_0}{\vstate}$ for some reachable $\vstate_0$.  The latter judgement only holds when $\vstate_0$ is itself reachable.

This judgement includes rules for modular verification. From $\initsym$, we can begin verification of the entry statement (\textsc{SVerifyInit}) or of any method (\textsc{SVerifyMethod}). When verifying a method, the method's post-condition is used as the formula of the verification state. After completely executing the method's body, i.e. having reached $\kskip$, we consume the formula contained in the verification state (\textsc{SVerifyFinal}), which is the method's post-condition.

We modularly verify loop bodies following a similar pattern. As described in \S\ref{sec:precise-exec}, symbolic execution steps over loop bodies in the same way it steps over method calls. However, we introduce a verification rule (\textsc{SVerifyLoopBody}) that allows symbolic execution of a loop body, beginning with a new symbolic state. We reuse the symbolic store from the initial symbolic state, except that all variables modified by the loop body are replaced by fresh values. Verification proceeds similar to method verification, except that we use the loop invariant for the formula of the new verification state---we produce the loop invariant, symbolically execute the loop body, and finally consume the loop invariant. Thus symbolic execution, which steps over the loop, ensures that the loop invariant holds for the initial iteration, while this verification rule ensures that the loop invariant is preserved after every iteration.

We also include another verification rule for loops, \textsc{SVerifyLoop}, in order to match the behavior of \gco. This rule and its correspondence with \gco is described further in \S\ref{sec:loops}.

Statements are executed by symbolic execution as described in \S\ref{sec:precise-exec}. Given a reachable verification state $\triple{\sstate}{s}{\phi}$ and the symbolic execution $\sstate \vdash s \to s' \dashv \sstate'$, the state $\triple{\sstate'}{s'}{\phi}$ is reachable, i.e. $\prog \vdash \triple{\sstate}{s}{\phi} \to \triple{\sstate'}{s'}{\phi}$.

\begin{figure}
  {\footnotesize\disableTttResize
  \begin{mathpar}
    \inferrule[SVerifyInit]
      { }
      { \strans{\quadruple{s}{M}{P}{S}}{\initsym}{\triple{\ssempty}{s}{\ktrue}} } \and
    \inferrule[SVerifyMethod]
      { m \in M \\
        \multiple{x} = \fparams(m) \\\\
        \sproduce{\ssempty[\senv = [\multiple{x \mapsto \ffresh}]]}{\fpre(m)}{\sstate} }
      { \strans{\quadruple{s}{M}{P}{S}}{\initsym}{\triple{\sstate}{\sseq{\fbody(m)}{s}}{\fpost(m)}} } \and
    \inferrule[SVerifyLoopBody]
      { \strans{\prog}{\_}{\triple{\sstate_0}{\sseq{\swhile{e}{\gform}{s}}{s'}}{\gform_0}} \\\\
      \sproduce{\quintuple{\bot}{\senv(\sstate_0)[\multiple{x \mapsto \ffresh}]}{\emptyset}{\emptyset}{\pc(\sstate_0)}}{\gform}{\sstate} \\\\
      \multiple{x} = \fmodified(s) \\
      \spceval{\sstate}{e}{t}{\scheck} }
      {
        \prog \vdash \triple{\sstate_0}{\sseq{\swhile{e}{\gform}{s}}{s'}}{\gform_0} \to \\\\
        \triple{\sstate[\pc = \pc(\sstate) \kand t]}{\sseq{s}{\kskip}}{\gform}
      } \and
    \inferrule[SVerifyLoop]
      {
        \strans{\prog}{\_}{\triple{\sstate_0}{\sseq{\swhile{e}{\gform}{s}}{s'}}{\gform_0}} \\\\
        \scons{\sstate_0}{\gform}{\sstate_0'}{\_} \\
        \sproduce{\sstate_0'[\senv = \senv(\sstate_0)[\multiple{x \mapsto \ffresh}]]}{\gform}{\sstate_0''} \\\\
        \multiple{x} = \fmodified(s) \\
      }
      {\prog \vdash \triple{\sstate_0}{\sseq{\swhile{e}{\gform}{s}}{s'}}{\gform_0} \to \\\\
      \triple{\sstate_0}{\sseq{\swhile{e}{\gform}{s}}{s'}}{\gform_0} } \and
    \inferrule[SVerifyStep]
      { \strans{\prog}{\_}{\triple{\sstate}{s}{\phi}} \\
        \sexec{\sstate}{s}{s'}{\sstate'} }
      { \strans{\prog}{\triple{\sstate}{s}{\phi}}{\triple{\sstate'}{s'}{\phi}} } \and
    \inferrule[SVerifyFinal]
      { \strans{\prog}{\_}{\triple{\sstate}{\kskip}{\phi}} \\
        \scons{\sstate}{\phi}{\sstate'} }
      { \strans{\prog}{\triple{\sstate}{\kskip}{\phi}}{\finalsym} }
  \end{mathpar}}
  \caption{Selected verification rules}
  \label{fig:verify-rules}
\end{figure}

\subsection{Example}\label{sec:static-example}

\begin{figure}
  \begin{minipage}[t]{.45\linewidth}
\begin{lstlisting}[name=list-example]
struct List { int value; List next }

predicate acyclic(List l) =
  acc(l.value) * acc(l.next) *
  (if l.next == NULL then true
   else acyclic(l.next))

List singleton(int value)
  requires true
  ensures (acyclic(result) *
           result != NULL)
{ (*@ $\cdots$ @*) }
\end{lstlisting}
  \end{minipage}
  \begin{minipage}[t]{.45\linewidth}
\begin{lstlisting}[name=list-example]
List append(List l, int value)
  requires acyclic(l) * l != NULL(*@\label{ln:precise-append-pre}@*)
  ensures acyclic(result) * result != NULL(*@\label{ln:precise-append-post}@*)
{(*@\label{ln:precise-append-begin}@*)
  unfold acyclic(l);(*@\label{ln:precise-append-unfold}@*)
  if (l.next == NULL)(*@\label{ln:precise-append-if}@*)
    n = singleton(value);(*@\label{ln:precise-append-single}@*)
  else(*@\label{ln:precise-append-else}@*)
    n = append(l.next, value);(*@\label{ln:precise-append-rec}@*)
  l.next = n;(*@\label{ln:precise-append-next}@*)
  fold acyclic(l);(*@\label{ln:precise-append-fold}@*)
  result = l;(*@\label{ln:precise-append-result}@*)
}(*@\label{ln:precise-append-end}@*)
\end{lstlisting}
  \end{minipage}
  \caption{Code and supporting declarations for appending to an acyclic linked list}
  \label{fig:list-example}
\end{figure}

{\parindent0pt
We now illustrate verification of the \ttt{append} method defined in Figure \ref{fig:list-example}, which appends a given value to the end of a list using recursion. The \ttt{append} method is ensured to be memory safe and preserve acyclicity of the list through verification. We begin with an empty state and initialize all parameters with fresh values:
{\disableTttResize
\begin{lstlisting}[aboveskip=.25em, belowskip=0.5em,numbers=none]
(*@\lightgray{$\sstate_1 = \triple{\emptyset}{\senv}{\ktrue} \quad \senv = [\ttt{l} \mapsto \nu_1, \ttt{v} \mapsto \nu_2]$} @*)
\end{lstlisting}}
Then the pre-condition \ttt{acyclic(l) * l != NULL} is produced:
{\disableTttResize
\begin{lstlisting}[aboveskip=.25em, belowskip=0.5em,numbers=none]
(*@\lightgray{$\sstate_2 = \triple{\set{ \pair{\ttt{acyclic}}{\nu_1} }}{\senv}{\nu_1 \kneq \knull}$} @*)
\end{lstlisting}}
Unfolding \ttt{acyclic(l)} (line \ref{ln:precise-append-unfold}) consumes the predicate from the state and produces its body. The body of \ttt{acyclic(l)} contains a logical conditional resulting in two possible execution paths for produce -- one where $\nu_4$ is null and one where $\nu_4$ is not null, where $\nu_4$ is the symbolic value for \ttt{l.next}:
{\disableTttResize
\begin{lstlisting}[aboveskip=.25em, belowskip=0.5em, firstnumber=17,stepnumber=17]
unfold acyclic(l);
(*@\lightpurple{$\sstate_{A3} = \triple{\set{ \triple{\ttt{value}}{\nu_1}{\nu_3}, \triple{\ttt{next}}{\nu_1}{\nu_4} }}{\senv}{\nu_1 \kneq \knull \kand \nu_4 \keq \knull}$} @*)
(*@\lightblue{$\sstate_{B3} = \triple{\set{ \triple{\ttt{value}}{\nu_1}{\nu_3}, \triple{\ttt{next}}{\nu_1}{\nu_4} , \pair{\ttt{acyclic}}{\nu_4}} }{\senv}{\nu_1 \kneq \knull \kand \nu_4 \kneq \knull} $} @*)
\end{lstlisting}}
We follow both execution paths, using color-coding to distinguish them. Next, when executing the \ttt{if} statement (line \ref{ln:precise-append-if}), we first evaluate the condition. Since \ttt{l.next} is framed by the state, evaluation of the condition succeeds and execution branches along the \ttt{if}. We first consider executing the \ttt{then} branch of the \ttt{if}, where $\nu_4 \keq \knull$ is added it to the path condition:
{\disableTttResize
\begin{lstlisting}[aboveskip=.25em, belowskip=0.5em, firstnumber=18,stepnumber=18]
if (l.next == NULL)
    (*@\lightpurple{$\sstate_{A4} = \triple{\cdots}{\cdots}{\nu_1 \kneq \knull \kand \nu_4 \keq \knull \kand \nu_4 \keq \knull}$} @*)
    (*@\lightblue{$\sstate_{B4} = \triple{\cdots}{\cdots}{\nu_1 \kneq \knull \kand \nu_4 \kneq \knull \kand \nu_4 \keq \knull}$} @*)
\end{lstlisting}}
However, the path condition $\nu_1 \kneq \knull \kand \nu_4 \kneq \knull \kand \nu_4 \keq \knull$ is unsatisfiable, thus we can safely prune this execution path and only continue with the first. We proceed to symbolically execute the call to \ttt{singleton} (line \ref{ln:precise-append-single}) by consuming the (empty) pre-condition, and producing the post-condition. The result is represented by a fresh symbolic value $\nu_5$:
{\disableTttResize
\begin{lstlisting}[aboveskip=.25em, belowskip=0.5em, firstnumber=19,stepnumber=19]
    n = singleton(value);
    (*@\lightpurple{$\sstate_{A5} = \triple{\set{ \triple{\ttt{value}}{\nu_1}{\nu_3}, \triple{\ttt{next}}{\nu_1}{\nu_4}, \pair{\ttt{acyclic}}{\nu_5}} }{\senv[\ttt{n} \mapsto \nu_5]}{\nu_1 \kneq \knull \kand \nu_4 \keq \knull \kand \nu_5 \kneq \knull}$} @*)
\end{lstlisting}}
Symbolic execution of this path then jumps to line \ref{ln:precise-append-next}, but to preserve code order we now demonstrate verification of the \ttt{else} branch (line \ref{ln:precise-append-else}). To do this, we use states $\sstate_{A3}$ and $\sstate_{B3}$, and add the negation of the condition to verify the \ttt{else} body:
{\disableTttResize
\begin{lstlisting}[aboveskip=.25em, belowskip=0.5em,firstnumber=20,stepnumber=20]
else
    (*@\lightpurple{$\sstate_{A4}' = \triple{\cdots}{\cdots}{\nu_1 \kneq \knull \kand \nu_4 \keq \knull \kand \nu_4 \kneq \knull}$} @*)
    (*@\lightblue{$\sstate_{B4}' = \triple{\set{ \triple{\ttt{value}}{\nu_1}{\nu_3}, \triple{\ttt{next}}{\nu_1}{\nu_4}, \triple{\ttt{acyclic}}{\nu_4} }}{\senv}{\nu_1 \kneq \knull \kand \nu_4 \kneq \knull}$} @*)
\end{lstlisting}}
Here again this results in an unsatisfiable path condition $\nu_1 \kneq \knull \kand \nu_4 \keq \knull \kand \nu_4 \kneq \knull$, so we prune that path. We continue with the other path and execute the recursive call to \ttt{append} (line \ref{ln:precise-append-rec}), which consumes the pre-condition (removing $\pair{\ttt{acyclic}}{\nu_4}$) and produces the post-condition, using the fresh value $\nu_6$ to represent the result (adding $\pair{\ttt{acyclic}}{\nu_6}$):
{\disableTttResize
\begin{lstlisting}[aboveskip=.25em, belowskip=0.5em, firstnumber=21,stepnumber=21]
    n = append(l.next, value);
    (*@\lightblue{$\sstate_{B5}' = \triple{\set{ \triple{\ttt{value}}{\nu_1}{\nu_3}, \triple{\ttt{next}}{\nu_1}{\nu_4}, \pair{\ttt{acyclic}}{\nu_6}} }{\senv[\ttt{n} \mapsto \nu_6]}{\nu_1 \kneq \knull \kand \nu_4 \kneq \knull \kand \nu_6 \kneq \knull}$} @*)
\end{lstlisting}}
Now we have completed verifying both branches of the \ttt{if} statement. Note that we do not actually join execution at this point; instead, we jump to line \ref{ln:precise-append-next} immediately after executing the program up to $\sstate_{A5}$ and $\sstate_{B5}'$ along both paths. We follow both of these paths for the rest of verification. The field assignment on line 22 consumes \ttt{acc(l.next)} and produces a new corresponding heap chunk with \ttt{n}'s value:
{\disableTttResize
\begin{lstlisting}[aboveskip=.25em, belowskip=0.5em, firstnumber=22,stepnumber=22]
l.next = n;
(*@\lightpurple{$\sstate_{A6} = \triple{\set{ \triple{\ttt{value}}{\nu_1}{\nu_3}, \triple{\ttt{next}}{\nu_1}{\nu_5}, \pair{\ttt{acyclic}}{\nu_5}} }{\senv[\ttt{n} \mapsto \nu_5]}{\nu_1 \kneq \knull \kand \nu_4 \keq \knull \kand \nu_5 \kneq \knull}$} @*)
(*@\lightblue{$\sstate_{B6}' = \triple{\set{ \triple{\ttt{value}}{\nu_1}{\nu_3}, \triple{\ttt{next}}{\nu_1}{\nu_6}, \pair{\ttt{acyclic}}{\nu_6}} }{\senv[\ttt{n} \mapsto \nu_6]}{\nu_1 \kneq \knull \kand \nu_4 \kneq \knull \kand \nu_6 \kneq \knull}$} @*)
\end{lstlisting}}
Folding \ttt{acyclic} results in twice the number of execution paths since it consumes \ttt{acylic(l)}'s body, which includes an logical conditional. 
However, again, information from the path conditions in $\sstate_{A6}$ and $\sstate_{B6}'$ allow us to prune some of these paths. 
We elide these pruned paths and only show the taken ones. After consuming \ttt{acyclic(l)}'s body, execution produces \ttt{acyclic(l)} into the state:
{\disableTttResize
\begin{lstlisting}[aboveskip=.25em, belowskip=0pt, firstnumber=23,stepnumber=23]
fold acyclic(l);
(*@\lightpurple{$\sstate_{A7} = \triple{\set{ \pair{\ttt{acyclic}}{\nu_1}} }{\senv[\ttt{n} \mapsto \nu_5]}{\nu_1 \kneq \knull \kand \nu_4 \keq \knull \kand \nu_5 \kneq \knull}$} @*)
(*@\lightblue{$\sstate_{B7}' = \triple{\set{ \pair{\ttt{acyclic}}{\nu_1}} }{\senv[\ttt{n} \mapsto \nu_6]}{\nu_1 \kneq \knull \kand \nu_4 \kneq \knull \kand \nu_6 \kneq \knull}$} @*)
\end{lstlisting}
\begin{lstlisting}[aboveskip=0pt, belowskip=0.5em, firstnumber=24,stepnumber=24]
result = l;
(*@\lightpurple{$\sstate_{A8} = \triple{\set{ \pair{\ttt{acyclic}}{\nu_1}} }{\senv[\ttt{n} \mapsto \nu_5, \ttt{result} \mapsto \nu_1]}{\nu_1 \kneq \knull \kand \nu_4 \keq \knull \kand \nu_5 \kneq \knull}$} @*)
(*@\lightblue{$\sstate_{B8}' = \triple{\set{ \pair{\ttt{acyclic}}{\nu_1}} }{\senv[\ttt{n} \mapsto \nu_6, \kresult \mapsto \nu_1]}{\nu_1 \kneq \knull \kand \nu_4 \kneq \knull \kand \nu_6 \kneq \knull}$} @*)
\end{lstlisting}}
Finally, in both $\sstate_{A8}$ and $\sstate_{B8}'$, we can consume the post-condition \ttt{acyclic(result) * result != NULL}. Therefore, we have verified all possible symbolic execution paths of \ttt{append}'s body, and thus verified \ttt{append}.
}

\section{\gvl}\label{sec:gradual}

\svl reflects the core components of Viper---\ttt{eval}, \ttt{consume}, \ttt{produce}, and \ttt{exec}. We now formally define \gvl, an extension of \svl which supports gradual specifications. We then define static verification for \gvl that allows optimistic assumptions to satisfy proof goals and generates checks to be verified at run time to cover these assumptions as in \citet{divincenzo2022gradual}.

Note, the syntax of \gvl differs slightly from that of \gvc (the frontend for \gco), particularly with its omission of C-style pointers. However, due to the restrictions of \co, all usages of pointers in \co can be translated to use object references. This and other translations are done by \gco during its conversion to an intermediate language \gviper, which is used in the backend verifier. In order to simplify our model, \gvl is very similar to the language of \gvc, but incorporates elements of the \gviper language when this simplifies the definition of our verification algorithm.

\subsection{Gradual formulas}\label{sec:gradual-formulas}

We first extend the syntax of our language to include \textit{imprecise} formulas---formulas of the form $\simprecise{\phi}$. An imprecise formula may represent any logically consistent strengthening of the precise portion $\phi$ \cite{wise2020gradual}. For example, the imprecise formula $\simprecise{x\,\ttt{>}\,0}$ consistently implies $x \keq 2$, but does not consistently imply $x \keq 0$. Then, a \textit{gradual formula} $\gform$ may be precise or imprecise, and \textit{gradual programs} are programs that contain gradual formulas. The abstract syntax of \gvl extends \svl's syntax with gradual formulas: 

{\footnotesize\disableTttResize
  \begin{minipage}{.45\linewidth}
    \begin{flalign*}
      \mathcal{P} &::= p(\multiple{T~x}) \kassign \gform &\\
      \Phi        &::= \krequires~\gform~\kensures~\gform &
    \end{flalign*}
  \end{minipage}
  \begin{minipage}{.45\linewidth}
    \begin{flalign*}
      s      &::= \cdots \mid \swhile{e}{\gform}{s} & \\
      \gform &::= \phi \mid \simprecise{\phi} &
    \end{flalign*}
  \end{minipage}}

~

\noindent Note, imprecise formulas are always considered self-framed, because they can always be strengthened to be self-framing. Therefore we require all method pre- and postconditions, loop invariants, and predicate bodies to be \textit{specifications}---formulas which are either imprecise or self-framed.

Also, note that IDF is particularly well-suited for gradual specifications, in comparison to separation logic \citep{reynolds2002separation}, since IDF allows separately specifying access permission and heap values. This allows specification of heap values while leaving more complex accessibility assertions unspecified, as in the formula $\simprecise{\ttt{x.f} \kneq \knull}$.



\subsection{Representation}\label{sec:gradual-representation}

In this section we extend the data structures from \S\ref{sec:precise-representation} to support \textit{imprecise states}---states in which it is permissible to make optimistic assumptions---and define our representation of run-time checks.

\begin{itemize}

\item A symbolic state $\sstate$ is now a quintuple $\quintuple{\imp}{\sheap}{\oheap}{\senv}{\pc}$ where $\imp$ is an \textit{imprecise flag}, $\sheap$ is a \textit{precise heap}, $\oheap$ is an \textit{optimistic heap}, $\senv$ is the \textit{symbolic store}, and $\pc$ is the \textit{path condition}. As before, we use the notation $\imp(\sstate)$, $\sheap(\sstate)$, etc. to reference specific components of a symbolic state. $\senv$ and $\pc$ are defined in \S\ref{sec:precise-representation} but we redefine the other components.

\item An \textit{imprecise flag} $\imp \in \set{\top, \bot}$ is a flag indicating whether a symbolic state is imprecise ($\top$) or precise ($\bot$). $\imp(\sstate)$ denotes that $\imp(\sstate) = \top$ (and thus $\sstate$ is an imprecise state), while $\neg \imp(\sstate)$ denotes that an $\imp(\sstate) = \bot$ (and thus $\sstate$ is precise). Imprecise states are produced by consuming or producing an imprecise specification. Once imprecise, a state always remains imprecise.

\item A \textit{precise heap} $\sheap$ is a symbolic heap as described in section \ref{sec:precise-representation}. Thus it is a finite set of heap chunks where all heap chunks represent distinct locations in the heap at run time.

\item An \textit{optimistic heap} $\oheap$ is a finite set of field chunks. Field chunks contained in the optimistic heap may represent the same location in the heap at run time, i.e. the optimistic heap does not preserve the separation invariant like the precise heap. The optimistic heap of a well-formed symbolic state must be empty unless it is an imprecise state. 

\end{itemize}

\subsection{Run-time checks}\label{sec:gradual-rt-checks}
A \textit{run-time check} $r \in \SCheck$ denotes an assertion that validates assumptions made during static verification of imprecise programs. It is a symbolic expression, symbolic permission, pair of symbolic permission sets, or $\bot$:
{\footnotesize
$$r \quad::=\quad t \quad\mid\quad \sperm \quad\mid\quad \fsep(\sperms_1, \sperms_2) \quad\mid\quad \bot$$}

\noindent A set of run-time checks is denoted $\scheck \in \powerset{\SCheck}$.
In a run-time check, a symbolic expression $t$ asserts that the value represented by $t$ at run time is $\ktrue$, a symbolic permission $\sperm$ asserts ownership of a field or a predicate instance, and a pair $\fsep(\sperms_1, \sperms_2)$ asserts that the sets of permissions represented by $\sperms_1$ and $\sperms_2$ are disjoint. $\bot$ represents a static verification failure. We represent static verification failure as an unsatisfiable run-time check, instead of failing verification entirely, to accommodate imprecision.

Note that our run-time checks contain symbolic values. This is unlike \gco \citep{divincenzo2022gradual}, where checks produced have their symbolic values replaced by corresponding program variables. This replacement is needed to support the implementation of run-time checks and adds a significant amount of complexity to their algorithms. Fortunately, as we will see later, we can abstract away this connection of symbolic values to program variables (aka. concrete values) using \emph{valuations}; and so we can produce abstracted checks here, avoiding additional complexity in our formalism. 
Additionally, at each branch point \citet{divincenzo2022gradual} check whether \textit{all} possible branches fail and, if so, halt static verification. We do not specify this behavior; however, this is possible by checking for $\bot \in \scheck$ at each step of symbolic execution.

\subsection{Evaluating expressions}\label{sec:gradual-eval}

We now extend our previous judgement for symbolic evaluation from \S\ref{sec:precise-eval} to allow optimistic symbolic evaluation of expressions. We specify a set $\scheck$ of run-time checks necessary for a given evaluation, thus our judgement is now $\seval{\sstate}{e}{t}{\sstate'}{\scheck}$.

Field chunks may be referenced in the optimistic heap by \textsc{SEvalFieldOptimistic} in Figure \ref{fig:gradual-eval-rules}. These field chunks have already been validated, thus we do not need additional run-time checks. A field may also be optimistically evaluated by \textsc{SEvalFieldImprecise}, even if it does not exist in $\sheap$ or $\oheap$. This adds a new field chunk with a $\ffresh$ value to $\oheap$. This requires a run-time check which asserts permission to access the field. Finally, \textsc{SEvalFieldFailure} applies in a precise state when a field is referenced but no matching heap chunk exists. This results in a failure of static verification, represented by $\bot$, for that execution branch.

We also modify the existing set of rules described in \S\ref{sec:precise-eval} to collect run-time checks from recursive evaluations. Likewise, we modify the deterministic evaluation judgement to add similar rules as those described above, allowing it to also generate run-time checks, thus its form is $\spceval{\sstate}{e}{t}{\scheck}$.

\begin{figure}
  {\footnotesize\disableTttResize
  \begin{mathpar}
    \inferrule[SEvalFieldOptimistic]
      { \seval{\sstate}{e}{t_e}{\sstate'}{\scheck} \\\\
        \nexistential{\triple{f}{t_e'}{t} \in \sheap(\sstate')}{\pc(\sstate') \implies t_e' \keq t_e} \\\\
        \triple{f}{t_e'}{t} \in \oheap(\sstate') \\
        \pc(\sstate') \implies t_e' \keq t_e }
      { \seval{\sstate}{e.f}{t}{\sstate'}{\scheck} } \and
    \inferrule[SEvalFieldImprecise]
      { \imp(\sstate) \\
        \seval{\sstate}{e}{t_e}{\sstate'}{\scheck} \\\\
        \nexistential{\triple{f}{t_e'}{t} \in \sheap(\sstate') \cup \oheap(\sstate')}{\pc(\sstate') \implies t_e' \keq t_e} \\\\
        t = \ffresh \\
        \oheap' = \oheap(\sstate) \cup \set{ \triple{f}{t_e}{t} } }
      { \seval{\sstate}{e.f}{t}{\sstate'[\oheap = \oheap']}{\scheck \cup \set{ \pair{t_e}{f} }} } \and
    \inferrule[SEvalFieldFailure]
      { \neg\imp(\sstate) \\ 
        \seval{\sstate}{e}{t_e}{\sstate'}{\scheck} \\
        \nexistential{\triple{f}{t_e'}{t} \in \sheap(\sstate') \cup \oheap(\sstate')}{\pc(\sstate') \implies t_e' \keq t_e} }
      { \seval{\sstate}{e.f}{\ffresh}{\sstate'}{\set{\bot}} }
  \end{mathpar}}
  \caption{Selected rules for evaluation during gradual verification}
  \label{fig:gradual-eval-rules}
\end{figure}

\subsection{Consuming formulas}\label{sec:gradual-consume}

We extend our previous judgment for consuming formulas from \S\ref{sec:precise-consume} to handle imprecise formulas and imprecise states. As in \S\ref{sec:gradual-eval}, we add a parameter $\scheck$ to the consume judgments. Additionally, we collect all permissions for the given formula into a set of symbolic permissions $\sperms$ so that separation checks may be added where necessary. Thus the new judgments are of the form $\scons{\sstate}{\gform}{\sstate'}{\scheck}$ and $\sconsume{\sstate}{\sstate_E}{\gform}{\sstate'}{\scheck}{\sperms}$; these two forms are related by \textsc{SConsume} and correspond to the forms described in \S\ref{sec:precise-consume}. We list selected rules in Figure \ref{fig:gradual-consume-rules}.

Consuming an imprecise formula empties the precise and optimistic heaps (\textsc{SConsumeImprecise}). This is because the imprecision may represent access to arbitrary fields. For example, a method with an imprecise precondition could modify any field that the callee owns, thus we cannot make any assumptions about field permissions or values after the method returns. Consuming an imprecise formula results in an imprecise state, thus removed field chunks can be referenced optimistically, with the possible addition of a run-time check.

We must also consider the case of consuming an imprecise formula in a precise state. Since optimistic assumptions are not permitted in a precise state, we cannot assume any of the assertions contained in the imprecise formula. However, the imprecise formula may reference fields without a corresponding accessibility predicate. Thus, when consuming an imprecise formula, we use an imprecise state as the symbolic state for evaluation, but use the original (possibly precise) state for assertions.

In an imprecise state we may optimistically consume an expression, even if it is not implied by the current path condition. We then add the value as a run-time check to be asserted at run-time.

Consumption of accessibility predicates must be modified to handle imprecise states, where field chunks in $\oheap$ may overlap with field chunks in $\sheap$. We must remove all fields that \emph{may} represent the same heap reference when removing a field chunk from $\sheap$. To do this, we use the helper functions $\fremfp$ and $\fremf$.
$\fremfp$ is used when removing heap chunks from the precise heap.
For precise states, $\fremfp$ removes the field chunk that coincides exactly with the heap location being consumed (thus computing the same result as the rules in \S\ref{sec:precise-consume}).
For imprecise states, it also removes all field chunks that could possibly coincide with the specified heap location. $\fremf$ is used when removing chunks from the optimistic heap and behaves similarly, but also removes all predicate chunks, since predicates occurring in the precise heap could overlap with heap chunks in the imprecise heap.
%
%
Some optimizations could be made -- for instance, if a predicate's unfolding will never reference a field $f$, we could preserve an instance of this predicate when consuming $\kacc(e.f)$. However, we leave such optimizations to future work.

We can also optimistically assume an accessibility predicate in an imprecise state, even if a matching field chunk does not exist in $\sheap$ or $\oheap$. Since this assumes ownership of the field, we add the corresponding symbolic permission to $\scheck$.
Finally, like accessibility predicates, we allow optimistic consumption of predicate instances. In this case the symbolic permission representing the predicate instance is added as a run-time check.

When consuming any accessibility predicate or predicate instance, the symbolic permission is always added to a set $\sperms$. This allows specification of checks for separation. When consuming $\kacc(x.f) * \kacc(y.f)$, if $\kacc(x.f)$ is optimistically assumed while $\kacc(y.f)$ is statically verified, the run-time check for $\kacc(x.f)$ does not imply that its permission is disjoint from that of $y.f$. Therefore additional checks for separation are added when consuming a separating conjunction such as $\phi_1 * \phi_2$. If no run-time check for permissions exists, all permissions must have been consumed from $\sheap$ or $\oheap$ and thus separation may be assumed. However, if a symbolic permission is contained in $\scheck$ we can no longer assume separation. Thus we add a run-time check $\fsep(\sperms_1, \sperms_2)$ where $\sperms_1$ is the set of symbolic permissions collected while consuming $\phi_1$ and likewise for $\sperms_2$ and $\phi_2$.

\begin{figure}
  {\footnotesize\disableTttResize
  \begin{mathpar}
    \inferrule[SConsume]
      { \sconsume{\sstate}{\sstate_E}{\gform}{\sstate'}{\scheck}{\sperms} }
      { \scons{\sstate}{\gform}{\sstate'}{\scheck} } \and
    \inferrule[SConsumeImprecise]
      { \sconsume{\sstate}{\sstate_E[\imp = \top]}{\phi}{\sstate'}{\scheck}{\sperms} }
      { \sconsume{\sstate}{\sstate_E}{\simprecise{\phi}}{\quintuple{\top}{\pc(\sstate')}{\senv(\sstate')}{\emptyset}{\emptyset}}{\scheck}{\sperms}  } \and
    \inferrule[SConsumeValueImprecise]
      { \imp(\sstate) \\
        \spceval{\sstate_E}{e}{t}{\scheck} \\
        \pc(\sstate) \notimplies t }
      { \sconsume{\sstate}{\sstate_E}{e}{\sstate[\pc = \pc(\sstate) \kand t]}{\scheck \cup \set{t}}{\emptyset} } \and
    \inferrule[SConsumeValueFailure]
      { \neg \imp(\sstate) \\
        \spceval{\sstate_E}{e}{t}{\scheck} \\
        \pc(\sstate) \notimplies t }
      { \sconsume{\sstate}{\sstate_E}{e}{\sstate}{\set{\bot}}{\emptyset} } \\
    \inferrule[SConsumeAcc]
      { \spceval{\sstate}{e}{t_e}{\scheck} \\
        \pc(\sstate) \implies t_e' \keq t_e \\\\
        \triple{f}{t_e'}{t} \in \sheap(\sstate) \\\\
        \sheap' = \fremfp(\sheap(\sstate), \sstate, t_e, f) \\\\
        \oheap' = \fremf(\oheap(\sstate), \sstate, t_e, f) \\\\
        \sstate' = \sstate[\sheap = \sheap', \oheap = \oheap'] }
      { \sconsume{\sstate}{\sstate_E}{\kacc(e.f)}{\sstate'}{\scheck}{\set{\pair{t_e}{f}}} } \and
    \inferrule[SConsumeAccOptimistic]
      { \spceval{\sstate}{e}{t_e}{\scheck} \\
        \pc(\sstate) \implies t_e' \keq t_e \\\\
        \nexistential{\triple{f}{t_e'}{t} \in \sheap(\sstate)}{\pc(\sstate) \implies t_e' \keq t_e} \\\\
        \triple{f}{t_e'}{t} \in \oheap(\sstate) \\\\
        \sheap' = \fremf(\sheap(\sstate), \sstate, t_e, f) \\\\
        \oheap' = \fremf(\oheap(\sstate), \sstate, t_e, f) \\\\
        \sstate' = \sstate[\sheap = \sheap', \oheap = \oheap'] }
      { \sconsume{\sstate}{\sstate_E}{\kacc(e.f)}{\sstate'}{\scheck}{\set{\pair{t_e}{f}}} } \and
    \inferrule[SConsumeAccImprecise]
      { \imp(\sstate) \\
        \spceval{\sstate}{e}{t_e}{\scheck} \\\\
        \nexistential{\triple{f}{t_e'}{t} \in \sheap(\sstate) \cup \oheap(\sstate)}{\pc(\sstate) \implies t_e' \keq t_e} \\\\
        \sstate' = \sstate[\sheap = \fremf(\sheap(\sstate), \sstate, t_e, f), \oheap = \fremf(\oheap(\sstate), \sstate, t_e, f)]}
      { \sconsume{\sstate}{\sstate_E}{\kacc(e.f)}{\sstate'}{\scheck \cup \set{ \pair{t_e}{f} }}{\set{\pair{t_e}{f}}} } \and
    \inferrule[SConsumeAccFailure]
      { \neg \imp(\sstate) \\
        \spceval{\sstate}{e}{t_e}{\scheck} \\\\
        \nexistential{\triple{f}{t_e'}{t} \in \sheap(\sstate) \cup \oheap(\sstate)}{\pc(\sstate) \implies t_e' \keq t_e} }
      { \sconsume{\sstate}{\sstate_E}{\kacc(e.f)}{\sstate}{\set{\bot}}{\set{\pair{t_e}{f}}} }
  \end{mathpar}}
  \caption{Selected rules for consuming gradual formulas}
  \label{fig:gradual-consume-rules}
\end{figure}

\subsection{Producing formulas}

Since a formula is only produced when we can assume its validity, producing a gradual formula does not require any optimistic assumptions, thus we do not need to calculate any run-time checks. When producing an imprecise formula, we produce the precise portion and set $\imp = \top$. All other rules from \S\ref{sec:precise-produce} are left unchanged.

\subsection{Executing statements}

All rules from \S\ref{sec:precise-exec} are left unchanged. While it may seem natural to calculate run-time checks while determining execution transitions (as in the \ttt{exec} algorithm of \citet{divincenzo2022gradual}), we found that this unnecessarily complicates statements of soundness since symbolic execution steps are not equivalent to dynamic execution steps. For example, a method call occurs in one step during symbolic execution but may never complete during dynamic execution, therefore it may be impossible to determine which symbolic execution step applies. However, assertion of run-time checks must occur before a dynamic execution step may proceed. Therefore we cleanly delineate between symbolic execution transitions, specified by the judgement $\sexec{\sstate}{s}{s'}{\sstate'}$, and the calculation of run-time checks.

\subsection{Guarding execution}

As described above, we must the assert run-time checks \emph{before} the corresponding dynamic execution step occurs. Therefore we define \textit{guard} judgements to calculate run-time checks which ensure that execution can safely proceed. A guard for a method call calculates the checks necessary to ensure that the method's pre-condition is satisfied, while a guard for a field assignment calculates the checks necessary to ensure permission to access the assignee and evaluate the value to be stored.

A guard judgement $\sguard{\vstate}{\sstate'}{\scheck}{\sperms}$ denotes that, at the execution state represented by $\vstate$, when the execution path matches the path condition in $\sstate'$, the run-time checks $\scheck$ must be checked. Selected guard rules are defined in Figure \ref{fig:gradual-guard-rules}.

In a guard judgement, $\sperms$ determines the exclusion frame---a set of permissions which must not escape the executing method's context. Its necessity and behavior is explained in \S\ref{sec:unsoundness}.

\begin{figure}
  {\footnotesize\disableTttResize
  \begin{mathpar}
    \inferrule[SGuardAssign]
      { \seval{\sstate}{e}{\_}{\sstate'}{\scheck} }
      { \sguard{\triple{\sstate}{\sseq{x \kassign e}{s}}{\gform}}{\sstate'}{\scheck}{\emptyset} } \and
    \inferrule[SGuardAssignField]
      { \seval{\sstate}{e}{\_}{\sstate'}{\scheck} \\\\
        \scons{\sstate'}{\kacc(x.f)}{\sstate''}{\scheck'} }
      { \sguard{\triple{\sstate}{\sseq{x.f \kassign e}{s}}{\gform}}{\sstate''}{\scheck \cup \scheck'}{\emptyset} } \and
    \inferrule[SGuardCall]
      { \multiple{\seval{\sstate}{e}{t}{\sstate'}{\scheck}} \\
        \multiple{x} = \fparams(m) \\\\
        \scons{\sstate'[\senv = [\multiple{x \mapsto t}]]}{\fpre(m)}{\sstate''}{\scheck'} }
      { \triple{\sstate}{\sseq{y \kassign m(\multiple{e})}{s}}{\gform} \pfunc \sstate'' \dashv \\\\ \scheck \cup \scheck', \frem(\sstate'', \fpre(m)) } \\
    \frem(\sstate, \, \gform) := \begin{cases}
      \emptyset &\text{if $\gform$ is completely precise} \\
      \set{\pair{p}{\multiple{t}} : \pair{p}{\multiple{t}} \in \sheap(\sstate) } ~\cup \\
      \quad \set{ \pair{t}{f} : \triple{f}{t}{t'} \in \sheap(\sstate) \cup \oheap(\sstate) }  &\text{otherwise}
    \end{cases}
  \end{mathpar}}
  \caption{Selected guard rules}
  \label{fig:gradual-guard-rules}
\end{figure}

\subsection{Example}
{\parindent0pt
\begin{figure}
  \begin{minipage}[t]{.45\linewidth}
\begin{lstlisting}[name=gradual-example]
List append(List l, int value)
  requires ? * true(*@ \label{ln:imprecise-append-pre} @*)
  ensures ? * acyclic(result)(*@ \label{ln:imprecise-append-post} @*)
{
  if (l.next == NULL)(*@ \label{ln:imprecise-append-if} @*)
    n = singleton(value);(*@ \label{ln:imprecise-append-single} @*)
\end{lstlisting}
  \end{minipage}
  \begin{minipage}[t]{.45\linewidth}
\begin{lstlisting}[name=gradual-example]
  else(*@ \label{ln:imprecise-append-else} @*)
    n = append(l.next, value);(*@ \label{ln:imprecise-append-rec} @*)
  l.next = n;(*@ \label{ln:imprecise-append-next} @*)
  result = l;(*@ \label{ln:imprecise-append-result} @*)
}(*@ \label{ln:imprecise-append-end} @*)
\end{lstlisting}
  \end{minipage}
  \caption{\gvl code for appending to an acyclic linked list}
  \label{fig:gradual-example}
\end{figure}

We now illustrate verification of the gradually-specified method in Figure \ref{fig:gradual-example}. We assume the definition of \ttt{List} and \ttt{acyclic} from Figure \ref{fig:list-example}. The gradual specification of \ttt{append} ensures that all returned lists are acyclic.

Symbolic states are tuples of the form $\quintuple{\imp}{\sheap}{\oheap}{\senv}{\pc}$. As in \S\ref{sec:static-example}, we begin verification of \ttt{append} by assigning fresh symbolic values to all parameters and producing the pre-condition $\simprecise{\ktrue}$, which results in an imprecise state:
{\disableTttResize
\begin{lstlisting}[aboveskip=.25em, belowskip=0.5em,numbers=none]
(*@\lightgray{$\sstate_1 = \quintuple{\top}{\emptyset}{\emptyset}{\senv}{\ktrue} \quad \senv = [\ttt{l} \mapsto \nu_1, \ttt{v} \mapsto \nu_2]$} @*)
\end{lstlisting}}
At each statement we compute the guard to find the necessary run-time checks. The guard for the \ttt{if} statement at line \ref{ln:imprecise-append-if} evaluates \ttt{l.next == NULL}. A heap chunk for \ttt{l.next} is optimistically added to $\oheap$ with a fresh value $\nu_3$. This also results in a run-time check for the symbolic permission $\pair{\nu_1}{\ttt{next}}$.

The next state $\sstate_{A2}$ is computed by symbolic execution. This again evaluates \ttt{l.next == NULL} in the state $\sstate_1$, which again requires the addition of an optimistic heap chunk. Since $\nu_3$ was not previously used in $\sstate_1$ it can be used again as a fresh value for symbolic execution.
{\disableTttResize
\begin{lstlisting}[aboveskip=.25em, belowskip=0.5em,firstnumber=4,stepnumber=5]
(*@\lightgray{$\scheck = \set{ \pair{\nu_1}{\ttt{next}}$} } @*)
if (l.next == NULL) 
    (*@\lightpurple{$\sstate_{A2} = \quintuple{\top}{\emptyset}{\set{\triple{\ttt{next}}{\nu_1}{\nu_3}}}{\senv}{\nu_3 \keq \knull}$} @*)
\end{lstlisting}}
The guard for line \ref{ln:imprecise-append-single} consumes the pre-condition of \ttt{singleton}, which requires no run-time checks. Symbolic execution consumes the same pre-condition and produces the post-condition; here we use the fresh value $\nu_4$ for the returned value:
{\disableTttResize
\begin{lstlisting}[aboveskip=.25em, belowskip=0.5em,firstnumber=5,stepnumber=6]
    (*@\lightpurple{$\scheck = \emptyset$} @*)
    n = singleton(value);
    (*@\lightpurple{$\sstate_{A3} = \quintuple{\top}{\set{\pair{\ttt{acyclic}}{\nu_4}}}{\set{\triple{\ttt{next}}{\nu_1}{\nu_3}}}{\senv[\ttt{n} \mapsto \nu_4]}{\nu_3 \keq \knull}$} @*)
\end{lstlisting}}
As in \S\ref{sec:static-example}, we follow code order, instead of following each execution path individually, and distinguish separate execution paths with color-coding. The guard at line \ref{ln:imprecise-append-if} computes the checks for both branches of \ttt{if}, thus the guard is not computed at line \ref{ln:imprecise-append-else}. We can symbolically execute the \ttt{else} branch by adding the negation of the path condition we used previously:
{\disableTttResize
\begin{lstlisting}[aboveskip=.25em, belowskip=0.5em,firstnumber=7,stepnumber=7]
else
    (*@\lightblue{$\sstate_{B2} = \quintuple{\top}{\emptyset}{\set{\triple{\ttt{next}}{\nu_1}{\nu_3}}}{\senv}{\nu_3 \kneq \knull}$} @*)
\end{lstlisting}}
The guard for line \ref{ln:imprecise-append-rec} consumes the pre-condition of \ttt{append}, which is $\simprecise{\ktrue}$. Since the body of this imprecise formula is only $\ktrue$, no run-time checks are necessary. However, since this is an imprecise formula, we clear the precise and optimistic heaps (\textsc{SConsumeImprecise} in Figure \ref{fig:gradual-consume-rules}). Symbolic execution then produces the post-condition; here we use the fresh value $\nu_5$ for the returned value:
{\disableTttResize
\begin{lstlisting}[aboveskip=.25em, belowskip=0.5em,firstnumber=7,stepnumber=8]
    (*@\lightyellow{$\scheck_B = \emptyset$} @*)
    n = append(l.next, value);
    (*@\lightblue{$\sstate_{B3} = \quintuple{\top}{\set{\pair{\ttt{acyclic}}{\nu_5}}}{\emptyset}{\senv[\ttt{n} \mapsto \nu_5]}{\nu_3 \kneq \knull \kand \nu_5 \kneq \knull}$} @*)
\end{lstlisting}}
We resume symbolic execution of both paths at line \ref{ln:imprecise-append-next}. Executing \ttt{l.next = n} in the state $\sstate_{A3}$ does not require any run-time checks since it contains the heap chunk representing \ttt{l.next}. However, executing the same statement in $\sstate_{B3}$ requires optimistic assumption of the symbolic permission $\pair{\ttt{next}}{\nu_1}$ which requires a run-time check and removes the predicate instance. \citet{divincenzo2022gradual} describe the implementation of such conditional run-time checks, but here we represent it with separate symbolic execution paths:
{\disableTttResize
\begin{lstlisting}[aboveskip=.25em, belowskip=0pt,firstnumber=8,stepnumber=9]
(*@\lightpurple{$\scheck_A = \emptyset$}, \lightblue{$\scheck_B = \set{\pair{\ttt{next}}{\nu_1}}$} @*)
l.next = n;
(*@\lightpurple{$\sstate_{A4} = \quintuple{\top}{\set{\pair{\ttt{acyclic}}{\nu_4}}}{\set{\triple{\ttt{next}}{\nu_1}{\nu_4}}}{\senv[\ttt{n} \mapsto \nu_4]}{\nu_3 \keq \knull}$} @*)
(*@\lightblue{$\sstate_{B4} = \quintuple{\top}{\set{\triple{\ttt{next}}{\nu_1}{\nu_5}}}{\emptyset}{\senv[\ttt{n} \mapsto \nu_5]}{\nu_3 \kneq \knull \kand \nu_5 \kneq \knull}$} @*)
\end{lstlisting}
\begin{lstlisting}[aboveskip=0pt, belowskip=0.5em,firstnumber=9,stepnumber=10]
(*@\lightpurple{$\scheck_A = \emptyset$}, \lightblue{$\scheck_B = \emptyset$} @*)
result = l;
(*@\lightpurple{$\sstate_{A5} = \quintuple{\top}{\nu_3 \keq \knull}{\set{\pair{\ttt{acyclic}}{\nu_4}}}{\set{\triple{\ttt{next}}{\nu_1}{\nu_4}}}{\senv[\ttt{n} \mapsto \nu_4, \ttt{result} \mapsto \nu_1]}$} @*)
(*@\lightblue{$\sstate_{B4} = \quintuple{\top}{\nu_3 \kneq \knull \kand \nu_5 \kneq \knull}{\set{\triple{\ttt{next}}{\nu_1}{\nu_5}}}{\emptyset}{\senv[\ttt{n} \mapsto \nu_5, \ttt{result} \mapsto \nu_1]}$} @*)
\end{lstlisting}}
Finally, the applicable guard at line \ref{ln:imprecise-append-end} consumes the post-condition $\ttt{acyclic(result)}$. Since neither state contains a matching predicate instance, in both paths a run-time check is added for the symbolic permission $\pair{\ttt{acyclic}}{\nu_1}$:
{\disableTttResize
\begin{lstlisting}[aboveskip=.25em, belowskip=0.5em,firstnumber=10,stepnumber=11]
(*@\lightpurple{$\scheck_A = \set{\pair{\ttt{acyclic}}{\nu_1}}$}, \lightblue{$\scheck_B = \set{\pair{\ttt{acyclic}}{\nu_1}}$} @*)
}
\end{lstlisting}}
Now we have verified the method and computed all necessary run-time checks.
}

\section{Executing \gvl}


Since soundness of static verification requires specification of program execution, we define execution semantics for \gvl programs (including \svl programs, which are a subset). This includes dynamic semantics for formulas, and execution semantics which dynamically assert the validity of every specification. Therefore, these semantics define valid execution for \gvl programs.

As explained in \S\ref{sec:introduction}, execution semantics are based on those of \co \citep{arnold2010c0}, while the semantics of formulas are based on those of \gvlrp \citep{wise2020gradual}.

\subsection{Representation}\label{sec:dynamic-representation}

In this section, we formally define the data structures used during execution of \gvl programs:

\begin{itemize}
\item A \textit{value} $v \in \Value$ is an integer, boolean, or object reference.

\item An \textit{object reference} $\ell \in \textsc{Ref}$ is an identifier for a particular object. As with symbolic values, we assume that an infinite number of distinct values can be generated by the $\ffresh$ function. The type of value represented by $\ffresh$ is disambiguated by its usage.

\item An \textit{environment} $\env$ is a partial function mapping variable names to values, i.e. $\env : \Var \pfunc \Value$.

\item A \textit{heap} $\heap : \textsc{Ref} \times \Field \to \Value$ is a function mapping object reference and field pairs to values. We assume that the heap function is total, i.e. all reference and field pairs have some corresponding value, but heap access is restricted during execution by a set of \textit{access permissions} $\perms \in \mathcal{P}(\textsc{Ref} \times \textsc{Field})$. This reflects the semantics of IDF \citep{smans2012implicit}. A heap location $\pair{\ell}{f}$ is \textit{owned} if it is contained in the currently-applicable set of access permissions.

\item A \textit{stack frame} is a triple $\triple{\perms}{\env}{s}$ containing of a set of owned permissions $\perms$, a local environment $\env$, and a statement $s$. A \textit{stack} $\stack$ is a list of stack frames -- either $\triple{\perms}{\env}{s} \cdot \stack$ for some other $\stack$, or the empty stack, denoted $\nilsym$. For a non-empty stack $\stack$, $\perms(\stack)$, $\env(\stack)$, and $s(\stack)$ refer to their respective components of the head element.

\item A \textit{dynamic state} $\Gamma$ may be a symbol $\initsym$ or $\finalsym$, or pair $\pair{\heap}{\stack}$ containing a heap $\heap$ and a non-empty stack $\stack$. $\heap(\Gamma)$ and $\stack(\Gamma)$ reference individual components of $\Gamma$ when $\Gamma$ is not a symbol, while $\perms(\Gamma)$, $\env(\Gamma)$, and $s(\Gamma)$ reference a component of the head element of $\stack(\Gamma)$. $\perms(\initsym)$ and $\heap(\initsym)$ are defined to be $\emptyset$.
\end{itemize}

\subsection{Evaluating expressions}

Given a heap $\heap$ and environment $\env$, the evaluation of expression $e$ to a value $v$ is represented by a judgement of the form $\eval{\heap}{\env}{e}{v}$. This follows standard evaluation rules---variables are evaluated to the corresponding value in $\env$ and field references are evaluated to the corresponding value in $\heap$. The boolean operators $\kand$ and $\kor$ are short-circuiting---when evaluating $e_1 \kand e_2$, $e_2$ is only evaluated when $e_1$ is not true.

\subsection{Asserting formulas}\label{sec:dynamic-formulas}

A judgement of the form $\assertion{\heap}{\perms}{\env}{\gform}$ denotes that a gradual formula $\gform$ is satisfied given a heap $\heap$, a set of accessible permissions $\perms$, and an environment $\env$. Selected rules are shown in Figure \ref{fig:dynamic-rules}. Boolean expressions are satisfied when they evaluate to true. An accessibility predicate $\kacc(e.f)$ is satisfied when the field referenced by $e.f$ is in the set of accessible permissions. A predicate instance $p(\overline{e})$ is satisfied when the predicate body is satisfied using an environment mapping each predicate parameter $x$ to the corresponding argument $e$. A separating conjunction $\phi_1 * \phi_2$ is satisfied when $\phi_1$ is satisfied using a permission set $\perms_1$ and $\phi_2$ is satisfied using a permission set $\perms_2$, where $\perms_1$ and $\perms_2$ are disjoint subsets of $\perms$. Finally, an imprecise formula $\simprecise{\phi}$ is satisfied exactly when $\phi$ is satisfied.

A judgement of the form $\frm{\heap}{\perms}{\env}{e}$ denotes that the expression $e$ is \textit{framed} by the given set of permission $\perms$. This denotes that all heap locations necessary to evaluate $e$ are included in $\perms$. Selected rules are shown in Figure \ref{fig:dynamic-rules}.

Note that a predicate instance $p(\overline{e})$ is satisfied iff the predicate body is satisfied. Thus dynamic execution of \gvl uses \textit{equirecursive} semantics for predicates \citep{summers2013formal}. We also define \textit{equirecursive framing} of formulas by the judgement $\efrm{\heap}{\perms}{\env}{\gform}$. A formula is equirecursively framed if its \emph{unrolling}, the recursive expansion of referenced predicate bodies, is framed.

A formula $\gform$ is \textit{self-framed} if $\universal{\heap, \perms, \env}{\assertion{\heap}{\perms}{\env}{\gform} \implies \efrm{\heap}{\perms}{\env}{\gform}}$. As specified in \S\ref{sec:gradual-formulas}, a specification is a formula which is imprecise or self-framed.

\subsection{Footprints}

The \emph{footprint} of a formula is the set of permissions necessary to assert a formula \cite{reynolds2002separation}. There are two types of footprints for gradual formulas:

The \emph{exact footprint} of a formula is the minimal set of permissions necessary to assert and frame a formula. Given a heap $\heap$ and environment $\env$, $\efoot{\heap}{\env}{\gform}$ denotes the exact footprint of a formula $\gform$.

The \emph{maximal footprint} (often abbreviated to \emph{footprint}) of a formula contains the exact footprint and all permissions that are consistently implied by the formula. The maximal footprint of a completely precise formula is its exact footprint, but the maximal footprint of an imprecise formula contains all accessible permissions. Given a heap $\heap$, set of owned permissions $\perms$, and environment $\env$, $\foot{\heap}{\perms}{\env}{\gform}$ denotes the maximal footprint of a formula $\gform$.

\subsection{Executing statements}\label{sec:dynamic-exec}

We represent the dynamic execution of program statements as small-step execution semantics denoted by the judgement $\dexec{\heap}{\stack}{\xperms}{\heap}{\stack}$, where the statement $s$ is executing with the initial state $\pair{\heap}{\stack}$, and then transitions to the next statement $s'$ with the new state $\pair{\heap}{\stack}$. $\xperms$ specifies the \emph{exclusion frame}, which is described below. Selected rules are shown in Figure \ref{fig:dynamic-rules}. Execution will be stuck (i.e., no further derivation will apply) when an error is encountered. For example, execution is stuck when a formula is not satisfied or if some expression is not framed.

A method call is executed by evaluating all arguments, asserting the pre-condition, and adding a new stack frame containing the footprint of the pre-condition and the method body. After the method body is completely executed, the post-condition is asserted and the $\kresult$ value in the callee's environment is passed to the caller's environment.

A loop is executed similarly to a method, but uses the loop invariant instead of a method contract. When the loop condition is true, an iteration is executed by asserting the invariant and adding a new stack frame for the loop body. When the body is complete, we return to the original loop statement, allowing further iterations as long as the condition remains true. When the condition is false, the invariant is still asserted but execution skips over the statement. These rules are specified in \supplement{the supplement \citep{supplement}}{\S\ref{sec:exec-rules}}.

$\xperms$ specifies the \textit{exclusion frame} -- a set of permissions which may not be passed to the callee or loop body. It is used only for executing method calls and loops. We later explain why this is necessary for soundness in \S\ref{sec:unsoundness}.

$\kfold$ and $\kunfold$ statements are ignored at run-time. Explicit folding and unfolding of predicate instances is not required because the run-time uses equirecursive semantics for predicates.

The entire set of possible execution steps for a program $\prog$ is determined by judgements of the form $\dtrans{\prog}{\xperms}{\Gamma}{\Gamma'}$, which denote that execution transitions from $\Gamma$ to $\Gamma'$, using the exclusion frame $\xperms$. From the $\initsym$ state, execution may only step to the entry statement, and then execution follows the rules described above.

\begin{figure}
  {\footnotesize\disableTttResize
  \begin{mathpar}
    \inferrule[AssertImprecise]
      { \assertion{\heap}{\perms}{\env}{\phi} \\
        \efrm{\heap}{\perms}{\env}{\phi} }
      { \assertion{\heap}{\perms}{\env}{\simprecise{\phi}} } \and
    \inferrule[AssertAcc]
      { \eval{\heap}{\env}{e}{\ell} \\
        \pair{\ell}{f} \in \perms }
      { \assertion{\heap}{\perms}{\env}{\kacc(e.f)} } \and
    \inferrule[AssertPredicate]
      { \multiple{x} = \fpredparams(p) \\
        \multiple{\eval{\heap}{\env}{e}{v}} \\\\
        \assertion{\heap}{\perms}{[\multiple{x \mapsto v}]}{\fpred(p)} }
      { \assertion{\heap}{\perms}{\env}{p(\multiple{e})} } \and
    \inferrule[AssertConjunction]
      { \assertion{\heap}{\perms_1}{\env}{\phi_1} \\
        \assertion{\heap}{\perms_2}{\env}{\phi_2} \\\\
        \perms_1 \cap \perms_2 = \emptyset \\
        \perms_1 \cup \perms_2 \subseteq \perms }
      { \assertion{\heap}{\perms}{\env}{\phi_1 * \phi_2} } \and
    \inferrule[ExecAssignField]
      { \eval{\heap}{\env}{x}{\ell} \\
        \eval{\heap}{\env}{e}{v} \\
        \assertion{\heap}{\perms}{\env}{\kacc(x.f)} \\
        \frm{\heap}{\perms}{\env}{e} \\
        \heap' = \heap[\pair{\ell}{f} \mapsto v]}
      { \dexec{\heap}{\triple{\perms}{\env}{\sseq{x.f \kassign e}{s}} \cdot \stack}{\xperms}{\heap'}{\triple{\perms}{\env}{s} \cdot \stack} } \and
    \inferrule[ExecCallEnter]
      { \multiple{x} = \fparams(m) \\
        \multiple{\eval{\heap}{\env}{e}{v}} \\
        \multiple{\frm{\heap}{\perms}{\env}{e}} \\\\
        \env' = [\multiple{x \mapsto v}] \\
        \assertion{\heap}{\perms \setminus \xperms}{\env'}{\fpre(m)} \\
        \perms' = \foot{\heap}{\perms \setminus \xperms}{\env}{\fpre(m)} }
      { \dexec{\heap}{\triple{\perms}{\env}{\sseq{y = m(\multiple{e})}{s}} \cdot \stack}{\xperms}{\heap}{\triple{\perms'}{\env'}{\sseq{\fbody(m)}{\kskip}} \cdot \triple{\perms \setminus \perms'}{\env}{\sseq{y = m(\multiple{e})}{s}} \cdot \stack} } \and
    \inferrule[ExecCallExit]
      { \assertion{\heap}{\perms'}{\env'}{\fpost(m)} \\
        \env'' = \env[y \mapsto \env'(\kresult)] \\
        \perms'' = \perms \cup \foot{\heap}{\perms'}{\env'}{\fpost(m)} }
      { \dexec{\heap}{\triple{\perms'}{\env'}{\kskip} \cdot \triple{\perms}{\env}{\sseq{y = m(\multiple{e})}{s}} \cdot \stack}{\xperms}{\heap}{\triple{\perms}{\env}{s} \cdot \stack} }
  \end{mathpar}}
  \caption{Selected formal rules for dynamic semantics of \gvl.}
  \label{fig:dynamic-rules}
\end{figure}

\section{Correspondence}\label{sec:correspondence}

Before formalizing soundness, we must specify the correspondence between verification and dynamic states. We include invariants which depend on concrete values, such as separation, in this correspondence relation. Finally, we specify the behavior of run-time checks in a dynamic state.

\subsection{State correspondence}

A dynamic environment $\env$ \textit{models} a symbolic store $\senv$ via a valuation $V$ when $\simenv{V}{\senv}{\env}$. This denotes that $\universal{x \mapsto t \in \senv}{x \mapsto V(t) \in \env}$.

A heap $\heap$ and set of permissions $\perms$ model a precise heap $\sheap$ when $\simheap{V}{\sheap}{\heap}{\perms}$. This denotes that for all field chunks $\triple{f}{t}{t'} \in \sheap$, $\heap(V(t), f) = V(t')$ and $\pair{V(t)}{f} \in \perms$. Also, for all predicate chunks $\pair{p}{\multiple{t}}$, the corresponding predicate body is true using given arguments $\multiple{V(t)}$. Additionally, the footprint represented by each heap chunk must be disjoint.

The footprint of a heap chunk, given valuation $V$ and heap $\heap$, is denoted $\vfoot{V}{\heap}{h}$. The footprint of a field chunk $\triple{f}{t}{t'}$ is $\set{\pair{V(t)}{f}}$. The footprint of a predicate chunk $\pair{p}{\multiple{t}}$ is the exact footprint of the predicate when applied to the arguments $\multiple{V(t)}$.

$\heap$ and $\perms$ model an optimistic heap $\oheap$ when $\simheap{V}{\oheap}{\heap}{\perms}$. This has the same requirements as that for $\sheap$, except that heap chunks are allowed to overlap.

$\heap$, $\perms$, and $\perms$ model a symbolic state when $\simstate{V}{\sstate}{\heap}{\perms}{\env}$. This denotes that $\heap$ and $\perms$ model both $\sheap(\sstate)$ and $\oheap(\sstate)$, $\env$ models $\senv$, and the path condition is true---$V(\pc(\sstate)) = \ktrue$.

We also refer to these relations as correspondence---$\simstate{V}{\sstate}{\heap}{\perms}{\env}$ denotes that the symbolic state $\sstate$ \textit{corresponds} to $\heap$, $\perms$, and $\env$.

Finally, a verification state $\vstate$ corresponds to a dynamic state $\Gamma$ with valuation $V$ if $\vstate = \Gamma$ (i.e., $\vstate$ and $\Gamma$ are the same symbol), or $\simstate{V}{\sstate(\vstate)}{\heap(\heap)}{\perms(\Gamma)}{\env(\Gamma)}$ and $s(\Gamma) = s(\vstate)$. In other words, the heap and head stack frame of $\Gamma$ model the symbolic state of $\vstate$, and the statement in the head stack frame is syntactically the same statement as that of $\vstate$.

\subsection{Run-time checks}

We also define the semantics of run-time checks using valuations. The judgement $\rtassert{V}{\heap}{\perms}{r}$ denotes the assertion of a run-time check $r$, given a valuation $V$, a heap $\heap$, and a set of owned permissions $\perms$. Likewise, $\rtassert{V}{\heap}{\perms}{\scheck}$ denotes the assertion of all run-time checks contained in $\scheck$. Formal rules are given in Figure \ref{fig:rt-check-rules}.

\begin{figure}
  {\footnotesize\disableTttResize
  \begin{mathpar}
    \inferrule
      { V(t) = \ktrue }
      { \rtassert{V}{\heap}{\perms}{t} } \and
    \inferrule
      { \pair{V(t)}{f} \in \perms }
      { \rtassert{V}{\heap}{\perms}{\pair{t}{f}} } \and
    \inferrule
      { \vfoot{V}{\heap}{\sperms_1} \cap \vfoot{V}{\heap}{\sperms_2} = \emptyset }
      { \rtassert{V}{\heap}{\perms}{\fsep(\sperms_1, \sperms_2)} } \and
    \inferrule
      { \assertion{\heap}{\perms}{[\multiple{x \mapsto V(t)}]}{\fpred(p)} \\\\
        \multiple{x} = \fpredparams(p) }
      { \rtassert{V}{\heap}{\perms}{\pair{p}{\multiple{t}}} }
  \end{mathpar}}
  \caption{Rules for run-time check assertions}
  \label{fig:rt-check-rules}
\end{figure}

\section{Soundness}\label{sec:soundness}

We can now state the soundness of our static verifier. We slightly modify a traditional progress/preservation statement of soundness in order to accommodate run-time checks.

\subsection{Corresponding valuations}

For most symbolic execution judgements, we define a \textit{corresponding valuation}, inspired by the valuations used in \citet{khoo10mix}. This defines how symbolic values used in the judgement are mapped to concrete values. To calculate the corresponding valuation we require an initial valuation, which defines the valuation for all symbolic values contained by the input symbolic state, and a dynamic heap, which defines the valuation for optimistically-added fields. A corresponding valuation $V'$ must extend the initial valuation $V$, i.e. $V'(t) = V(t)$ for all $t \in \dom(V)$.

We denote the corresponding valuation for a judgement $\mathcal{J}$, initial valuation $V$, and heap $\heap$ by $V[\mathcal{J} \mid \heap]$. The definition for each judgement type is defined in \supplement{the supplement \citep{supplement}, along with the proofs for that judgement}{the appendix which contains the corresponding proofs for that judgement}. Each corresponding valuation is defined by induction on the judgement derivation, specifying the corresponding valuation for each derivation rule.  The judgment is nondeterministic if only the input state is considered, but knowing the output state resolves this nondeterminism.
When the judgement and heap are clear from context, we simply reference the \textit{corresponding valuation extending} $V$.


\subsection{Valid states}

A \textit{valid state} is a dynamic state which is completely characterized by verification states. If $\Gamma = \initsym$ this is trivially true. For a dynamic state $\pair{\heap}{\stack}$, we require that the head stack frame corresponds to a reachable verification state. We also require that all other stack frames are \textit{partially validated} by some reachable verification state.

If a stack frame is executing a method call, partial validation is characterized by the stack frame and heap modeling a reachable symbolic state for that program point, with the callee's precondition consumed.
For the full definition refer to \supplement{\citet{supplement}}{definition \ref{def:partial-valid}}.

\subsection{Progress and preservation}

Our statement of progress is split into two parts. First, theorem \ref{thm:progress-1} states that if $\Gamma$ is a valid state and $\Gamma$ satisfies the run-time checks calculated by a guard with a path condition that matches the current dynamic state, then dynamic execution proceeds. Second, theorem \ref{thm:progress-2} states that we can always find the guard necessary to apply theorem \ref{thm:progress-1}---a guard whose path condition matches. Thus, theorem \ref{thm:progress-2} represents completeness of symbolic execution with respect to possible dynamic execution paths. Together these theorems show that, in a valid state, the only possible way for execution to be stuck is when the run-time checks cannot be asserted.

\begin{theorem}[Progress part 1] \label{thm:progress-1}
  For some program $\prog$, let $\Gamma$ be some dynamic state validated by $\vstate$ and $V$. If $\sguard{\vstate}{\sstate}{\scheck}{\sperms}$, $V'$ is the corresponding valuation extending $V$, $V'(\pc(\sstate)) = \ktrue$, and $\rtassert{V'}{\heap}{\perms(\Gamma)}{\scheck}$, then $\dtrans{\prog}{\vfoot{V'}{\heap(\Gamma)}{\sperms}}{\Gamma}{\Gamma'}$ for some $\Gamma'$.
\end{theorem}

\begin{theorem}[Progress part 2] \label{thm:progress-2}
  For some program $\prog$, let $\Gamma$ be some dynamic state validated by $\vstate$ and $V$. Then $\sguard{\vstate}{\sstate}{\scheck}{\sperms}$ for some $\sstate$, $\scheck$, and $\sperms$ such that $V'(\pc(\sstate)) = \ktrue$ where $V'$ is the corresponding valuation extending $V$.
\end{theorem}

Finally, our statement of preservation (theorem \ref{thm:preservation}) assumes the antecedent and conclusion of theorem \ref{thm:progress-1}---the initial state is valid and satisfies the run-time checks of some matching guard---as well as a dynamic execution step to $\Gamma'$. By theorem \ref{thm:progress-2}, we know that there is such a guard statement; i.e., we can always find the necessary set of run-time checks. Then preservation states that the resulting dynamic state $\Gamma'$ is also valid.

\begin{theorem}[Preservation] \label{thm:preservation}
  For some program $\prog$, let $\Gamma$ be some dynamic state validated by $\vstate$ and $V$. If $\sguard{\vstate}{\sstate}{\scheck}{\sperms}$, $V'$ is the corresponding valuation extending $V$, $V'(\pc(\sstate)) = \ktrue$, $\rtassert{V'}{\heap}{\perms(\Gamma)}{\scheck}$, and $\dtrans{\prog}{\vfoot{V'}{\heap(\Gamma)}{\sperms}}{\Gamma}{\Gamma'}$, then $\Gamma'$ is a valid state.
\end{theorem}

Note that our assumptions for preservation require dynamic execution to not only assert the run-time checks represented symbolically by $\scheck$, but also respect the exclusion frame represented symbolically by $\sperms$.
The necessity and implications of this requirement are discussed in \S\ref{sec:unsoundness}.

Taken together, these theorems demonstrate that dynamic execution will never be stuck as long as the run-time checks calculated by static verification succeed. Further, it shows that we calculate run-time checks for all possible execution paths. Since the dynamic execution semantics ensures all necessary specifications are satisfied, this implies that the calculated run-time checks are sufficient.

\section{Challenges to formalism of static verification}

Our specification of static verification in \S\ref{sec:precise} and \S\ref{sec:gradual} is formalised using non-deterministic inference rules. This differs greatly from the specifications of \citet{schwerhoff16silicon} and \citet{divincenzo2022gradual}, which both use a CPS-style definition for algorithms.  The latter form is useful when specifying an implementation, but makes it difficult to formulate a syntactic soundness proof.
Furthermore, operational semantics allow a higher level of abstraction than pseudo-code definitions. However, we must carefully consider whether our operational semantics represent the system which is implemented.

\subsection{Previous approaches}

During development of our soundness proof, we attempted several formulations of soundness. Initially we abstractly defined a \emph{symbolic stack}---a list of symbolic states with the form of a dynamic stack. This approach allowed us to easily state correspondence of the entire dynamic state---each dynamic stack frame models a corresponding symbolic stack frame.

We found it challenging, however, to prove that this correspondence is maintained.  When the dynamic stack takes a step, we must verify that there is a corresponding symbolic stack. To address this issue, we defined an execution semantics for symbolic stacks.  Unfortunately, this increased the distance between our formalism and implementation, and now we also need to show that all symbolic states are reachable during static verification.  Perhaps due to the complexity of this approach, the proof of correspondence remained quite difficult even after defining this execution semantics.

Instead, we defined a valid state primarily by the correspondence of the currently executing dynamic stack frame with some \textit{reachable} symbolic state---a symbolic state which is computed during static verification, with no input from dynamic execution. This resulted in a much simpler definition of \textit{valid state}.

However, in order to completely prove preservation, we also must specify the behavior of intermediate stack frames -- frames contained in the dynamic stack below the currently-executing frame. Thus we provide a recursive definition for a valid \textit{partial state}. For intermediate frames containing a method call that is waiting to complete, this requires the frame to model a symbolic state that results from consuming the callee method's pre-condition from a reachable verification state. We use this to prove that the dynamic state after the method returns models the symbolic state after symbolically executing the method call.

\subsection{Verification of loops}\label{sec:loops}

Almost all of our symbolic execution rules are finitely non-deterministic. That is, given an input state, there are a finite number of derivations that can apply. This is necessary since all possible states must be computed during static verification.

While this property matches the finite branching of symbolic execution, we make an exception in the case of loops---specifically, the \textsc{SVerifyLoop} rule (Figure \ref{fig:verify-rules}). It consumes the loop pre-condition, havocs all variables modified by the loop body (i.e., replaces them with $\ffresh$ values), and produces the loop post-condition. Thus it replaces all symbolic values that could be modified by the loop body with fresh values. The loop is left in place, which means that the rule can be immediately applied again to derive yet another state. However, this is harmless because repeated applications of this rule result in isomorphic symbolic states---states which represent the same state but with different symbolic values. Since the exact symbolic values do not matter, these are equivalent states from the perspective of static verification.
%
%
Therefore, even though we allow unbounded non-determinism, an implementation such as \gco can compute all possible states (as determined by our formal model) up to this equivalence. In other words, unbounded non-determinism is an artifact of our formalization that does not affect an implementation.



This exception is motivated by a disconnect between our formal model and the implementation of \gco \citep{divincenzo2022gradual}.
In our formalism, run-time checks are computed as symbolic values and lack a representation in terms of the source. Furthermore, we interpret these run-time checks by means of the valuation function, which we only extend with fresh values as dynamic execution proceeds. Therefore, the references in a run-time check are fixed -- for example, the validity of a check does not change when the heap is updated, since the heap reference has already been fully evaluated against the symbolic heap.

Consider the example in Figure \ref{fig:loop-example}. A new object reference $\ell_1$ is allocated by \ttt{create} at line \ref{ln:loop-create-initial}. Then we consume the loop pre-condition $\simprecise{\ktrue}$, which results in an imprecise state with empty symbolic heaps. Thus we cannot statically assert access to \ttt{x.value} in the loop body (line \ref{ln:loop-assign}). However, we optimistically assume access and produce a run-time check representing $\kacc(\ttt{x.value})$. In our formalism, \ttt{x} is symbolically evaluated to a symbolic value $\nu_1$ and the corresponding valuation contains the mapping $\nu_1 \mapsto \ell_1$. Thus the symbolic run-time check is $\pair{\nu_1}{\ttt{value}}$, which succeeds since the dynamic state owns $\pair{\ell_1}{\ttt{value}}$. This permission is then lost when \ttt{consume} is executed (line \ref{ln:loop-consume}), but a new reference $\ell_2$ is allocated at line \ref{ln:loop-create-while}, and the dynamic environment is updated with $\ttt{x} \mapsto \ell_2$.

During the next iteration of the loop, if we directly applied the same run-time check, this would again require the run-time check $\pair{\nu_1}{\ttt{value}}$. However, this would fail since the dynamic state no longer owns $\pair{\ell_1}{\ttt{value}}$. But the run-time check should reference $\ell_2$, since the check is intended to represent $\kacc(\ttt{x.value})$, and $\ttt{x} \mapsto \ell_2$ in the dynamic state.

This contrasts with the implementation of run-time checks in \gco, which translates the symbolic checks into source expressions. For the example described, \gco directly inserts the assertion $\kacc(\ttt{x}.\ttt{value})$. The expression \ttt{x.value} is then re-evaluated every time this assertion is checked.

\textsc{SVerifyLoop} fixes this mismatch by allowing our formal model to be updated with new symbolic values. With this rule, we can continue execution using a new symbolic state where we havoc \ttt{x}, since it is modified by the loop body, and consume the loop invariant $\simprecise{\ktrue}$ again. Thus we begin with a symbolic state with empty symbolic heaps and a symbolic store containing $\ttt{x} \mapsto \nu_2$, where $\nu_2$ is a fresh value. We define a new valuation $V'$ where new symbolic values are mapped to the current dynamic state, i.e. $V'(\nu_2) = \ell_2$ since $\ttt{x} \mapsto \ell_2$ in the dynamic state. The new symbolic state is isomorphic to the state used during the initial symbolic execution, since it also began execution of the body with empty symbolic heaps. We will again optimistically evaluate \ttt{x.value}, which produces a new run-time check for the symbolic permission $\pair{\nu_2}{\ttt{value}}$, thus we will check access to $\pair{\ell_2}{\ttt{value}}$, and therefore our run-time checks succeed.

Finally, since this rule introduces \emph{more} symbolic states in our formal model, this means that our soundness theorems are stronger than they would be otherwise; i.e., the soundness result holds for strictly more cases. As our example demonstrates, we want to consider these additional cases since they are already permitted by \gco due to its source-level run-time checks. Therefore, this additional rule allows us to abstract away the re-evaluation of source-level checks, allowing us to reason with fixed symbolic values.


\begin{figure}
\begin{minipage}[t]{.45\linewidth}
\begin{lstlisting}[name=loop-example]
Cell create()
requires true ensures ?
{ result = alloc(Cell); }

int consume(Cell c)
requires acc(c.value) ensures true
{ (*@ $\cdots$ @*) }
\end{lstlisting}
\end{minipage}
\begin{minipage}[t]{.45\linewidth}
\begin{lstlisting}[name=loop-example]
int main() {
  x = create();(*@\label{ln:loop-create-initial}@*)
  while (true) invariant ? * true {(*@\label{ln:loop-while-begin}@*)
    x.value = 1;(*@\label{ln:loop-assign}@*)
    consume(x);(*@\label{ln:loop-consume}@*)
    x = create();(*@\label{ln:loop-create-while}@*)
  }(*@\label{ln:loop-while-end}@*)
  result = 0;
}
\end{lstlisting}
\end{minipage}
   \caption{Example illustrating the neccessity of the \textsc{SVerifyLoop} rule.}
   \label{fig:loop-example}
\end{figure}

\section{Unsoundness of \gco}\label{sec:unsoundness}

While attempting to prove the soundness of \gco, we discovered that its implementation \citep{divincenzo2022gradual} allows unsound behavior, and have communicated this to the authors. This unsoundness results from the combination of imprecise specifications, static verification with isorecursive predicates, and run-time checking with equirecursive predicates.

\subsection{Example}

\begin{figure}
\begin{minipage}{.45\linewidth}
\begin{lstlisting}[name=gvc-unsound]
struct Cell { int value; }
predicate imprecise() = ? * true
void set(Cell c, int v)
  requires imprecise()
  ensures true
{
  unfold imprecise();(*@\label{ln:unsound-unfold}@*)
  c.value = v;(*@\label{ln:unsound-set-assign}@*)
}
\end{lstlisting}
\end{minipage}
\begin{minipage}{.45\linewidth}
\begin{lstlisting}[name=gvc-unsound]
int test()
  requires true
  ensures result == 0
{(*@\label{ln:unsound-test-begin}@*)
  fold imprecise();(*@\label{ln:unsound-fold}@*)
  Cell c = alloc(Cell);(*@\label{ln:unsound-alloc}@*)
  c.value = 0;(*@\label{ln:unsound-test-assign}@*)
  set(c, 1);(*@\label{ln:unsound-test-set}@*)
  result = c.value;(*@\label{ln:unsound-test-result}@*)
}(*@\label{ln:unsound-test-end}@*)
\end{lstlisting}
\end{minipage}
\caption[Example exhibiting unsoundness of \citet{divincenzo2022gradual}]
{Example exhibiting unsoundness of \citet{divincenzo2022gradual}.\protect\footnotemark}
\label{fig:unsound-example}
\end{figure}

We show an example which exhibits this behavior in Figure \ref{fig:unsound-example}. At line \ref{ln:unsound-fold} $\ttt{imprecise}()$ is folded, thus $\simprecise{\ktrue}$ is consumed, and the predicate chunk is added to the symbolic heap. At this point $\sheap = \set{\langle \ttt{imprecise} \rangle}$. In lines \ref{ln:unsound-alloc}-\ref{ln:unsound-test-assign} a new \ttt{Cell} is allocated and its value is initialized to $0$. At this point $\sheap = \set{\langle \ttt{imprecise} \rangle, \triple{\ttt{value}}{t_1}{0}}$, where $\ttt{c} \mapsto t_1$.

\footnotetext{\ttt{void} methods are used for clarity since they can be trivially translated to the formally defined grammar.}
At line \ref{ln:unsound-test-set}, the \ttt{set} method is called. Thus the precondition---\ttt{imprecise()}---is consumed, resulting in $\sheap = \set{\triple{\ttt{value}}{t_1}{0}}$. The postcondition---$\ktrue$---is then produced, which does not change the symbolic state. In this symbolic state $\ttt{c.value} \mapsto 0$. Then line \ref{ln:unsound-test-result} adds the mapping $\ttt{result} \mapsto 0$ to the symbolic store, which allows the postcondition $\kresult \keq 0$ to be consumed successfully. Now \ttt{test} is valid and no run-time checks are required in its body.  Symbolic execution of the \ttt{set} method shows that this method is also valid but requires a check representing $\kacc(\ttt{c.value})$ at line \ref{ln:unsound-test-assign}.

Now we consider dynamic execution of the \ttt{test} method. We first use no exclusion frame (i.e. using $\emptyset$ for every occurrence of $\xperms$ in the rules).

The fold at line \ref{ln:unsound-fold} is ignored, a new \ttt{Cell} is allocated and initialized at lines \ref{ln:unsound-alloc}-\ref{ln:unsound-test-assign}, and the \ttt{set} method is called at line \ref{ln:unsound-test-set}. The formula \ttt{imprecise()} is not completely precise, therefore $\foot{\heap}{\perms}{\env}{\ttt{imprecise}()} = \perms$. Thus all of the caller's owned permissions are passed to \ttt{set}. The assertion for \ttt{imprecise()} succeeds since its equirecursive unrolling is simply $\simprecise{\ktrue}$. Likewise, the assertion for $\kacc(\ttt{c.value})$ when executing line \ref{ln:unsound-set-assign} also succeeds since the required permissions were passed from \ttt{test}. After returning from \ttt{set}, $\ttt{c.value} \mapsto 1$ in the dynamic state, thus $\ttt{result} \mapsto 1$ after executing line \ref{ln:unsound-test-result}. However, the postcondition $\ttt{result} \keq 0$ cannot be asserted, therefore execution is stuck.

Since \citet{divincenzo2022gradual} does not implement an exclusion frame, execution proceeds as described above, except that only the calculated run-time checks are asserted. Therefore the \ttt{test} method returns $1$, which contradicts its contract. \citet{wise2020gradual} follows the dynamic execution behavior described above, but since it checks every assertion at run-time, execution halts and soundness is preserved.

\subsection{Diagnosis}

%
%
%
As described above, the caller's permissions are passed to \ttt{set}, thus the set of permissions owned by \ttt{test} is empty during execution of \ttt{set}. But we calculated that after consuming $\fpre(\ttt{set})$ the symbolic heap still contains the field chunk representing $\ttt{c}.\ttt{value}$. Therefore heap chunks which are included in the frame of \ttt{set} during dynamic execution are not removed by consume during symbolic execution, thus symbolic execution does not accurately represent dynamic execution.

\subsection{Possible solutions}

At first this appears to be an error of static verification, and thus we could address this by making static verification more conservative.
More specifically, we could require a stronger invariant of the precise heap: the \textit{maximal} footprint represented by predicate chunks cannot overlap. This contrasts with our current definition, where the exact footprint represented by a predicate chunk must be disjoint from that of all other predicate chunks.

For example, we could clear the symbolic heaps when consuming any formula that is not \textit{completely} precise (i.e., the recursive unfolding contains an imprecise formula). When this occurs, we would also need to use an imprecise state, so that the existence of the removed permissions can be optimistically assumed. This would result in empty symbolic heap after line \ref{ln:unsound-test-set} in Figure \ref{fig:unsound-example}, and a run-time check for the value of $\kresult$ would be required before returning from \ttt{test}, thus soundness is preserved.

This would allow maximal footprints of predicate chunks to overlap in the symbolic heap, but when consuming a predicate instance, all potentially overlapping predicate chunks would be removed. Thus, after some predicate instance is consumed, its maximal footprint would not overlap with any permission represented by a heap chunk contained in the symbolic heap.

Alternatively, we could achieve soundness by removing any predicate instance that is not completely precise when additional permissions are added to the precise heap. Similar to the previous option, we would also need to use an imprecise state when this occurs. In the example, that would (perhaps unintuitively) \emph{remove} the \ttt{imprecise()} predicate when adding permissions for the \ttt{alloc} statement. This would ensure that the maximal footprint of heap chunks in the symbolic heap never overlap.

Unfortunately, both of these options reduce the number of assertions that can be statically discharged when verifying gradual programs, thus more run-time checks would be necessary. Furthermore, the run-time checks require checking a predicate instance, which can be quite costly since this traverses the entire unfolding of the predicate.

Furthermore, allowing the predicate instance folded at line \ref{ln:unsound-fold} to affect permissions allocated afterward, at line \ref{ln:unsound-alloc}, seems counter-intuitive. This invalidates the intuitive assumption that the set of permissions represented by a folded predicate instqance will not change while it remains folded. Furthermore this behavior breaks the semantics of \ttt{?}, as specified in \citet{wise2020gradual}, since no logically consistent strengthening of the \ttt{imprecise} predicate allows it to include permissions allocated after its body is folded.

This indicates that the semantics of dynamic execution should be modified to exclude access permission for \ttt{c.value}, which is allocated after \ttt{imprecise()} is folded, from being passed to \ttt{set}, which is a precise formula that only requires \ttt{imprecise()}. Then execution would fail at line \ref{ln:unsound-set-assign} in Figure \ref{fig:unsound-example}. To accomplish this, we have introduced the concept of an \emph{exclusion frame} -- a set of permissions which cannot be passed to a callee. This exclusion frame is calculated by symbolic execution, and passed to dynamic execution in much the same way as run-time checks. It is represented by $\sperms$ in the guard judgement (\S\ref{fig:gradual-guard-rules}), which also calculates $\scheck$, and is translated to dynamic permissions using a valuation.

The guard rules in Figure \ref{fig:gradual-guard-rules} calculate the exclusion frame by the $\frem$ helper function, after consuming the pre-condition of a method. If the pre-condition is completely precise, then $\sperms = \emptyset$, thus execution of an \svl program is not affected. Otherwise, $\sperms$ contains all permissions currently contained in the symbolic heaps. In Figure \ref{fig:unsound-example}, since the pre-condition of \ttt{set} is not completely precise, $\sperms = \set{ \pair{t_1}{\ttt{value}} }$ when calculating the guard statement at line \ref{ln:unsound-test-set}. At run-time this is translated to $\xperms = \set{ \pair{\ell}{\ttt{value}} }$ where $\ttt{c} \mapsto \ell$. Then all permissions \textit{except} $\pair{\ell}{\ttt{value}}$ are passed to \ttt{set}. Thus the run-time check for $\ttt{acc(c.value)}$ at line \ref{ln:unsound-set-assign} cannot be asserted.

This addresses the intuitive and semantic problems described above. The isorecursive instance of \ttt{imprecise} referenced in the pre-condition of \ttt{set} \textit{should not} represent access to $\ttt{c.value}$ since it was folded before \ttt{c} was allocated. Under this interpretation we would expect a failure at line \ref{ln:unsound-set-assign}, since \ttt{set} does not require the necessary permissions. This also matches the semantics of \ttt{?}, as defined in \cite{wise2020gradual}, since the predicate instance folded at line \ref{ln:unsound-fold} cannot consistently imply access to the heap location allocated at line \ref{ln:unsound-alloc}.


\subsection{Implementation}

There are important implementation challenges that must be addressed before this change can be implemented in \gco \citep{divincenzo2022gradual}. Currently, \gco constructs sets of permissions at run-time---before calling a method, for example---by recursively unfolding the neccessary specification and collecting all permissions. However, this method cannot be used to create the exclusion frame, since these permissions are not necessarily represented by a specification. But we expect that a translation algorithm can be developed which generates the source code necessary to compute the exclusion frame at run time. This is similar to the existing translation algorithm described by \citet{divincenzo2022gradual}, which translates symbolic run-time checks into source code that implements the desired assertion.

Also, note that we calculate the exclusion frame using information from symbolic execution of a particular statement. In other words, if method \ttt{m} calls $\ttt{m}'$, we can calculate the exclusion frame necessary for calling $\ttt{m}'$ without considering the exclusion frame used to call \ttt{m}. This implies that exclusion frames can be dropped when entering a completely precise method, and then instantiated again when a precise method calls an imprecise method. This is similar to how \gco does not pass permission sets to precise methods, but reconstructs the permissions when a precise method calls an imprecise methods. Applying this technique to exclusion frames, as described, would ensure that exclusion frames do not affect the run-time performance of methods that are specified with completely precise specifications.



\section{Future work}\label{sec:future-work}

There are many possible directions in which this work can be extended.
%
%
We have not yet proven the gradual guarantees for gradual verification, as formalized in \citet{wise2020gradual}. These guarantees formalize the notion that, given a valid program, gradual specifications may be used in place of all static specifications without introducing errors (both during verification and at run time). This ensures that any errors do not arise from imprecision, but rather from an invalid program or specification, or (for precise specifications) from incompleteness of verification. Our formalization appears to satisfy this since, as described in \S\ref{sec:gradual}, we extend the underlying static verification algorithm mainly by adding optimistic capabilities while leaving the bulk of static verification intact. However, we have not completed a formal proof.


Our formalization could also be used to extend gradual verification. Notably, gradual verification has not been implemented for quantified specifications or concurrent programs. Ghost code/parameters (i.e., code only necessary for supporting logical proofs) is also not supported in gradual verification, since the ``ghost'' code could be necessary for run-time checks. Our high-level definition of the gradual verifier could enable further development to support these techniques. Likewise, our formalization does not capture several important concepts in Viper such as domains, fractional permissions, and joining of symbolic execution paths. Formalizing the usage of these techniques in Viper and proving their soundness would provide further assurance of the correctness of Viper and provide a starting point for integrating these techniques with gradual verification. Our formalization provides a basis for formally proving properties of verification techniques (in our case, gradual verification) with a model that closely resembles the implementation (in our case, \gco). Thus modifications to our formal model can be more easily implemented and used, while modifications to the implementation can be reflected in the formal model and proven sound.


\section{Related work}

As mentioned previously, implementations of verification using symbolic execution, such as Viper \cite{schwerhoff16silicon}, \gco \cite{divincenzo2022gradual}, Smallfoot \cite{berdine2006smallfoot}, Chalice \cite{leino2009chalice}, and jStar \cite{distefano2008jstar}, often lack formal soundness proofs. A notable exception is VeriFast \cite{jacobs2011verifast}, which implements verification using symbolic execution. The core of its verifier was proven sound in \citet{jacobs2015featherweight}. This soundness proof utilizes techniques from abstract interpretation, which may simplify proofs of verifiers using symbolic execution. However, VeriFast uses separation logic instead of IDF.

Several previous verifiers using WLP or verification condition generation (VCG) have been directly proven sound \citep{vogels2009machine, vogels2010machine, herms2012certified, smans2012implicit}. Several similar verifiers produce a proof during verification which may be checked to the validate soundness of an individual verification result \citep{filliatre2013why3, parthasarathy2021formally}.

Viper \citep{viper16} and \gco \citep{divincenzo2022gradual} rely on an SMT solver to implement their verification algorithms. While we have proved soundness of our formal model, this soundness is contingent on the soundness of the SMT solver. Other work has extended soundness to include soundness of the entire verification system. Notably, VeriSmall \citep{verismall22}, Diaframe \citep{diaframe22}, and RefinedC \citep{refinedc21} are all encoded in Iris/Coq, making them either foundational or self-verifying.

As described before, soundness of gradual verification based on WLP has been proven in both \citet{wise2020gradual} and \citet{bader2018gradual}. However, \citet{wise2020gradual} depends on dynamically checking all assertions, while \citet{bader2018gradual} does not handle abstract heap predicates.

\section{Conclusion}

The recent implementation of gradual verification in \citet{divincenzo2022gradual} promises a dramatic reduction in the effort required to verify programs. However, this requires confidence in the correctness of their gradual verification system, \gco, as well as its underlying static verification system, Viper. In this work, we formalized symbolic execution in (a subset of) Viper and proved it sound, in addition to formalizing gradual verification in \gco and proving it sound. During this work we found a soundness bug in \gco, which we communicated to \citet{divincenzo2022gradual} along with possible solutions. This illustrates that, while correctness in gradual verifiers can be guaranteed, it should not be assumed without rigorous proof. There are a few interesting directions we could take this work: (1) proving that \gco adheres to the gradual guarantee as formalized by \citet{wise2020gradual}, which is a very important property of gradual verifiers that should be straightforward to prove with our formal system, and (2) using our formalism to explore new directions in gradual verification like quantification or concurrency, and prove systems utilizing them sound. In general, we hope that this work serves as a strong basis for future proof work in static and gradual verification when using symbolic execution.

\begin{acks}
  We thank Jana Dunfield for her helpful feedback.

  This work was supported by the
  \grantsponsor{NSF}{National Science Foundation}{https://www.nsf.gov/}
  under Grant No. \grantnum[https://www.nsf.gov/awardsearch/showAward?AWD_ID=1901033]{NSF}{CCF-1901033}
  and a \grantsponsor{Google}{Google PhD Fellowship}{https://research.google/outreach/phd-fellowship/}.
\end{acks}

\clearpage

\bibliography{references}
\clearpage

\appendix

\renewcommand*\contentsname{Appendices}
\tableofcontents

\addtocontents{toc}{\protect\setcounter{tocdepth}{2}}
\setcounter{theorem}{0}
\let\SavedTheHtheorem=\theHtheorem
\def\theHtheorem{proofs.\SavedTheHtheorem}

\section{Grammar}\label{sec:grammar}


\begin{grammar}
  \firstcase
    {\gprogram}
    {\multiple{\gstruct} ~ \multiple{\gpredicate} ~ \multiple{\gmethod} ~ \gstatement}
    {Program definition}

  \firstcase
    {\gstruct}
    {S~ \sblock{\multiple{\gtype ~ f}}}
    {Struct definition}

  \firstcase
    {\gpredicate}
    {p({\multiple{\gtype ~\gvar}}) = \gform}
    {Predicate definition}

  \firstcase
    {\gmethod}
    {\smethdef{T}{m}{\multiple{T ~ x}}{\gcontract}{s}}
    {Method definition}
  
  \firstcase
    {\gcontract}
    {\krequires ~ \gform ~ \kensures ~ \gform}
    {Method contract}

  \firstcase
    {\gtype}
    {S \gralt \kint \gralt \kbool \gralt \kchar}
    {Type}

  \firstcase
    {\gstatement}
    {\sseq{\gstatement}{\gstatement}}
    {Statement sequence}
  \otherform
    {\kskip}
    {No-op}
  \otherform
    {x \kassign e}
    {Variable assignment}
  \otherform
    {x.f \kassign e}
    {Field assignment}
  \otherform
    {x \kassign \salloc{S}}
    {Allocation}
  \otherform
    {x \kassign m({\multiple{e}})}
    {Method invocation}
  \otherform
    {\sassert{\gform}}
    {Assertion}
  \otherform
    {\sif{e}{s}{s}}
    {Conditional}
  \otherform
    {\swhile{e}{\gform}{s}}
    {Loop}
  \otherform
    {\sfold{p(\multiple{e})}}
    {Fold predicate}
  \otherform
    {\sunfold{p(\multiple{e})}}
    {Unfold predicate}

  \firstcase
    {\gexpression}
    {l \gralt x \gralt e.f \gralt e \oplus e}
    {Expression}
  \otherform
    {e \kor e \gralt e \kand e \gralt \kneg e}
    {}

  \firstcase
    {\gvar}
    {\kresult \gralt id}
    {Variable}

  \firstcase
    {l}
    {n \gralt c}
    {Value}
  \otherform{\knull \gralt \ktrue \gralt \kfalse}{}

  \firstcase
    {\gform}
    {\simprecise{\phi} \gralt \phi}
    {Gradual formula}

  \firstcase
    {\phi}
    {\phi * \phi \gralt p(\multiple{e}) \gralt e}
    {Precise formula}
    \otherform
    {\sif{e}{\phi}{\phi}}
    {}
  \otherform
    {\kacc(e.f)}
    {}
\end{grammar}

Where $n \in \mathbb{Z}$, $c \in \textsc{Char}$, $id \in \textsc{Identifier}$, $f \in \Field$, $m \in \Method$, $p \in \Predicate$, $\oplus \in \{ +, -, <, >, \le, \ge, = \}$

\begin{definition}\label{def:well-formed-prog}
  A program is \textbf{well-formed} if all the following requirements are satisified:
  \begin{itemize}
    \item It is properly typed.
    \item All loop invariants, method pre-conditions, and method post-conditions are specifications (definition \ref{def:specification}).
    \item The free variables of any method are a subset of its parameters.
    \item The special variable $\kresult$ is not a free variable of any method.
    \item No parameters appear on the left side of a variable assignment.
    \item Formulas in pre-conditions only reference parameters.
    \item Formulas in post-conditions only reference parameters and the special variable $\kresult$.
    \item If a pre-condition is imprecise, the post-condition is also imprecise.
  \end{itemize}
\end{definition}

\section{Run-time Semantics}


\subsection{Definitions}

The rules in the following section reference an ambient program with elements denoted as follows:
\begin{itemize}
  \item Predicates: $p \in \Predicate$
  \item Methods: $m \in \Method$
  \item Structs: $S \in \Struct$
  \item Types: $T \in \Type$
  \item Variables: $x \in \Var$
  \item Field identifiers: $f \in \Field$
  \item Locations (opaque values): $\ell \in \Location$
  \item Literals (integers, characters, booleans, null): $l \in \Literal$
  \item Values: $v \in \Value = \Location \cup \Literal$
  \item Gradual formulas: $\gform \in \GFormula$
  \item Precise formulas: $\phi \in \Formula$
  \item Statements: $s \in \Stmt$
  \item Heap: $\heap : \Location \times \Field \to \Value$
  \item Permissions: $\perms \in \powerset{\Location \times \Field}$
  \item Environment: $\env : \Var \pfunc \Value$
\end{itemize}

The following functions are defined to access elements in the program:
\begin{itemize}
  \item $\fdefault : \Type \to \Value$ -- Gets the default value of the given type ($0$, $\knull$, etc.)
  \item $\fpre : \Method \to \GFormula$ -- Gets the precondition from the declaration of the specified method.
  \item $\fpost : \Method \to \GFormula$ -- Gets the postcondition from the declaration of the specified method.
  \item $\fbody : \Method \to \Stmt$ -- Gets the body from the declaration of the specified method.
  \item $\fparams : \Method \to \multiple{\Var}$ -- Gets the list of parameters from the declaration of the specified predicate.
  \item $\fpred : \Predicate \to \GFormula$ -- Gets the body from the declaration of the specified predicate.
  \item $\fpredparams : \Predicate \to \multiple{\Var}$ -- Gets the list of parameters from the declaration of the specified predicate.
  \item $\fstruct : \Struct \to \multiple{\Field}$ -- Gets the list of fields from the declaration of the specified struct.
\end{itemize}

\subsection{Evaluation}

The relation
$$\eval{\heap}{\env}{e}{v}$$
denotes the evaluation of an expression $e \in \Expr$ to a value $v \in \Value$
where $\heap : \Value \times \Field \to \Value$ represents the heap, and $\env : \Var \pfunc \Value$ represents the local variable environment.

\semantics[EvalLiteral]
  { }
  {\eval{\heap}{\env}{l}{l}}
\semantics[EvalVar]
  { }
  {\eval{\heap}{\env}{x}{\env(x)}}
\semantics[EvalAndA]
  {\eval{\heap}{\env}{e_1}{\kfalse}}
  {\eval{\heap}{\env}{e_1 \kand e_2}{\kfalse}}
\semantics[EvalAndB]
  {\eval{\heap}{\env}{e_1}{\ktrue} \\ \eval{\heap}{\env}{e_2}{v_2}}
  {\eval{\heap}{\env}{e_1 \kand e_2}{v_2}}
\semantics[EvalOrA]
  {\eval{\heap}{\env}{e_1}{\ktrue}}
  {\eval{\heap}{\env}{e_1 \kor e_2}{\ktrue}}
\semantics[EvalOrB]
  {\eval{\heap}{\env}{e_1}{\kfalse} \\ \eval{\heap}{\env}{e_2}{v_2}}
  {\eval{\heap}{\env}{e_1 \kor e_2}{v_2}}
\semantics[EvalOp]
  {\eval{\heap}{\env}{e_1}{v_1} \\ \eval{\heap}{\env}{e_2}{v_2}}
  {\eval{\heap}{\env}{e_1 \oplus e_2}{v_1 \oplus v_2}}
\semantics[EvalNeg]
  {\eval{\heap}{\env}{e}{v}}
  {\eval{\heap}{\env}{\kneg e}{\neg v}}
\semantics[EvalField]
  {\eval{\heap}{\env}{e}{\ell}}
  {\eval{\heap}{\env}{e.f}{\heap(\ell, f)}}

\subsection{Formulas}

$\Formula$ is the set of all $\gpreciseform$ elements in the grammar, while $\GFormula$ is the set of all $\gform$ elements in the grammar.

\begin{definition}\label{def:precise}
  An \textbf{imprecise} formula $\gform$ is any formula in $\GFormula$ of the form $\simprecise{\phi}$ where $\phi \in \Formula$.
  
  Otherwise, a formula $\phi$ is \textbf{precise} and $\phi \in \Formula$.
\end{definition}

\begin{definition}\label{def:precise-formula}
  A formula $\gform$ is \textbf{completely precise} if there is no $\heap, \perms, \env$ such that \textnormal{\refrule{AssertImprecise}} applies at some step in the derivation of $\assertion{\heap}{\perms}{\env}{\gform}$.

  In other words, a completely precise formula is precise and all predicate bodies referenced in its equi-recursive unrolling are also precise.
\end{definition}

\begin{definition}\label{def:specification}
  A formula $\gform$ is a \textbf{specification} if either $\gform$ is imprecise or $\gform$ is precise and self-framed (definition \ref{def:self-framed}).
\end{definition}

\subsection{Footprints}

\begin{definition}\label{def:efoot}
  The \textbf{exact footprint} of a formula $\gform \in \GFormula$, denoted $\efoot{\heap}{\env}{\gform}$, or of an expression $e$, denoted $\efoot{\heap}{\env}{e}$, is the set of permissions that must be accessed when asserting $\gform$ or evaluating $e$.

  By lemmas \ref{lem:efoot-subset-spec} and \ref{lem:efoot-assert}, if $\gform$ is a specification, this set is the lower bound of permissions that satisfy $\gform$.
\end{definition}

The calculation of exact footprints is defined as follows:
\begin{align*}
  \efoot{\heap}{\env}{l} &:= \emptyset \\
  \efoot{\heap}{\env}{x} &:= \emptyset \\
  \efoot{\heap}{\env}{e.f} &:= \efoot{\heap}{\env}{e}; \pair{\ell}{f} &\text{if }\eval{\heap}{\env}{e}{\ell} \\
  \efoot{\heap}{\env}{e_1 \oplus e_2} &:= \efoot{\heap}{\env}{e_1} \cup \efoot{\heap}{\env}{e_2} \\
  \efoot{\heap}{\env}{e_1 \kor e_2} &:= \efoot{\heap}{\env}{e_1} &\text{if }\eval{\heap}{\env}{e_1}{\ktrue} \\
  \efoot{\heap}{\env}{e_1 \kor e_2} &:= \efoot{\heap}{\env}{e_1} \cup \efoot{\heap}{\env}{e_2} &\text{if }\eval{\heap}{\env}{e_1}{\kfalse} \\
  \efoot{\heap}{\env}{e_1 \kand e_2} &:= \efoot{\heap}{\env}{e_1} &\text{if }\eval{\heap}{\env}{e_1}{\kfalse} \\
  \efoot{\heap}{\env}{e_1 \kand e_2} &:= \efoot{\heap}{\env}{e_1} \cup \efoot{\heap}{\env}{e_2} &\text{if }\eval{\heap}{\env}{e_1}{\ktrue} \\
  \efoot{\heap}{\env}{\kneg e} &:= \efoot{\heap}{\env}{e} \\
  \efoot{\heap}{\env}{\simprecise{\phi}} &:= \efoot{\heap}{\env}{\phi} \\
  \efoot{\heap}{\env}{\phi_1 * \phi_2} &:= \efoot{\heap}{\env}{\phi_1} \cup \efoot{\heap}{\env}{\phi_2} \\
  \efoot{\heap}{\env}{p(\multiple{e})} &:= \efoot{\heap}{[\multiple{x \mapsto v}]}{\fpred(p)} ~\cup &\text{if } \multiple{x} = \fpredparams(p) \\
  &\quad\quad \bigcup \multiple{\efoot{\heap}{\env}{e}} &\text{and }\multiple{\eval{\heap}{\env}{e}{v}} \\
  \efoot{\heap}{\env}{\sif{e}{\phi_1}{\phi_2}} &:= \efoot{\heap}{\env}{e} \cup \efoot{\heap}{\env}{\phi_1} &\text{if }\eval{\heap}{\env}{e_1}{\ktrue} \\
  \efoot{\heap}{\env}{\sif{e}{\phi_1}{\phi_2}} &:= \efoot{\heap}{\env}{e} \cup \efoot{\heap}{\env}{\phi_2} &\text{if }\eval{\heap}{\env}{e_1}{\kfalse} \\
  \efoot{\heap}{\env}{\kacc(e.f)} &:= \efoot{\heap}{\env}{e}; \pair{\ell}{f} &\text{if }\eval{\heap}{\env}{e}{\ell}
\end{align*}

\begin{definition}\label{def:footprint}
  The \textbf{maximal footprint} of a formula, denoted $\foot{\heap}{\perms}{\env}{\gform}$, is the set of all permissions that $\gform$ may represent in the context of a heap $\heap$, permission set $\perms$, and variable environment $\env$.

  The footprint of a completely precise formula (definition \ref{def:precise-formula}) is its exact footprint, while the footprint of a formula which is not completely precise is the current set of permissions.

  \begin{equation*}
    \foot{\heap}{\perms}{\env}{\gform} := \begin{cases}
      \efoot{\heap}{\env}{\gform} & \text{if $\gform$ is completely precise} \\
      \perms & \text{otherwise}
    \end{cases}
  \end{equation*}
\end{definition}

\subsection{Framing}\label{sec:framing}

The relation $\frm{\heap}{\perms}{\env}{e}$ denotes that $e \in \Expr$ is framed by the permissions contained in $\perms \in \powerset{\Perm}$.

\semantics[FrameLiteral]
  { }
  {\frm{\heap}{\perms}{\env}{l}}
\semantics[FrameVar]
  { }
  {\frm{\heap}{\perms}{\env}{x}}
\semantics[FrameField]
  {
    \frm{\heap}{\perms}{\env}{e} \\
    \assertion{\heap}{\perms}{\env}{\kacc(e.f)}
  }
  {
    \frm{\heap}{\perms}{\env}{e.f}
  }
\semantics[FrameOp]
  {
    \frm{\heap}{\perms}{\env}{e_1} \\
    \frm{\heap}{\perms}{\env}{e_2}
  }
  {
    \frm{\heap}{\perms}{\env}{e_1 \oplus e_2}
  }
\semantics[FrameOrA]
  {
    \eval{\heap}{\env}{e_1}{\ktrue} \\
    \frm{\heap}{\perms}{\env}{e_1}
  }
  {
    \frm{\heap}{\perms}{\env}{e_1 \kor e_2}
  }
\semantics[FrameOrB]
  {
    \eval{\heap}{\env}{e_1}{\kfalse} \\
    \frm{\heap}{\perms}{\env}{e_1} \\
    \frm{\heap}{\perms}{\env}{e_2}
  }
  {
    \frm{\heap}{\perms}{\env}{e_1 \kor e_2}
  }
\semantics[FrameAndA]
  {
    \eval{\heap}{\env}{e_1}{\kfalse} \\
    \frm{\heap}{\perms}{\env}{e_1}
  }
  {
    \frm{\heap}{\perms}{\env}{e_1 \kand e_2}
  }
\semantics[FrameAndB]
  {
    \eval{\heap}{\env}{e_1}{\ktrue} \\
    \frm{\heap}{\perms}{\env}{e_1} \\
    \frm{\heap}{\perms}{\env}{e_2}
  }
  {
    \frm{\heap}{\perms}{\env}{e_1 \kand e_2}
  }
\semantics[FrameNeg]
  {\frm{\heap}{\perms}{\env}{e}}
  {\frm{\heap}{\perms}{\env}{\kneg e}}

The relation $\ifrm{\heap}{\perms}{\env}{\phi}$ denotes that $\phi \in \Formula$ is framed by the permissions in $\perms \in \powerset{\Perm}$ using an iso-recursive interpretation of predicates (i.e., without unrolling predicate instances).

\semantics[IFrameExpression]
  {
    \frm{\heap}{\perms}{\env}{e}
  }
  {
    \ifrm{\heap}{\perms}{\env}{e}
  }
\semantics[IFrameConjunction]
  {
    \ifrm{\heap}{\perms}{\env}{\phi_1} \\
    \ifrm{\heap}{\perms}{\env}{\phi_2}
  }
  {
    \ifrm{\heap}{\perms}{\env}{\phi_1 * \phi_2}
  }
\semantics[IFramePredicate]
  {
    \multiple{\frm{\heap}{\perms}{\env}{e}}
  }
  {
    \ifrm{\heap}{\perms}{\env}{p(\multiple{e})}
  }
\semantics[IFrameConditionalA]
  {
    \eval{\heap}{\env}{e}{\ktrue} \\
    \frm{\heap}{\perms}{\env}{e} \\
    \ifrm{\heap}{\perms}{\env}{\phi_1}
  }
  {
    \ifrm{\heap}{\perms}{\env}{\sif{e}{\phi_1}{\phi_2}}
  }
\semantics[IFrameConditionalB]
  {
    \eval{\heap}{\env}{e}{\kfalse} \\
    \frm{\heap}{\perms}{\env}{e} \\
    \ifrm{\heap}{\perms}{\env}{\phi_2}
  }
  {
    \ifrm{\heap}{\perms}{\env}{\sif{e}{\phi_1}{\phi_2}}
  }
\semantics[IFrameAcc]
  {
    \frm{\heap}{\perms}{\env}{e}
  }
  {
    \ifrm{\heap}{\perms}{\env}{\kacc(e.f)}
  }

Define the relation $\efrm{\heap}{\perms}{\env}{\phi}$  denotes that $\phi \in \Formula$ is framed by the permissions in $\perms \in \powerset{\Perm}$ using an equi-recursive interpretation of predicates (i.e., unrolling predicate instances).

\semantics[EFrameExpression]
  {
    \frm{\heap}{\perms}{\env}{e}
  }
  {
    \efrm{\heap}{\perms}{\env}{e}
  }
\semantics[EFrameConjunction]
  {
    \efrm{\heap}{\perms}{\env}{\phi_1} \\
    \efrm{\heap}{\perms}{\env}{\phi_2}
  }
  {
    \efrm{\heap}{\perms}{\env}{\phi_1 * \phi_2}
  }
\semantics[EFramePredicate]
  {
    \multiple{\frm{\heap}{\perms}{\env}{e}} \\
    \multiple{\eval{\heap}{\env}{e}{v}} \\
    \multiple{x} = \fpredparams(p) \\
    \efrm{\heap}{\perms}{[\multiple{x \mapsto v}]}{\fpred(p)}
  }
  {
    \efrm{\heap}{\perms}{\env}{p(\multiple{e})}
  }
\semantics[EFrameConditionalA]
  {
    \eval{\heap}{\env}{e}{\ktrue} \\
    \frm{\heap}{\perms}{\env}{e} \\
    \efrm{\heap}{\perms}{\env}{\phi_1}
  }
  {
    \efrm{\heap}{\perms}{\env}{\sif{e}{\phi_1}{\phi_2}}
  }
\semantics[EFrameConditionalB]
  {
    \eval{\heap}{\env}{e}{\kfalse} \\
    \frm{\heap}{\perms}{\env}{e} \\
    \efrm{\heap}{\perms}{\env}{\phi_2}
  }
  {
    \efrm{\heap}{\perms}{\env}{\sif{e}{\phi_1}{\phi_2}}
  }
\semantics[EFrameAcc]
  {
    \frm{\heap}{\perms}{\env}{e}
  }
  {
    \efrm{\heap}{\perms}{\env}{\kacc(e.f)}
  }

\begin{definition}\label{def:self-framed}
  A \textbf{self-framed} formula is a precise formula $\phi \in \Formula$ such that for all $\heap, \perms, \env$,
  $$\assertion{\heap}{\perms}{\env}{\phi} \implies \ifrm{\heap}{\perms}{\env}{\phi}.$$
\end{definition}

\subsection{Assertions}

The relation $\assertion{\heap}{\perms}{\env}{\gform}$ denotes the validity of $\gform \in \GFormula$ for the state represented by $\triple{\heap}{\perms}{\env}$.

\semantics[AssertImprecise]
  {
    \assertion{\heap}{\perms}{\env}{\phi} \\
    \efrm{\heap}{\perms}{\env}{\phi}
  }
  {\assertion{\heap}{\perms}{\env}{\simprecise{\phi}}}
\semantics[AssertValue]
  {
    \eval{\heap}{\env}{e}{\ktrue}
  }
  {
    \assertion{\heap}{\perms}{\env}{\gexpression}
  }
\semantics[AssertIfA]
  {
    \eval{\heap}{\env}{e}{\ktrue} \\
    \assertion{\heap}{\perms}{\env}{\phi_1}
  }
  {
    \assertion{\heap}{\perms}{\env}{\sif{e}{\phi_1}{\phi_2}}
  }
\semantics[AssertIfB]
  {
    \eval{\heap}{\env}{e}{\kfalse} \\
    \assertion{\heap}{\perms}{\env}{\phi_2}
  }
  {
    \assertion{\heap}{\perms}{\env}{\sif{e}{\phi_1}{\phi_2}}
  }
\semantics[AssertAcc]
  {
    \eval{\heap}{\env}{e}{\ell} \\
    \pair{\ell}{f} \in \perms
  }
  {
    \assertion{\heap}{\perms}{\env}{\kacc(e.f)}
  }
\semantics[AssertConjunction]
  {
    \assertion{\heap}{\perms_1}{\env}{\phi_1} \\
    \assertion{\heap}{\perms_2}{\env}{\phi_2} \\
    \perms_1 \cup \perms_2 \subseteq \perms \\
    \perms_1 \cap \perms_2 = \emptyset
  }
  {
    \assertion{\heap}{\perms}{\env}{\phi_1 * \phi_2}
  }
\semantics[AssertPredicate]
  {
    \multiple{x} = \fpredparams(p) \\
    \multiple{\eval{\heap}{\env}{e}{v}} \\
    \assertion{\heap}{\perms}{[\multiple{x \mapsto v}]}{\fpred(p)}
  }
  {
    \assertion{\heap}{\perms}{\env}{p(\multiple{e})}
  }

\subsection{Execution}\label{sec:exec-rules}

\begin{definition}
  A \textbf{stack} $\stack$ is a list of the form
  $$\triple{\perms_n}{\env_n}{s_n} \cdot \ldots \cdot \triple{\perms_1}{\env_1}{s_1} \cdot \nilsym$$
  where $n\ge 1$, $\perms_n, \cdots, \perms_1$ are permission sets, $\env_n, \cdots, \env_1$ are variable environments, and $s_n, \cdots, s_1$ are statements.

  $\perms(\stack)$, $\env(\stack)$, and $s(\stack)$ may be used to denote the values $\perms_n$, $\env_n$, and $s_n$, respectively.
\end{definition}

\begin{definition}
  An \textbf{exclusion frame} $\xperms$ a set of permissions that may not be transferred to a callee stack frame. This is necessary to ensure that the permissions represented by the imprecise specifications of a callee cannot overlap with some predicate instance that is owned by the caller.
\end{definition}

Small-step execution is denoted by the judgement
$$\dexec{\heap}{\stack}{\xperms}{\heap'}{\stack'}$$
for stacks $\stack, \stack'$, heap $\heap$, and exclusion frame $\xperms$.

\semantics[ExecSeq]
  { }
  {
    \dexec
      {\heap}{\triple{\perms}{\env}{\sseq{\kskip}{s}} \cdot \stack}
      {\xperms}
      {\heap}{\triple{\perms}{\env}{s} \cdot \stack}
  }
\semantics[ExecAssign]
  {
    \eval{\heap}{\env}{e}{v} \\
    \frm{\heap}{\perms}{\env}{e}
  }
  {
    \dexec
      {\heap}{\triple{\perms}{\env}{\sseq{x = e}{s}}  \cdot \stack}
      {\xperms}
      {\heap}{\triple{\perms}{\env[x \mapsto v]}{s} \cdot \stack}
  }
\semantics[ExecAssignField]
  {
    \eval{\heap}{\env}{x}{\ell} \\
    \eval{\heap}{\env}{e}{v} \\
    \assertion{\heap}{\perms}{\env}{\kacc(x.f)} \\
    \frm{\heap}{\perms}{\env}{e} \\
    \heap' = \heap[\pair{\ell}{f} \mapsto v]
  }
  {
    \dexec
      {\heap}{\triple{\perms}{\env}{\sseq{x.f = e}{s}} \cdot \stack}
      {\xperms}
      {\heap'}{\triple{\perms}{\env}{s} \cdot \stack}
  }
\semantics[ExecAlloc]
  {
    \ell = \ffresh \\
    \multiple{T ~f} = \fstruct(S) \\
    \heap' = \heap[\multiple{\pair{\ell}{f} \mapsto \fdefault(T)}] \\
    \perms' = \perms \cup \set{\multiple{\pair{\ell}{f}}}
  }
  {
    \dexec
      {\heap}{\triple{\perms}{\env}{\sseq{x = \salloc{S}}{s}} \cdot \stack}
      {\xperms}
      {\heap'}{\triple{\perms'}{\env[x \mapsto \ell]}{s} \cdot \stack}
  }
\semantics[ExecCallEnter]
  {
    \multiple{x} = \fparams(m) \\
    \multiple{\eval{\heap}{\env}{e}{v}} \\
    \multiple{\frm{\heap}{\perms}{\env}{e}} \\
    \env' = [\multiple{x \mapsto v}] \\
    \assertion{\heap}{\perms \setminus \xperms}{\env'}{\fpre(m)} \\
    \perms' = \foot{\heap}{\perms \setminus \xperms}{\env'}{\fpre(m)}
  }
  {
    \pair
      {\heap}
      {\triple{\perms}{\env}{\sseq{y \kassign m(\multiple{e})}{s}} \cdot \stack},
    \,\xperms
    \\\\ \to \\\\
    \pair
      {\heap}
      {
        \triple{\perms'}{\env'}{\sseq{\fbody(m)}{\kskip}} \cdot
        \triple{\perms \setminus \perms'}{\env}{\sseq{y \kassign m(\multiple{e})}{s}} \cdot \stack
      }
  }
\semantics[ExecCallExit]
  {
    \assertion{\heap}{\perms'}{\env'}{\fpost(m)} \\
    \env'' = \env[y \mapsto \env'(\kresult)] \\
    \perms'' = \perms \cup \foot{\heap}{\perms'}{\env'}{\fpost(m)}
  }
  {
    \dexec
      {\heap}
      {
        \triple{\perms'}{\env'}{\kskip} \cdot
        \triple{\perms}{\env}{\sseq{y \kassign m(\multiple{e})}{s}} \cdot
        \stack
      }
      {\xperms}
      {\heap}
      {\triple{\perms''}{\env''}{s} \cdot \stack}
  }
\semantics[ExecAssert]
  {
    \assertion{\heap}{\perms}{\env}{\simprecise{\phi}}
  }
  {
    \dexec
      {\heap}{\triple{\perms}{\env}{\sseq{\sassert{\gform}}{s}} \cdot \stack}
      {\xperms}
      {\heap}{\triple{\perms}{\env}{s} \cdot \stack}
  }
\semantics[ExecIfA]
  {
    \eval{\heap}{\env}{e}{\ktrue} \\
    \frm{\heap}{\perms}{\env}{e}
  }
  {
    \dexec
      {\heap}{\triple{\heap}{\env}{\sseq{\sif{e}{s_1}{s_2}}{s}} \cdot \stack}
      {\xperms}
      {\heap}{\triple{\heap}{\env}{\sseq{s_1}{s}} \cdot \stack}
  }
\semantics[ExecIfB]
  {
    \eval{\heap}{\env}{e}{\kfalse} \\
    \frm{\heap}{\perms}{\env}{e}
  }
  {
    \dexec
      {\heap}{\triple{\heap}{\env}{\sseq{\sif{e}{s_1}{s_2}}{s}} \cdot \stack}
      {\xperms}
      {\heap}{\triple{\heap}{\env}{\sseq{s_2}{s}} \cdot \stack}
  }
\semantics[ExecWhileEnter]
  {
    \eval{\heap}{\env}{e}{\ktrue} \\
    \frm{\heap}{\perms}{\env}{e} \\
    \assertion{\heap}{\perms \setminus \xperms}{\env}{\gform} \\
    \perms' = \foot{\heap}{\perms \setminus \xperms}{\env}{\gform}
  }
  {
    \pair
      {\heap}
      {\triple{\perms}{\env}{\sseq{\swhile{e}{\gform}{s'}}{s}} \cdot \stack},
    \,\xperms
    \\\\ \to \\\\
    \pair
      {\heap}
      {\triple{\perms'}{\env}{\sseq{s'}{\kskip}} \cdot \triple{\perms \setminus \perms'}{\env}{\sseq{\swhile{e}{\gform}{s'}}{s}} \cdot \stack}
  }
\semantics[ExecWhileSkip]
  {
    \eval{\heap}{\env}{e}{\kfalse} \\
    \frm{\heap}{\perms}{\env}{e} \\
    \assertion{\heap}{\perms \setminus \xperms}{\env}{\gform}
  }
  {
    \dexec
      {\heap}{\triple{\perms}{\env}{\sseq{\swhile{e}{\gform}{s'}}{s}} \cdot \stack}
      {\xperms}
      {\heap}{\triple{\perms}{\env}{s} \cdot \stack}
  }
\semantics[ExecWhileFinish]
  {
    \assertion{\heap}{\perms'}{\env'}{\gform} \\
    \perms'' = \perms \cup \foot{\heap}{\perms'}{\env'}{\gform}
  }
  {
    \pair
      {\heap}
      {
        \triple{\perms'}{\env'}{\kskip} \cdot
        \triple{\perms}{\env}{\sseq{\swhile{e}{\gform}{s'}}{s}} \cdot
        \stack
      },
    \,\xperms
    \\\\ \to \\\\
    \pair
      {\heap}
      {\triple{\perms''}{\env'}{\sseq{\swhile{e}{\gform}{s'}}{s}} \cdot \stack}
  }
\semantics[ExecFold]
  { }
  {
    \dexec
      {\heap}
      {\triple{\perms}{\env}{\sseq{\sfold{p(\multiple{e})}}{s}} \cdot S}
      {\xperms}
      {\heap}
      {\triple{\perms}{\env}{s} \cdot S}
  }
\semantics[ExecUnfold]
  { }
  {
    \dexec
      {\heap}
      {\triple{\perms}{\env}{\sseq{\sunfold{p(\multiple{e})}}{s}} \cdot S}
      {\xperms}
      {\heap}
      {\triple{\perms}{\env}{s} \cdot S}
  }

\subsection{Reachable transitions}\label{sec:dynamic-reachability}

\begin{definition}
  An \textbf{execution state} $\Gamma$ is either one of the abstract symbols $\finalsym$ or $\initsym$, or a pair $\pair{\heap}{\stack}$ of a heap $\heap$ and a stack $\stack$.
\end{definition}

\begin{definition}
  An execution state $\Gamma$ is \textbf{well-formed} if $\Gamma$ is either one of the abstract symbols $\initsym$ or $\finalsym$, or of the form $\pair{\heap}{\triple{\perms_n}{\env_n}{s_n} \cdot \ldots \cdot \triple{\perms_1}{\env_1}{s_1} \cdot \nilsym}$ and
  \begin{itemize}
    \item $\perms_i \cap \perms_j = \emptyset$ for all $1 \le i < j \le n$.
    \item $s_n = \sseq{s}{\kskip}$ for some statement $s$ or $s_n = \kskip$.
    \item For all $1 \le i < n$, $s_i = \sseq{s}{\sseq{s'}{\kskip}}$ or $s_i  = \sseq{s}{\kskip}$ for some statements $s$ and $s'$ where $s$ is of the form $m(\multiple{e})$ for some $m$, $\multiple{e}$ or $\swhile{e}{\gform}{s_{body}}$ for some $e$, $\gform$, $s_{body}$.
  \end{itemize}
\end{definition}

Examining the execution rules shows that well-formedness of states is preserved by the execution rules defined above.

A dynamic execution transition $\Gamma \to \Gamma'$ is reachable under a program $\prog$, using the exclusion frame $\xperms$, when the following judgement holds:
$$\dtrans{\prog}{\xperms}{\Gamma}{\Gamma'}$$
\semantics[ExecInit]
{ }
{
  \dtrans{\quadruple{s}{M}{P}{S}}{\_}{\initsym}{\pair{\emptyset}{\triple{\emptyset}{\emptyset}{s} \cdot \nilsym}}
}
\semantics[ExecStep]
{
  \dtrans{\prog}{\_}{\_}{\pair{\heap}{\stack}} \\
  \dexec{\heap}{\stack}{\xperms}{\heap'}{\stack'}
}
{
  \dtrans{\prog}{\xperms}{\pair{\heap}{\stack}}{\pair{\heap'}{\stack'}}
}
\semantics[ExecFinal]
{ }
{
  \dtrans{\prog}{\_}{\pair{\_}{\triple{\_}{\_}{\kskip} \cdot \nilsym}}{\finalsym}
}
\begin{definition}\label{def:dstate-reachable}
  An execution state $\dstate$ is \textbf{reachable} from program $\prog$ if $\dstate = \initsym$ or $\dtrans{\prog}{\_}{\_}{\dstate}$.
\end{definition}

\section{Symbolic Execution}


\subsection{Definitions}

\begin{definition}
  A \textbf{symbolic value} $\nu \in \SValue$ is an abstract value that represents an unknown character, boolean, integer, or location value. We leave the concrete type of $\SValue$ undefined, but assume that an infinite number of distinct new values can be produced by the $\ffresh$ function.
\end{definition}

\begin{definition}
  A \textbf{symbolic expression} $t \in \SExpr$ is a symbolic value or symbolic expressions combined using operators. Note that the binary operators $\oplus$ are the same as in \S \ref{sec:grammar}.
    $$t ::= \nu \gralt l \gralt \kneg t \gralt t_1 \kand t_2 \gralt t_1 \kor t_2 \gralt t_1 \oplus t_2$$
\end{definition}

\begin{definition}
  A \textbf{path condition} $\pc \in \SExpr$ is a symbolic expression consisting of conjuncts added at every branch point during a particular symbolic execution path.
\end{definition}

\begin{definition}
  An \textbf{imprecise flag} $\imp \in \set{ \top, \bot }$ is a flag that indicates whether a state is imprecise.
\end{definition}

\begin{definition}
  A \textbf{symbolic evironment} $\senv : \Var \pfunc \SExpr$ is a partial function mapping variable names to symbolic expressions.
\end{definition}

\begin{definition}
  A \textbf{field chunk} $\triple{f}{t}{t'} \in \SField$ denotes the mapping of the field $f$ of instance $t$ to the value $t'$.
\end{definition}

\begin{definition}
  A \textbf{predicate chunk} $\pair{p}{\multiple{t}} \in \SPredicate$ represents an isorecursive predicate $p$ with symbolically-evaluated arguments $\multiple{t}$.
\end{definition}

\begin{definition}
  A \textbf{heap chunk} $h \in \SField \cup \SPredicate$ is either a field chunk or a predicate chunk.
\end{definition}

\begin{definition}
  A \textbf{precise symbolic heap} (usually abbreviated as precise heap) $\sheap \in \powerset{\SField \cup \SPredicate}$ is a set of heap chunks where all heap chunks must occupy distinct heap locations at run time.
\end{definition}

\begin{definition}
  An \textbf{optimistic symbolic heap} (usually abbreviated as optimistic heap) $\oheap \in \powerset{\SField}$ is a set of field chunks where distinct chunks may coincide on the heap at run time (i.e. object references that are distinct symbolic expressions may be represent the same object value at run time).
\end{definition}

\begin{definition}
  A \textbf{symbolic permission} $\sperm \in \SPerm$ represents a particular heap location or predicate instance.
  $$\sperm ::= \pair{f}{t} \gralt \pair{p}{\multiple{t}}$$
  A set of symbolic permissions is denoted by $\sperms$.
\end{definition}

\begin{definition}
  A \textbf{run-time check} $r \in \SCheck$ is a symbolic value that must be asserted at run time, a symbolic permission whose access must be asserted at run time, a set of symbolic permissions whose disjointness must be asserted at run time, or an unsatisfiable check.
  $$r ::= t \gralt \sperm \gralt \fsep(\sperms_1, \sperms_2) \gralt \bot$$
  A set of run-time checks is denoted by $\scheck$.
\end{definition}

\begin{definition}
  A \textbf{symbolic state} $\sstate \in \SState$ consists of an imprecise flag, a path condition, a precise heap, an optimistic heap, and a symbolic environment.
  $$\sstate ::= \quintuple{\imp}{\pc}{\sheap}{\oheap}{\senv}$$
  $\imp(\sstate)$, $\pc(\sstate)$, $\sheap(\sstate)$, $\oheap(\sstate)$, and $\senv(\sstate)$ each denote a reference to the respective component of $\sstate$.
\end{definition}

\subsection{Valuations}\label{sec:valuations}

In order to prove soundness with respect to the dynamic semantics, we must first define a correspondence between the two representations.

\begin{definition}
  A \textbf{valuation} $V : \SValue \pfunc \Value$ is a partial function mapping symbolic values to concrete values.
\end{definition}

A base valuation $V : \SValue \pfunc \Value$ is implicitly extended to $V : \SExpr \pfunc \Value$ for all possible symbolic expressions composed of literals and symbolic values in the domain of $V$:
\begin{align*}
  V(l) &:= l \\
  V(t_1 \kadd t_2) &:= V(t_1) + V(t_2) \\
  V(t_1 \ksub t_2) &:= V(t_1) - V(t_2) \\
  V(t_1 \kmul t_2) &:= V(t_1) \cdot V(t_2) \\
  V(t_1 \kdiv t_2) &:= \frac{V(t_1)}{V(t_2)} \\
  V(t_1 \keq t_2) &:= \begin{cases}
    \ktrue &\text{if}~ V(t_1) = V(t_2) \\
    \kfalse &\text{otherwise}
  \end{cases} \\
  V(\kneg t) &:= \begin{cases}
    \ktrue &\text{if}~ V(t_1) = \kfalse \\
    \kfalse &\text{otherwise}
  \end{cases} \\
  V(t_1 \kor t_2) &:= \begin{cases}
    \ktrue &\text{if}~ V(t_1) = \ktrue ~\text{or}~ V(t_2) = \ktrue \\
    \kfalse &\text{otherwise}
  \end{cases} \\
  V(t_1 \kand t_2) &:= \begin{cases}
    \ktrue &\text{if}~ V(t_1) = \ktrue ~\text{and}~ V(t_2) = \ktrue \\
    \kfalse &\text{otherwise}
  \end{cases}
\end{align*}

\begin{definition}\label{def:implication}
  A symbolic expression $t_1$ \textbf{implies} another symbolic expression $t_2$ (denoted $t_1 \implies t_2$) if, for all valuations for which $V(t_1)$ and $V(t_2)$ are defined, $(V(t_1) = \ktrue) \implies (V(t_2) = \ktrue)$.
\end{definition}

\subsection{Footprints}

$\vfoot{V}{\heap}{\sperm}$ and $\vfoot{V}{\heap}{\sperms}$ denote the footprint (i.e. set of permissions) necessary to satisfy the given symbolic permission or symbolic permission set, respectively, given some heap $\heap$.
\begin{align*}
  \vfoot{V}{\heap}{\pair{f}{t}} &:= \pair{V(t)}{f} \\
  \vfoot{V}{\heap}{\pair{p}{\multiple{t}}} &:= \efoot{\heap}{[\multiple{x \mapsto V(t)}]}{\fpred(p)} \\
  \vfoot{V}{\heap}{\sperms} &:= \bigcup_{\sperm \in \sperms} \vfoot{V}{\heap}{\sperm}
\end{align*}

\subsection{Correspondence}

The relations $\simheap{V}{\sheap}{\heap}{\perms}$, $\simheap{V}{\oheap}{\heap}{\perms}$, $\simenv{V}{\senv}{\env}$, and $\simstate{V}{\sstate}{\heap}{\perms}{\env}$ denote correspondence between symbolic states and run-time states:
\begin{align}
  \begin{split}\label{eq:sheap-correspondence}
    \simheap{V}{\sheap}{\heap}{\perms} \iffdef
      & (\universal{\triple{f}{t}{t'} \in \sheap}{\heap(V(t), f) = V(t')}) ~\wedge \\
      & (\universal{\triple{f}{t}{t'} \in \sheap}{\pair{V(t)}{f} \in \perms}) ~\wedge \\
      & (\universal{\pair{p}{\tlist} \in \sheap}{\assertion{\heap}{\perms}{[\multiple{x \mapsto V(t)}]}{\fpred(p)}}) ~\wedge \\
      & (\universal{h_1, h_2 \in \sheap^2}{h_1 \ne h_2 \implies \vfoot{V}{\heap}{h_1} \cap \vfoot{V}{\heap}{h_2} = \emptyset})
  \end{split} \\
  \begin{split}
    \simheap{V}{\oheap}{\heap}{\perms} \iffdef
      & (\universal{\triple{f}{t}{t'} \in \oheap}{\heap(V(t), f) = V(t')}) ~\wedge \\
      & (\universal{\triple{f}{t}{t'} \in \oheap}{\pair{V(t)}{f} \in \perms})
  \end{split} \label{eq:oheap-correspondence} \\
  \simenv{V}{\senv}{\env} \iffdef &\universal{x \in \dom(\senv)}{\env(x) = V(t)} \label{eq:senv-correspondence} \\
  \begin{split}\label{eq:sstate-correspondence}
    \simstate{V}{\sstate}{\heap}{\perms}{\env} \iffdef
      & (\simheap{V}{\sheap(\sstate)}{\heap}{\perms}) ~\wedge \\
      & (\simheap{V}{\oheap(\sstate)}{\heap}{\perms}) ~\wedge \\
      & (\simenv{V}{\senv(\sstate)}{\env}) ~\wedge \\
      & (V(\pc(\sstate)) = \ktrue)
  \end{split}
\end{align}

\subsection{Run-time checks}

The judgement $\rtassert{V}{\heap}{\perms}{r}$ denotes that a symbolic runtime check $r \in \SCheck$ is satisfied at run time by a heap $\heap$ and permission set $\perms$ through a valuation $V$. Note that there is no rule for $\bot$; by design it can never be satisfied.
\semantics[CheckValue]
  {V(t) = \ktrue}
  {\rtassert{V}{\heap}{\perms}{t}}
\semantics[CheckAcc]
  {\pair{V(t)}{f} \in \perms}
  {\rtassert{V}{\heap}{\perms}{\pair{f}{t}}}
\semantics[CheckPred]
  {
    \multiple{x} = \fpredparams(p) \\
    \assertion{\heap}{\perms}{[\multiple{x \mapsto V(t)}]}{\fpred(p)}
  }
  {
    \rtassert{V}{\heap}{\perms}{\pair{p}{\multiple{t}}}
  }
\semantics[CheckSep]
  {
    \vfoot{V}{\heap}{\sperms_1} \cap \vfoot{V}{\heap}{\sperms_2} = \emptyset
  }
  {
    \rtassert{V}{\heap}{\perms}{\fsep(\sperms_1, \sperms_2)}
  }

This judgement is naturally extended for a set of runtime checks $\scheck$:
$$\rtassert{V}{\heap}{\perms}{\scheck} \iffdef \universal{r \in \scheck}{\rtassert{V}{\heap}{\perms}{r}}$$

\subsection{Evaluation}\label{sec:seval-rules}

The judgement
$$\seval{\sstate}{e}{t}{\sstate'}{\scheck}$$
denotes the evaluation of the expression $e \in \Expr$ in the symbolic state $\sstate \in \SState$. It yields the symbolic expression $t \in \SExpr$, a new symbolic state $\sstate'$ which must be satisfied to produce the resulting value, and a set of run-time checks $\scheck \in \powerset{\SCheck}$.

Note that for any given $\sstate$, there may be multiple values of $t, \sstate', \scheck$ for which the relation is satisfied. Therefore, the path condition of $\sstate'$ must be satisfied before assuming that $t$ corresponds to an actual value.

Also note that unsatisfiable paths are not pruned during evaluation. These paths may be pruned by checking the satisfiability of $\pc(\sstate')$.

For a list of expressions $\multiple{e}$, $\multiple{\seval{\sstate}{e}{t}{\sstate'}{\scheck}}$ represents a sequence of judgements
$$\seval{\sstate_0}{e_1}{t_1}{\sstate_1}, \cdots, \seval{\sstate_{n-1}}{e_n}{t_n}{\sstate_n}{\scheck_n}$$
where $\sstate_0 = \sstate$, $e_1, \cdots, e_n = \multiple{e}$, and $\scheck = \scheck_1 \cup \cdots \cup \scheck_n$.

\semantics[SEvalLiteral]
  { }
  {\seval{\sstate}{l}{l}{\sstate}{\emptyset}}
\semantics[SEvalVar]
  { }
  {\seval{\sstate}{x}{\senv(\sstate)(x)}{\sstate}{\emptyset}}
\semantics[SEvalOrA]
  {
    \seval{\sstate}{e_1}{t_1}{\sstate'}{\scheck} \\
    \sstate'' = \sstate'[\pc = \pc(\sstate') \kand t_1]
  }
  {
    \seval{\sstate}{e_1 \kor e_2}{t_1}{\sstate''}{\scheck}
  }
\semantics[SEvalOrB]
  {
    \seval{\sstate}{e_1}{t_1}{\sstate'}{\scheck_1} \\
    \seval{\sstate'[\pc = \pc(\sstate') \kand \kneg t_1]}{e_2}{t_2}{\sstate''}{\scheck_2}
  }
  {
    \seval{\sstate}{e_1 \kor e_2}{t_2}{\sstate''}{\scheck_1 \cup \scheck_2}
  }
\semantics[SEvalAndA]
  {
    \seval{\sstate}{e_1}{t_1}{\sstate'}{\scheck} \\
    \sstate'' = \sstate'[\pc = \pc(\sstate') \kand \kneg t_1]
  }
  {
    \seval{\sstate}{e_1 \kand e_2}{t_1}{\sstate''}{\scheck}
  }
\semantics[SEvalAndB]
  {
    \seval{\sstate}{e_1}{t_1}{\sstate'}{\scheck_1} \\
    \seval{\sstate'[\pc = \pc(\sstate') \kand t_1]}{e_2}{t_2}{\sstate''}{\scheck_2}
  }
  {
    \seval{\sstate}{e_1 \kand e_2}{t_2}{\sstate''}{\scheck_1 \cup \scheck_2}
  }
\semantics[SEvalOp]
  {
    \seval{\sstate}{e_1}{t_1}{\sstate'}{\scheck_1} \\
    \seval{\sstate'}{e_2}{t_2}{\sstate''}{\scheck_2}
  }
  {
    \seval{\sstate}{e_1 \oplus e_2}{t_1 \oplus t_2}{\sstate''}{\scheck_1 \cup \scheck_2}
  }
\semantics[SEvalNeg]
  {
    \seval{\sstate}{e}{t}{\sstate'}{\scheck}
  }
  {
    \seval{\sstate}{\kneg e}{\kneg t}{\sstate'}{\scheck}
  }
\semantics[SEvalField]
  {
    \seval{\sstate}{e}{t_e}{\sstate'}{\scheck} \\
    \pc(\sstate') \implies t_e \keq t_e' \\
    \triple{f}{t_e'}{t} \in \sheap(\sstate')
  }
  {
    \seval{\sstate}{e.f}{t}{\sstate'}{\scheck}
  }
\semantics[SEvalFieldOptimistic]
  {
    \seval{\sstate}{e}{t_e}{\sstate'}{\scheck} \\
    \nexistential{t_e', t}{\triple{f}{t_e'}{t} \in \sheap(\sstate) \wedge \pc(\sstate') \implies t_e' \keq t_e} \\
    \triple{f}{t_e'}{t} \in \oheap(\sstate) \\
    \pc(\sstate') \implies t_e' \keq t_e
  }
  {
    \seval{\sstate}{e.f}{t}{\sstate'}{\scheck}
  }
\semantics[SEvalFieldImprecise]
  {
    \imp(\sstate) \\
    \seval{\sstate}{e}{t_e}{\sstate'}{\scheck} \\
    \nexistential{t_e', t}{\triple{f}{t_e'}{t} \in \sheap(\sstate) \cup \oheap(\sstate) \wedge \pc(\sstate') \implies t_e' \keq t_e} \\
    t = \ffresh \\
    \sstate'' = \sstate'[\oheap = \oheap(\sstate'); \triple{f}{t_e}{t}]
  }
  {
    \seval{\sstate}{e.f}{t}{\sstate''}{\scheck; \pair{t_e}{f}}
  }
\semantics[SEvalFieldFailure]
  {
    \neg \imp(\sstate) \\
    \seval{\sstate}{e}{t_e}{\sstate'}{\scheck} \\
    \nexistential{t_e', t}{\triple{f}{t_e'}{t} \in \sheap(\sstate) \cup \oheap(\sstate) \wedge \pc(\sstate') \implies t_e' \keq t_e} \\
    t = \ffresh
  }
  {
    \seval{\sstate}{e.f}{t}{\sstate'}{\set{\bot}}
  }

\subsection{Deterministic evaluation}

We also define separate evaluation semantics for evaluation when branching is unwanted. This is denoted by the judgement
$$\spceval{\sstate}{e}{t}{\scheck}.$$

This operation mirrors the behavior of \ttt{pc-eval}, and as such, does not modify the symbolic state. Logical operations are not short-circuited as in the previous section; instead, they are formally encoded as conjuncts of a single $\SExpr$.

\semantics[SEvalPCLiteral]
  { }
  {\spceval{\sstate}{l}{l}{\emptyset}}
\semantics[SEvalPCVar]
  { }
  {\spceval{\sstate}{x}{\senv(\sstate)(x)}{\emptyset}}
\semantics[SEvalPCOr]
  {
    \spceval{\sstate}{e_1}{t_1}{\scheck_1} \\
    \spceval{\sstate}{e_2}{t_2}{\scheck_2}
  }
  {
    \spceval{\sstate}{e_1 \kor e_2}{t_1 \kor t_2}{\scheck_1 \cup \scheck_2}
  }
\semantics[SEvalPCAnd]
  {
    \spceval{\sstate}{e_1}{t_1}{\scheck_1} \\
    \spceval{\sstate}{e_2}{t_2}{\scheck_2}
  }
  {
    \spceval{\sstate}{e_1 \kand e_2}{t_1 \kand t_2}{\scheck_1 \cup \scheck_2}
  }
\semantics[SEvalPCOp]
  {
    \spceval{\sstate}{e_1}{t_1}{\scheck_1} \\
    \spceval{\sstate}{e_2}{t_2}{\scheck_2}
  }
  {
    \spceval{\sstate}{e_1 \oplus e_2}{t_1 \oplus t_2}{\scheck_1 \cup \scheck_2}
  }
\semantics[SEvalPCNeg]
  {
    \spceval{\sstate}{e}{t}{\scheck}
  }
  {
    \spceval{\sstate}{\kneg e}{\kneg t}{\scheck}
  }
\semantics[SEvalPCField]
  {
    \spceval{\sstate}{e}{t_e}{\scheck} \\
    \pc(\sstate) \implies t_e' \keq t_e \\
    \triple{f}{t_e'}{t} \in \sheap(\sstate)
  }
  {
    \spceval{\sstate}{e.f}{t}{\scheck}
  }
\semantics[SEvalPCFieldOptimistic]
  {
    \spceval{\sstate}{e}{t_e}{\scheck} \\
    \nexistential{t_e, t}{\triple{f}{t_e'}{t} \in \sheap(\sstate) \wedge \pc(\sstate) \implies t_e' \keq t_e} \\
    \pc(\sstate) \implies t_e' \keq t_e \\
    \triple{f}{t_e'}{t} \in \oheap(\sstate)
  }
  {
    \spceval{\sstate}{e.f}{t}{\scheck}
  }
\semantics[SEvalPCFieldImprecise]
  {
    \imp(\sstate) \\
    \spceval{\sstate}{e}{t_e}{\scheck} \\\\
    \nexistential{t_e, t}{\triple{f}{t_e'}{t} \in \sheap(\sstate) \cup \oheap(\sstate) \wedge \pc(\sstate) \implies t_e' \keq t_e} \\
    t = \ffresh
  }
  {
    \spceval{\sstate}{e.f}{t}{\scheck; \pair{t_e}{f}}
  }
\semantics[SEvalPCFieldMissing]
  {
    \neg \imp(\sstate) \\
    \spceval{\sstate}{e}{t_e}{\scheck} \\\\
    \nexistential{t_e, t}{\triple{f}{t_e'}{t} \in \sheap(\sstate) \cup \oheap(\sstate) \wedge \pc(\sstate) \implies t_e' \keq t_e} \\
    t = \ffresh
  }
  {
    \spceval{\sstate}{e.f}{t}{\scheck; \bot}
  }

\subsection{Produce}\label{sec:produce-rules}

The \textbf{produce} operation adds the information contained in a formula $\gform$ to the symbolic state $\sstate$, resulting in a new symbolic state $\sstate'$. This is denoted by the judgement
$$\sproduce{\sstate}{\phi}{\sstate'}.$$

\semantics[SProduceImprecise]
  {\sproduce{\sstate[\imp = \top]}{\phi}{\sstate'}}
  {\sproduce{\sstate}{\simprecise{\phi}}{\sstate'}}
\semantics[SProduceExpr]
  {
    \spceval{\sstate}{e}{t}{\_} \\
    \sstate' = \sstate[\pc = \pc(\sstate) \kand t]
  }
  {
    \sproduce{\sstate}{e}{\sstate'}
  }
\semantics[SProducePredicate]
  {
    \multiple{\spceval{\sstate}{e}{t}{\_}} \\
    \sstate' = \sstate[\sheap = \sheap(\sstate); \pair{p}{\tlist}]
  }
  {
    \sproduce{\sstate}{p(\multiple{e})}{\sstate'}
  }
\semantics[SProduceField]
  {
    \spceval{\sstate}{e}{t_e}{\_} \\
    t = \ffresh \\
    \sstate' = \sstate[\sheap = \sheap(\sstate); \triple{f}{t_e}{t}]
  }
  {
    \sproduce{\sstate}{\kacc(e.f)}{\sstate'}
  }
\semantics[SProduceConjunction]
  {
    \sproduce{\sstate}{\phi_1}{\sstate'} \\
    \sproduce{\sstate'}{\phi_2}{\sstate''}
  }
  {
    \sproduce{\sstate}{\phi_1 * \phi_2}{\sstate''}
  }
\semantics[SProduceIfA]
  {
    \spceval{\sstate}{e}{t}{\_} \\
    \sproduce{\sstate[\pc = \pc(\sstate) \kand t]}{\phi_1}{\sstate'}
  }
  {
    \sproduce{\sstate}{\sif{e}{\phi_1}{\phi_2}}{\sstate'}
  }
\semantics[SProduceIfB]
  {
    \spceval{\sstate}{e}{t}{\_} \\
    \sproduce{\sstate[\pc = \pc(\sstate) \kand \kneg t]}{\phi_2}{\sstate'}
  }
  {
    \sproduce{\sstate}{\sif{e}{\phi_1}{\phi_2}}{\sstate'}
  }

\subsection{Consume}\label{sec:consume-rules}

The \textbf{consume} operation checks whether a formula $\gform$ is established by the symbolic state, collects runtime checks that are minimally sufficient to establish $\gform$, and removes permissions asserted in $\gform$ from the symbolic state. This is denoted by the judgement
$$\sconsume{\sstate}{\sstate_E}{\gform}{\sstate'}{\scheck}{\sperms}$$
where $\sstate$ is the symbolic state containing the currently remaining permissions during consume, and $\sstate_E$ is the symbolic state containing the original permissions which may used for evaluating expressions.

Note that consume does not branch on operations such as $\kand$ that are normally short-circuiting. This is because the evaluation operation $\downarrow$ does not modify the path condition but keeps track of logical operators as symbolic values. Values in specifications must be framed anyway (explicitly or inferred), so we must always have the permissions necessary (whether statically or dynamically checked) for evaluating all branches.

\semantics[SConsumeImprecision]
  {
    \sconsume{\sstate}{\sstate_E[\imp = \top]}{\phi}{\sstate'}{\scheck}{\sperms}
  }
  {
    \sconsume{\sstate}{\sstate_E}{\simprecise{\phi}}{\quintuple{\top}{\pc(\sstate')}{\senv(\sstate')}{\emptyset}{\emptyset}}{\scheck}{\sperms}
  }
\semantics[SConsumeValue]
  {
    \spceval{\sstate_E}{e}{t}{\scheck} \\
    \pc(\sstate) \implies t
  }
  {
    \sconsume{\sstate}{\sstate_E}{e}{\sstate}{\scheck}{\emptyset}
  }
\semantics[SConsumeValueImprecise]
  {
    \imp(\sstate) \\
    \spceval{\sstate_E}{e}{t}{\scheck} \\
    \pc(\sstate) \notimplies t
  }
  {
    \sconsume{\sstate}{\sstate_E}{e}{\sstate[\pc = \pc(\sstate) \kand t]}{\scheck; t}{\emptyset}
  }
\semantics[SConsumeValueFailure]
  {
    \neg \imp(\sstate) \\
    \spceval{\sstate_E}{e}{t}{\scheck} \\
    \pc(\sstate) \notimplies t
  }
  {
    \sconsume{\sstate}{\sstate_E}{e}{\sstate}{\set{\bot}}{\emptyset}
  }
\semantics[SConsumePredicate]
  {
    \multiple{\spceval{\sstate_E}{e}{t}{\scheck}} \\
    \multiple{\pc(\sstate) \implies t \keq t'} \\
    \sheap(\sstate) = \sheap'; \pair{p}{\multiple{t'}}
  }
  {
   \sconsume{\sstate}{\sstate_E}{p(\multiple{e})}{\sstate[\sheap = \sheap', \oheap = \emptyset]}{\bigcup \multiple{\scheck}}{\set{\pair{p}{\multiple{t}}}}
  }
\semantics[SConsumePredicateImprecise]
  {
    \imp(\sstate) \\
    \multiple{\spceval{\sstate_E}{e}{t}{\scheck}} \\
    \nexistential{\pair{p}{\multiple{t'}} \in \sheap(\sstate)}{\bigwedge \multiple{\pc(\sstate) \implies t \keq t'}}
  }
  {
    \sconsume{\sstate}{\sstate_E}{p(\multiple{e})}{\sstate[\heap = \emptyset, \oheap = \emptyset]}{\bigcup \multiple{\scheck}; \pair{p}{\multiple{t}}}{\set{\pair{p}{\multiple{t}}}}
  }
\semantics[SConsumePredicateFailure]
  {
    \neg \imp(\sstate) \\
    \multiple{\spceval{\sstate_E}{e}{t}{\scheck}} \\
    \nexistential{\pair{p}{\multiple{t'}} \in \sheap(\sstate)}{\bigwedge \multiple{\pc(\sstate) \implies t \keq t'}}
  }
  {
    \sconsume{\sstate}{\sstate_E}{p(\multiple{e})}{\sstate}{\set{\bot}}{\set{\pair{p}{\multiple{t}}}}
  }
\semantics[SConsumeAcc]
  {
    \spceval{\sstate_E}{e}{t_e}{\scheck} \\
    \pc(\sstate) \implies t_e' \keq t_e \\
    \triple{f}{t_e'}{t} \in \sheap(\sstate) \\
    \sheap' = \fremfp(\sheap(\sstate), \sstate, t_e, f) \\
    \oheap' = \fremf(\oheap(\sstate), \sstate, t_e, f)
  }
  {
    \sconsume{\sstate}{\sstate_E}{\kacc(e.f)}{\sstate[\sheap = \sheap', \oheap = \oheap']}{\scheck}{\set{\pair{t_e}{f}}}
  }
\semantics[SConsumeAccOptimistic]
  {
    \spceval{\sstate_E}{e}{t_e}{\scheck} \\
    \pc(\sstate) \implies t_e' \keq t_e \\
    \triple{f}{t_e'}{t} \in \sheap(\sstate) \\
    \nexistential{t_e', t}{\triple{f}{t_e}{t} \in \sheap(\sstate) \wedge (\pc(\sstate) \implies t_e' \keq t_e)} \\
    \sheap' = \fremf(\sheap(\sstate), \sstate, t_e, f) \\
    \oheap' = \fremf(\oheap(\sstate), \sstate, t_e, f)
  }
  {
    \sconsume{\sstate}{\sstate_E}{\kacc(e.f)}{\sstate[\sheap = \sheap', \oheap = \oheap']}{\scheck}{\set{\pair{t_e}{f}}}
  }
\semantics[SConsumeAccImprecise]
  {
    \imp(\sstate) \\
    \spceval{\sstate_E}{e}{t_e}{\scheck} \\
    \nexistential{t_e', t}{\triple{f}{t_e}{t} \in \sheap(\sstate) \cup \oheap(\sstate) \wedge (\pc(\sstate) \implies t_e' \keq t_e)} \\
    \sheap' = \fremf(\sheap(\sstate), \sstate, t_e, f) \\
    \oheap' = \fremf(\oheap(\sstate), \sstate, t_e, f)
  }
  {
    \sconsume{\sstate}{\sstate_E}{\kacc(e.f)}{\sstate[\sheap = \sheap', \oheap = \oheap']}{\scheck; \pair{t_e}{f}}{\set{\pair{t_e}{f}}}
  }
\semantics[SConsumeAccFailure]
  {
    \neg \imp(\sstate) \\
    \spceval{\sstate_E}{e}{t_e}{\scheck} \\
    \nexistential{t_e', t}{\triple{f}{t_e}{t} \in \sheap(\sstate) \cup \oheap(\sstate) \wedge (\pc(\sstate) \implies t_e' \keq t_e)}
  }
  {
    \sconsume{\sstate}{\sstate_E}{\kacc(e.f)}{\sstate}{\set{\bot}}{\set{\pair{t_e}{f}}}
  }
\semantics[SConsumeConjunction]
  {
    \sconsume{\sstate}{\sstate_E}{\phi_1}{\sstate'}{\scheck_1}{\sperms_1} \\
    \sconsume{\sstate'}{\sstate_E[\pc = \pc(\sstate')]}{\phi_2}{\sstate''}{\scheck_2}{\sperms_2} \\\\
    (\scheck_1 \cup \scheck_2) \cap \SPerm = \emptyset
  }
  {
    \sconsume{\sstate}{\sstate_E}{\phi_1 * \phi_2}{\sstate''}{\scheck_1 \cup \scheck_2}{\sperms_1 \cup \sperms_2}
  }
\semantics[SConsumeConjunctionImprecise]
  {
    \sconsume{\sstate}{\sstate_E}{\phi_1}{\sstate'}{\scheck_1}{\sperms_1} \\
    \sconsume{\sstate'}{\sstate_E[\pc = \pc(\sstate')]}{\phi_2}{\sstate''}{\scheck_2}{\sperms_2} \\\\
    (\scheck_1 \cup \scheck_2) \cap \SPerm \ne \emptyset
  }
  {
    \sconsume{\sstate}{\sstate_E}{\phi_1 * \phi_2}{\sstate''}{\scheck_1 \cup \scheck_2; \fsep(\sperms_1, \sperms_2)}{\sperms_1 \cup \sperms_2}
  }
\semantics[SConsumeConditionalA]
  {
    \spceval{\sstate_E}{e}{t}{\scheck} \\
    \pc' = \pc(\sstate) \kand t \\
    \sconsume{\sstate[\pc = \pc']}{\sstate_E[\pc = \pc']}{\phi_1}{\sstate'}{\scheck'}{\sperms}
  }
  {
    \sconsume{\sstate}{\sstate_E}{\sif{e}{\phi_1}{\phi_2}}{\sstate'}{\scheck \cup \scheck'}{\sperms}
  }
\semantics[SConsumeConditionalB]
  {
    \spceval{\sstate_E}{e}{t}{\scheck} \\
    \pc' = \pc(\sstate) \kand \kneg t \\
    \sconsume{\sstate[\pc = \pc']}{\sstate_E[\pc = \pc']}{\phi_2}{\sstate'}{\scheck'}{\sperms}
  }
  {
    \sconsume{\sstate}{\sstate_E}{\sif{e}{\phi_1}{\phi_2}}{\sstate'}{\scheck \cup \scheck'}{\sperms}
  }

The functions $\fremf$, $\fremfp$, and $\falias$ are defined as follows:
\begin{align*}
  \fremf(\sheap, \sstate, t, f) &= \set{\triple{f'}{t'}{t''} \in \sheap : \neg \falias(\sstate, t, f, t', f')} \\
  \fremfp(\sheap, \sstate, t, f) &= \fremf(\sheap, \sstate, t, f) \cup \set{\pair{p}{\multiple{t}} \in \sheap} \\
  \falias(\sstate, t, f, t', f') &= \begin{cases}
    f = f' \wedge (g(\sstate) \implies t \keq t') & \neg \imp(\sstate) \\
    (f = f') \wedge \fsat(g(\sstate) \kand t \keq t') & \imp(\sstate)
  \end{cases}
\end{align*}

For ease of notation, the judgement $\scons{\sstate}{\gform}{\sstate'}{\scheck}$ applies the rules above, initializing additional parameters and ignoring the parameters that are internal to consume.
\semantics[SConsume]
  {
    \sconsume{\sstate}{\sstate}{\gform}{\sstate'}{\scheck}{\_}
  }
  {
    \scons{\sstate}{\gform}{\sstate'}{\scheck}
  }

\subsection{Execute}\label{sec:sexec-rules}

Symbolic execution is denoted by the small-step judgement
$$\sexec{\sstate}{s}{s'}{\sstate'}.$$
where $\sstate$ is the symbolic state prior to execution, $s$ is the statement to execute, $s'$ is the statement remaining after this execution step and $\sstate'$ is the symbolic state after the execution step.

Note that as in \S\ref{sec:seval-rules} there may be multiple resulting symbolic states $\sstate'$ for which the judgement applies, and therefore the path condition of $\sstate'$ must be satisfied before assuming that $\sstate'$ corresponds to a particular execution step.

\semantics[SExecSeq]
  {
  }
  {
    \sexec{\sstate}{\sseq{\kskip}{s}}{s}{\sstate}
  }
\semantics[SExecAssign]
  {
    \seval{\sstate}{e}{t}{\sstate'}{\_} \\
    \senv' = \senv(\sstate)[x \mapsto t]
  }
  {
    \sexec{\sstate}{\sseq{x = e}{s}}{s}{\sstate'[\senv = \senv']}
  }
\semantics[SExecAssignField]
  {
    \seval{\sstate}{e}{t}{\sstate'}{\_} \\
    \scons{\sstate'}{\kacc(x.f)}{\sstate''}{\_} \\
    \sheap' = \sheap(\sstate); \triple{\senv(\sstate'')(x)}{f}{t}
  }
  {
    \sexec{\sstate}{\sseq{x.f = e}{s}}{s}{\sstate''[\sheap = \sheap']}
  }
\semantics[SExecAlloc]
  {
    t = \ffresh \\
    \multiple{T ~ f} = \fstruct(S) \\
    \sheap' = \sheap(\sstate); \multiple{\triple{f}{t}{\fdefault(T)}}
  }
  {
    \sexec{\sstate}{\sseq{x = \salloc{S}}{s}}{s}{\sstate[\sheap = \sheap']}
  }
\semantics[SExecCall]
  {
    \multiple{\seval{\sstate}{e}{t_e}{\sstate'}{\_}} \\
    \multiple{x} = \fparams(m) \\
    \scons{\sstate'[\senv = [\multiple{x \mapsto t_e}]]}{\fpre(m)}{\sstate'}{\_} \\
    t = \ffresh \\
    \sproduce{\sstate'[\senv = [\multiple{x \mapsto t_e}, \kresult \mapsto t]]}{\fpost(m)}{\sstate''}
  }
  {
    \sexec{\sstate}{\sseq{y \kassign m(\multiple{e})}{s}}{s}{\sstate''[\senv = \senv(\sstate)[y \mapsto t]]}
  }
\semantics[SExecAssert]
  {
    \scons{\sstate}{\simprecise{\phi}}{\sstate'}{\_} \\
    \sproduce{\sstate'}{\simprecise{\phi}}{\sstate''}
  }
  {
    \sexec{\sstate}{\sseq{\sassert{\phi}}{s}}{s}{\sstate[\pc = \pc(\sstate'')]}
  }
\semantics[SExecFold]
  {
    \multiple{\seval{\sstate}{e}{t}{\sstate'}{\_}} \\
    \multiple{x} = \fpredparams(p) \\
    \scons{\sstate'[\senv = [\multiple{x \mapsto t}]]}{\fpred(p)}{\sstate''}{\_} \\
    \sstate''' = \sstate''[\senv = \senv(\sstate), \sheap = \sheap(\sstate''); \pair{p}{\multiple{t}}]
  }
  {
    \sexec{\sstate}{\sseq{\sfold{p(\multiple{e})}}{s}}{s}{\sstate'''}
  }
\semantics[SExecUnfold]
  {
    \multiple{\seval{\sstate}{e}{t}{\sstate'}{\_}} \\
    \scons{\sstate'}{p(\multiple{e})}{\sstate''}{\_} \\
    \multiple{x} = \fpredparams(p) \\
    \sproduce{\sstate''[\senv = [\multiple{x \mapsto t}]]}{\fpred(p)}{\sstate'''}
  }
  {
    \sexec{\sstate}{\sseq{\sunfold{p(e_1, \cdots, e_n)}}{s}}{s}{\sstate'''[\senv = \senv(\sstate)]}
  }
\semantics[SExecIfA]
  {
    \seval{\sstate}{e}{t}{\sstate'}{\_}
  }
  {
    \sexec{\sstate}{\sseq{\sif{e}{s_1}{s_2}}{s}}{\sseq{s_1}{s}}{\sstate'[\pc = \pc(\sstate') \kand t]}
  }
\semantics[SExecIfB]
  {
    \seval{\sstate}{e}{t}{\sstate'}{\_}
  }
  {
    \sexec{\sstate}{\sseq{\sif{e}{s_1}{s_2}}{s}}{\sseq{s_2}{s}}{\sstate'[\pc = \pc(\sstate') \kand \kneg t]}
  }
\semantics[SExecWhileSkip]
  {
    \scons{\sstate}{\gform}{\sstate'}{\_} \\
    \multiple{x} = \fmodified(s) \\
    \sproduce{\sstate'[\senv = \senv(\sstate')[\multiple{x \mapsto \ffresh}]]}{\gform}{\sstate''} \\
    \spceval{\sstate''}{e}{t}{\_} \\
  }
  {
    \sexec{\sstate}{\sseq{\swhile{e}{\gform}{s}}{s'}}{s'}{\sstate''[\pc = \pc(\sstate'') \kand \kneg t]}
  }

\begin{definition}\label{def:frem}
  The helper function $\frem(\sstate, \gform)$ returns the set of all permissions remaining in $\sstate$ if $\gform$ is not completely precise:
  $$\frem(\sstate, \gform) := \begin{cases}
    \emptyset &\text{if $\gform$ is completely precise} \\
    \set{ \pair{t}{f} : \triple{f}{t}{t'} \in \sheap(\sstate) \cup \oheap(\sstate) } ~\cup \\
    \quad \set{ \pair{p}{\multiple{t}} : \pair{p}{\multiple{t}} \in \sheap(\sstate) } &\text{otherwise}
  \end{cases}$$
  This is used to symbolically calculate the required exclusion frame.
\end{definition}

\subsection{Verification states}

\begin{definition}
  A \textbf{verification state} $\vstate$ is either an abstract symbol or a triple consisting of a symbolic state, a statement, and a post-condition.
  $$\vstate ::= \initsym \gralt \finalsym \gralt \triple{\sstate}{s}{\gform}$$
\end{definition}

The reachability of verification transitions, as determined by modular verification for a given program $\prog$, is determined by the judgement
$$\strans{\prog}{\vstate}{\vstate'}$$
where $\vstate$ is the beginning verification state and $\vstate'$ is the next verification state.

\semantics[SVerifyInit]
{
}
{
  \strans{\quadruple{s}{M}{P}{S}}{\initsym}{\triple{\quintuple{\bot}{\emptyset}{\emptyset}{\emptyset}{\ktrue}}{s}{\ktrue}}
}
\semantics[SVerifyMethod]
{
  m \in M \\
  \multiple{x} = \fparams(m) \\
  \sproduce{\quintuple{\bot}{\emptyset}{\emptyset}{[\multiple{x \mapsto \ffresh}]}{\ktrue}}{\fpre(m)}{\sstate}
}
{
  \strans{\quadruple{s}{M}{P}{S}}{\initsym}{\triple{\sstate}{\sseq{\fbody(m)}{\kskip}}{\fpost(m)}}
}
\semantics[SVerifyLoopBody]
{
  \strans{\prog}{\_}{\triple{\sstate_0}{\sseq{\swhile{e}{\gform}{s}}{s'}}{\gform_0}} \\
  \multiple{x} = \fmodified(s) \\
  \sproduce{\quintuple{\bot}{\senv(\sstate_0)[\multiple{x \mapsto \ffresh}]}{\emptyset}{\emptyset}{\pc(\sstate_0)}}{\gform}{\sstate_0'} \\
  \spceval{\sstate_0'}{e}{t}{\scheck}
}
{
  \strans{\prog}{\triple{\sstate}{\sseq{\swhile{e}{\gform}{s}}{s'}}{\gform_0}}{\triple{\sstate_0'[\pc = \pc(\sstate_0') \kand t]}{\sseq{s}{\kskip}}{\gform}}
}
\semantics[SVerifyLoop]
{
  \strans{\prog}{\_}{\triple{\sstate}{\sseq{\swhile{e}{\gform'}{s}}{s'}}{\gform}} \\
  \scons{\sstate}{\gform'}{\sstate'}{\_} \\
  \multiple{x} = \fmodified(s) \\
  \sproduce{\sstate'[\senv = \senv(\sstate)[\multiple{x \mapsto \ffresh}]]}{\gform'}{\sstate''}
}
{
  \prog \vdash \triple{\sstate}{\sseq{\swhile{e}{\gform'}{s}}{s'}}{\gform} \to \\\\
  \hspace{3em} \triple{\sstate''}{\sseq{\swhile{e}{\gform'}{s}}{s'}}{\gform}
}
\semantics[SVerifyStep]
{
  \strans{\prog}{\_}{\triple{\sstate}{s}{\gform}} \\
  \sexec{\sstate}{s}{s'}{\sstate'}
}
{
  \strans{\prog}{\triple{\sstate}{s}{\gform}}{\triple{\sstate'}{s'}{\gform}}
}
\semantics[SVerifyFinal]
{
  \strans{\prog}{\_}{\triple{\sstate}{\kskip}{\gform}} \\
  \scons{\sstate}{\gform}{\sstate'}{\_}
}
{
  \prog \vdash \triple{\sstate}{\kskip}{\gform} \to \finalsym
}

\begin{definition}\label{def:vstate-reachable}
  A verification state $\vstate$ is \textbf{reachable} from program $\prog$ if $\vstate = \initsym$ or $\prog \vdash \_ \to \vstate, \_, \_, \_$.

  A verification state $\vstate$ is \textbf{reachable} from program $\prog$ \textbf{with valuation} $V$ if $\vstate$ is reachable from $\prog$ and $V$ is defined for all symbolic values contained in $\vstate$.
\end{definition}

\subsection{Guards}

A guard judgement determines the set of runtime checks that must be satisfied at the given verification state before taking the next execution step, and determines the exclusion frame required to preserve the validity of heap chunks contained by the current symbolic state. This is denoted by the judgement
$$\vstate \rightharpoonup \sstate', \scheck, \sperms$$
where $\vstate$ is the current verification state, $\sstate'$ is the intermediate state, $\scheck$ is the set of required checks, and $\sperms$ is the exclusion frame (represented as a set of symbolic permissions).

\semantics[SGuardInit]
  {
  }
  {
    \sguard{\initsym}{\quintuple{\bot}{\emptyset}{\emptyset}{\emptyset}{\ktrue}}{\emptyset}{\emptyset}
  }
\semantics[SGuardSeq]
  {
  }
  {
    \sguard{\triple{\sstate}{\sseq{\kskip}{s}}{\gform}}{\sstate}{\emptyset}{\emptyset}
  }
\semantics[SGuardAssign]
  {
    \seval{\sstate}{e}{\_}{\sstate'}{\scheck}
  }
  {
    \sguard{\triple{\sstate}{\sseq{x = e}{s}}{\gform}}{\sstate}{\scheck}{\emptyset}
  }
\semantics[SGuardAssignField]
  {
    \seval{\sstate}{e}{\_}{\sstate'}{\scheck'} \\
    \scons{\sstate'}{\kacc(x.f)}{\sstate''}{\scheck''}
  }
  {
    \sguard{\triple{\sstate}{\sseq{x.f = e}{s}}{\gform}}{\sstate''}{\scheck' \cup \scheck''}{\emptyset}
  }
\semantics[SGuardAlloc]
  {
  }
  {
    \sguard{\triple{\sstate}{\sseq{x = \salloc{S}}{s}}{\gform}}{\sstate}{\emptyset}{\emptyset}
  }
\semantics[SGuardCall]
  {
    \multiple{
      \seval{\sstate}{e}{t}{\sstate'}{\scheck}
    } \\
    \multiple{x} = \fparams(m) \\
    \scons{\sstate'[\senv = [\multiple{x \mapsto t}]]}{\fpre(m)}{\sstate''}{\scheck'}
  }
  {
    \sguard{\triple{\sstate}{\sseq{y \kassign m(\multiple{e})}{s}}{\gform}}{\sstate''[\senv = \senv(\sstate)]}{\multiple{\scheck} \cup \scheck'}{\frem(\sstate'', \fpre(m))}
  }
\semantics[SGuardAssert]
  {
    \scons{\sstate}{\simprecise{\phi}}{\sstate'}{\scheck}
  }
  {
    \sguard{\triple{\sstate}{\sseq{\sassert{\phi}}{s}}{\gform}}{\sstate'}{\scheck}{\emptyset}
  }
\semantics[SGuardFold]
  {
    \multiple{\seval{\sstate}{e}{t}{\sstate'}{\scheck}} \\
    \multiple{x} = \fpredparams(p) \\
    \scons{\sstate'[\senv = [\multiple{x \mapsto t}]]}{\fpred(p)}{\sstate''}{\scheck'}
  }
  {
    \sguard{\triple{\sstate}{\sseq{\sfold{p(\multiple{e})}}{s}}{\gform}}{\sstate''[\senv = \senv(\sstate)]}{\scheck' \cup \bigcup \multiple{\scheck}}{\emptyset}
  }
\semantics[SGuardUnfold]
  {
    \multiple{\seval{\sstate}{e}{t}{\sstate'}{\scheck}} \\
    \scons{\sstate'}{p(\multiple{e})}{\sstate''}{\scheck'}
  }
  {
    \sguard{\triple{\sstate}{\sseq{\sunfold{p(\multiple{e})}}{s}}{\gform}}{\sstate''}{\scheck' \cup \bigcup \multiple{\scheck}}{\emptyset}
  }
\semantics[SGuardIf]
  {
    \seval{\sstate}{e}{\_}{\sstate'}{\scheck}
  }
  {
    \sguard{\triple{\sstate}{\sseq{\sif{e}{s_1}{s_2}}{s}}{\gform}}{\sstate'}{\scheck}{\emptyset}
  }
\semantics[SGuardWhile]
  {
    \scons{\sstate}{\gform}{\sstate'}{\scheck'} \\
    \multiple{x} = \fmodified(s) \\
    \sproduce{\sstate'[\senv = \senv(\sstate')[\multiple{x \mapsto \ffresh}]]}{\gform}{\sstate''} \\
    \spceval{\sstate''}{e}{\_}{\scheck''}
  }
  {
    \sguard{\triple{\sstate}{\sseq{\swhile{e}{\gform}{s}}{s'}}{\gform'} }{\sstate'[\pc = \pc(\sstate'')]}{\scheck' \cup \scheck''}{\frem(\sstate', \gform)}
  }

\semantics[SGuardFinish]
  {
    \scons{\sstate}{\gform}{\sstate'}{\scheck}
  }
  {
    \sguard{\triple{\sstate}{\kskip}{\gform}}{\sstate'}{\scheck}{\emptyset}
  }

\subsection{Valid states}

\begin{definition}\label{def:vstate-corresponds}
  A verification state $\vstate$ \textbf{corresponds} with valuation $V$ to an execution state $\Gamma$ if $\vstate = \Gamma$, or $\Gamma = \pair{\heap}{\triple{\perms}{\env}{s} \cdot \stack}$ for some $\heap, \perms, \env, s$ and $\vstate = \triple{\sstate}{s}{\_}$ for some $\sstate$, such that $\simstate{V}{\sstate}{\heap}{\perms}{\env}$.
\end{definition}

\begin{definition}\label{def:partial-valid}
  A \textit{partial state} $\Gamma = \pair{\heap}{\stack}$ is \textbf{validated} by a verification state $\vstate$ with valuation $V$ if one of the following cases apply:

  \begin{defparts}
    \defpart\label{def:partial-valid-nil}
      $\stack = \nilsym$

    \defpart\label{def:partial-valid-call}
      $\stack = \triple{\env}{\perms}{\sseq{y \kassign m(e_1, \cdots, e_k)}{s}} \cdot \stack^*$ for some $\env$, $\perms$, $y$, $m$, $k$, $e_1, \cdots, e_k$, $s$, and $\stack^*$, and there exists some $\vstate'$, $V'$, $x_1, \cdots, x_k$, $t_1, \cdots, t_k$, $\sstate_0, \cdots, \sstate_k$ and $\sstate'$ such that:
      \begin{gather}
        \text{The partial state $\pair{\heap}{\stack^*}$ is validated by $\vstate'$ and $V'$}, \\
        \vstate' \text{ is reachable from $\prog$ with valuation $V'$}, \quad s(\vstate') = s(\stack), \\
        x_1, \cdots, x_k = \fparams(m), \\
        \sstate_0 = \sstate(\vstate'), \quad \seval{\sstate_0}{e_1}{t_1}{\sstate_1}{\_}, \quad\cdots,\quad \seval{\sstate_{k-1}}{e_k}{t_k}{\sstate_k}{\_}, \label{eq:partial-valid-call-eval}\\
        \universal{1 \le i \le k}{V(\senv(\vstate)(x_i)) = V'(t_i)}, \label{eq:partial-valid-call-params} \\
        \scons{\sstate_k}{\fpre(m)}{\sstate'}{\_}, \quad \simstate{V'}{\sstate'[\senv = \senv(\sstate_0)]}{\heap}{\perms}{\env}, \quad\text{and} \label{eq:partial-valid-call-sim} \\
        \gform(\vstate) = \fpost(m)
      \end{gather}

    \defpart\label{def:partial-valid-while}
      $\stack = \triple{\env}{\perms}{\sseq{\swhile{e}{\gform}{s}}{s'}} \cdot \stack^*$ for some $\env$, $\perms$, $e$, $\gform$, $s$, $s'$, $\stack^*$, and there exists some $\vstate'$, $V'$, and $\sstate'$ such that:
      \begin{gather}
        \text{The partial state $\pair{\heap}{\stack^*}$ is validated by $\vstate'$ and $V'$} \label{eq:partial-valid-while-valid}\\
        \vstate' \text{ is reachable from $\prog$ with valuation $V'$}, \quad s(\vstate') = s(\stack) \label{eq:partial-valid-while-reachable}\\
        \scons{\sstate}{\gform}{\sstate'}{\_}, \quad\text{and}\quad
        \simstate{V'}{\sstate'}{\heap}{\perms}{\env} \label{eq:partial-valid-while-sim} \\
        \gform(\vstate) = \gform \label{eq:partial-valid-while-post}
      \end{gather}
  \end{defparts}
\end{definition}

\begin{definition}\label{def:state-valid}
  For a program $\prog$, a dynamic state $\Gamma$ is \textbf{validated} by $\vstate$ and valuation $V$ if all the following are true:
  
  \begin{defparts}
    \defpart\label{def:state-valid-reachable} $\vstate$ is reachable from $\prog$ with $V$
    
    \defpart\label{def:state-valid-correspond} $\Gamma$ corresponds to $\vstate$ with $V$

    \defpart\label{def:state-valid-partial} If $\Gamma = \pair{\heap}{\triple{\perms}{\env}{s} \cdot \stack^*}$,
    then the partial state $\pair{\heap}{\stack^*}$ is validated by $\vstate$ and $V$.

  \end{defparts}
\end{definition}

\begin{definition}
  $\Gamma$ is a \textbf{valid state} if $\Gamma$ is validated by some $\vstate$.
\end{definition}





\section{Soundness}

This section contains the proof of soundness, culminating in theorems  \ref{thm:dtrans-progress}, \ref{thm:guard-progress}, and \ref{thm:dtrans-preservation}.

\subsection{Cross-cutting lemmas}


\begin{lemma}[Relating expression framing and exact footprints]\label{lem:efoot-framing}
  $\efoot{\heap}{\env}{e} \subseteq \perms$ if and only if $\frm{\heap}{\perms}{\env}{e}$.
\end{lemma}
\begin{proof}
  By induction on the syntax forms of $e$:
  \begin{enumcases}
    \case $l$: Then $\efoot{\heap}{\env}{l} = \emptyset$, thus $\efoot{\heap}{\env}{l} \subseteq \perms$, and $\frm{\heap}{\perms}{\env}{l}$ by \refrule{FrameLiteral}.
    \case $v$: Then $\efoot{\heap}{\env}{v} = \emptyset$, thus $\efoot{\heap}{\env}{v} \subseteq \perms$, and $\frm{\heap}{\perms}{\env}{v}$ by \refrule{FrameVar}.
    \case $e.f$:

      Suppose that $\efoot{\heap}{\env}{e.f} \subseteq \perms$. Then by definition $\efoot{\heap}{\env}{e} \subseteq \perms$ and $\pair{\ell}{f} \in \perms$ for some $\ell$ such that $\eval{\heap}{\env}{e}{\ell}$. By induction $\frm{\heap}{\perms}{\env}{e}$, by \refrule{AssertAcc} $\assertion{\heap}{\perms}{\env}{\kacc(e.f)}$, and thus by \refrule{FrameField}, $\frm{\heap}{\perms}{\env}{e.f}$.

      Suppose that $\frm{\heap}{\perms}{\env}{e.f}$. Then by \refrule{FrameField}, $\frm{\heap}{\perms}{\env}{e}$ and $\assertion{\heap}{\perms}{\env}{\kacc(e.f)}$, thus by induction $\efoot{\heap}{env}{e} \subseteq \perms$ and by \refrule{AssertAcc}, $\pair{\ell}{f} \in \perms$ for some $\ell$ such that $\eval{\heap}{\env}{e}{\ell}$. Thus $\efoot{\heap}{\env}{e.f} = \efoot{\heap}{\env}{e}; \pair{\ell}{f} \subseteq \perms$.

    \case $e_1 \oplus e_2$:
      \begin{align*}
        &\efoot{\heap}{\env}{e_1 \oplus e_2} \subseteq \perms \\
        &\iff \efoot{\heap}{\env}{e_1} \subseteq \perms ~\text{and}~ \efoot{\heap}{\env}{e_2} \subseteq \perms &\text{by definition} \\
        &\iff \frm{\heap}{\perms}{\env}{e_1} ~\text{and}~ \frm{\heap}{\perms}{\env}{e_2} &\text{by induction} \\
        &\iff \frm{\heap}{\perms}{\env}{e_1 \oplus e_2} &\text{by \refrule{FrameOp}}
      \end{align*}

    \case\label{case:efoot-framing-or} $e_1 \kor e_2$:
      \begin{align*}
        &\efoot{\heap}{\env}{e_1 \kor e_2} \subseteq \perms \implies  \\
        &~\textbf{either}~ \eval{\heap}{\env}{e_1}{\ktrue} ~\text{and}~ \efoot{\heap}{\env}{e_1} \subseteq \perms &\text{by definition}\\
        &\quad\implies \eval{\heap}{\env}{e_1}{\ktrue} ~\text{and}~ \frm{\heap}{\perms}{\env}{e_1} &\text{by induction} \\
        &\quad\implies \frm{\heap}{\perms}{\env}{e_1 \kor e_2} &\text{by \refrule{FrameOrA}} \\
        &~\textbf{or}~ \eval{\heap}{\env}{e_1}{\kfalse} ~\text{and}~ \efoot{\heap}{\env}{e_1} \cup \efoot{\heap}{\env}{e_2} \subseteq \perms &\text{by definition}\\
        &\quad\implies \eval{\heap}{\env}{e_1}{\kfalse}, ~\frm{\heap}{\perms}{\env}{e_1}, \\
          &\hspace{3.5em} \text{and}~ \frm{\heap}{\perms}{\env}{e_2} &\text{by induction} \\
        &\quad\implies \frm{\heap}{\perms}{\env}{e_1 \kor e_2} &\text{by \refrule{FrameOrB}}
      \end{align*}
      \begin{align*}
        &\frm{\heap}{\perms}{\env}{e_1 \kor e_2} \implies \\
        &~\textbf{either}~ \eval{\heap}{\env}{e_1}{\ktrue} ~\text{and}~ \frm{\heap}{\perms}{\env}{e_1} &\text{by \refrule{FrameOrA}} \\
        &\quad\implies \eval{\heap}{\env}{e_1}{\ktrue} ~\text{and}~ \efoot{\heap}{\env}{e_1} \subseteq \perms &\text{by induction} \\
        &\quad\implies \efoot{\heap}{\env}{e_1 \kor e_2} \subseteq \perms &\text{by definition} \\
        &~\textbf{or}~ \eval{\heap}{\env}{e_1}{\kfalse}, ~\frm{\heap}{\perms}{\env}{e_1}, \\
          &\hspace{2.25em}\text{and}~ \frm{\heap}{\perms}{\env}{e_2} &\text{by \refrule{FrameOrB}}\\
        &\quad\implies \eval{\heap}{\env}{e_1}{\kfalse} ~\text{and}~ \efoot{\heap}{\env}{e_1} \cup \efoot{\heap}{\env}{e_2} \subseteq \perms &\text{by induction}\\
        &\quad\implies \efoot{\heap}{\env}{e_1 \kor e_2} \subseteq \perms &\text{by definition}
      \end{align*}

    \case $e_1 \kand e_2$: Similar to case \ref{case:efoot-framing-or}.

    \case $\kneg e$:
      \begin{align*}
        &\efoot{\heap}{\env}{\kneg e} \subseteq \perms \\
        &\iff \efoot{\heap}{\env}{e} \subseteq \perms &\text{by definition} \\
        &\iff \frm{\heap}{\perms}{\env}{e} &\text{by induction} \\
        &\iff \frm{\heap}{\perms}{\env}{\kneg e} &\text{by \refrule{FrameNeg}}
      \end{align*}

  \end{enumcases}
\end{proof}

\begin{lemma}[Relating formula assertion/framing and exact footprints]\label{lem:efoot-subset-framed}
  If $\assertion{\heap}{\perms'}{\env}{\gform}$, $\efrm{\heap}{\perms}{\env}{\gform}$, and $\perms' \subseteq \perms$, then $\efoot{\heap}{\env}{\gform} \subseteq \perms$.
\end{lemma}
\begin{proof}
  By induction on $\assertion{\heap}{\perms'}{\env}{\gform}$:
  \begin{enumcases}
    \case \refrule{AssertImprecise} -- $\assertion{\heap}{\perms'}{\env}{\simprecise{\phi}}$:

      By \refrule{AssertImprecise}, $\assertion{\heap}{\perms'}{\env}{\phi}$ and $\efrm{\heap}{\perms}{\env}{\phi}$, therefore by induction $\efoot{\heap}{\env}{\simprecise{\phi}} = \efoot{\heap}{\env}{\phi} \subseteq \perms$.

    \case \refrule{AssertValue} -- $\assertion{\heap}{\perms'}{\env}{e}$:

      Since $\efrm{\heap}{\perms}{\env}{e}$, by \refrule{EFrameExpression} $\frm{\heap}{\perms}{\env}{e}$. Thus by lemma \ref{lem:efoot-framing}, $\efoot{\heap}{\env}{\gform} \subseteq \perms$.

    \case\label{case:efoot-subset-framed-ifa} \refrule{AssertIfA} -- $\assertion{\heap}{\perms'}{\env}{\sif{e}{\phi_1}{\phi_2}}$

      By \refrule{AssertIfA}, $\eval{\heap}{\env}{e}{\ktrue}$. Then, since $\efrm{\heap}{\perms}{\env}{\sif{e}{\phi_1}{\phi_2}}$, \\
      $\efrm{\heap}{\perms}{\env}{\phi_1}$ by \refrule{EFrameConditionalA}. Also, by \refrule{AssertIfA}, $\assertion{\heap}{\perms'}{\env}{\phi_1}$. Thus, by induction, \\
      $\efoot{\heap}{\env}{\sif{e}{\phi_1}{\phi_2}} = \efoot{\heap}{\env}{\phi_1} \subseteq \perms$.

    \case \refrule{AssertIfB} -- $\assertion{\heap}{\perms}{\env}{\sif{e}{\phi_1}{\phi_2}}$: Similar to case \ref{case:efoot-subset-framed-ifa}.

    \case \refrule{AssertAcc} -- $\assertion{\heap}{\perms'}{\env}{\kacc(e.f)}$

      Since $\efrm{\heap}{\perms}{\env}{\kacc(e.f)}$, by \refrule{EFrameAcc} $\frm{\heap}{\perms}{\env}{e}$, thus by lemma \ref{lem:efoot-framing} $\efoot{\heap}{\env}{e} \subseteq \perms$.
      By \refrule{AssertAcc}, $\eval{\heap}{\env}{e}{\ell}$ and $\pair{\ell}{f} \in \perms' \subseteq \perms$. Thus by definition $\efoot{\heap}{\env}{\kacc(e.f)} = \efoot{\heap}{\env}{e}; \pair{\ell}{f} \subseteq \perms$.

    \case \refrule{AssertConjunction} -- $\assertion{\heap}{\perms'}{\env}{\phi_1 * \phi_2}$

      Since $\efrm{\heap}{\perms}{\env}{\phi_1 * \phi_2}$, by \refrule{EFrameConjunction} $\efrm{\heap}{\perms}{\env}{\phi_1}$ and $\efrm{\heap}{\perms}{\env}{\phi_2}$. Also, by \refrule{AssertConjunction} $\assertion{\heap}{\perms_1}{\env}{\phi_1}$ and $\assertion{\heap}{\perms_2}{\env}{\phi_2}$ where $\perms_1 \cup \perms_2 \subseteq \perms' \subseteq \perms$. Therefore by induction $\efoot{\heap}{\env}{\phi_1 * \phi_2} = \efoot{\heap}{\env}{\phi_1} \cup \efoot{\heap}{\env}{\phi_2} \subseteq \perms$.

    \case \refrule{AssertPredicate} -- $\assertion{\heap}{\perms'}{\env}{p(\multiple{e})}$
    
      Since $\efrm{\heap}{\perms}{\env}{p(\multiple{e})}$, by \refrule{EFramePredicate} $\multiple{\eval{\heap}{\env}{e}{v}}$ and \\
      $\efrm{\heap}{\perms}{[\multiple{x \mapsto v}]}{\fpred(p)}$ where $\multiple{x} = \fpredparams(p)$. Also, by \refrule{AssertPredicate} $\assertion{\heap}{\perms'}{[\multiple{x \mapsto v}]}{\fpred(p)}$. Thus by induction, $\efoot{\heap}{[\multiple{x \mapsto v}]}{\fpred(p)} \subseteq \perms$.

      Also, by \refrule{EFramePredicate} $\multiple{\frm{\heap}{\perms}{\env}{e}}$ thus by \ref{lem:efoot-framing} $\efoot{\heap}{\env}{e} \subseteq \perms$ for all $e$.

      Therefore by definition $\efoot{\heap}{\env}{e} = \efoot{\heap}{[\multiple{x \mapsto v}]}{\fpred(p)} \cup \bigcup\multiple{\efoot{\heap}{\env}{e}} \subseteq \perms$.
  \end{enumcases}
\end{proof}

\begin{lemma}[Relating iso- and equi-recursive framing]\label{lem:ifrm-implies-efrm}
  If $\ifrm{\heap}{\perms}{\env}{\gform}$, $\assertion{\heap}{\perms'}{\env}{\gform}$, and $\perms' \subseteq \perms$, then $\efrm{\heap}{\perms}{\env}{\gform}$.
\end{lemma}
\begin{proof}
  By induction on $\assertion{\heap}{\perms'}{\env}{\gform}$:
  \begin{enumcases}
    \case \refrule{AssertImprecise} -- $\assertion{\heap}{\perms'}{\env}{\simprecise{\phi}}$:

      By \refrule{AssertImprecise} $\efrm{\heap}{\perms}{\env}{\phi}$, thus by \refrule{EFrameImprecise} $\efrm{\heap}{\perms}{\env}{\simprecise{\phi}}$

    \case \refrule{AssertValue} -- $\assertion{\heap}{\perms'}{\env}{e}$:

      Since $\ifrm{\heap}{\perms}{\env}{e}$, by \refrule{IFrameValue} $\frm{\heap}{\perms}{\env}{e}$. Then by \refrule{EFrameValue} $\efrm{\heap}{\perms}{\env}{e}$.

    \case\label{case:equi-implies-iso-framing-ifa} \refrule{AssertIfA} -- $\assertion{\heap}{\perms'}{\env}{\sif{e}{\phi_1}{\phi_2}}$:

      By \refrule{AssertIfA} $\eval{\heap}{\env}{e}{\ktrue}$ and $\assertion{\heap}{\perms'}{\env}{\phi_1}$. Then by \refrule{IFrameConditionalA} $\ifrm{\heap}{\perms}{\env}{\phi_1}$. Thus by induction $\efrm{\heap}{\perms}{\env}{\phi_1}$ and then by \refrule{EFrameConditionalA} \\
      $\efrm{\heap}{\perms}{\env}{\sif{e}{\phi_1}{\phi_2}}$.

    \case \refrule{AssertIfB}: Similar to case \ref{case:equi-implies-iso-framing-ifa}.

    \case \refrule{AssertAcc} -- $\assertion{\heap}{\perms'}{\env}{\kacc(e.f)}$

      Since $\ifrm{\heap}{\perms}{\env}{\kacc(e.f)}$, by \refrule{IFrameAcc} $\frm{\heap}{\perms}{\env}{e}$. Thus by \refrule{EFrameAcc} $\efrm{\heap}{\perms}{\env}{\kacc(e.f)}$.

    \case \refrule{AssertConjunction} -- $\assertion{\heap}{\perms'}{\env}{\phi_1 * \phi_2}$:

      Since $\ifrm{\heap}{\perms}{\env}{\phi_1 * \phi_2}$, by \refrule{IFrameConjunction} $\ifrm{\heap}{\perms}{\env}{\phi_1}$ and $\ifrm{\heap}{\perms}{\env}{\phi_2}$. Also, by \refrule{AssertConjunction} $\assertion{\heap}{\perms_1}{\env}{\phi_1}$ and $\assertion{\heap}{\perms_2}{\env}{\phi_2}$ where $\perms_1, \perms_2 \subseteq \perms$. Therefore by induction $\efrm{\heap}{\perms}{\env}{\phi_1}$ and $\efrm{\heap}{\perms}{\env}{\phi_2}$, and thus by \refrule{EFrameConjunction} $\efrm{\heap}{\perms}{\env}{\phi_1 * \phi_2}$.

    \case \refrule{AssertPredicate} -- $\assertion{\heap}{\perms'}{\env}{p(\multiple{e})}$:

      Since $\ifrm{\heap}{\perms}{\env}{p(\multiple{e})}$, by \refrule{IFramePredicate} $\multiple{\frm{\heap}{\perms}{\env}{e}}$.

      Also, by \refrule{AssertPredicate} $\multiple{\eval{\heap}{\env}{e}{v}}$ and $\assertion{\heap}{\perms'}{[\multiple{x \mapsto v}]}{\fpred(p)}$. Now, since $p$ is a predicate, $\fpred(p)$ must be a specification, and thus one of the following cases applies:

      \subcase $\fpred(p)$ is precise and self-framed: Then $\ifrm{\heap}{\perms}{[\multiple{x \mapsto v}]}{\fpred(p)}$ by definition \ref{def:self-framed}. Thus by induction $\efrm{\heap}{\perms}{[\multiple{x \mapsto v}]}{\fpred(p)}$. Then by \refrule{EFramePredicate} $\efrm{\heap}{\perms}{\env}{p(\multiple{e})}$.

      \subcase $\fpred(p)$ is imprecise: Then \refrule{AssertImprecise} must apply in order to derive \\
      $\assertion{\heap}{\perms'}{[\multiple{x \mapsto v}]}{\fpred(p)}$, and thus $\efrm{\heap}{\perms'}{[\multiple{x \mapsto v}]}{\fpred(p)}$, and thus by lemma \ref{lem:efrm-monotonicity} $\efrm{\heap}{\perms}{[\multiple{x \mapsto v}]}{\fpred(p)}$. Then by \refrule{EFramePredicate} $\efrm{\heap}{\perms}{\env}{p(\multiple{e})}$.
  \end{enumcases}
\end{proof}

\begin{lemma}[Relating specification assertion and exact footprints]\label{lem:efoot-subset-spec}
  If $\gform$ is a specification and $\assertion{\heap}{\perms}{\env}{\gform}$, then $\efoot{\heap}{\env}{\gform} \subseteq \perms$.
\end{lemma}
\begin{proof}
  Since $\gform$ is a specification, one of the following cases must apply:
  \begin{enumcases}
    \case $\gform$ is precise and self-framed: Then by definition \ref{def:self-framed} $\ifrm{\heap}{\perms}{\env}{\gform}$, and since $\assertion{\heap}{\perms}{\env}{\gform}$, by lemma \ref{lem:ifrm-implies-efrm} $\efrm{\heap}{\perms}{\env}{\gform}$. Therefore, by lemma \ref{lem:efoot-subset-framed} $\efoot{\heap}{\env}{\gform} \subseteq \perms$.

    \case $\gform$ is imprecise: Then $\gform = \simprecise{\phi}$ for some $\phi \in \Formula$, and \refrule{AssertImprecise} must apply in order to derive $\assertion{\heap}{\perms}{\env}{\gform}$. Therefore $\efrm{\heap}{\perms}{\env}{\phi}$, and thus by lemma \ref{lem:efoot-subset-framed} $\efoot{\heap}{\env}{\phi} \subseteq \perms$. Thus by definition $\efoot{\heap}{\env}{\gform} = \efoot{\heap}{\env}{\phi} \subseteq \perms$.
  \end{enumcases}
\end{proof}

\begin{lemma}[Relating specification assertion and footprints]\label{lem:foot-subset-spec}
  If $\gform$ is a specification and $\assertion{\heap}{\perms}{\env}{\gform}$, then $\foot{\heap}{\perms}{\env}{\gform} \subseteq \perms$.
\end{lemma}
\begin{proof}
  If $\gform$ is completely precise then $\foot{\heap}{\perms}{\env}{\gform} = \efoot{\heap}{\env}{\gform} \subseteq \perms$.

  Otherwise, $\foot{\heap}{\perms}{\env}{\gform} = \perms$.
\end{proof}

\begin{lemma}[Relating specification exact footprints and footprints]\label{lem:foot-subset-efoot-spec}
  If $\gform$ is a specification and $\assertion{\heap}{\perms}{\env}{\gform}$, then $\efoot{\heap}{\env}{\gform} \subseteq \foot{\heap}{\perms}{\env}{\gform}$.
\end{lemma}
\begin{proof}
  If $\gform$ is completely precise then $\efoot{\heap}{\env}{\gform} = \foot{\heap}{\perms}{\env}{\gform}$.

  Otherwise, by lemma \ref{lem:efoot-subset-spec} $\efoot{\heap}{\env}{\gform} \subseteq \perms = \foot{\heap}{\perms}{\env}{\gform}$.
\end{proof}

\begin{lemma}[Monotonicity of expression framing WRT permissions]\label{lem:framing-monotonicity}
  If $\frm{\heap}{\perms}{\env}{e}$ and $\perms \subseteq \perms'$, then $\frm{\heap}{\perms'}{\env}{e}$.
\end{lemma}
\begin{proof}
  By \ref{lem:efoot-framing}, $\efoot{\heap}{\env}{e} \subseteq \perms$, thus $\efoot{\heap}{\env}{e} \subseteq \perms'$, and therefore $\frm{\heap}{\perms'}{\env}{e}$ by \ref{lem:efoot-framing}.
\end{proof}

\begin{lemma}[Monotonicity of equi-recursive framing WRT permissions]\label{lem:efrm-monotonicity}
  If $\efrm{\heap}{\perms}{\env}{\gform}$ and $\perms \subseteq \perms'$, then $\efrm{\heap}{\perms'}{\env}{\gform}$.
\end{lemma}
\begin{proof}
  By induction on $\efrm{\heap}{\perms}{\env}{\gform}$:
  \begin{enumcases}
    \case\refrule{EFrameExpression} -- $\efrm{\heap}{\perms}{\env}{e}$: By lemma \ref{lem:framing-monotonicity}.
    \case\refrule{EFrameConjunction} -- $\efrm{\heap}{\perms}{\env}{\phi_1 * \phi_2}$: By induction
    \case\refrule{EFramePredicate} -- $\efrm{\heap}{\perms}{\env}{p(\multiple{e})}$: By induction and lemma \ref{lem:framing-monotonicity}.
    \case\refrule{EFrameConditionalA} -- $\efrm{\heap}{\perms}{\env}{\sif{e}{\phi_1}{\phi_2}}$: By induction and lemma \ref{lem:framing-monotonicity}.
    \case\refrule{EFrameConditionalB} -- $\efrm{\heap}{\perms}{\env}{\sif{e}{\phi_1}{\phi_2}}$: By induction and lemma \ref{lem:framing-monotonicity}.
    \case\refrule{EFrameAcc} -- $\efrm{\heap}{\perms}{\env}{\kacc(e)}$: By lemma \ref{lem:framing-monotonicity}.
  \end{enumcases}
\end{proof}

\begin{lemma}[Monotonicity of assertions WRT permissions]\label{lem:assert-monotonicity}
  If $\assertion{\heap}{\perms}{\env}{\gform}$ and $\perms \subseteq \perms'$, then $\assertion{\heap}{\perms'}{\env}{\gform}$.
\end{lemma}
\begin{proof}
  By induction on $\assertion{\heap}{\perms}{\env}{\gform}$:
  \begin{enumcases}
    \case\refrule{AssertImprecise} -- $\assertion{\heap}{\perms}{\env}{\simprecise{\phi}}$: By induction and lemma \ref{lem:efrm-monotonicity}.
    \case\refrule{AssertValue} -- $\assertion{\heap}{\perms}{\env}{e}$: Trivial.
    \case\refrule{AssertIfA} -- $\assertion{\heap}{\perms}{\env}{\sif{e}{\phi_1}{\phi_2}}$: By induction.
    \case\refrule{AssertIfB} -- $\assertion{\heap}{\perms}{\env}{\sif{e}{\phi_1}{\phi_2}}$: By induction.
    \case\refrule{AssertAcc} -- $\assertion{\heap}{\perms}{\env}{\kacc(e.f)}$: Then $\eval{\heap}{\env}{e}{\ell}$ and $\pair{\ell}{f} \in \perms$, thus $\pair{\ell}{f} \in \perms'$. Therefore $\assertion{\heap}{\perms'}{\env}{\kacc(e.f)}$ by \refrule{AssertAcc}.
    \case\refrule{AssertConjunction} -- $\assertion{\heap}{\perms}{\env}{\phi_1 * \phi_2}$:
      Then $\assertion{\heap}{\perms}{\env}{\phi_1}$ and $\assertion{\heap}{\perms_2}{\env}{\phi_2}$ where $\perms_1 \cap \perms_2 = \emptyset$ and $\perms_1 \cup \perms_2 \subseteq \perms \subseteq \perms'$. Therefore $\assertion{\heap}{\perms'}{\env}{\phi_1 * \phi_2}$ by \refrule{AssertConjunction}.
    \case\refrule{AssertPredicate} -- $\assertion{\heap}{\perms}{\env}{p(\multiple{e})}$: By induction.
  \end{enumcases}
\end{proof}

\begin{lemma}[Exact footprint preserves equi-recursive framing]\label{lem:efoot-efrm}
  If $\efrm{\heap}{\perms}{\env}{\gform}$, then \\
  $\efrm{\heap}{\efoot{\heap}{\env}{\gform} \cap \perms}{\env}{\gform}$.
\end{lemma}
\begin{proof}
  By induction on $\efrm{\heap}{\perms}{\env}{\gform}$:
  \begin{enumcases}
    \case\label{case:efoot-efrm-expr} \refrule{EFrameExpression} -- $\efrm{\heap}{\perms}{\env}{e}$:
      By \refrule{EFrameExpression} $\frm{\heap}{\perms}{\env}{e}$, thus $\efoot{\heap}{\env}{e} \subseteq \perms$ and $\frm{\heap}{\efoot{\heap}{\env}{e}}{\env}{e}$ by lemma \ref{lem:efoot-framing}. Then $\efoot{\heap}{\env}{e} \cap \perms = \efoot{\heap}{\env}{e}$. Therefore $\efrm{\heap}{\efoot{\heap}{\env}{e} \cap \perms}{\env}{e}$ by \refrule{EFrameExpression}.

    \case \refrule{EFrameConjunction} -- $\efrm{\heap}{\perms}{\env}{e_1 * e_2}$:
      By \refrule{EFrameConjunction} $\efrm{\heap}{\perms}{\env}{e_1}$ and $\efrm{\heap}{\perms}{\env}{e_2}$, thus by induction $\efrm{\heap}{\efoot{\heap}{\env}{e_1} \cap \perms}{\env}{e_1}$ and $\efrm{\heap}{\efoot{\heap}{\env}{e_2} \cap \perms}{\env}{e_2}$.
      
      By definition $\efoot{\heap}{\env}{e_1 * e_2} \cap \perms = (\efoot{\heap}{\env}{e_1} \cap \perms) \cup (\efoot{\heap}{\env}{e_2} \cap \perms)$, thus $\efrm{\heap}{\efoot{\heap}{\env}{e_1 * e_2} \cap \perms}{\env}{e_1}$ and $\efrm{\heap}{\efoot{\heap}{\env}{e_1 * e_2} \cap \perms}{\env}{e_2}$ by lemma \ref{lem:efrm-monotonicity}. Therefore $\efrm{\heap}{\efoot{\heap}{\env}{e_1 * e_2} \cap \perms}{\env}{e_1 * e_2}$ by \refrule{EFrameConjunction}.

    \case \refrule{EFramePredicate} -- $\efrm{\heap}{\perms}{\env}{p(\multiple{e})}$:
      By \refrule{EFramePredicate} $\multiple{\frm{\heap}{\perms}{\env}{e}}$, thus by lemma \ref{lem:efoot-framing} $\multiple{\efoot{\heap}{\env}{e} \subseteq \perms}$, thus $\multiple{\efoot{\heap}{\env}{e} \cap \perms = \efoot{\heap}{\env}{e}}$, and then $\multiple{\frm{\heap}{\efoot{\heap}{\env}{e} \cap \perms}{\env}{e}}$.

      By \refrule{EFramePredicate} $\multiple{\eval{\heap}{\env}{e}{v}}$ and $\efrm{\heap}{\perms}{[\multiple{x \mapsto v}]}{\fpred(p)}$ where \\
      $\multiple{x} = \fpredparams(p)$. Thus by induction $\efrm{\heap}{\efoot{\heap}{[\multiple{x \mapsto v}]}{\fpred(p)} \cap \perms}{[\multiple{x \mapsto v}]}{\fpred(p)}$.

      Now by definition $\efoot{\heap}{\perms}{p(\multiple{e})} \cap \perms = (\efoot{\heap}{[\multiple{x \mapsto v}]}{\fpred(p)} \cap \perms) \cup \bigcup\multiple{(\efoot{\heap}{\env}{e} \cap \perms)}$. Therefore by lemma \ref{lem:efrm-monotonicity} $\efrm{\heap}{\efoot{\heap}{\env}{p(\multiple{e})} \cap \perms}{[\multiple{x \mapsto v}]}{\fpred(p)}$ and \\
      $\multiple{\frm{\heap}{\efoot{\heap}{\env}{p(\multiple{e})} \cap \perms}{\env}{e}}$.

      Thus $\efrm{\heap}{\efoot{\heap}{\env}{p(\multiple{e})} \cap \perms}{\env}{p(\multiple{e})}$ by \refrule{EFramePredicate}.

    \case\label{case:efoot-efrm-ifa} \refrule{EFrameConditionalA} -- $\efrm{\heap}{\perms}{\env}{\sif{e}{\phi_1}{\phi_2}}$:
      By \refrule{EFrameConditionalA} $\efrm{\heap}{\perms}{\env}{\phi_1}$ and $\frm{\heap}{\perms}{\env}{e}$, thus by induction $\efrm{\heap}{\efoot{\heap}{\env}{\phi_1} \cap \perms}{\env}{\phi_1}$, and by lemma \ref{lem:efoot-framing} $\efoot{\heap}{\env}{e} \subseteq \perms$, thus $\efoot{\heap}{\env}{e} \cap \perms = \efoot{\heap}{\env}{e}$, and finally $\frm{\heap}{\efoot{\heap}{\env}{e} \cap \perms}{\env}{e}$.

      Also by \refrule{EFrameConditionalA} $\eval{\heap}{\env}{e}{\ktrue}$, thus by definition $\efoot{\heap}{\env}{\sif{e}{\phi_1}{\phi_2}} \cap \perms = (\efoot{\heap}{\env}{e} \cap \perms) \cup (\efoot{\heap}{\env}{\phi_1} \cap \perms)$. Therefore $\efrm{\heap}{\efoot{\heap}{\env}{\sif{e}{\phi_1}{\phi_2}} \cap \perms}{\env}{\phi_1}$ by lemma \ref{lem:efrm-monotonicity} and $\frm{\heap}{\efoot{\heap}{\env}{\sif{e}{\phi_1}{\phi_2}} \cap \perms}{\env}{e}$ by \ref{lem:framing-monotonicity}, thus \\
      $\efrm{\heap}{\efoot{\heap}{\env}{\sif{e}{\phi_1}{\phi_2}} \cap \perms}{\env}{\sif{e}{\phi_1}{\phi_2}}$ by \refrule{EFrameConditionalA}.

    \case \refrule{EFrameConditionalB} -- $\efrm{\heap}{\perms}{\env}{\sif{e}{\phi_1}{\phi_2}}$: Similar to case \ref{case:efoot-efrm-ifa}.

    \case \refrule{EFrameAcc} -- $\efrm{\heap}{\perms}{\env}{\kacc(e.f)}$:
      By \refrule{EFrameAcc} $\frm{\heap}{\perms}{\env}{e}$, thus by lemma \ref{lem:efoot-framing} $\efoot{\heap}{\env}{e} \subseteq \perms$, thus $\efoot{\heap}{\env}{e} \cap \perms = \efoot{\heap}{\env}{e}$, and thus $\frm{\heap}{\efoot{\heap}{\env}{e} \cap \perms}{\env}{e}$. By definition $\efoot{\heap}{\env}{e} \subseteq \efoot{\heap}{\env}{\kacc(e.f)}$, thus $\frm{\heap}{\efoot{\heap}{\env}{\kacc(e.f)} \cap \perms}{\env}{e}$ by \ref{lem:framing-monotonicity}. Therefore $\efrm{\heap}{\efoot{\heap}{\env}{\kacc(e.f)} \cap \perms}{\env}{\kacc(e.f)}$ by \refrule{EFrameAcc}.
  \end{enumcases}
\end{proof}

\begin{lemma}[Exact footprint preserves assertions]\label{lem:efoot-assert}
  If $\assertion{\heap}{\perms}{\env}{\gform}$, then $\assertion{\heap}{\efoot{\heap}{\env}{\gform} \cap \perms}{\env}{\gform}$.
\end{lemma}
\begin{proof}
  By induction on $\assertion{\heap}{\perms}{\env}{\gform}$:
  \begin{enumcases}
    \case \refrule{AssertImprecise} -- $\assertion{\heap}{\perms}{\env}{\simprecise{\phi}}$:
      By \refrule{AssertImprecise} $\assertion{\heap}{\perms}{\env}{\phi}$, thus by induction $\assertion{\heap}{\efoot{\heap}{\env}{\phi} \cap \perms}{\phi}$.

      Also by \refrule{AssertImprecise} $\efrm{\heap}{\perms}{\env}{\phi}$, thus by lemma \ref{lem:efoot-efrm} $\efrm{\heap}{\efoot{\heap}{\env}{\phi} \cap \perms}{\phi}$.

      By definition $\efoot{\heap}{\env}{\simprecise{\phi}} = \efoot{\heap}{\env}{\phi}$, thus $\assertion{\heap}{\efoot{\heap}{\env}{\simprecise{\phi}} \cap \perms}{\phi}$ and $\efrm{\heap}{\efoot{\heap}{\env}{\simprecise{\phi}} \cap \perms}{\phi}$. Therefore $\assertion{\heap}{\efoot{\heap}{\env}{\simprecise{\phi}} \cap \perms}{\simprecise{\phi}}$ by \refrule{AssertImprecise}.

    \case \refrule{AssertValue} -- $\assertion{\heap}{\perms}{\env}{e}$:
      By \refrule{AssertValue} $\eval{\heap}{\env}{e}{\ktrue}$, thus $\assertion{\heap}{\efoot{\heap}{\env}{e}}{\env}{e}$ by \refrule{AssertValue}.

    \case\label{case:efoot-assert-ifa}\refrule{AssertIfA} -- $\assertion{\heap}{\perms}{\env}{\sif{e}{\phi_1}{\phi_2}}$:
      By \refrule{AssertIfA} $\assertion{\heap}{\perms}{\env}{\phi_1}$, thus by induction $\assertion{\heap}{\efoot{\heap}{\env}{\phi_1} \cap \perms}{\env}{\phi_1}$.

      Also by \refrule{AssertIfA} $\eval{\heap}{\env}{e}{\ktrue}$, thus $\efoot{\heap}{\env}{\sif{e}{\phi_1}{\phi_2}} = \efoot{\heap}{\env}{e} \cup \efoot{\heap}{\env}{\phi_1}$. Therefore $\assertion{\heap}{\efoot{\heap}{\env}{\sif{e}{\phi_1}{\phi_2}} \cap \perms}{\env}{\phi_1}$ by \ref{lem:assert-monotonicity}.

      Thus $\assertion{\heap}{\efoot{\heap}{\env}{\sif{e}{\phi_1}{\phi_2}} \cap \perms}{\env}{\sif{e}{\phi_1}{\phi_2}}$ by \refrule{AssertIfA}.

    \case\refrule{AssertIfB} -- $\assertion{\heap}{\perms}{\env}{\sif{e}{\phi_1}{\phi_2}}$:
      Similar to case \ref{case:efoot-assert-ifa}.

    \case\refrule{AssertAcc} -- $\assertion{\heap}{\perms}{\env}{\kacc(e.f)}$:
      By \refrule{AssertAcc} $\eval{\heap}{\env}{e}{\ell}$ and $\pair{\ell}{f} \in \perms$. By definition $\efoot{\heap}{\env}{\kacc(e.f)} = \efoot{\heap}{\env}{e} \cup \set{\pair{\ell}{f}}$, thus $\pair{\ell}{f} \in \efoot{\heap}{\env}{\kacc(e.f)}$. Therefore $\assertion{\heap}{\efoot{\heap}{\env}{\kacc(e.f)} \cap \perms}{\env}{\kacc(e.f)}$ by \refrule{AssertAcc}.

    \case\refrule{AssertConjunction} -- $\assertion{\heap}{\perms}{\env}{\phi_1 * \phi_2}$:
      Let $\perms_1' = \efoot{\heap}{\env}{\phi_1} \cap \perms_1$ and $\perms_2' = \efoot{\heap}{\env}{\phi_2} \cap \perms_2$. By \refrule{AssertConjunction} $\assertion{\heap}{\perms_1}{\env}{\phi_1}$ and $\assertion{\heap}{\perms_2}{\env}{\phi_2}$, thus by induction $\assertion{\heap}{\perms_1'}{\env}{\phi_1}$ and $\assertion{\heap}{\perms_2'}{\env}{\phi_2}$.

      Also by \refrule{AssertConjunction} $\perms_1 \cap \perms_2 = \emptyset$ and $\perms_1 \cup \perms_2 \subseteq \perms$, thus $\perms_1' \cap \perms_2' = \emptyset$ and $\perms_1' \cup \perms_2' \subseteq \efoot{\heap}{\env}{\phi_1 * \phi_2} = \efoot{\heap}{\env}{\phi_1} \cup \efoot{\heap}{\env}{\phi_2}$.

      Therefore $\assertion{\heap}{\efoot{\heap}{\env}{\phi_1 * \phi_2}}{\env}{\phi_1 * \phi_2}$ by \refrule{AssertConjunction}.

    \case\refrule{AssertPredicate} -- $\assertion{\heap}{\perms}{\env}{p(\multiple{e})}$:
      By \refrule{AssertPredicate} $\multiple{\eval{\heap}{\env}{e}{v}}$ and \\
      $\assertion{\heap}{\perms}{[\multiple{x \mapsto v}]}{\fpred(p)}$ where $\multiple{x} = \fpredparams(p)$. Thus by induction \\
      $\assertion{\heap}{\efoot{\heap}{[\multiple{x \mapsto v}]}{\fpred(p)} \cap \perms}{[\multiple{x \mapsto v}]}{\fpred(p)}$.

      Now by definition $\efoot{\heap}{[\multiple{x \mapsto v}]}{\fpred(p)} \subseteq \efoot{\heap}{\env}{p(\multiple{e})}$, thus by lemma \ref{lem:assert-monotonicity} $\assertion{\heap}{\efoot{\heap}{\env}{p(\multiple{e})} \cap \perms}{[\multiple{x \mapsto v}]}{\fpred(p)}$. Therefore $\assertion{\heap}{\efoot{\heap}{\env}{p(\multiple{e})} \cap \perms}{\env}{p(\multiple{e})}$ by \refrule{AssertPredicate}.
  \end{enumcases}
\end{proof}

\begin{lemma}[Supersets of exact footprints preserve assertions]\label{lem:assert-efoot-subset}
  If $\assertion{\heap}{\perms}{\env}{\gform}$ and $\efoot{\heap}{\env}{\gform} \subseteq \perms'$, then $\assertion{\heap}{\perms'}{\env}{\gform}$.
\end{lemma}
\begin{proof}
  By lemma \ref{lem:efoot-assert} $\assertion{\heap}{\efoot{\heap}{\env}{\gform} \cap \perms}{\env}{\gform}$. Then $\efoot{\heap}{\env}{\gform} \cap \perms \subseteq \perms'$, thus $\assertion{\heap}{\perms'}{\env}{\gform}$ by lemma \ref{lem:assert-monotonicity}.
\end{proof}

\begin{lemma}[Footprints preserve specification assertions]\label{lem:foot-assert}
  If $\gform$ is a specification and $\assertion{\heap}{\perms}{\env}{\gform}$, then $\assertion{\heap}{\foot{\heap}{\perms}{\env}{\gform}}{\env}{\gform}$.
\end{lemma}
\begin{proof}
  By lemma \ref{lem:foot-subset-efoot-spec}, $\efoot{\heap}{\env}{\gform} \subseteq \foot{\heap}{\perms}{\env}{\gform}$. Therefore $\assertion{\heap}{\foot{\heap}{\perms}{\env}{\gform}}{\env}{\gform}$ by lemma \ref{lem:assert-efoot-subset}.
\end{proof}

\begin{lemma}[Modeling implies ownership]\label{lem:sim-heap-contains}
  If $\simstate{V}{\sstate}{\heap}{\perms}{\env}$ then
  $$\universal{h \in \sheap(\sstate) \cup \oheap(\sstate)}{\vfoot{V}{\heap}{h} \subseteq \perms}.$$
\end{lemma}
\begin{proof}
  Let $h$ be an arbitrary element of $\sheap(\sstate)$ or $\oheap(\sstate)$. Then one of the following cases applies:

  \begin{enumcases}
    \case $h = \triple{f}{t}{t'}$ for some $f, t, t'$:

      Since $\simheap{V}{\sheap(\sstate)}{\heap}{\perms}$ (and likewise for $\oheap(\sstate)$), then $\pair{V(t)}{f} \in \perms$. Therefore $\vfoot{V}{\heap}{\triple{f}{t}{t'}} = \set{ \pair{V(t)}{f} } \subseteq \perms$.

    \case $h = \pair{p}{\multiple{t}}$ for some $p, \multiple{t}$:

      Then $h \in \sheap(\sstate)$ since $\oheap(\sstate)$ does not contain predicate chunks. Since $\simheap{V}{\sheap(\sstate)}{\heap}{\perms}$, \\
      $\assertion{\heap}{\perms}{[\multiple{x \mapsto V(t)}]}{\fpred(p)}$. $\fpred(p)$ must be a specification, thus by lemma \ref{lem:efoot-subset-spec} \\
      $\vfoot{V}{\heap}{\pair{p}{\multiple{t}}} = \efoot{\heap}{[\multiple{x \mapsto V(t)}]}{\fpred(p)} \subseteq \perms$.
  \end{enumcases}
\end{proof}

\begin{lemma}[Correspondence with exclusion implies disjointness]\label{lem:sim-heap-disjoint}
  If $\simstate{V}{\sstate}{\heap}{\perms \setminus \perms'}{\env}$, then
  $$\universal{h \in \sheap(\sstate)}{\vfoot{V}{\heap}{h} \cap \perms' = \emptyset}.$$
\end{lemma}
\begin{proof}
  Let $h$ be an arbitrary element of $\sheap(\sstate)$. Then $\vfoot{V}{\heap}{h} \subseteq \perms \setminus \perms'$ by lemma \ref{lem:sim-heap-contains}, thus $\vfoot{V}{\heap}{h} \cap \perms' = \emptyset$.
\end{proof}

\begin{lemma}[Disjointness implies correspondence with exclusion]\label{lem:disjoint-sim-heap-subset}
  If $\simheap{V}{\sheap}{\heap}{\perms}$ and $\universal{h \in \sheap}{\vfoot{V}{\heap}{h} \cap \perms' = \emptyset}$, then $\simheap{V}{\sheap}{\heap}{\perms \setminus \perms'}$.
\end{lemma}
\begin{proof}
  Since $\simheap{V}{\sheap}{\heap}{\perms}$,
  $$\universal{\triple{f}{t}{t'} \in \sheap}{\heap(V(t), f) = V(t')}.$$

  Let $\triple{f}{t}{t'}$ be an arbitrary field chunk in $\sheap$. Then, by definition, $\vfoot{V}{\heap}{\triple{f}{t}{t'}} = \set{\pair{V(t)}{f}}$. Thus $\set{\pair{V(t)}{f}} \cap \perms' = \emptyset$ and then $\pair{V(t)}{f} \notin \perms'$. Therefore $\pair{V(t)}{f} \in (\perms \setminus \perms')$. Thus
  $$\universal{\triple{f}{t}{t'} \in \sheap}{\pair{V(t)}{f} \in (\perms \setminus \perms')}.$$

  Let $\pair{p}{\multiple{t}}$ be an arbitrary predicate in $\sheap$. Since $\simheap{V}{\sheap}{\heap}{\perms}$, $\assertion{\heap}{\perms}{[\multiple{x \mapsto V(t)}]}{\fpred(p)}$.
  By definition, $\vfoot{V}{\heap}{\pair{p}{\multiple{t}}} = \efoot{\heap}{[\multiple{x \mapsto V(t)}]}{\fpred(p)}$. Thus $\efoot{\heap}{[\multiple{x \mapsto V(t)}]}{\fpred(p)} \cap \perms' = \emptyset$, and therefore $\efoot{\heap}{[\multiple{x \mapsto V(t)}]}{\fpred(p)} \subseteq (\perms \setminus \perms')$. Thus by lemma \ref{lem:assert-efoot-subset} $\assertion{\heap}{\perms \setminus \perms'}{[\multiple{x \mapsto V(t)}]}{\fpred(p)}$. Therefore
  $$\universal{\pair{p}{\multiple{t}} \in \sheap}{\assertion{\heap}{\perms \setminus \perms'}{[\multiple{x \mapsto V(t)}]}{\fpred(p)}}.$$

  Finally, since $\simheap{V}{\sheap}{\heap}{\perms}$,
  $$\universal{h_1, h_2 \in \sheap^2}{h_1 \ne h_2 \implies \efoot{V}{\heap}{h_1} \cap \efoot{V}{\heap}{h_2} = \emptyset}.$$

  Thus all conditions in \eqref{eq:sheap-correspondence} are satisfied, therefore $\simheap{V}{\sheap}{\heap}{\perms \setminus \perms'}$.
\end{proof}

\begin{lemma}\label{lem:sim-sheap-monotonicity}
  If $\simheap{V}{\sheap}{\heap}{\perms}$ and $\perms \subseteq \perms'$, then $\simheap{V}{\sheap}{\heap}{\perms}$.
\end{lemma}
\begin{proof}
  Let $\triple{f}{t}{t'} \in \sheap$. Since $\simheap{V}{\sheap}{\heap}{\perms}$, $\heap(V(t), f) = V(t')$ and $\pair{V(t)}{f} \in \perms \subseteq \perms'$. Therefore
  \begin{gather*}
    \universal{\triple{f}{t}{t'} \in \sheap}{\heap(V(t), f) = V(t')} \quad\text{and} \\
    \universal{\triple{f}{t}{t'} \in \sheap}{\pair{V(t)}{f} \in \perms'}.
  \end{gather*}

  Let $\pair{p}{\multiple{t}} \in \sheap$. Since $\simheap{V}{\sheap}{\heap}{\perms}$, $\assertion{\heap}{\perms}{[\multiple{x \mapsto V(t)}]}{\fpred(p)}$. Then by lemma \ref{lem:assert-monotonicity}, $\assertion{\heap}{\perms'}{[\multiple{x \mapsto V(t)}]}{\fpred(p)}$. Therefore
  $$\universal{\pair{p}{\multiple{t}} \in \sheap}{\assertion{\heap}{\perms'}{[\multiple{x \mapsto V(t)}]}{\fpred(p)}}.$$

  Now, since $\simheap{V}{\sheap}{\heap}{\perms}$,
  $$\universal{h_1, h_2 \in \sheap}{h_1 \ne h_2 \implies \vfoot{V}{\heap}{h_1} \cap \vfoot{V}{\heap}{h_2} = \emptyset}.$$

  Therefore $\simheap{V}{\sheap}{\heap}{\perms'}$.
\end{proof}

\begin{lemma}\label{lem:sim-oheap-monotonicity}
  If $\simheap{V}{\sheap}{\heap}{\perms}$ and $\perms \subseteq \perms'$, then $\simheap{V}{\sheap}{\heap}{\perms}$.
\end{lemma}
\begin{proof}
  Similar to proof of \ref{lem:sim-sheap-monotonicity}.
\end{proof}

\begin{lemma}\label{lem:simstate-monotonicity}
  If $\simstate{V}{\sstate}{\heap}{\perms}{\env}$ and $\perms \subseteq \perms'$, then $\simstate{V}{\sstate}{\heap}{\perms'}{\env}$.
\end{lemma}
\begin{proof}
  Since $\simstate{V}{\sstate}{\heap}{\perms}{\env}$, $V(\pc(\sstate)) = \ktrue$ and $\simenv{V}{\senv(\sstate)}{\env}$.

  Also by lemma \ref{lem:sim-sheap-monotonicity} $\simheap{V}{\sheap(\sstate)}{\heap}{\perms'}$ and by lemma \ref{lem:sim-oheap-monotonicity} $\simheap{V}{\oheap(\sstate)}{\heap}{\perms'}$.

  Therefore $\simstate{V}{\sstate'}{\heap}{\perms'}{\env}$.
\end{proof}

\begin{lemma}[Statement rearrangement]\label{lem:stmt-rearrangement}
  If $s = s_1; s_2$, then $s = s_1'; s_2'$ such that $s_1'$ is not a sequence statement and $s_2' = s_2$ or $s_2' = s_1''; s_2$ where $s_1 = s_1'; s_1''$.
\end{lemma}

\begin{proof}
  Suppose $s = s_1; s_2$. Complete the proof by induction on the syntax forms of $s_1$:

  \begin{enumcases}
    \case $s_1$ is not a sequence statement -- Let $s_1' = s_1$ and $s_2' = s_2$. Then $s = s_1'; s_2'$ where $s_1'$ is not a sequence statement and $s_2 = s_2'$.

    \case $s_1 = s_{11}; s_{12}$ -- Then $s = (s_{11}; s_{12}); s_2$.

      Applying the inductive hypothesis on $s_1$, $s_1 = s_{11}'; s_{12}'$ such that $s_{11}'$ is not a sequence statement and one of the following cases apply:
      
      \subcase $s_{12}' = s_{12}$ -- Let $s_1' = s_{11}'$, $s_1'' = s_{12}$, and $s_2' = s_{12}; s_2$. Then
        \begin{align*}
          s &= s_1; s_2 = (s_{11}'; s_{12}'); s_2 = (s_1'; s_{12}); s_2 = s_1'; (s_{12}; s_2) = s_1'; s_2' \\
          s_2' &= s_{12}; s_2 = s_1''; s_2 \\
          s_1 &= s_{11}'; s_{12}' = s_1'; s_{12} = s_1'; s_1''
        \end{align*}
        which satisfies the original statement.
      
      \subcase $s_{12}' = s_{11}''; s_{12}$ where $s_{11} = s_{11}'; s_{11}''$ -- Let $s_1' = s_{11}'$, $s_1'' = s_{11}''; s_{12}$, and $s_2' = s_{11}''; s_{12}; s_2$. Then
        \begin{align*}
          s &= s_1; s_2 = (s_{11}'; s_{12}'); s_2 = (s_{11}'; (s_{11}''; s_{12})); s_2 = s_{11}'; (s_{11}''; s_{12}; s_2) = s_1'; s_2' \\
          s_2' &= s_{11}''; s_{12}; s_2 = (s_{11}''; s_{12}); s_2 = s_1''; s_2 \\
          s_1 &= s_{11}'; s_{12}' = s_{11}'; (s_{11}''; s_{12}) = s_1'; s_1''
        \end{align*}
        which satisfies the original statement.
  \end{enumcases}
\end{proof}

\begin{lemma}[Run-time check monotonicity WRT permissions]\label{lem:scheck-perms-monotonicity-elem}
  If $\rtassert{V}{\heap}{\perms}{r}$ and $\perms \subseteq \perms'$, then $\rtassert{V}{\heap}{\perms'}{r}$.
\end{lemma}
\begin{proof}
  By cases on $\rtassert{V}{\heap}{\perms}{r}$:
  \begin{enumcases}
    \case\refrule{CheckValue} -- $\rtassert{V}{\heap}{\perms}{t}$:
      By \refrule{CheckValue} $V(t) = \ktrue$, thus by \refrule{CheckValue} $\rtassert{V}{\heap}{\perms'}{t}$.
    \case\refrule{CheckAcc} -- $\rtassert{V}{\heap}{\perms}{\pair{f}{t}}$:
      By \refrule{CheckAcc} $\pair{V(t)}{f} \in \perms$, thus $\pair{V(t)}{f} \in \perms'$ and therefore $\rtassert{V}{\heap}{\perms'}{\pair{t}{f}}$ by \refrule{CheckAcc}.
    \case\refrule{CheckPred} -- $\rtassert{V}{\heap}{\perms}{\pair{p}{\multiple{t}}}$:
      By \refrule{CheckPred} $\assertion{\heap}{\perms}{[\multiple{x \mapsto V(t)}]}{\fpred(p)}$ where $\multiple{x} = \fpredparams(p)$. Then by lemma \ref{lem:assert-monotonicity}, $\assertion{\heap}{\perms'}{[\multiple{x \mapsto V(t)}]}{\fpred(p)}$. Therefore $\rtassert{V}{\heap}{\perms'}{\pair{p}{\multiple{t}}}$ by \refrule{CheckPred}.
    \refrule{CheckSep} -- $\rtassert{V}{\heap}{\perms}{\fsep(\sperms_1, \sperms_2)}$:
      By \refrule{CheckSep} $\vfoot{V}{\heap}{\sperms_1} \cap \vfoot{V}{\heap}{\sperms_2} = \emptyset$, therefore $\rtassert{V}{\heap}{\perms}{\fsep(\sperms_1, \sperms_2)}$ by \refrule{CheckSep}.
  \end{enumcases}
\end{proof}

\begin{lemma}[Run-time checks monotonicity WRT permissions]\label{lem:scheck-perms-monotonicity}
  If $\rtassert{V}{\heap}{\perms}{\scheck}$ and $\perms \subseteq \perms'$, then $\rtassert{V}{\heap}{\perms'}{\scheck}$.
\end{lemma}
\begin{proof}
  Then by definition $\universal{r \in \scheck}{\rtassert{V}{\heap}{\perms}{r}}$ and then by lemma \ref{lem:scheck-perms-monotonicity-elem} $\universal{r \in \scheck}{\rtassert{V}{\heap}{\perms'}{r}}$. Thus $\rtassert{V}{\heap}{\perms'}{\scheck}$.
\end{proof}

\begin{lemma}[Run-time check set implies subsets]\label{lem:scheck-monotonicity}
  If $\rtassert{V}{\heap}{\perms}{\scheck}$ and $\scheck' \subseteq \scheck$, then $\rtassert{V}{\heap}{\perms}{\scheck'}$.
\end{lemma}
\begin{proof}
  By definition $\universal{r \in \scheck}{\rtassert{V}{\heap}{\perms}{r}}$ and $\universal{r \in \scheck'}{r \in \scheck}$, therefore $\universal{r \in \scheck'}{\rtassert{V}{\heap}{\perms}{r}}$. Thus $\rtassert{V}{\heap}{\perms}{\scheck'}$.
\end{proof}

\subsection{Evaluation}


\begin{definition}\label{def:eval-valuation}
  For a judgement $\seval{\sstate}{e}{t}{\sstate'}{\scheck}$, given an initial valuation $V$ and heap $\heap$, the \textbf{corresponding valuation} is denoted
  $$V[\seval{\sstate}{e}{t}{\sstate'}{\scheck} \mid \heap].$$
  This valuation is defined as follows, depending on the rule that proves the derivation. Values are referenced using the respective name from the rule definition.

  Note that the corresponding valuation always extends the initial valuation, and is defined for all fresh symbolic values in the judgement.

  \begin{itemize}
    \item \refrule{SEvalLiteral}:
      $$V[\seval{\sstate}{l}{l}{\sstate}{\_} \mid \heap] := V$$
    \item \refrule{SEvalVar}:
      $$V[\seval{\sstate}{x}{\_}{\sstate}{\_} \mid \heap] := V$$
    \item \refrule{SEvalNeg}:
      $$V[\seval{\sstate}{\kneg e}{\kneg t}{\sstate'}{\scheck} \mid \heap] := V[\seval{\sstate}{e}{t}{\sstate'}{\scheck} \mid \heap]$$
    \item \refrule{SEvalOrA}:
      $$V[\seval{\sstate}{e_1 \kor e_2}{t_1}{\sstate''}{\scheck} \mid \heap] := V[\seval{\sstate}{e_1}{t_1}{\sstate'}{\scheck} \mid \heap]$$
    \item \refrule{SEvalOrB}:
      \begin{align*}
        &V[\seval{\sstate}{e_1 \kor e_2}{t_2}{\sstate''}{\scheck_1 \cup \scheck_2} \mid \heap] := \\
        &\hspace{2em} V[\seval{\sstate}{e_1}{t_1}{\sstate'}{\scheck_1} \mid \heap] \\
        &\hspace{3em} [\seval{\sstate'[\pc = \pc(\sstate') \kand \kneg t_1]}{e_2}{t_2}{\sstate''}{\scheck_2} \mid \heap]
      \end{align*}
    \item \refrule{SEvalAndA}:
      $$V[\seval{\sstate}{e_1 \kand e_2}{t_1}{\sstate''}{\scheck} \mid \heap] := V[\seval{\sstate}{e_1}{t_1}{\sstate'}{\scheck} \mid \heap]$$
    \item \refrule{SEvalAndB}:
      \begin{align*}
        &V[\seval{\sstate}{e_1 \kand e_2}{t_2}{\sstate''}{\scheck_1 \cup \scheck_2} \mid \heap] := \\
        &\hspace{2em} V[\seval{\sstate}{e_1}{t_1}{\sstate'}{\scheck_1} \mid \heap] \\
        &\hspace{3em} [\seval{\sstate'[\pc = \pc(\sstate') \kand t_1]}{e_2}{t_2}{\sstate''}{\scheck_2} \mid \heap]
      \end{align*}
    \item \refrule{SEvalOp}:
      \begin{align*}
        &V[\seval{\sstate}{e_1 \oplus e_2}{t_1 \oplus t_2}{\sstate''}{\scheck_1 \cup \scheck_2} \mid \heap] := \\
        &\hspace{2em} V[\seval{\sstate}{e_1}{t_1}{\sstate'}{\scheck_1} \mid \heap] \\
        &\hspace{3em} [\seval{\sstate'}{e_2}{t_2}{\sstate''}{\scheck_2} \mid \heap]
      \end{align*}
    \item \refrule{SEvalField} or \refrule{SEvalFieldOptimistic}:
      $$V[\seval{\sstate}{e.f}{t}{\sstate''}{\_} \mid \heap] := V[\seval{\sstate}{e}{t_e}{\sstate'}{\scheck} \mid \heap]$$
    \item \refrule{SEvalFieldImprecise}:
      $$V[\seval{\sstate}{e.f}{t}{\sstate''}{\_} \mid \heap] := V[\seval{\sstate}{e}{t_e}{\sstate'}{\scheck} \mid \heap][t \mapsto \heap(V(t_e), f)]$$
    \item \refrule{SEvalFieldFailure}:
      $$V[\seval{\sstate}{e.f}{t}{\sstate'}{\_} \mid \heap] := V[\seval{\sstate}{e}{t_e}{\sstate'}{\scheck} \mid \heap][t \mapsto \heap(V(t_e), f)]$$
  \end{itemize}
\end{definition}

\begin{lemma}\label{lem:eval-subpath}
  If $\seval{\sstate}{e}{\_}{\sstate'}{\_}$ then $\pc(\sstate') \implies \pc(\sstate)$.
\end{lemma}
\begin{proof}
  By induction on $\seval{\sstate}{e}{\_}{\sstate'}{\_}$:

  \begin{enumcases}
    \case \refrule{SEvalLiteral} -- $\seval{\sstate}{l}{\_}{\sstate}{\_}$; \refrule{SEvalVar} -- $\seval{\sstate}{x}{\_}{\sstate}{\_}$: Trivial since $\pc(\sstate) \implies \pc(\sstate)$.

    \case\label{case:eval-subpath-ora} \refrule{SEvalOrA} -- $\seval{\sstate}{e_1 \vee e_2}{\_}{\sstate''}{\_}$:
      By \refrule{SEvalOrA} $\seval{\sstate}{e_1}{\_}{\sstate'}{\_}$, thus by induction $\pc(\sstate') \implies \pc(\sstate)$. Therefore $\pc(\sstate'') = \pc(\sstate') \wedge t_1 \implies \pc(\sstate') \implies \pc(\sstate)$.

    \case\label{case:eval-subpath-orb} \refrule{SEvalOrB} -- $\seval{\sstate}{e_1 \vee e_2}{\_}{\sstate''}{\_}$:
      By \refrule{SEvalOrB} $\seval{\sstate}{e_1}{\_}{\sstate'}{\_}$ and $\seval{\sstate'}{e_2}{\_}{\sstate''}{\_}$, thus by induction $\pc(\sstate'') \implies \pc(\sstate') \implies \pc(\sstate)$. Therefore $\pc(\sstate''') = \pc(\sstate'') \kand \kneg t_1 \implies \pc(\sstate'') \implies \pc(\sstate)$.

    \case \refrule{SEvalAndA} -- $\seval{\sstate}{e_1 \wedge e_2}{\_}{\sstate''}{\_}$: Similar to case \ref{case:eval-subpath-ora}.

    \case \refrule{SEvalAndB} -- $\seval{\sstate}{e_1 \wedge e_2}{\_}{\sstate''}{\_}$: Similar to case \ref{case:eval-subpath-orb}.

    \case \refrule{SEvalOp} -- $\seval{\sstate}{e_1 \oplus e_2}{\_}{\sstate''}{\_}$; \refrule{SEvalNeg} -- $\seval{\sstate}{\kneg e}{\_}{\sstate'}{\_}$; \refrule{SEvalField}, \refrule{SEvalFieldOptimistic} -- $\seval{\sstate}{e.f}{\_}{\sstate'}{\_}$; \refrule{SEvalFieldImprecise} -- $\seval{\sstate}{e.f}{\_}{\sstate''}{\_}$: By induction.
  \end{enumcases}
\end{proof}

\begin{lemma}\label{lem:eval-unchanged}
  If $\seval{\sstate}{e}{t}{\sstate'}{\scheck}$ then $\imp(\sstate') = \imp(\sstate)$, $\sheap(\sstate') = \sheap(\sstate)$, and $\senv(\sstate') = \senv(\sstate)$.
\end{lemma}
\begin{proof}
  Trivial by induction on $\seval{\sstate}{e}{t}{\sstate'}{\scheck}$.
\end{proof}

\begin{lemma}\label{lem:seval-soundness}
  Let $V$ be some initial valuation and $\triple{\heap}{\perms}{\env}$ be some well-formed evaluation state such that $\simstate{V}{\sstate}{\heap}{\perms}{\env}$.

  If $\seval{\sstate}{e}{t}{\sstate'}{\scheck}$, $\rtassert{V'}{\heap}{\perms}{\scheck}$, and $V'(\pc(\sstate')) = \ktrue$ where $V' \supseteq V[\seval{\sstate}{e}{t}{\sstate'}{\scheck} \mid \heap]$, then
  $$\simstate{V'}{\sstate'}{\heap}{\perms}{\env}, \quad
    \eval{\heap}{\env}{e}{V'(t)}, \quad \text{and}~
    \frm{\heap}{\perms}{\env}{e}.$$
\end{lemma}
\begin{proof}
  By induction on $\seval{\sstate}{e}{t}{\sstate'}{\scheck}$:
  
  \begin{enumcases}
    \case \refrule{SEvalLiteral} -- $\seval{\sstate}{l}{l}{\sstate}{\emptyset}$:
      Then $V' = V$, thus $\simstate{V'}{\sstate}{\heap}{\perms}{\env}$ by assumption.

      By \refrule{EvalLiteral} $\eval{\heap}{\env}{l}{l}$, and by definition $V'(l) = l$.

      By \refrule{FrameLiteral} $\frm{\heap}{\perms}{\env}{l}$.

    \case \refrule{SEvalVar} -- $\seval{\sstate}{x}{\senv(\sstate)(x)}{\sstate}{\emptyset}$:
      Then $V' = V$, thus $\simstate{V'}{\sstate}{\heap}{\perms}{\env}$ by assumption.

      By \refrule{EvalVar} $\eval{\heap}{\env}{x}{\env(x)}$, and $\env(x) = V'(\senv(\sstate)(x))$ since $\simenv{V'}{\senv(\sstate)}{\env}$.

      By \refrule{FrameVar} $\frm{\heap}{\perms}{\env}{x}$.

    \case\label{case:seval-sim-ora} \refrule{SEvalOrA} -- $\seval{\sstate}{e_1 \kor e_2}{t_1}{\sstate''}{\scheck}$:

      By \refrule{SEvalOrA} $\seval{\sstate}{e_1}{t_1}{\sstate'}{\scheck}$.

      Suppose $\rtassert{V'}{\heap}{\perms}{\scheck}$ and $V'(\pc(\sstate'')) = \ktrue$. Then $V'(\pc(\sstate')) = \ktrue$ since $\pc(\sstate'') = \pc(\sstate') \kand t_1 \implies \pc(\sstate')$ by \refrule{SEvalOrA}. Then by induction $\simstate{V'}{\sstate'}{\heap}{\perms}{\env}$.

      Now $\simstate{V'}{\sstate''}{\heap}{\perms}{\env}$ since $\sstate'' = \sstate'[\pc = \pc(\sstate') \kand t_1]$ by \refrule{SEvalOrA} and $V'(\pc(\sstate'')) = \ktrue$ by assumption.

      By induction $\eval{\heap}{\env}{e_1}{V'(t_1)}$. But now $V'(t_1) = \ktrue$ since $\pc(\sstate'') = \pc(\sstate') \kand t_1 \implies t_1$. Thus $\eval{\heap}{\env}{e_1}{\ktrue}$.

      Then by \refrule{EvalOrA} $\eval{\heap}{\env}{e_1 \kor e_2}{\ktrue}$ and $V'(t_1) = \ktrue$.

      By induction $\frm{\heap}{\perms}{\env}{e_1}$. Thus $\frm{\heap}{\perms}{\env}{e_1 \kor e_2}$ by \refrule{FrameOrA} since $\eval{\heap}{\env}{e_1}{\ktrue}$.

    \case\label{case:seval-sim-orb} \refrule{SEvalOrB} -- $\seval{\sstate}{e_1 \vee e_2}{t_2}{\sstate''}{\scheck_1 \cup \scheck_2}$:

      By \refrule{SEvalOrB} $\seval{\sstate}{e_1}{t_1}{\sstate'}{\scheck_1}$ and $\seval{\hat{\sstate}'}{e_2}{t_2}{\sstate''}{\scheck_2}$, where $\hat{\sstate}' = \sstate'[\pc = \pc(\sstate') \kand \kneg t_1]$.

      Now suppose $\rtassert{V'}{\heap}{\perms}{\scheck_1 \cup \scheck_2}$ and $V'(\pc(\sstate'')) = \ktrue$. By lemma \ref{lem:scheck-monotonicity} $\rtassert{V'}{\heap}{\perms}{\scheck_1}$ (thus $\rtassert{V'}{\heap}{\perms}{\scheck_1}$) and $\rtassert{V'}{\heap}{\perms}{\scheck_2}$.

      Also, by lemma \ref{lem:eval-subpath}, $\pc(\sstate'') \implies \pc(\hat{\sstate}') = \pc(\sstate') \kand \kneg t_1 \implies \pc(\sstate')$. Therefore $V'(\pc(\sstate')) = V'(\pc(\sstate')) = \ktrue$, and $V'(\pc(\sstate'')) = \ktrue$ by assumption.

      Now by induction $\simstate{V'}{\sstate'}{\heap}{\perms}{\env}$. Also $V'(\pc(\hat{\sstate}')) = \ktrue$ since $\pc(\sstate'') \implies \pc(\sstate')$. Therefore $\simstate{V'}{\hat{\sstate}'}{\heap}{\perms}{\env}$. Then also by induction $\simstate{V'}{\sstate''}{\heap}{\perms}{\env}$.

      By induction $\eval{\heap}{\env}{e_1}{V'(t_1)}$ and $\eval{\heap}{\env}{e_2}{V'(t_2)}$. But now $\pc(\sstate'') \implies \pc(\hat{\sstate}') = \pc(\sstate') \kand \kneg t_1 \implies \kneg t_1$, thus $V'(t_1) = \kfalse$. Therefore $\eval{\heap}{\env}{e_1}{\kfalse}$.

      Thus $\eval{\heap}{\env}{e_1 \kor e_2}{V'(t_2)}$ by \refrule{EvalOrB}.

      By induction $\frm{\heap}{\perms}{\env}{e_1}$ and $\frm{\heap}{\perms}{\env}{e_2}$. Thus $\frm{\heap}{\perms}{\env}{e_1 \kor e_2}$ by \refrule{FrameOrB} since $\eval{\heap}{\env}{e_1}{\kfalse}$.

    \case \refrule{SEvalAndA} -- $\seval{\sstate}{e_1 \wedge e_2}{t_1}{\sstate''}{\scheck}$: Similar to case \ref{case:seval-sim-ora}.

    \case \refrule{SEvalAndB} -- $\seval{\sstate}{e_1 \wedge e_2}{t_2}{\sstate''}{\scheck_1 \cup \scheck_2}$: Similar to case \ref{case:seval-sim-orb}.

    \case \refrule{SEvalOp} -- $\seval{\sstate}{e_1 \oplus e_2}{t_1 \oplus t_2}{\sstate''}{\scheck_1 \cup \scheck_2}$:

      By \refrule{SEvalOp} $\seval{\sstate}{e_1}{t_1}{\sstate'}{\scheck_1}$ and $\seval{\sstate'}{e_2}{t_2}{\sstate''}{\scheck_2}$.

      Now suppose $\rtassert{V'}{\heap}{\perms}{\scheck_1 \cup \scheck_2}$ and $V'(\pc(\sstate'')) = \ktrue$. By lemma \ref{lem:scheck-monotonicity} $\rtassert{V'}{\heap}{\perms}{\scheck_1}$ and $\rtassert{V'}{\heap}{\perms}{\scheck_2}$.

      Also, by lemma \ref{lem:eval-subpath}, $\pc(\sstate'') \implies \pc(\sstate')$. Therefore $V'(\pc(\sstate')) = V'(\pc(\sstate')) = \ktrue$, and $V'(\pc(\sstate'')) = \ktrue$ by assumption.

      Now by induction $\simstate{V'}{\sstate'}{\heap}{\perms}{\env}$, and then also by induction $\simstate{V'}{\sstate''}{\heap}{\perms}{\env}$.

      By induction $\eval{\heap}{\env}{e_1}{V'(t_1)}$ and $\eval{\heap}{\env}{e_2}{V'(t_2)}$. Therefore $\eval{\heap}{\env}{e_1 \oplus e_2}{V'(t_1) \oplus V'(t_2)}$ and $V'(t_1) \oplus V'(t_2) = V'(t_1 \oplus t_2)$.

      By induction $\frm{\heap}{\perms}{\env}{e_1}$ and $\frm{\heap}{\perms}{\env}{e_2}$. Thus $\frm{\heap}{\perms}{\env}{e_1 \oplus e_2}$ by \refrule{FrameOp}.

    \case \refrule{SEvalNeg} -- $\seval{\sstate}{\kneg e}{\kneg t}{\sstate'}{\scheck}$:
    
      By \refrule{SEvalNeg} $\seval{\sstate}{e}{t}{\sstate'}{\scheck}$.

      Suppose that $\rtassert{V'}{\heap}{\perms}{\scheck}$ and $V'(\pc(\sstate')) = \ktrue$.

      Now $\simstate{V'}{\sstate'}{\heap}{\perms}{\env}$ by induction.

      Also by induction $\eval{\heap}{\env}{e}{V'(t)}$. Thus $\eval{\heap}{\env}{\kneg e}{\neg V'(t)}$ by \refrule{EvalNeg} and $V'(\kneg t) = \neg V'(t)$.

      By induction $\frm{\heap}{\perms}{\env}{e}$, and thus $\frm{\heap}{\perms}{\env}{\kneg e}$ by \refrule{FrameNeg}.

    \case\label{case:seval-sim-field} \refrule{SEvalField}: $\seval{\sstate}{e.f}{t}{\sstate'}{\scheck}$

      By \refrule{SEvalField} $\seval{\sstate}{e}{t_e}{\sstate'}{\scheck}$.

      Suppose that $\rtassert{V'}{\heap}{\perms}{\scheck}$ and $V'(\pc(\sstate')) = \ktrue$.

      Then $\simstate{V'}{\sstate'}{\heap}{\perms}{\env}$ by induction.

      Also by induction $\eval{\heap}{\env}{e}{V'(t_e)}$. By \refrule{SEvalField} $\pc(\sstate') \implies t_e \keq t_e'$, therefore $V'(t_e) = V'(t_e')$. Also by \refrule{SEvalField} $\triple{f}{t_e'}{t} \in \sheap(\sstate')$. Therefore $V'(t) = \heap(V'(t_e'), f) = \heap(V'(t_e), f)$ since $\simstate{V'}{\sstate'}{\heap}{\env}{\perms}$.

      Thus $\eval{\heap}{\env}{e.f}{\heap(V'(t_e), f)}$ by \refrule{EvalField} and $\heap(V'(t_e), f) = V_{\sstate'}(t)$.

      By induction $\frm{\heap}{\perms}{\env}{e}$. Also, since $\triple{f}{t_e'}{t} \in \sheap(\sstate')$ and $\simstate{V'}{\sstate'}{\heap}{\env}{\perms}$, $\pair{V'(t_e')}{f} = \pair{V'(t_e)}{f} \in \perms$. By \refrule{AssertAcc}, $\assertion{\heap}{\perms}{\env}{\kacc(e.f)}$ since $\eval{\heap}{\env}{e}{V'(t_e)}$. Thus by \refrule{FrameField} $\frm{\heap}{\perms}{\env}{e.f}$.

    \case \refrule{SEvalFieldOptimistic} -- $\seval{\sstate}{e.f}{t}{\sstate'}{\scheck}$: Similar to case \ref{case:seval-sim-field}.

    \case \refrule{SEvalFieldImprecise} -- $\seval{\sstate}{e.f}{t}{\sstate'}{\scheck; \pair{t_e}{f}}$:

      By \refrule{SEvalFieldImprecise} $\seval{\sstate}{e}{t_e}{\sstate'}{\scheck}$, $t = \ffresh$, and $\sstate'' = \sstate'[\oheap = \oheap(\sstate'); \triple{f}{t_e}{t}]$.

      Suppose $\rtassert{V'}{\heap}{\perms}{\scheck; \pair{t_e}{f}}$ and $V'(\pc(\sstate'')) = \ktrue$, thus $\rtassert{V'}{\heap}{\perms}{\scheck}$ by lemma \ref{lem:scheck-monotonicity}. Then $V'(\pc(\sstate'')) = V'(\pc(\sstate')) = \ktrue$.

      Then by induction $\simstate{V'}{\sstate'}{\heap}{\perms}{\env}$, thus $\simstate{V'}{\sstate'}{\heap}{\perms}{\env}$.

      By lemma \ref{lem:scheck-monotonicity} $\rtassert{V'}{\heap}{\perms}{\set{\pair{t_e}{f}}}$ and thus $\rtassert{V'}{\heap}{\perms}{\pair{t_e}{f}}$. Then by \refrule{CheckAcc} $\pair{V'(t_e)}{f} \in \perms$.
      
      By definition \ref{def:eval-valuation} $\heap(V'(t_e), f) = V'(t)$, and also $\pair{V'(t_e)}{f} \in \perms$ and $\simheap{V'}{\oheap(\sstate')}{\heap}{\perms}$, thus $\simheap{V'}{\oheap(\sstate'); \triple{f}{t_e}{t}}{\heap}{\perms}$.
      
      Then $\simstate{V'}{\sstate''}{\heap}{\perms}{\env}$ since $\sstate'$ and $\sstate''$ differ only in their $\oheap$ components and $\oheap(\sstate'') = \oheap(\sstate'); \triple{f}{t_e}{t}$.

      By induction $\eval{\heap}{\env}{e}{V'(t_e)}$. As shown before, $\heap(V'(t_e), f) = V'(t)$. Thus by \refrule{EvalField} $\eval{\heap}{\env}{e.f}{V'(t)}$.

      By induction $\frm{\heap}{\perms}{\env}{e}$. Also, as shown before, $\pair{V'(t_e)}{f} \in \perms$. By \refrule{AssertAcc}, $\assertion{\heap}{\perms}{\env}{\kacc(e.f)}$ since $\eval{\heap}{\env}{e}{V'(t_e)}$. Thus by \refrule{FrameField} $\frm{\heap}{\perms}{\env}{e.f}$.

    \case \refrule{SEvalFieldFailure} -- $\seval{\sstate}{e.f}{t}{\sstate'}{\set{\bot}}$:
    
      $\rtassert{V}{\heap}{\perms}{\set{\bot}}$ cannot hold. Since the assumptions cannot be satisfied, the statement holds vacuously.

  \end{enumcases}
\end{proof}

\begin{lemma}[Progress]\label{lem:eval-progress}
  If $V$ is some initial valuation and $V(\pc(\sstate)) = \ktrue$, then for some $\sstate'$, $t$, and $\scheck$,
  $$\seval{\sstate}{e}{t}{\sstate'}{\scheck} ~\text{and}~ V'(\pc(\sstate')) = \ktrue$$
  where $V' = V[\seval{\sstate}{e}{t}{\sstate'}{\scheck} \mid \heap]$ for some $\heap$.
\end{lemma}
\begin{proof}
  By induction on $e$:

  \begin{enumcases}
    \case $l \in \Literal$:
      By \refrule{SEvalLiteral} $\seval{\sstate}{l}{\_}{\sstate}{\_}$. Then $V'(\pc(\sstate)) = \ktrue$ by assumptions.

    \case $x \in \Var$:
      By \refrule{SEvalVar} $\seval{\sstate}{x}{\_}{\sstate}{\_}$. Then $V'(\pc(\sstate)) = \ktrue$ by assumptions.

    \case $e.f$ where $e \in \Expr, f \in \Field$:
      By induction, $\seval{e}{t_e}{\sstate}{\scheck}$ and $V'(\pc(\sstate')) = \ktrue$ where $V'$ is the corresponding valuation. Then one of the following must apply:
      
      \subcase If $\existential{t_e', t}{\big[\pc(\sstate') \implies t_e' \keq t_e \big] \wedge \triple{f}{t_e'}{t} \in \sheap(\sstate')}$, then \refrule{SEvalField} applies and thus $\seval{\sstate}{e.f}{\_}{\sstate'}{\_}$. Let $\sstate'' = \sstate'$.
      \subcase Otherwise, $\nexistential{t_e', t}{\big[\pc(\sstate') \implies t_e' \keq t_e \big] \wedge \triple{f}{t_e'}{t} \in \sheap(\sstate')}$. Then, if $\existential{t_e', t}{\big[\pc(\sstate') \implies t_e' \keq t_e \big] \wedge \triple{f}{t_e'}{t} \in \oheap(\sstate')}$, \refrule{SEvalFieldOptimistic} applies and thus $\seval{\sstate}{e.f}{\_}{\sstate'}{\_}$. Let $\sstate'' = \sstate'$.
      \subcase Otherwise, $\nexistential{t_e', t}{\big[\pc(\sstate') \implies t_e' \keq t_e \big] \wedge \triple{f}{t_e'}{t} \in \sheap(\sstate') \cup \oheap(\sstate')}$. Then if $\imp(\sstate')$, \refrule{SEvalFieldImprecise} applies. In this case, $\seval{\sstate}{e.f}{\_}{\sstate''}{\_}$ where $\pc(\sstate'') = \pc(\sstate')$.
      \subcase Otherwise, $\neg \imp(\sstate')$ and $\nexistential{t_e', t}{\big[\pc(\sstate') \implies t_e' \keq t_e \big] \wedge \triple{f}{t_e'}{t} \in \sheap(\sstate') \cup \oheap(\sstate')}$. Thus \refrule{SEvalFieldFailure} applies and thus $\seval{\sstate}{e.f}{\_}{\sstate'}{\_}$. Let $\sstate'' = \sstate'$.

      In all of these subcases, $\seval{\sstate}{e.f}{\_}{\sstate''}{\_}$ where $\pc(\sstate'') = \pc(\sstate')$, and thus $V'(\pc(\sstate'')) = V'(\pc(\sstate')) = \ktrue$. By definition \ref{def:eval-valuation}, in all of these subcases $V[\seval{\sstate}{e.f}{\_}{\sstate''}{\_} \mid \heap]$ extends $V'$.

    \case $e_1 \oplus e_2$; $e_1, e_2 \in \Expr$:
      By induction $\seval{\sstate}{e_1}{\_}{\sstate'}{\_}$ for some $e_1, \sstate'$ where $V'$ is the corresponding valuation and $V'(\pc(\sstate')) = \ktrue$. Then by induction $\seval{\sstate'}{e_2}{\_}{\sstate''}$ where $V''$ is the corresponding valuation extending $V'$ and $V''(\pc(\sstate'')) = \ktrue$. Finally, by \refrule{SEvalOp} $\seval{\sstate}{e_1 \oplus e_2}{\_}{\sstate''}{\_}$. By definition \ref{def:eval-valuation}, $V[\seval{\sstate}{e_1 \oplus e_2}{\_}{\sstate''}{\_} \mid \heap]$ extends $V''$.

    \case\label{case:seval-progress-or} $e_1 \kor e_2$; $e_1, e_2 \in \Expr$:
      By induction $\seval{\sstate}{e_1}{t_1}{\sstate'}{\_}$ for some $e_1, t_e, \sstate'$ where $V'$ is the corresponding valuation and $V'(\pc(\sstate')) = \ktrue$. Then one of the following cases must apply since the program is well-typed:

      \subcase $V'(t_1) = \ktrue$: Then by \refrule{SEvalOrA}, $\seval{\sstate}{e_1 \vee e_2}{t_1}{\sstate''}{\_}$ where $\pc(\sstate'') = \pc(\sstate') \kand t_1$. Let $V''$ be the corresponding valuation. Now $V''(\pc(\sstate'')) = V''(\pc(\sstate')) \wedge V''(t_1) = \ktrue$, which completes the proof.
      
      \subcase $V'(t_1) = \kfalse$: Let $\hat{\sstate}' = \sstate'[\pc = \pc(\sstate') \kand \kneg t_1]$. Then $V'(\pc(\hat{\sstate}')) = V'(\sstate') \wedge \neg V'(t_1) = \ktrue$. By induction $\seval{\hat{\sstate}'}{e_2}{t_2}{\sstate''}{\_}$ for some $e_2, t_2, \sstate''$ where $V''$ is the corresponding valuation and $V''(\pc(\sstate'')) = \ktrue$. Finally, by \refrule{SEvalOrB}, $\seval{\sstate}{e_1 \kor e_2}{\_}{\sstate''}{\_}$. By definition \ref{def:eval-valuation} $V[\seval{\sstate}{e_1 \kor e_2}{\_}{\sstate''}{\_} \mid \heap]$ extends $V''$.

    \case $e_1 \kand e_2$; $e_1, e_2 \in \Expr$: Similar to case \ref{case:seval-progress-or}.

    \case $\kneg e$; $e \in \Expr$:
      By induction $\seval{\sstate}{e}{\_}{\sstate'}{\_}$ and $V'(\sstate') = \ktrue$ where $V'$ is the corresponding derivation. Then by \refrule{SEvalNeg}, $\seval{\sstate}{\kneg e}{\_}{\sstate'}{\_}$ and the corresponding valuation extends $V'$ by definition \ref{def:eval-valuation}.
  \end{enumcases}
\end{proof}

\subsection{Deterministic Evaluation}


\begin{definition}\label{def:spceval-valuation}
  For a judgement $\spceval{\sstate}{e}{t}{\scheck}$, given an initial valuation $V$ and heap $\heap$, the \textbf{corresponding valuation} is denoted
  $$V[\spceval{\sstate}{e}{t}{\scheck} \mid \heap].$$
  This valuation is defined as follows, depending on the rule that proves the derivation.

  Note that the corresponding valuation always extends the initial valuation and is defined for all $\ffresh$ symbolic values in the judgement.
  \begin{itemize}
    \item \refrule{SEvalPCLiteral}:
      $$V[\spceval{\sstate}{l}{l}{\emptyset} \mid \heap] := V$$
    \item \refrule{SEvalPCVar}:
      $$V[\spceval{\sstate}{x}{\senv(\sstate)(x)}{\emptyset} \mid \heap] := V$$
    \item \refrule{SEvalPCOr}:
      \begin{align*}
        V[\spceval{\sstate}{e_1 \kor e_2}{t_1 \kor t_2}{\scheck_1 \cup \scheck_2} \mid \heap] :=
          V&[\spceval{\sstate}{e_1}{t_1}{\scheck_1} \mid \heap] \\
          &[\spceval{\sstate}{e_2}{t_2}{\scheck_2} \mid \heap]
      \end{align*}
    \item \refrule{SEValPCAnd}:
      \begin{align*}
        V[\spceval{\sstate}{e_1 \kand e_2}{t_1 \kand t_2}{\scheck_1 \cup \scheck_2} \mid \heap] :=
          V&[\spceval{\sstate}{e_1}{t_1}{\scheck_1} \mid \heap] \\
          &[\spceval{\sstate}{e_2}{t_2}{\scheck_2} \mid \heap]
      \end{align*}
    \item \refrule{SEvalPCOp}:
      \begin{align*}
        V[\spceval{\sstate}{e_1 \oplus e_2}{t_1 \oplus t_2}{\scheck_1 \cup \scheck_2} \mid \heap] :=
          V&[\spceval{\sstate}{e_1}{t_1}{\scheck_1} \mid \heap] \\
          &[\spceval{\sstate}{e_2}{t_2}{\scheck_2} \mid \heap]
      \end{align*}
    \item \refrule{SEvalPCField} or \refrule{SEvalPCFieldOptimistic}:
      $$V[\spceval{\sstate}{e.f}{t}{\scheck} \mid \heap] := V[\spceval{\sstate}{e}{t_e}{\scheck} \mid \heap]$$
    \item \refrule{SEvalPCFieldImprecise} or \refrule{SEvalPCFieldMissing}:
      $$V[\spceval{\sstate}{e.f}{t}{\scheck} \mid \heap] := V'[t \mapsto \heap(V'(t_e), f)]$$
      where $V' = V[\spceval{\sstate}{e}{t_e}{\scheck} \mid \heap]$.
  \end{itemize}
\end{definition}

\begin{lemma}[Soundness] \label{lem:pc-eval-soundness}
  Suppose $\simstate{V}{\sstate}{\heap}{\perms}{\env}$.

  If $\spceval{\sstate}{e}{t}{\scheck}$, $V' \supseteq V[\spceval{\sstate}{e}{t}{\scheck} \mid \heap]$, and $\rtassert{V'}{\heap}{\perms}{\scheck}$, then
  $$\eval{\heap}{\env}{e}{V'(t)} \quad\text{and}\quad
    \frm{\heap}{\perms}{\env}{e}.$$
\end{lemma}

\begin{proof}
  By induction on $\spceval{\sstate}{e}{t}{\scheck}$:

  \begin{enumcases}
    \case \refrule{SEvalPCLiteral} -- $\spceval{\sstate}{l}{l}{\emptyset}$:

      Then $V$ is the corresponding valuation. By \refrule{EvalLiteral} $\eval{\heap}{\env}{l}{l}$ and by definition $V(l) = l$.

      By \refrule{FrameLiteral} $\frm{\heap}{\perms}{\env}{l}$.

    \case \refrule{SEvalPCVar} -- $\spceval{\sstate}{x}{\senv(\sstate)(x)}{\emptyset}$:

      Then $V$ is the corresponding valuation. Then $\eval{\heap}{\env}{x}{\env(x)}$ by \refrule{EvalLiteral}, and $\env(x) = V(\senv(\sstate)(x))$ since $\simenv{V}{\senv(\sstate)}{\env}$.

      By \refrule{FrameVar}, $\frm{\heap}{\perms}{\env}{x}$.

    \case\label{case:evalpc-soundness-or} \refrule{SEvalPCOr} -- $\spceval{\sstate}{e_1 \vee e_2}{t_1 \kor t_2}{\scheck_1 \cup \scheck_2}$:

      By \refrule{SEvalPCOr} $\spceval{\sstate}{e_1}{t_1}{\scheck_1}$ and $\spceval{\sstate}{e_2}{t_2}{\scheck_2}$. Let $V_1$ and $V_2$ be the respective corresponding valuations, with $V_2$ extending $V_1$ and $V_1$ extending $V$. Then $V_2$ is the corresponding valuation for this case.

      Suppose that $\rtassert{V_2}{\heap}{\perms}{\scheck_1 \cup \scheck_2}$. Then $\rtassert{V_1}{\heap}{\perms}{\scheck_1}$ and $\rtassert{V_2}{\heap}{\perms}{\scheck_2}$ by lemma \ref{lem:scheck-monotonicity}.

      Then one of the following cases applies, assuming a well-typed program:

      \subcase $V_1(t_1) = \ktrue$:
        Then $\eval{\heap}{\env}{e_1}{\ktrue}$ by induction, and then by \refrule{EvalOrA} $\eval{\heap}{\env}{e_1 \kor e_2}{\ktrue}$. Also, $V_2(t_1 \kor t_2) = (V_1(t_1) \kor V_2(t_2)) = \ktrue \vee V_2(t_2) = \ktrue$. Therefore $\eval{\heap}{\env}{e_1 \kor e_2}{V_2(t_1 \kor t_2)}$.

        Also by induction $\frm{\heap}{\perms}{\env}{e_1}$. Since $\eval{\heap}{\env}{e_1}{\ktrue}$, by \refrule{FrameOrA} $\frm{\heap}{\perms}{\env}{e_1 \kor e_2}$.

      \subcase $V_1(t_1) = \kfalse$:
        Then $\eval{\heap}{\env}{e_1}{\kfalse}$ and $\eval{\heap}{\env}{e_2}{V_2(t_2)}$ by induction. Thus $\eval{\heap}{\env}{e_1 \kor e_2}{V_2(t_2)}$ by \refrule{EvalOrB}. Also, $V_2(t_1 \kor t_2) = V_1(t_1) \vee V_2(t_2) = \kfalse \vee V_2(t_2) = V_2(t_2)$. Therefore $\eval{\heap}{\env}{e_1 \kor e_2}{V_2(t_1 \kor t_2)}$.

        Also by induction $\frm{\heap}{\perms}{\env}{e_1}$ and $\frm{\heap}{\perms}{\env}{e_2}$. Since $\eval{\heap}{\env}{e_1}{\kfalse}$, by \refrule{FrameOrB} $\frm{\heap}{\perms}{\env}{e_1 \kor e_2}$.

    \case \refrule{SEvalPCAnd} -- $\spceval{\sstate}{e_1 \kand e_2}{t_1 \kand t_2}{\scheck_1 \cup \scheck_2}$:
      Similar to case \ref{case:evalpc-soundness-or}.

    \case \refrule{SEvalPCOp} -- $\spceval{\sstate}{e_1 \oplus e_2}{t_1 \oplus t_2}{\scheck_1 \cup \scheck_2}$

      Use the same proof as in case \ref{case:evalpc-soundness-or} up to subcases.

      Then $\eval{\heap}{\env}{e_1}{V_1(t_1)}$ and $\eval{\heap}{\env}{e_2}{V_2(t_2)}$ by induction. Thus $\eval{\heap}{\env}{e_1 \oplus e_2}{V_1(t_1) \oplus V_2(t_2)}$ by \refrule{EvalOp}. Also, $V_1(t_1) \oplus V_2(t_2) = V_2(t_1 \oplus t_2)$ by definition, therefore $\eval{\heap}{\env}{e_1 \oplus e_2}{V_2(t_1 \oplus t_2)}$.

      Also by induction $\frm{\heap}{\perms}{\env}{e_1}$ and $\frm{\heap}{\perms}{\env}{e_2}$. Therefore by \refrule{FrameOp}, $\frm{\heap}{\perms}{\env}{e_1 \oplus e_2}$.

    \case \refrule{SEvalPCNeg} -- $\spceval{\sstate}{\kneg e}{\kneg t}{\scheck}$:

      By \refrule{SEvalPCNeg} $\spceval{\sstate}{e}{t}{\scheck}$. Let $V'$ be the corresponding valuation, thus $V'$ is the corresponding valuation for this case. Suppose that $\rtassert{V'}{\heap}{\perms}{\scheck}$.

      Then $\eval{\heap}{\env}{e}{V'(t)}$ by induction, and then $\eval{\heap}{\env}{\kneg e}{\kneg V'(t)}$ by \refrule{EvalNeg}. Also, $\neg V'(t) = V'(\kneg t)$ by definition. Therefore $\eval{\heap}{\env}{\kneg e}{V'(\kneg t)}$.

      Also by induction $\frm{\heap}{\perms}{\env}{e}$, thus $\frm{\heap}{\perms}{\env}{\kneg e}$ by \refrule{FrameNeg}.

    \case\label{case:seval-pc-soundness-field} \refrule{SEvalPCField} -- $\spceval{\sstate}{e.f}{t}{\scheck}$:

      By \refrule{SEvalPCField} $\spceval{\sstate}{e}{t_e}{\scheck}$. Let $V'$ be the corresponding valuation, thus $V'$ is the corresponding valuation for this case. Suppose that $\rtassert{V'}{\heap}{\perms}{\scheck}$.

      By \refrule{SEvalPCField} $\pc(\sstate') \implies t_e' \keq t_e$, thus $V'(t_e') = V'(t_e)$. By induction $\eval{\heap}{\env}{e}{V'(t_e)}$ [$= V'(t_e')$]. Then by \refrule{EvalField}, $\eval{\heap}{\env}{e.f}{\heap(V'(t_e'), f)}$.

      Also by \refrule{SEvalPCField} $\triple{f}{t_e'}{t} \in \sheap(\sstate)$. Then $\heap(V'(t_e'), f) = V'(t)$ since $\simheap{V'}{\sheap(\sstate)}{\heap}{\perms}$. Therefore $\eval{\heap}{\env}{e.f}{V'(t)}$.

      Also, $\pair{V'(t_e')}{f} \in \perms$ since $\simheap{V'}{\sheap(\sstate)}{\heap}{\perms}$. Now $\assertion{\heap}{\perms}{\env}{\kacc(e.f)}$ by \refrule{AssertAcc}. By induction $\frm{\heap}{\perms}{\env}{e}$, therefore $\frm{\heap}{\perms}{\env}{e.f}$ by \refrule{FrameField}.

    \case \refrule{SEvalPCFieldOptimistic} -- $\spceval{\sstate}{e.f}{t}{\scheck}$:
      Same as case \ref{case:seval-pc-soundness-field}, replacing $\sheap$ with $\oheap$.

    \case \refrule{SEvalPCFieldImprecise} -- $\spceval{\sstate}{e.f}{t}{\scheck \cup \set{\kacc(t_e.f)}}$:

      By \refrule{SEvalPCFieldImprecise} $\spceval{\sstate}{e}{t_e}{\scheck}$. Let $V'$ be the corresponding valuation for this case. Suppose that $\simstate{V}{\sstate}{\heap}{\perms}{\env}$ and $\rtassert{V'}{\heap}{\perms}{\scheck; \pair{t_e}{f}}$.

      By lemma \ref{lem:scheck-monotonicity} $\rtassert{V'}{\heap}{\perms}{\scheck}$ and $\rtassert{V'}{\heap}{\perms}{\set{\pair{t_e}{f}}}$.

      By \refrule{SEvalPCFieldImprecise} $t = \ffresh$. Since $\rtassert{V'}{\heap}{\perms}{\pair{t_e}{f}}$, by \refrule{CheckAcc} $\pair{V'(t_e)}{f} \in \perms$.
      
      By induction $\eval{\heap}{\env}{e}{V'(t_e)}$, thus $\eval{\heap}{\env}{e.f}{\heap(V'(t_e), f)}$ by \refrule{EvalField}. But also $V'(t) = \heap(V'(t_e), f)$ by definition, therefore $\eval{\heap}{\env}{e.f}{V'(t)}$.

      Finally, $\eval{\heap}{\env}{e}{V'(t_e)}$ and $\pair{V'(t_e)}{f} \in \perms$, thus $\assertion{\heap}{\perms}{\env}{\kacc(e.f)}$ by \refrule{AssertAcc}. Also, $\frm{\heap}{\perms}{\env}{e}$ by induction. Therefore $\frm{\heap}{\perms}{\env}{e.f}$ by \refrule{FrameField}.

    \case \refrule{SEvalPCFieldMissing} -- $\spceval{\sstate}{e.f}{t}{\scheck \cup \set{\bot}}$:

      Let $V'$ be the corresponding valuation for this case. $\rtassert{V'}{\heap}{\perms}{\set{\bot}}$ cannot hold. Since the conditions cannot be satisfied, the statement vacuously holds.
  \end{enumcases}
\end{proof}

\begin{lemma}\label{lem:spceval-correspondence}
  Suppose that $\simstate{V}{\sstate}{\heap}{\perms}{\env}$.

  If $\eval{\heap}{\env}{e}{v}$ and $\spceval{\sstate}{e}{t}{\_}$ then $v = V'(t)$, where $V' = V[\spceval{\sstate}{e}{t}{\_} \mid \heap]$.
\end{lemma}
\begin{proof}
  By induction on $\spceval{\sstate}{e}{t}{\_}$:

  \begin{enumcases}
    \case \refrule{SEvalPCLiteral} -- $\spceval{\sstate}{l}{l}{\emptyset}$:

      Suppose $\eval{\heap}{\env}{l}{v}$ for some $v$. By \refrule{EvalLiteral} $v = l$, and by definition $V'(l) = l$.

    \case \refrule{SEvalPCVar} -- $\spceval{\sstate}{x}{\senv(\sstate)(x)}{\emptyset}$:

      Suppose $\eval{\heap}{\env}{x}{v}$ for some $v$. By \refrule{EvalVar} $v = \env(x)$. Also $\env(x) = V'(\senv(\sstate)(x))$ since $\simenv{V}{\senv(\sstate)}{\env}$.

    \case\label{case:spceval-correspondence-or} \refrule{SEvalPCOr} -- $\spceval{\sstate}{e_1 \kor e_2}{t_1 \kor t_2}{\scheck_1 \cup \scheck_2}$:

      By \refrule{SEvalPCOr} $\spceval{\sstate}{e_1}{t_1}{\scheck_1}$ and $\spceval{\sstate}{e_2}{t_2}{\scheck_2}$. Let $V_1$ and $V_2$ be the respective corresponding valuations, with $V_2$ extending $V_1$ and $V_1$ extending $V$. Then $V_2$ is the corresponding valuation for this case.

      Suppose that $\eval{\heap}{\env}{e_1 \kor e_2}{v}$. Then one of the following rules must be used:

      \subcase \refrule{EvalOrA} -- $\eval{\heap}{\env}{e_1 \kor e_2}{\ktrue}$ thus $v = \ktrue$:

        Then $v = \ktrue$. By \refrule{EvalOrA} $\eval{\heap}{\env}{e_1}{\ktrue}$, thus by induction $V_1(t_1) = \ktrue$. Also $V_2(t_1 \kor t_2) = (V_1(t_1) \vee V_2(t_2)) = \ktrue \vee V'(t_2) = \ktrue$. Therefore $v = V_2(t_1 \kor t_2)$.

      \subcase \refrule{EvalOrB} -- $\eval{\heap}{\env}{e_2 \kor e_2}{v_2}$ thus $v = v_2$:

        By \refrule{EvalOrB} $\eval{\heap}{\env}{e_1}{\kfalse}$ and $\eval{\heap}{\env}{e_2}{v}$. Thus $V_1(t_1) = \kfalse$ and by $V_2(t_2) = v_2$ by induction. Also, $V_2(t_1 \kor t_2) = V_1(t_1) \vee V_2(t_2) = \kfalse \vee v_2 = v_2$. Therefore $v = V_2(t_1 \kor t_2)$.

    \case \refrule{SEvalPCAnd} -- $\spceval{\sstate}{e_1 \kand e_2}{t_1 \kand t_2}{\scheck_1 \cup \scheck_2}$:
      Similar to case \ref{case:spceval-correspondence-or}.

    \case \refrule{SEvalPCOp} -- $\spceval{\sstate}{e_1 \oplus e_2}{t_1 \oplus t_2}{\scheck_1 \cup \scheck_2}$:
      Follow the same proof as in case \ref{case:spceval-correspondence-or}, up to subcases.

      Then by \refrule{EvalOp} $v = v_1 \oplus v_2$ for some $v_1, v_2$ such that $\eval{\heap}{\env}{e_1}{v_1}$ and $\eval{\heap}{\env}{e_2}{v_2}$. By induction $V_1(t_1) = v_1$ and $V_2(t_2) = v_2$. Also, $V_2(t_1 \oplus t_2) = V_1(t_1) \oplus V_2(t_2) = v_1 \oplus v_2 = v$.

    \case \refrule{SEvalPCNeg} -- $\spceval{\sstate}{\kneg e}{\kneg t}{\scheck}$:

      By \refrule{SEvalPCNeg} $\spceval{\sstate}{e}{t}{\scheck}$. Let $V'$ be the corresponding valuation, thus $V'$ is the corresponding valuation for this case.
    
      Suppose $\eval{\heap}{\env}{\kneg e}{v}$ for some $v$. Then by \refrule{EvalNeg} $v = \neg v'$ where $\eval{\heap}{\env}{e}{v'}$. By induction $v' = V'(t)$ and therefore $v = \neg v' = \neg V'(t) = V'(\kneg t)$.

    \case\label{case:spceval-correspondence-field} \refrule{SEvalPCField}: $\spceval{\sstate}{e.f}{t}{\scheck}$

      By \refrule{SEvalPCField} $\spceval{\sstate}{e}{t_e}{\scheck}$. Let $V'$ be the corresponding valuation, thus $V'$ is the corresponding valuation for this case.

      Suppose $\eval{\heap}{\env}{e.f}{v}$ for some $v$. Then by \refrule{EvalField} $\eval{\heap}{\env}{e}{v_e}$ for some $v_e$. Therefore $V'(t_e) = v_e$ by induction. By \refrule{SEvalPCField} $\pc(\sstate') \implies t_e' \keq t_e$ and thus $v_e = V'(t_e) = V'(t_e')$.

      Now by \refrule{EvalField} $v = \heap(v_e, f)$. Also, $\triple{f}{t_e'}{t} \in \sheap(\sstate)$ by \refrule{SEvalPCField}. Therefore $\heap(V'(t_e'), f) = V'(t)$ since $\simheap{V}{\sheap(\sstate)}{\heap}{\perms}$ and $V'$ extends $V$. Finally, $V'(t) = \heap(V'(t_e'), f) = \heap(v_e, f) = v$.

    \case \refrule{SEvalPCFieldOptimistic}: $\spceval{\sstate}{e.f}{t}{\scheck}$

      Same as case \ref{case:spceval-correspondence-field}, replacing $\sheap$ with $\oheap$.

    \case\label{label:spceval-correspondence-imprecise} \refrule{SEvalPCFieldImprecise} -- $\spceval{\sstate}{e.f}{t}{\scheck; \pair{t_e}{f}}$:

      By \refrule{SEvalPCFieldImprecise} $\spceval{\sstate}{e}{t_e}{\scheck}$. Let $V_e$ be the corresponding valuation, and let $V'$ be the corresponding valuation for this case, which extends $V_e$.

      Suppose $\eval{\heap}{\env}{e.f}{v}$ for some $v$. Then by \refrule{EvalField} $\eval{\heap}{\env}{e}{v_e}$ for some $v_e$. Therefore $V_e(t_e) = v_e$ by induction.

      Now by \refrule{EvalField} $v = \heap(v_e, f)$. By definition \ref{def:spceval-valuation} $V'(t) = \heap(V_e(t_e), f) = \heap(v_e, f) = v$.

    \case \refrule{SEvalPCFieldMissing} -- $\spceval{\sstate}{e.f}{t}{\set{\bot}}$:
      Same as case \ref{label:spceval-correspondence-imprecise}.
  \end{enumcases}
\end{proof}

\begin{lemma}\label{lem:pc-eval-progress}
  For any $\sstate, e$, $\spceval{\sstate}{e}{t}{\_}$ for some $t$.
\end{lemma}
\begin{proof}
  By induction on the syntax forms of $e$:

  \begin{enumcases}
    \case $l \in \Literal$

      Then $l \in \SExpr$ and by \refrule{SEvalPCLiteral}, $\spceval{\sstate}{l}{l}{\_}$.

    \case $x \in \Var$:

      Since this is a well-formed program, all variables must be assigned before use. Therefore $\senv(\sstate)(x)$ must be defined, and by \refrule{SEvalPCVar}, $\spceval{\sstate}{e}{\senv(\sstate)(x)}{\_}$.

    \case $e.f$ -- $e \in \Expr$, $f \in \Field$:

      By induction, $\existential{t_e \in \SExpr}{\spceval{\sstate}{e}{t_e}{\_}}$. Then one of the following must apply:

      \subcase $\existential{t_e'}{\triple{f}{t_e'}{t} \in \sheap(\sstate) \wedge \pc(\sstate) \implies t_e' \keq t_e}$:
        Then \refrule{SEvalPCField} applies and thus $\spceval{\sstate}{e.f}{t}{\_}$.

      \subcase $\nexistential{t_e'}{\triple{f}{t_e'}{t} \in \sheap(\sstate) \wedge \pc(\sstate) \implies t_e' \keq t_e}$ and $\existential{t_e'}{\triple{f}{t_e'}{t} \in \oheap(\sstate) \wedge \pc(\sstate) \implies t_e' \keq t_e}$:
        Then \refrule{SEvalPCFieldOptimistic} applies and thus $\spceval{\sstate}{e.f}{t}{\_}$.

      \subcase $\nexistential{t_e'}{\triple{f}{t_e'}{t} \in \sheap(\sstate) \cup \oheap(\sstate) \wedge \pc(\sstate) \implies t_e' \keq t_e}$ and $\imp(\sstate)$:
        Then \refrule{SEvalPCFieldImprecise} applies and thus $\spceval{\sstate}{e.f}{t}{\_}$ where $t = \ffresh$.

      \subcase $\nexistential{t_e'}{\triple{f}{t_e'}{t} \in \sheap(\sstate) \cup \oheap(\sstate) \wedge \pc(\sstate) \implies t_e' \keq t_e}$ and $\neg \imp(\sstate)$:
        Then \refrule{SEvalPCFieldMissing} applies and thus $\spceval{\sstate}{e.f}{t}{\_}$ where $t = \ffresh$.

    \case $e_1 \oplus e_2$ -- $e_1, e_2 \in \Expr$:

      By induction, $\existential{t_1 \in \SExpr}{\spceval{\sstate}{e_1}{t_1}{\_}}$ and $\existential{t_2 \in \SExpr}{\spceval{\sstate}{e_2}{t_2}{\_}}$. Then, by \refrule{SEvalPCOp}, $\spceval{\sstate}{e_1 \oplus e_2}{t_1 \oplus t_2}{\_}$.

    \case $e_1 \kor e_2$ -- $e_1, e_2 \in \Expr$:

      By induction, $\existential{t_1 \in \SExpr}{\spceval{\sstate}{e_1}{t_1}{\_}}$ and $\existential{t_2 \in \SExpr}{\spceval{\sstate}{e_2}{t_2}{\_}}$. Then, by \refrule{SEvalPCOr}, $\spceval{\sstate}{e_1 \kor e_2}{t_1 \kor t_2}{\_}$.

    \case $e_1 \kand e_2$ -- $e_1, e_2 \in \Expr$:

      By induction, $\existential{t_1 \in \SExpr}{\spceval{\sstate}{e_1}{t_1}{\_}}$ and $\existential{t_2 \in \SExpr}{\spceval{\sstate}{e_2}{t_2}{\_}}$. Then, by \refrule{SEvalPCAnd}, $\spceval{\sstate}{e_1 \kand e_2}{t_1 \kand t_2}{\_}$.

    \case $\kneg e$ -- $e \in \Expr$:

      By induction, $\existential{t \in \SExpr}{\spceval{\sstate}{e}{t}{\_}}$. Then, by \refrule{SEvalPCNeg}, $\spceval{\sstate}{\kneg e}{\kneg t}{\_}$.

  \end{enumcases}
\end{proof}

\subsection{Produce}


\begin{definition}
  For a judgement $\sproduce{\sstate}{\gform}{\sstate'}$, given an initial valuation $V$ and heap $\heap$, the \textbf{corresponding valuation} is denoted
  $$V[\sproduce{\sstate}{\gform}{\sstate} \mid \heap].$$
  This valuation is defined as follows, depending on the rule that proves the derivation. Values are referenced using the respective name from the rule definition.

  Note that the corresponding valuation always extends the initial valuation and is defined for all $\ffresh$ symbolic values in the judgement.
  \begin{itemize}
    \item \refrule{SProduceImprecise}:
      $$V[\sproduce{\sstate}{\simprecise{\phi}}{\sstate'} \mid \heap] := V[\sproduce{\sstate[\imp = \top]}{\phi}{\sstate'} \mid \heap]$$
    \item \refrule{SProduceExpr}:
      $$V[\sproduce{\sstate}{e}{\sstate'} \mid \heap] := V[\spceval{\sstate}{e}{t}{\_} \mid \heap]$$
    \item \refrule{SProducePredicate}:
      $$V[\sproduce{\sstate}{p(\multiple{e})}{\sstate'}] := V[\multiple{\spceval{\sstate}{e}{t}{\_} \mid \heap}]$$
    \item \refrule{SProduceField}:
      $$V[\sproduce{\sstate}{\kacc(e.f)}{\sstate'} \mid \heap] := V'[t \mapsto \heap(V'(t_e), f)]$$
      where $V' = V[\spceval{\sstate}{e}{t}{\_} \mid \heap]$.
    \item \refrule{SProduceConjunction}:
      $$V[\sproduce{\sstate}{\phi_1 * \phi_2}{\sstate''} \mid \heap] := V[\sproduce{\sstate}{\phi_1}{\sstate'} \mid \heap][\sproduce{\sstate'}{\phi_2}{\sstate''} \mid \heap]$$
    \item \refrule{SProduceIfA}:
      \begin{align*}
        &V[\sproduce{\sstate}{\sif{e}{\phi_1}{\phi_2}}{\sstate'} \mid \heap] := \\&\quad V[\spceval{\sstate}{e}{t}{\_} \mid \heap][\sproduce{\sstate[\pc = \pc(\sstate) \kand t]}{\phi_1}{\sstate'} \mid \heap]
      \end{align*}
    \item \refrule{SProduceIfB}:
      \begin{align*}
        &V[\sproduce{\sstate}{\sif{e}{\phi_1}{\phi_2}}{\sstate'} \mid \heap] := \\&\quad V[\spceval{\sstate}{e}{t}{\_} \mid \heap][\sproduce{\sstate[\pc = \pc(\sstate) \kand \kneg t]}{\phi_2}{\sstate'} \mid \heap]
      \end{align*}
  \end{itemize}
\end{definition}

\begin{lemma}\label{lem:produce-subpath}
  If $\sproduce{\sstate}{\phi}{\sstate'}$, then $\pc(\sstate') \implies \pc(\sstate)$.
\end{lemma}

\begin{proof}
  By induction on $\sproduce{\sstate}{\phi}{\sstate'}$:

  \begin{enumcases}
    \case \refrule{SProduceImprecise} -- $\sproduce{\sstate}{\simprecise{\phi}}{\sstate'}$:
      By \refrule{SProduceImprecise} $\sproduce{\sstate[\imp = \top]}{\phi}{\sstate'}$, thus by induction $\pc(\sstate') \implies \pc(\sstate)$.

    \case \refrule{SProduceExpr} -- $\sproduce{\sstate}{e}{\sstate[\pc = \pc(\sstate) \kand t]}$:
      Trivially $\pc(\sstate) \kand t \implies \pc(\sstate)$.

    \case\label{case:produce-subpath-pred} \refrule{SProducePredicate} -- $\sproduce{\sstate}{p(\multiple{e})}{\sstate'}$:
      By \refrule{SProducePredicate} $\sstate' = \sstate[\sheap = \cdots]$, thus $\pc(\sstate') = \pc(\sstate)$.

    \case \refrule{SProduceField} -- $\sproduce{\sstate}{\kacc(e.f)}{\sstate'}$: Similar to case \ref{case:produce-subpath-pred}.

    \case \refrule{SProduceConjunction} -- $\sproduce{\sstate}{\phi_1 * \phi_2}{\sstate''}$:
      By \refrule{SProduceConjunction} $\sproduce{\sstate}{\phi_1}{\sstate'}$ and $\sproduce{\sstate'}{\phi_2}{\sstate''}$. By induction $\pc(\sstate'') \implies \pc(\sstate')$ and $\pc(\sstate') \implies \pc(\sstate)$, therefore $\pc(\sstate'') \implies \pc(\sstate)$.

    \case\label{case:produce-subpath-ifa} \refrule{SProduceIfA} -- $\sproduce{\sstate}{\sif{e}{\phi_1}{\phi_2}}{\sstate'}$:
      By \refrule{SProduceIfA} $\sproduce{\sstate[\pc = \pc(\sstate) \kand t]}{\phi_1}{\sstate'}$, thus by induction $\pc(\sstate') \implies \pc(\sstate) \kand t \implies \pc(\sstate)$.

    \case \refrule{SProduceIfB} -- $\sproduce{\sstate}{\sif{e}{\phi_1}{\phi_2}}{\sstate'}$: Similar to case \ref{case:produce-subpath-ifa}.
  \end{enumcases}
\end{proof}

\begin{lemma}\label{lem:produce-unchanged}
  If $\sproduce{\sstate}{\gform}{\sstate'}$ then $\senv(\sstate') = \senv(\sstate)$.
\end{lemma}
\begin{proof}
  Trivial by induction on $\sproduce{\sstate}{\gform}{\sstate'}$.
\end{proof}

\begin{lemma}\label{lem:produce-soundness}
  Suppose $\simstate{V}{\sstate}{\heap}{\perms \setminus \efoot{\heap}{\env}{\gform}}{\env}$.
  
  If $\sproduce{\sstate}{\gform}{\sstate'}$, $\assertion{\heap}{\perms}{\env}{\gform}$, and $V'(\pc(\sstate')) = \ktrue$ where $V' = V[\sproduce{\sstate}{\gform}{\sstate'} \mid \heap]$, then
  $$\simstate{V'}{\sstate'}{\heap}{\perms}{\env}$$
\end{lemma}
\begin{proof}
  By induction on $\sproduce{\sstate}{\gform}{\sstate'}$:
  \begin{enumcases}
    \case \refrule{SProduceImprecise} -- $\sproduce{\sstate}{\simprecise{\phi}}{\sstate'}$:

      By \refrule{SProduceImprecise} $\sproduce{\sstate[\imp = \top]}{\phi}{\sstate'}$. Let $V'$ be the corresponding valuation, thus $V'$ is the corresponding valuation for this case.

      Then, since $\assertion{\heap}{\perms}{\env}{\simprecise{\phi}}$, by \refrule{AssertImprecise} $\assertion{\heap}{\perms}{\env}{\phi}$. Also, $\efoot{\heap}{\env}{\simprecise{\phi}} = \efoot{\heap}{\env}{\phi}$. Therefore $\simstate{V}{\sstate}{\heap}{\perms \setminus \efoot{\heap}{\env}{\phi}}{\env}$, and furthermore $\simstate{V}{\sstate[\imp = \top]}{\heap}{\perms \setminus \efoot{\heap}{\env}{\phi}}{\env}$.

      Therefore $\simstate{V'}{\sstate'}{\heap}{\perms}{\env}$ by induction.

    \case \refrule{SProduceExpr} -- $\sproduce{\sstate}{e}{\sstate[\pc = \pc(\sstate) \kand t]}$:
    
      By \refrule{SProduceExpr} $\spceval{\sstate}{e}{t}{\_}$. Let $V'$ be the corresponding valuation, therefore $V'$ is the corresponding valuation for this case.

      Let $\sstate' = \sstate[\pc(\sstate) \wedge t]$. Since $V'$ extends $V$ and $\simstate{V}{\sstate}{\heap}{\perms \setminus \efoot{\heap}{\env}{e}}{\env}$, $\simstate{V'}{\sstate}{\heap}{\perms}{\env}$. Then, since $\sstate'$ and $\sstate$ differ only in their $\pc$ components $\simstate{V'}{\sstate'}{\heap}{\perms}{\env}$ since $V'(\pc(\sstate')) = \ktrue$ by assumptions.

    \case \refrule{SProducePredicate} -- $\sproduce{\sstate}{p(\multiple{e})}{\sstate'}$:

      By \refrule{SProducePredicate}, for each $e$, $\spceval{\sstate}{e}{t}{\_}$ for some $t_i$. Let $V'$ be the corresponding valuation for this case, thus $V'$ extends the corresponding valuation corresponding for each $\spceval{\sstate}{e}{t}{\_}$.
      
      By \refrule{SProducePredicate} $\sstate' = \sstate[\sheap = \sheap(\sstate); \pair{p}{\multiple{t}}]$. Since $\sstate$ and $\sstate'$ differ only in their $\sheap$ components, proving that \eqref{eq:sheap-correspondence} holds for $\sheap(\sstate')$ is sufficient to prove that $\simstate{V'}{\sstate'}{\heap}{\perms}{\env}$.

      By assumptions, $\assertion{\heap}{\perms}{\env}{p(\multiple{e})}$ thus by \refrule{AssertPredicate} $\eval{\heap}{\env}{e}{v}$ for some $v$ for all $e$. Then by lemma \ref{lem:spceval-correspondence} $\eval{\heap}{\env}{e}{V'(t)}$ for all $e$. Then
      $$
        \efoot{\heap}{\env}{p(\multiple{e})} \supseteq \efoot{\heap}{\multiple{x \mapsto V'(t)}}{\fpred(p)} = \vfoot{V'}{\heap}{\pair{p}{\multiple{t}}},
      $$
      where $\multiple{x} = \fpredparams(p)$
      thus $\simheap{V}{\sheap(\sstate)}{\heap}{\perms \setminus \vfoot{V}{\heap}{\pair{p}{\multiple{t}}}}$ by lemma \ref{lem:simstate-monotonicity}. Now by lemma \ref{lem:sim-heap-disjoint} and since $V'$ extends $V$,
      $$\universal{h \in \sheap(\sstate)}{\vfoot{V'}{\heap}{h} \cap \vfoot{V'}{\heap}{\pair{p}{\multiple{t}}} = \emptyset}.$$
      Since $\pair{p}{\multiple{t}}$ is the only addition to $\sheap(\sstate')$ relative to $\sheap(\sstate)$ and $V'$ extends $V$,
      $$\universal{h_1, h_2 \in \sheap(\sstate')}{h_1 \ne h_2 \implies \vfoot{V'}{\heap}{h_1} \cap \vfoot{V'}{\heap}{h_2} = \emptyset}.$$
      Also by \refrule{AssertPredicate} $\assertion{\heap}{\perms}{[\multiple{x \mapsto V'(t)}]}{\fpred(p)}$. Therefore, since $\pair{p}{\multiple{t}}$ is the only addition to $\sheap(\sstate')$ relative to $\sheap(\sstate)$ and $V'$ extends $V$,
      $$\universal{\pair{p}{\multiple{t}} \in \sheap(\sstate')}{\assertion{\heap}{\perms}{[\multiple{x \mapsto V'(t)}]}{\fpred(p)}}.$$
      Therefore all requirements for \eqref{eq:sheap-correspondence} are satisfied, and therefore $\simstate{V'}{\sstate'}{\heap}{\perms}{\env}$.

    \case \refrule{SProduceField} -- $\sproduce{\sstate}{\kacc(e.f)}{\sstate'}$:

      By \refrule{SProduceField} $\spceval{\sstate}{e}{t_e}{\_}$ for some $t_e$. Let $V_e$ be the corresponding valuation, and let $V'$ be the corresponding valuation for this case, thus $V'$ extends $V_e$.
      
      Also by \refrule{SProduceField} $\sstate' = \sstate[\sheap = \sheap(\sstate); \triple{f}{t_e}{t}]$ where $t = \ffresh$. By definition $V'(t) = \heap(V_e(t_e), f)$.

      Since $\sstate$ and $\sstate'$ differ only in their $\sheap$ components, proving that \eqref{eq:sheap-correspondence} holds for $\sheap(\sstate')$ is sufficient to prove that $\simstate{V'}{\sstate'}{\heap}{\perms}{\env}$.

      By assumptions $\assertion{\heap}{\perms}{\env}{\kacc(e.f)}$, thus by \refrule{AssertAcc} $\eval{\heap}{\env}{e}{v_e}$ for some $v_e$ such that $\pair{v_e}{f} \in \perms$. By lemma \ref{lem:spceval-correspondence} $v_e = V_e(t_e)$, thus $\pair{V'(t_e)}{f} \in \perms$.

      As shown before, $V'(t) = \heap(V'(t_e), f)$. Therefore, since $\triple{f}{t_e}{t}$ is the only addition to $\sheap(\sstate')$ relative to $\sheap(\sstate)$ and $V'$ extends $V$,
      $$\universal{\triple{f}{t}{t'} \in \sheap(\sstate')}{\heap(V'(t), f) = V'(t')} ~\text{and}$$
      $$\universal{\triple{f}{t}{t'} \in \sheap(\sstate')}{\pair{V'(t)}{f} \in \perms}$$
      Since $\eval{\heap}{\env}{e}{V'(t)}$ as shown before,
      $$\efoot{\heap}{\env}{e} \supseteq \set{\pair{V'(t_e)}{f}} = \vfoot{V}{\heap}{(t_e.f; t)}.$$
      Therefore $\simheap{V}{\sheap(\sstate)}{\heap}{\perms \setminus \vfoot{V'}{\heap}{(t.f; t')}}$ by lemma \ref{lem:simstate-monotonicity}. Now by lemma \ref{lem:sim-heap-disjoint}, and since $V'$ extends $V$,
      $$\universal{h \in \sheap(\sstate')}{\vfoot{V'}{\heap}{h} \cap \vfoot{V'}{\heap}{\triple{f}{t_e}{t}} = \emptyset}.$$
      Finally, since $\triple{f}{t_e}{t}$ is the only addition to $\sheap(\sstate')$ relative to $\sheap(\sstate)$ and $V'$ extends $V$,
      $$\universal{h_1, h_2 \in \sheap(\sstate')^2}{h_1 \ne h_2 \implies \vfoot{V'}{\heap}{h_1} \cap \vfoot{V'}{\heap}{h_2} = \emptyset}.$$
      Therefore all requirements for \eqref{eq:sheap-correspondence} are satisfied and therefore $\simstate{V'}{\sstate'}{\heap}{\perms}{\env}$.

    \case \refrule{SProduceConjunction} -- $\sproduce{\sstate}{\phi_1 * \phi_2}{\sstate''}$:

      By \refrule{SProduceConjunction} $\sproduce{\sstate}{\phi_1}{\sstate'}$ and $\sproduce{\sstate'}{\phi_2}{\sstate''}$. Let $V_1$ and $V_2$ be the respective corresponding valuations, extending $V$ and $V_1$, respectively. Then $V_2$ is the corresponding valuation for this case.

      Since $\assertion{\heap}{\perms}{\env}{\phi_1 * \phi_2}$, by \refrule{AssertConjunction} $\assertion{\heap}{\perms}{\env}{\phi_1}$ and $\assertion{\heap}{\perms}{\env}{\phi_2}$ where $\efoot{\heap}{\env}{\phi_1} \cap \efoot{\heap}{\env}{\phi_2} = \emptyset$. Also, $\efoot{\heap}{\env}{\phi_1 * \phi_2} = \efoot{\heap}{\env}{\phi_1} \cup \efoot{\heap}{\env}{\phi_2}$.

      Let $\perms' = \perms \setminus \efoot{\heap}{\env}{\phi_2}$. Then $\assertion{\heap}{\perms'}{\env}{\phi_1}$ by lemma \ref{lem:assert-efoot-subset} since $\efoot{\heap}{\env}{\phi_1} \subseteq \perms'$. Also, $\simstate{V}{\sstate}{\heap}{\perms' \setminus \efoot{\heap}{\env}{\phi_1}}{\env}$ since $\perms' \setminus \efoot{\heap}{\env}{\phi_1} = \perms \setminus \efoot{\heap}{\env}{\phi_1 * \phi_2}$. Finally, by lemma \ref{lem:produce-subpath}, $\pc(\sstate'') \implies \pc(\sstate')$, and therefore $V_2(\pc(\sstate')) = V_1(\pc(\sstate')) = \ktrue$.

      Now $\simstate{V_1}{\sstate'}{\heap}{\perms'}{\env}$ by induction. Also, $\assertion{\heap}{\perms}{\env}{\phi_2}$ by lemma \ref{lem:assert-efoot-subset} since $\efoot{\heap}{\env}{\phi_2} \subseteq \perms$. Finally, $V_2(\pc(\sstate'')) = \ktrue$ by assumption, therefore $\simstate{V_2}{\sstate''}{\heap}{\perms}{\env}$ by induction.

    \case\label{case:produce-soundness-ifa} \refrule{SProduceIfA} -- $\sproduce{\sstate}{\sif{e}{\phi_1}{\phi_2}}{\sstate'}$:

      By \refrule{SProduceIfA} $\spceval{\sstate}{e}{t}{\_}$ and $\sproduce{\sstate[\pc = \pc(\sstate) \kand t]}{\phi_1}{\sstate'}$. Let $V_1$ and $V'$ be the respective corresponding valuations extending $V$ and $V_1$, respectively. Then $V'$ is the corresponding valuation in this case.

      By lemma \ref{lem:produce-subpath}, $\pc(\sstate') \implies \pc(\sstate) \kand t$. By assumptions $V'(\pc(\sstate')) = \ktrue$, thus $V'(\pc(\sstate) \kand t) = V_1(\pc(\sstate) \kand t) = \ktrue$.

      Then also $V_1(t) = \ktrue$. Since $\assertion{\heap}{\perms}{\env}{\sif{e}{\phi_1}{\phi_2}}$, $\eval{\heap}{\env}{e}{v}$ for some $v$ by \refrule{AssertIfA} or \refrule{AssertIfB}. But by lemma \ref{lem:spceval-correspondence} $v = V_1(t) = \ktrue$. Therefore \\
      $\efoot{\heap}{\env}{\phi_1} \subseteq \efoot{\heap}{\env}{\sif{e}{\phi_1}{\phi_2}}$ by definition. Therefore, since $V_1$ extends $V$,
      $$\simstate{V_1}{\sstate[\pc = \pc(\sstate) \kand t]}{\heap}{\perms \setminus \efoot{\heap}{\env}{\phi_1}}{\env}.$$
      Also, $\assertion{\heap}{\perms}{\env}{\phi_1}$ by \refrule{AssertIfA} since $\eval{\heap}{\env}{e}{\ktrue}$ and $\assertion{\heap}{\perms}{\env}{\sif{e}{\phi_1}{\phi_2}}$. Finally, $V'(\sstate') = \ktrue$ by assumption.
      
      Thus $\simstate{V'}{\sstate'}{\heap}{\perms}{\env}$ by induction.

    \case \refrule{SProduceIfB}: Similar to case \ref{case:produce-soundness-ifa}.
  \end{enumcases}
\end{proof}

\begin{lemma}[Progress] \label{lem:produce-progress}
  If $V(\pc(\sstate)) = \ktrue$, $\assertion{\heap}{\perms}{\env}{\gform}$, and $\simstate{V}{\sstate}{\heap}{\perms}{\env}$, then $\sproduce{\sstate}{\gform}{\sstate'}$ for some $\sstate'$ where $V'(\pc(\sstate')) = \ktrue$ and $V' = V[\sproduce{\sstate}{\gform}{\sstate'} \mid \heap]$.
\end{lemma}
\begin{proof}
  Suppose $V(\pc(\sstate)) = \ktrue$ and complete the proof by induction on the syntax forms of $\gform$:

  \begin{enumcases}
    \case $\simprecise{\phi}$ -- $\phi \in \Formula$:

      Let $\hat{\sstate} = \sstate[\imp = \top]$. Then $V(\pc(\hat{\sstate})) = V(\pc(\sstate)) = \ktrue$. Also, since $\assertion{\heap}{\perms}{\env}{\simprecise{\phi}}$, by \refrule{AssertImprecise} $\assertion{\heap}{\perms}{\env}{\phi}$. Thus by induction $\sproduce{\hat{\sstate}}{\phi}{\sstate'}$ for some $\sstate'$ where $V'(\pc(\sstate')) = \ktrue$. Then $\sproduce{\sstate}{\simprecise{\phi}}{\sstate'}$ by \refrule{SProduceImprecise}, and $V'(\pc(\sstate')) = \ktrue$.

    \case $\phi_1 * \phi_2$ -- $\phi_1, \phi_2 \in \Formula$:

      Since $\assertion{\heap}{\perms}{\env}{\phi_1 * \phi_2}$, $\assertion{\heap}{\perms}{\env}{\phi_1}$ and $\assertion{\heap}{\perms}{\env}{\phi_2}$ by \refrule{AssertConjunction} and lemma \ref{lem:assert-monotonicity}.

      By induction $\sproduce{\sstate}{\phi_1}{\sstate'}$, with corresponding valuation $V'$ where $V'(\pc(\sstate')) = \ktrue$. Then by induction $\sproduce{\sstate'}{\phi_2}{\sstate''}$, with corresponding valuation $V''$, for some $\sstate''$ where $V''(\pc(\sstate'')) = \ktrue$. By \refrule{SProduceConjunction}, $\sproduce{\sstate}{\phi_1 * \phi_2}{\sstate''}$, and $V''(\pc(\sstate'')) = \ktrue$.

    \case $p(\multiple{e})$ -- $p \in \Predicate, \multiple{e \in \Expr}$:

      For each $e$, by lemma \ref{lem:pc-eval-progress}, $\spceval{\sstate}{e}{t}{\_}$ for some $t \in \SExpr$. Let $\sstate' = \sstate[\sheap = \sheap(\sstate); \pair{p}{\multiple{t}}]$, then by \refrule{SProducePredicate}, $\sproduce{\sstate}{p(\multiple{e})}{\sstate'}$. Finally, letting $V'$ be the corresponding valuation extending $V$, $V'(\pc(\sstate')) = V(\pc(\sstate)) = \ktrue$.

    \case $e \in \Expr$:

      By lemma \ref{lem:pc-eval-progress}, $\spceval{\sstate}{e}{t}{\_}$ for some $t \in \SExpr$. Let $V_1$ be the corresponding valuation and $\sstate' = \sstate[\pc = \pc(\sstate) \kand t]$. Then by \refrule{SProduceExpr} $\sproduce{\sstate}{e}{\sstate'}$. Let $V'$ be the corresponding valuation extending $V$, thus $V'$ extends $V_1$.

      Since $\assertion{\heap}{\perms}{\env}{e}$, $\eval{\heap}{\env}{e}{\ktrue}$ by \refrule{AssertValue} and thus $V'(t) = \ktrue$ by lemma \ref{lem:spceval-correspondence}. Finally, $V'(\sstate') = V'(\sstate) \wedge V'(t) = V(\sstate) \wedge \ktrue = \ktrue$.

    \case $\sif{e}{\phi_1}{\phi_2}$ -- $e \in \Expr, \phi_1, \phi_2 \in \Formula$:

      By lemma \ref{lem:pc-eval-progress}, $\spceval{\sstate}{e}{t}{\_}$ for some $t \in \SExpr$. Let $V_1$ be the valuation corresponding to this derivation.

      Then one of the following rules must apply to produce $\assertion{\heap}{\perms}{\env}{\sif{e}{\phi_1}{\phi_2}}$:

      \subcase \refrule{AssertIfA}:

        Then $\eval{\heap}{\env}{e}{\ktrue}$. Then $V_1(t) = \ktrue$ by lemma \ref{lem:spceval-correspondence} and therefore $V_1(\pc(\sstate) \kand t) = \ktrue$. Also, $\assertion{\heap}{\perms}{\env}{\phi_1}$ by \refrule{AssertIfA}.

        Then by induction, for some $\sstate'$, $\sproduce{\sstate[\pc = \pc(\sstate) \kand t]}{\phi_1}{\sstate'}$ with corresponding valuation $V'$ (extending $V_1$) where $V'(\sstate') = \ktrue$.

        Now by \refrule{SProduceIfA}, $\sproduce{\sstate}{\sif{e}{\phi_1}{\phi_2}}{\sstate'}$. By definition $V'$ is the corresponding valuation extending $V$, and as shown before, $V'(\pc(\sstate')) = \ktrue$.

      \subcase \refrule{AssertIfB}:

        Then $\eval{\heap}{\env}{e}{\kfalse}$. Then $V_1(t) = \kfalse$ by lemma \ref{lem:spceval-correspondence} and therefore $V_1(\pc(\sstate) \kand \kneg t) = \ktrue$. Also, $\assertion{\heap}{\perms}{\env}{\phi_2}$ by \refrule{AssertIfB}. 

        Then by induction, for some $\sstate'$, $\sproduce{\sstate[\pc = \pc(\sstate) \kand \kneg t]}{\phi_2}{\sstate'}$ with corresponding valuation $V'$ (extending $V_1$) where $V'(\sstate) = \ktrue$.

        Now by \refrule{SProduceIfB}, $\sproduce{\sstate}{\sif{e}{\phi_1}{\phi_2}}{\sstate'}$. By definition $V'$ is the corresponding valuation extending $V$, and as shown before, $V'(\pc(\sstate')) = \ktrue$.

    \case $\kacc(e.f)$ where $e \in \Expr, f \in \Field$

      By lemma \ref{lem:pc-eval-progress}, $\spceval{\sstate}{e}{t}{\_}$ for some $t$. Let $\sstate' = \sstate[\sheap = \sheap(\sstate); \triple{f}{t}{\ffresh}]$. Then, by \refrule{SProduceField}, $\sproduce{\sstate}{\kacc(e.f)}{\sstate'}$. Let $V'$ be the corresponding valuation extending $V$, then $V'(\pc(\sstate')) = V(\pc(\sstate)) = \ktrue$.

  \end{enumcases}
\end{proof}

\subsection{Consume}


\begin{definition}\label{def:consume-valuation}
  For a judgement $\sconsume{\sstate}{\sstate_E}{\gform}{\sstate'}{\scheck}{\sperms}$, given an initial valuation $V$ and heap $\heap$, the \textbf{corresponding valuation} is denoted
  $$V[\sconsume{\sstate}{\sstate_E}{\gform}{\sstate'}{\scheck}{\sperms} \mid \heap].$$
  This valuation is defined as follows, depending on the rule that proves the derivation. Values are referenced using the respective name from the rule definition.

  Note that the corresponding valuation always extends the initial valuation and is defined for all $\ffresh$ symbolic values in the judgement.
  \begin{itemize}
    \item \refrule{SConsumeImprecision}:
      \begin{align*}
        &V[\sconsume{\sstate}{\sstate_E}{\simprecise{\phi}}{\quintuple{\top}{\pc(\sstate')}{\senv(\sstate')}{\emptyset}{\emptyset}}{\scheck}{\sperms} \mid \heap] := \\
        &\quad V[\sconsume{\sstate}{\sstate_E}{\phi}{\sstate'}{\scheck}{\sperms} \mid \heap]
      \end{align*}
    \item \refrule{SConsumeValue}:
      $$V[\sconsume{\sstate}{\sstate_E}{e}{\sstate}{\scheck}{\emptyset} \mid \heap] := V[\spceval{\sstate_E}{e}{t}{\scheck}]$$
    \item \refrule{SConsumeValueImprecise}:
      $$V[\sconsume{\sstate}{\sstate_E}{e}{\sstate[\pc = \pc(\sstate) \kand t]}{\scheck; t}{\emptyset} \mid \heap] := V[\spceval{\sstate_E}{e}{t}{\scheck} \mid \heap]$$
    \item \refrule{SConsumeValueFailure}:
      $$V[\sconsume{\sstate}{\sstate_E}{e}{\sstate}{\set{\bot}}{\emptyset} \mid \heap] := V[\spceval{\sstate_E}{e}{t}{\scheck} \mid \heap]$$
    \item \refrule{SConsumePredicate}:
      $$V[\sconsume{\sstate}{\sstate_E}{p(\multiple{e})}{\sstate[\sheap = \sheap', \oheap = \emptyset]}{\_}{\_} \mid \heap] := V[\multiple{\spceval{\sstate_E}{e}{t}{\scheck}} \mid \heap]$$
    \item \refrule{SConsumePredicateImprecise}:
      $$V[\sconsume{\sstate}{\sstate_E}{p(\multiple{e})}{\sstate[\heap = \emptyset, \oheap = \emptyset]}{\_}{\_} \mid \heap] := V[\multiple{\spceval{\sstate_E}{e}{t}{\scheck}} \mid \heap]$$
    \item \refrule{SConsumePredicateFailure}:
      $$V[\sconsume{\sstate}{\sstate_E}{p(\multiple{e})}{\sstate}{\set{\bot}}{\_} \mid \heap] := V[\multiple{\spceval{\sstate_E}{e}{t}{\scheck}} \mid \heap]$$
    \item \refrule{SConsumeAcc}, \refrule{SConsumeAccOptimistic}, \refrule{SConsumeAccImprecise}, \\ \refrule{SConsumeAccFailure}:
      \begin{align*}
        &V[\sconsume{\sstate}{\sstate_E}{\kacc(e.f)}{\sstate[\sheap = \sheap', \oheap = \oheap']}{\_}{\_} \mid \heap] := \\
        &\quad V[\spceval{\sstate_E}{e}{t_e}{\scheck} \mid \heap]
      \end{align*}
    \item \refrule{SConsumeConjunction}:
      \begin{align*}
        &V[\sconsume{\sstate}{\sstate_E}{\phi_1 * \phi_2}{\sstate''}{\_}{\_} \mid \heap] := \\
        &\quad V[\sconsume{\sstate}{\sstate_E}{\phi_1}{\sstate'}{\scheck_1}{\sperms_1} \mid \heap] \\
        &\hspace{1.85em} [\sconsume{\sstate'}{\sstate_E[\pc = \pc(\sstate')]}{\phi_2}{\sstate''}{\scheck_2}{\sperms_2} \mid \heap]
      \end{align*}
    \item \refrule{SConsumeConditionalA}:
      \begin{align*}
        &V[\sconsume{\sstate}{\sstate_E}{\sif{e}{\phi_1}{\phi_2}}{\sstate'}{\scheck \cup \scheck'}{\_} \mid \heap] := \\
        &\quad V[\sconsume{\sstate[\pc = \pc']}{\sstate_E[\pc = \pc']}{\phi_1}{\sstate'}{\scheck'}{\sperms} \mid \heap]
      \end{align*}
    \item \refrule{SConsumeConditionalB}:
      \begin{align*}
        &V[\sconsume{\sstate}{\sstate_E}{\sif{e}{\phi_1}{\phi_2}}{\sstate'}{\scheck \cup \scheck'}{\_} \mid \heap] := \\
        &\quad V[\sconsume{\sstate[\pc = \pc']}{\sstate_E[\pc = \pc']}{\phi_2}{\sstate'}{\scheck'}{\sperms} \mid \heap]
      \end{align*}
  \end{itemize}
\end{definition}

\begin{definition}\label{def:cons-valuation}
  For a judgement $\scons{\sstate}{\gform}{\sstate'}{\scheck}$, given an initial valuation $V$ and heap $\heap$, the \textbf{corresponding valuation} is denoted
  $$V[\scons{\sstate}{\gform}{\sstate'}{\scheck} \mid \heap].$$

  The is defined by
  $$V[\scons{\sstate}{\gform}{\sstate'}{\scheck} \mid \heap] := V[\sconsume{\sstate}{\sstate}{\gform}{\sstate'}{\scheck}{\_} \mid \heap]$$
  where $\sconsume{\sstate}{\sstate}{\gform}{\sstate'}{\scheck}{\_}$ is the judgement used when applying \textsc{SConsume} to derive $\scons{\sstate}{\gform}{\sstate'}{\scheck}$.
\end{definition}

\begin{lemma}[Consume results in more specific path condition (long form)]\label{lem:consume-subpath}
  If $\sconsume{\sstate}{\sstate_E}{\gform}{\sstate'}{\_}{\_}$, then $\pc(\sstate') \implies \pc(\sstate)$.
\end{lemma}
\begin{proof}
  By induction on $\sconsume{\sstate}{\sstate_E}{\gform}{\sstate'}{\_}{\_}$:

  \begin{enumcases}
    \case \refrule{SConsumeImprecision} -- $\sconsume{\sstate}{\sstate_E}{\simprecise{\phi}}{\quintuple{\top}{\pc(\sstate')}{\senv(\sstate')}{\emptyset}{\emptyset}}{\scheck}{\sperms}$:
      By \refrule{SConsumeImprecision} $\sconsume{\sstate}{\sstate_E[\imp = \top]}{\phi}{\sstate'}{\scheck}{\sperms}$. Then $\pc(\quintuple{\top}{\pc(\sstate')}{\senv(\sstate')}{\emptyset}{\emptyset}) = \pc(\sstate')$ and $\pc(\sstate') \implies \pc(\sstate)$ by induction.

    \case \refrule{SConsumeValue}, \refrule{SConsumeValueFailure}, \refrule{SConsumePredicateFailure}, \refrule{SConsumeAccFailure} -- $\sconsume{\sstate}{\sstate_E}{\_}{\sstate}{\_}{\_}$:
      Trivially $\pc(\sstate) \implies \pc(\sstate)$.

    \case \refrule{SConsumeValueImprecise} -- $\sconsume{\sstate}{\sstate_E}{e}{\sstate[\pc = \pc(\sstate) \kand t]}{\scheck; t}{\emptyset}$:
      $\pc(\sstate[\pc = \pc(\sstate) \kand t]) = \pc(\sstate) \kand t \implies \pc(\sstate)$.

    \case \refrule{SConsumePredicate}, \refrule{SConsumePredicateImprecise}, \refrule{SConsumeAcc}, \refrule{SConsumeAccOptimistic}, \refrule{SConsumeAccImprecise} -- $\sconsume{\sstate}{\sstate_E}{\_}{\sstate'}{\_}{\_}$:
      In each respective rule $\pc(\sstate') = \pc(\sstate)$, therefore $\pc(\sstate') \implies \pc(\sstate)$.

    \case \refrule{SConsumeConjunction}, \refrule{SConsumeConjunctionImprecise} -- $\sconsume{\sstate}{\sstate_E}{\phi_1 * \phi_2}{\sstate''}{\_}{\sperms_1 \cup \sperms_2}$:
      In each respective rule $\sconsume{\sstate}{\sstate_E}{\phi_1}{\sstate'}{\scheck_1}{\sperms_1}$ and $\sconsume{\sstate'}{\sstate_E[\pc = \pc(\sstate')]}{\phi_2}{\sstate''}$, therefore by induction $\pc(\sstate'') \implies \pc(\sstate') \implies \pc(\sstate)$.

    \case \refrule{SConsumeConditionalA}, \refrule{SConsumeConditionalB} -- $\sconsume{\sstate}{\sstate_E}{\sif{e}{\phi_1}{\phi_2}}{\sstate'}{\scheck \cup \scheck'}{\sperms}$:
      In each respective rule $\sconsume{\sstate[\pc = \pc(\sstate) \kand t]}{\sstate_E}{\phi'}{\sstate'}{\_}{\_}$ for some $t, \phi'$. Therefore by induction $\pc(\sstate') \implies \pc(\sstate) \kand t' \implies \pc(\sstate)$.

  \end{enumcases}
\end{proof}

\begin{lemma}[Consume results in more specific path condition (short form)]\label{lem:cons-subpath}
  If $\scons{\sstate}{\gform}{\sstate'}{\_}$, then $\pc(\sstate') \implies \pc(\sstate)$.
\end{lemma}
\begin{proof}
  By \refrule{SConsume} $\sconsume{\sstate}{\sstate}{\gform}{\sstate'}{\_}{\_}$, thus by lemma \ref{lem:consume-subpath}, $\pc(\sstate') \implies \pc(\sstate)$.
\end{proof}

\begin{lemma}[Soundness of $\fremf$ for precise heaps and imprecise states]\label{lem:sheap-remf-imp}
  If $\simstate{V}{\sstate}{\heap}{\perms}{\env}$ and $\imp(\sstate)$ then $\simheap{V}{\fremf(\sheap(\sstate), \sstate, t, f)}{\heap}{\perms \setminus \set{(V(t), f)}}$ for any $t$, $f$.
\end{lemma}
\begin{proof}
  Let $\sheap' = \fremf(\sheap(\sstate), \sstate, t, f)$ and $\perms' = \perms \setminus \set{(V(t), f)}$. Then by definition $\sheap' \subseteq \sheap(\sstate)$, therefore $\simheap{V}{\sheap'}{\heap}{\perms}$. In addition, $\sheap'$ contains no predicate values. Thus by \eqref{eq:sheap-correspondence} it suffices to show that $\universal{\triple{f'}{t'}{t''} \in \sheap'}{\pair{V(t')}{f'} \in \perms'}$.

  Let $\triple{f'}{t'}{t''}$ be some element of $\sheap'$. Then by definition $\neg \falias(\sstate, t, f, t', f')$; then by definition $(f \ne f') \vee \neg\fsat(\pc(\sstate) \kand [t \keq t'])$.

  Therefore $(f \ne f') \vee (V(\pc(\sstate) \kand t \keq t') = \kfalse)$, and thus $(f \ne f') \vee (V(\pc(\sstate)) = \kfalse) \vee (V(t) \ne V(t'))$. But $V(\pc(\sstate)) = \ktrue$, thus $(f \ne f') \vee (V(t) \ne V(t'))$. Therefore $\pair{V(t)}{f} \ne \pair{V(t')}{f'}$.
  
  Also $\pair{V(t')}{f'} \in \perms$ since $\simheap{V}{\sheap'}{\heap}{\perms}$. Therefore $\pair{V(t')}{f'} \in \perms \setminus \set{(V(t), f)} = \perms'$.

  Therefore $\simheap{V}{\sheap'}{\heap}{\perms'}$.
\end{proof}

\begin{lemma}\label{lem:consume-unchanged}
  If $\sconsume{\sstate}{\sstate_E}{\gform}{\sstate'}{\scheck}{\sperms}$ then $\senv(\sstate') = \senv(\sstate)$.
\end{lemma}
\begin{proof}
  Trivial by induction on $\sconsume{\sstate}{\sstate_E}{\gform}{\sstate'}{\scheck}{\sperms}$.
\end{proof}

\begin{lemma}\label{lem:cons-unchanged}
  If $\scons{\sstate}{\gform}{\sstate'}{\scheck}$ then $\senv(\sstate') = \senv(\sstate)$.
\end{lemma}
\begin{proof}
  By \refrule{SConsume} $\sconsume{\sstate}{\sstate}{\gform}{\sstate}{\scheck}{\_}$, thus by lemma \ref{lem:consume-unchanged} $\senv(\sstate') = \senv(\sstate)$.
\end{proof}

\begin{lemma}[Soundness of $\fremfp$ for precise heaps]\label{lem:sheap-remfp-prec}
  If $\simstate{V}{\sstate}{\heap}{\perms}{\env}$ and $\triple{f}{t}{\_} \in \sheap(\sstate)$, then $\simheap{V}{\fremfp(\sheap(\sstate), \sstate, t, f)}{\heap}{\perms \setminus \set{\pair{V(t)}{f}}}$ for any $t$, $f$.
\end{lemma}
\begin{proof}
  Let $\sheap' = \fremfp(\sheap(\sstate), \sstate, t, f)$, let $\sheap'' = \fremf(\sheap(\sstate), \sstate, t, f)$, and let $\perms' = \perms \setminus \set{\pair{V(t)}{f}}$. By definition of $\fremfp$ and $\fremf$, $\sheap'' \subseteq \sheap' \subseteq \sheap(\sstate)$.

  Let $\triple{f'}{t'}{\_}$ be an arbitrary field chunk in $\sheap'$. Thus $\neg \falias(\sstate, t, f, t', f')$ by definition of $\fremf$. Then $\pair{V(t)}{f} \ne \pair{V(t')}{f'}$:
  \begin{itemize}
    \item If $\imp(\sstate)$: Then $(f \ne f') \vee \neg \fsat(\pc(\sstate) \kand t \keq t')$, therefore $(f \ne f') \vee (V(\pc(\sstate)) \ne \ktrue) \vee (V(t) \ne V(t'))$. But $V(\pc(\sstate)) = \ktrue$, thus $(f \ne f') \vee (V(t) \ne V(t'))$. Therefore $\pair{V(t)}{f} \ne \pair{V(t')}{f'}$.
    \item Otherwise $\neg\imp(\sstate)$: Then $(f \ne f') \vee (\pc(\sstate) \notimplies t \keq t')$. Thus $f \ne f'$ or $V(t) \neq V(t')$. Thus $f \ne f'$ or $t \neq t'$ (using syntactic equivalence). But now $\set{\pair{V(t)}{f}} \cap \set{\pair{V(t')}{f}} = \vfoot{V}{\heap}{\triple{f}{t}{\_}} \cap \vfoot{V}{\heap}{\triple{f'}{t'}{\_}} = \emptyset$ since $\simheap{V}{\sheap(\sstate)}{\heap}{\perms}$, $\triple{f}{t}{\_}$ and $\triple{f'}{t'}{\_} \in \sheap(\sstate)$, and $\triple{f}{t}{\_} \ne \triple{f}{t'}{\_}$. Therefore $\pair{V(t)}{f} \ne \pair{V(t')}{f'}$.
  \end{itemize}
  Therefore $\pair{V(t')}{f'} \in \perms'$ since $\triple{f'}{t'}{\_} \in \sheap(\sstate)$.

  Also, let $\pair{p'}{\multiple{t'}}$ be an arbitrary predicate chunk in $\sheap'$. By definition of $\fremfp$, $\pair{p'}{\multiple{t'}} \in \sheap(\sstate)$; thus $\assertion{\heap}{\perms}{[\multiple{x \mapsto V(t')}]}{\fpred(p')}$.
  
  But also $\vfoot{V}{\heap}{\triple{f}{t}{\_}} \cap \vfoot{V}{\heap}{\pair{p'}{\multiple{t'}}} = \emptyset$ since $\triple{f}{t}{\_}$ and $\pair{p'}{\multiple{t'}} \in \sheap(\sstate)$ and $\triple{f}{t}{\_} \ne \pair{p'}{\multiple{t'}}$. Therefore $\efoot{\heap}{[\multiple{x \mapsto V(t')}]}{\fpred(p)} \subseteq \perms'$ by lemma \ref{lem:efoot-subset-spec} and since \\
  $\set{\pair{V(t)}{f}} \cap \efoot{\heap}{[\multiple{x \mapsto V(t')}]}{\fpred(p)} = \vfoot{V}{\heap}{\triple{f}{t}{\_}} \cap \vfoot{V}{\heap}{\pair{p'}{\multiple{t'}}} = \emptyset$. Thus \\
  $\assertion{\heap}{\perms'}{[\multiple{x \mapsto V(t')}]}{\fpred(p)}$ by lemma \ref{lem:assert-efoot-subset}.

  Finally, since $\sheap' \subseteq \sheap(\sstate)$, the remaining conditions of \eqref{eq:sheap-correspondence} are satisfied. Therefore \\
  $\simheap{V}{\sheap'}{\heap}{\perms'}$.
\end{proof}

\begin{lemma}[Soundness of $\fremf$ for optimistic heaps]\label{lem:oheap-remf}
  If $\sstate$ is well-formed and $\simstate{V}{\sstate}{\heap}{\perms}{\env}$ then $\simheap{V}{\fremf(\oheap(\sstate), \sstate, t, f)}{\heap}{\perms \setminus \set{\pair{V(t)}{f}}}$ for any $t$, $f$.
\end{lemma}
\begin{proof}
  \begin{enumcases}
    \case $\imp(\sstate)$: Similar to proof of lemma \ref{lem:sheap-remf-imp}, replacing $\sheap$ with $\oheap$.

    \case $\neg \imp(\sstate)$: Then $\oheap(\sstate) = \emptyset$ since $\sstate$ is well-formed. Then trivially $\simheap{V}{\oheap(\sstate)}{\heap}{\emptyset}$ and thus also $\simheap{V}{\oheap(\sstate)}{\heap}{\perms \setminus \set{\pair{V(t)}{f}}}$ by \ref{lem:sim-oheap-monotonicity}.
  \end{enumcases}
\end{proof}

\begin{lemma}[Soundness of $\sperms$ calculation]\label{lem:consume-assert-vfoot}
  If $\phi$ is a precise formula, $\simstate{V}{\sstate_E}{\heap}{\perms_E}{\env}$, $\sconsume{\sstate}{\sstate_E}{\phi}{\sstate'}{\scheck}{\sperms}$ with corresponding valuation $V'$, $V'(\pc(\sstate')) = \ktrue$, and $\assertion{\heap}{\perms'}{\env}{\phi}$ where $\perms' \subseteq \perms_E$, then $\assertion{\heap}{\vfoot{V'}{\heap}{\sperms}}{\env}{\phi}$ and $\vfoot{V}{\heap}{\sperms} \subseteq \perms'$.
\end{lemma}
\begin{proof}
  By induction on $\sconsume{\sstate}{\sstate_E}{\phi}{\sstate'}{\scheck}{\sperms}$:

  \begin{enumcases}
    \case \refrule{SConsumeImprecision} -- $\sconsume{\sstate}{\sstate_E}{\simprecise{\phi}}{\quintuple{\top}{\pc(\sstate')}{\senv(\sstate')}{\emptyset}{\emptyset}}{\scheck}{\sperms}$:
      $\simprecise{\phi}$ is not precise, therefore this rule cannot apply.

    \case \refrule{SConsumeValue}, \refrule{SConsumeValueImprecise}, \refrule{SConsumeValueFailure} -- $\sconsume{\sstate}{\sstate_E}{e}{\sstate}{\_}{\emptyset}$:

      Since $\assertion{\heap}{\perms'}{\env}{e}$, $\eval{\heap}{\env}{e}{\ktrue}$ by \refrule{AssertValue}. Therefore $\assertion{\heap}{\vfoot{V}{\heap}{\emptyset}}{\env}{e}$ by \refrule{AssertValue}.

      Also, $\vfoot{V}{\heap}{\emptyset} = \emptyset \subseteq \perms'$.

    \case \refrule{SConsumePredicate}, \refrule{SConsumePredicateImprecise}, \refrule{SConsumePredicateFailure} -- $\sconsume{\sstate}{\sstate_E}{p(\multiple{e})}{\_}{\_}{\set{\pair{p}{\multiple{t}}}}$:

      By the respective rule, $\multiple{\spceval{\sstate_E}{e}{t}{\_}}$ for some $\multiple{t}$. The corresponding valuation for this case extends the corresponding valuation for all of these derivation.

      Since $\assertion{\heap}{\perms}{\env}{p(\multiple{e})}$, $\multiple{\eval{\heap}{\env}{e}{v}}$ for some $\multiple{v}$ by \refrule{AssertPredicate}. By assumptions \\
      $\simstate{V}{\sstate_E}{\heap}{\perms_E}{\env}$. Thus by lemma \ref{lem:spceval-correspondence} $\multiple{v = V'(t)}$, i.e., $\multiple{\eval{\heap}{\env}{e}{V'(t)}}$.

      Let $\multiple{x} = \fpredparams(p)$. By \refrule{AssertPredicate} $\assertion{\heap}{\perms'}{[\multiple{x \mapsto V'(t)}]}{\fpred(p)}$. Also, $\vfoot{V}{\heap}{\pair{p}{\multiple{t}}} = \efoot{\heap}{[\multiple{x \mapsto V'(t)}]}{\fpred(p)}$. Therefore $\assertion{\heap}{\vfoot{V}{\heap}{\pair{p}{\multiple{t}} \cap \perms'}}{[\multiple{x \mapsto V'(t)}]}{\fpred(p)}$ by lemma \ref{lem:efoot-assert}.

      But $\fpred(p)$ must be a specification, thus $\vfoot{V}{\heap}{\pair{p}{\multiple{t}}} = \efoot{\heap}{[\multiple{x \mapsto V'(t)}]}{\fpred(p)} \subseteq \vfoot{V}{\heap}{\pair{p}{\multiple{t}}} \cap \perms'$ by lemma \ref{lem:efoot-subset-spec}. Therefore $\vfoot{V}{\heap}{\set{\pair{p}{\multiple{t}}}} \subseteq \perms'$.

      Then $\assertion{\heap}{\vfoot{V}{\heap}{\set{\pair{p}{\multiple{t}}}}}{[\multiple{x \mapsto V'(t)}]}{\fpred(p)}$, and thus $\assertion{\heap}{\vfoot{V}{\heap}{\set{\pair{p}{\multiple{t}}}}}{\env}{p(\multiple{e})}$ by \refrule{AssertPredicate}.

    \case \refrule{SConsumeAcc}, \refrule{SConsumeAccOptimistic}, \refrule{SConsumeAccImprecise}, \refrule{SConsumeAccFailure} -- $\sconsume{\sstate}{\sstate_E}{\kacc(e.f)}{\sstate'}{\_}{\set{\pair{t_e}{f}}}$:

      By the respective rule, $\spceval{\sstate_E}{e}{t_e}{\_}$ for some $t_e$. The corresponding valuation for this case extends the corresponding valuation for this derivation.

      Since $\assertion{\heap}{\perms}{\env}{\kacc(e.f)}$, $\eval{\heap}{\env}{e}{v}$ for some $v$ by \refrule{AssertAcc}. Thus by lemma \ref{lem:spceval-correspondence} $v = V'(t_e)$, i.e. $\eval{\heap}{\env}{e}{V'(t_e)}$.

      By \refrule{AssertAcc} $\pair{V'(t_e)}{f} \in \perms'$. Also, $\set{\pair{V'(t_e)}{f}}  = \vfoot{V'}{\heap}{\set{\pair{t_e}{f}}}$, therefore $\vfoot{V'}{\heap}{\set{\pair{t_e}{f}}} \subseteq \perms'$ and also $\assertion{\heap}{\vfoot{V'}{\heap}{\set{\pair{t_e}{f}}}}{\env}{\kacc(e.f)}$ by \refrule{AssertAcc}.

    \case \refrule{SConsumeConjunction}, \refrule{SConsumeConjunctionImprecise} -- $\sconsume{\sstate}{\sstate_E}{\phi_1 * \phi_2}{\sstate''}{\_}{\sperms_1 \cup \sperms_2}$:

      By the respective rule, $\sconsume{\sstate}{\sstate_E}{\phi_1}{\sstate'}{\_}{\sperms_1}$ and $\sconsume{\sstate'}{\sstate_E[\pc = \pc(\sstate')]}{\phi_2}{\sstate''}{\_}{\sperms_2}$. Let $V'$ be the corresponding valuation for this case, thus $V'$ extends the corresponding valuations for these judgements.

      By assumptions, $V'(\pc(\sstate'')) = \ktrue$ and $\pc(\sstate'') \implies \pc(\sstate')$ by \ref{lem:consume-subpath}. Therefore $\simstate{V'}{\sstate_E[\pc = \pc(\sstate')]}{\heap}{\perms_E}{\env}$.

      Then $\assertion{\heap}{\perms_1}{\env}{\phi_1}$ and $\assertion{\heap}{\perms_2}{\env}{\phi_2}$ where $\perms_1 \cup \perms_2 \subseteq \perms'$ and $\perms_1 \cap \perms_2 = \emptyset$ by \refrule{AssertConjunction}, since $\assertion{\heap}{\perms'}{\env}{\phi_1 * \phi_2}$. Then by induction $\assertion{\heap}{\vfoot{V'}{\heap}{\sperms_1}}{\env}{\phi_1}$, $\assertion{\heap}{\vfoot{V'}{\heap}{\sperms_2}}{\env}{\phi_2}$, $\vfoot{V'}{\heap}{\sperms_1} \subseteq \perms_1$, and $\vfoot{V'}{\heap}{\sperms_2} \subseteq \perms_1$.

      Now $\vfoot{V'}{\heap}{\sperms_1} \cap \vfoot{V'}{\heap}{\sperms_2} = \emptyset$ and $\vfoot{V'}{\heap}{\sperms_1} \cup \vfoot{V'}{\heap}{\sperms_2} = \vfoot{V'}{\heap}{\sperms_1 \cup \sperms_2} \subseteq \perms'$. Therefore $\assertion{\heap}{\vfoot{V'}{\heap}{\sperms_1 \cup \sperms_2}}{\env}{\phi_1 * \phi_2}$ by \refrule{AssertConjunction}.
  
    \case\label{case:consume-assert-vfoot-ifa} \refrule{SConsumeConditionalA} -- $\sconsume{\sstate}{\sstate_E}{\sif{e}{\phi_1}{\phi_2}}{\sstate'}{\scheck \cup \scheck'}{\sperms}$:

      By \refrule{SConsumeConditionalA} $\spceval{\sstate_E}{e}{t}{\_}$ for some $t$. Let $V_1$ be the corresponding valuation.
      
      Also, $\sconsume{\sstate[\pc = \pc']}{\sstate_E}{\phi_1}{\sstate'}{\_}{\sperms}$ where $\pc' = \pc \kand t$. Let $V'$ be the corresponding valuation for this case, which extends the corresponding valuation for this judgement and $V_1$.

      Since $\assertion{\heap}{\perms'}{\env}{\sif{e}{\phi_1}{\phi_2}}$, by \refrule{AssertIfA} $\eval{\heap}{\env}{e}{v}$ for some $v$. By lemma \ref{lem:spceval-correspondence} $v = V_1(t)$. Also, since $V'(\pc(\sstate')) = \ktrue$ and $\pc(\sstate') \implies \pc' = \pc(\sstate) \kand t$ by lemma \ref{lem:consume-subpath}, $V_1(t) = \ktrue$. Therefore $\eval{\heap}{\env}{e}{\ktrue}$.

      Also by \refrule{AssertIfA} $\assertion{\heap}{\perms'}{\env}{\phi_1}$. Therefore by induction $\assertion{\heap}{\vfoot{V}{\heap}{\sperms}}{\env}{\phi_1}$ and $\vfoot{V}{\heap}{\sperms} \subseteq \perms'$.

      Finally, $\assertion{\heap}{\vfoot{V}{\heap}{\sperms}}{\env}{\sif{e}{\phi_1}{\phi_2}}$ by \refrule{AssertIfA}.

    \case \refrule{SConsumeConditionalB} -- $\sconsume{\sstate}{\sstate_E}{\sif{e}{\phi_1}{\phi_2}}{\sstate'}{\scheck \cup \scheck'}{\sperms}$:
      Similar to case \ref{case:consume-assert-vfoot-ifa}.
  \end{enumcases}
\end{proof}

\begin{lemma}[Soundness of consume for precise formulas]\label{lem:consume-soundness-precise}
  Let $\phi$ be some precise formula, $V$ be some initial valuation, $\heap$ be some heap, $\env$ be some environment, $\perms_E$ and $\perms$ be sets of permissions such that $\perms \subseteq \perms_E$, and $\sstate$ and $\sstate_E$ be well-formed symbolic states such that $\simstate{V}{\sstate}{\heap}{\perms}{\env}$ and $\simstate{V}{\sstate_E}{\heap}{\perms_E}{\env}$.

  Then, if $\sconsume{\sstate}{\sstate_E}{\phi}{\sstate'}{\scheck}{\sperms}$ with corresponding valuation $V'$, $\rtassert{V'}{\heap}{\perms_E}{\scheck}$, and $V'(\pc(\sstate')) = \ktrue$, then
  \begin{equation}\label{eq:consume-prec-assert-e}
    \assertion{\heap}{\perms_E}{\env}{\phi}, \quad
    \simstate{V'}{\sstate'}{\heap}{\perms \setminus \vfoot{V}{\heap}{\sperms}}{\env}, \quad\text{and}~
    \ifrm{\heap}{\perms_E}{\env}{\phi}.
  \end{equation}

  Furthermore, if the above conditions hold and $\scheck \cap \SPerm = \emptyset$, then
  \begin{equation}\label{eq:consume-prec-assert}
    \assertion{\heap}{\perms}{\env}{\phi}.
  \end{equation}
\end{lemma}
\begin{proof}
  Suppose that $\sconsume{\sstate}{\sstate_E}{\phi}{\sstate'}{\scheck}{\sperms}$ with corresponding valuation $V'$, $\rtassert{V'}{\heap}{\perms_E}{\scheck}$, and $V'(\pc(\sstate')) = \ktrue$. Complete the proof by induction on $\sconsume{\sstate}{\sstate_E}{\phi}{\scheck}{\sperms}$:

  \begin{enumcases}
    \case \refrule{SConsumeImprecision} -- $\sconsume{\sstate}{\sstate_E}{\simprecise{\phi}}{\quintuple{\top}{\pc(\sstate')}{\senv(\sstate')}{\emptyset}{\emptyset}}{\scheck}{\sperms}$:
      Since $\simprecise{\phi}$ is imprecise, this rule cannot apply.

    \case \refrule{SConsumeValue} -- $\sconsume{\sstate}{\sstate_E}{e}{\sstate}{\scheck}{\emptyset}$:

      By \refrule{SConsumeValue} $\spceval{\sstate_E}{e}{t}{\scheck}$. Let $V'$ be the corresponding valuation, with initial valuation $V$. Then $V'$ is the corresponding valuation for this case.

      Thus by lemma \ref{lem:pc-eval-soundness} $\eval{\heap}{\env}{e}{V'(t)}$. Also by \refrule{SConsumeValue} $\pc(\sstate) \implies t$. Therefore $V'(t) = \ktrue$, and therefore $\eval{\heap}{\env}{e}{\ktrue}$. Thus $\assertion{\heap}{\perms_E}{\env}{e}$ by \refrule{AssertValue}.

      $\vfoot{V}{\heap}{\emptyset} = \emptyset$, therefore $\simstate{V'}{\sstate}{\heap}{\perms \setminus \vfoot{V}{\heap}{\emptyset}}{\env}$ since $V'$ extends $V$.

      Finally, $\frm{\heap}{\perms_E}{\env}{e}$ by lemma \ref{lem:pc-eval-soundness}. Therefore $\ifrm{\heap}{\perms_E}{\env}{e}$ by \refrule{IFrameExpression}, which completes the proof of \eqref{eq:consume-prec-assert-e}.

      As shown before, $\eval{\heap}{\env}{e}{\ktrue}$, thus $\assertion{\heap}{\perms}{\env}{e}$ by \refrule{AssertValue}, which proves \eqref{eq:consume-prec-assert}.

    \case \refrule{SConsumeValueImprecise} -- $\sconsume{\sstate}{\sstate_E}{e}{\sstate[\pc = \pc(\sstate) \kand t]}{\scheck; t}{\emptyset}$:

      By \refrule{SConsumeValueImprecise} $\spceval{\sstate_E}{e}{t}{\scheck}$. Let $V'$ be the valuation corresponding to this derivation, with initial valuation $V$. Then $V'$ is the valuation corresponding to this case.

      Let $\sstate' = \sstate[\pc = \pc(\sstate) \kand t]$.

      By assumptions, $\rtassert{V'}{\heap}{\perms_E}{\set{t}}$ by lemma \ref{lem:scheck-monotonicity}. Thus $V'(t) = \ktrue$ by \refrule{CheckValue}. Also, $\eval{\heap}{\env}{e}{V'(t)}$ by lemma \ref{lem:pc-eval-soundness}; therefore $\eval{\heap}{\env}{e}{\ktrue}$. Thus, by \refrule{AssertValue}, $\assertion{\heap}{\perms_E}{\env}{e}$.

      $\vfoot{V}{\heap}{\emptyset} = \emptyset$, therefore $\simstate{V'}{\sstate}{\heap}{\perms \setminus \vfoot{V}{\heap}{\emptyset}}{\env}$ since $V'$ extends $V$.

      Furthermore, since $V'(t) = \ktrue$, $V'(\pc(\sstate')) = V(\pc(\sstate)) \kand V'(t) = \ktrue$. Since $\sstate'$ and $\sstate$ differ only in their $\pc$ components, $\simstate{V'}{\sstate'}{\heap}{\perms \setminus \setminus \vfoot{V}{\heap}{\emptyset}}{\env}$.

      Finally, $\frm{\heap}{\perms_E}{\env}{e}$ by lemma \ref{lem:pc-eval-soundness}. Therefore $\ifrm{\heap}{\perms_E}{\env}{e}$ by \refrule{IFrameExpression}.

    \case \refrule{SConsumeValueFailure} -- $\sconsume{\sstate}{\sstate_E}{e}{\sstate}{\set{\bot}}{\emptyset}$:

      $\rtassert{V'}{\heap}{\perms_E}{\scheck \cup \set{\bot}}$ is a contradiction, thus the lemma vacuously holds.

    \case\label{case:consume-prec-pred} \refrule{SConsumePredicate} -- $\sconsume{\sstate}{\sstate_E}{p(\multiple{e})}{\sstate'}{\bigcup \multiple{\scheck}}{\set{\pair{p}{\multiple{t}}}}$:

      By \refrule{SConsumePredicate}, for each $e$, $\spceval{\sstate_E}{e}{t}{\scheck}$ for some $t$ and $\scheck$. Let $V'$ be corresponding valuation for this case, thus $V'$ extends the respective individual corresponding valuations and for each $\scheck$, $\rtassert{V'}{\heap}{\perms}{\scheck}$ by lemma \ref{lem:scheck-monotonicity}

      Therefore, for each $e$ and corresponding $t$, $\eval{\heap}{\env}{e}{V'(t)}$ by lemma \ref{lem:pc-eval-soundness}. By \refrule{SConsumePredicate}, for each $t$, $\pc(\sstate) \implies t \keq t'$ for some $t'$, thus $V'(t) = V'(t')$ since $V'(\pc(\sstate)) = \ktrue$. Therefore $\eval{\heap}{\env}{e_i}{V'(t_i')}$.

      By \refrule{SConsumePredicate} $\pair{p}{\multiple{t'}} \in \sheap(\sstate)$. Since $\simheap{V}{\sstate}{\heap}{\perms}$, $\assertion{\heap}{\perms}{[\multiple{x \mapsto V'(t')}]}{\fpred(p)}$ where $\multiple{x} = \fpredparams(p)$. Thus by \refrule{AssertPredicate} $\assertion{\heap}{\perms}{\env}{p(\multiple{e})}$, which proves \eqref{eq:consume-prec-assert}. Also, $\assertion{\heap}{\perms_E}{\env}{p(\multiple{e})}$ by lemma \ref{lem:assert-monotonicity}.

      Let $\perms' = \vfoot{V}{\heap}{\pair{p}{\multiple{t}}}$, therefore $\perms' = \efoot{\heap}{[\multiple{x \mapsto V'(t)}]}{\fpred(p)} = \efoot{\heap}{[\multiple{x \mapsto V'(t')}]}{\fpred(p)}$.

      By \refrule{SConsumePredicate} $\sstate' = \sstate[\sheap = \sheap', \oheap = \emptyset]$ where $\sheap = \sheap'; \pair{p}{\multiple{t}}$. Thus $\simheap{V'}{\sheap(\sstate')}{\heap}{\perms}$ since $\sheap(\sstate') \subset \sheap(\sstate)$.

      Then $\vfoot{V'}{\heap}{h} \cap \perms' = \emptyset$ for all $h \in \sheap(\sstate')$ since $\simheap{V'}{\sheap(\sstate)}{\heap}{\perms}$, $\pair{p}{\multiple{t}} \in \sheap(\sstate)$, and $\pair{p}{\multiple{t}} \notin \sheap(\sstate')$. Therefore, by lemma \ref{lem:disjoint-sim-heap-subset}, $\simheap{V'}{\sheap(\sstate')}{\heap}{\perms \setminus \perms'}$.

      Also, since $\oheap(\sstate') = \emptyset$, $\simheap{V'}{\oheap(\sstate')}{\heap}{\perms \setminus \perms'}$.

      Therefore $\simstate{V'}{\sstate'}{\heap}{\perms \setminus \perms'}{\env}$ since $\sstate'$ and $\sstate$ differ only in their $\sheap$ and $\oheap$ components.

      By lemmas \ref{lem:scheck-monotonicity} $\rtassert{V'}{\heap}{\perms_E}{\scheck}$ for each $\scheck$, thus $\frm{\heap}{\perms_E}{\env}{e}$ for each $e$ by lemma \ref{lem:pc-eval-soundness}, thus $\ifrm{\heap}{\perms_E}{\env}{p(\multiple{e})}$ by \refrule{IFramePredicate}. Thus \eqref{eq:consume-prec-assert-e} holds.

    \case \refrule{SConsumePredicateImprecise} -- $\sconsume{\sstate}{\sstate_E}{p(\multiple{e})}{\sstate'}{\bigcup \multiple{\scheck}; \pair{p}{\multiple{t}}}{\set{\pair{p}{\multiple{t}}}}$:

      By \refrule{SConsumePredicateImprecise}, for each $e$, $\spceval{\sstate_E}{e}{t}{\scheck}$ for some $t$ and $\scheck$. Let $V'$ be corresponding valuation for this case, thus $V'$ extends the respective individual corresponding valuations and for each $\scheck$, $\rtassert{V'}{\heap}{\perms}{\scheck}$ by lemma \ref{lem:scheck-monotonicity}

      Therefore, for each $e$ and corresponding $t$, $\eval{\heap}{\env}{e}{V'(t)}$ by lemma \ref{lem:pc-eval-soundness}. By \\
      \refrule{SConsumePredicateImprecise}, for each $t$, $\pc(\sstate) \implies t \keq t'$ for some $t'$, thus $V'(t) = V'(t')$ since $V'(\pc(\sstate)) = \ktrue$. Therefore $\eval{\heap}{\env}{e_i}{V'(t_i')}$.

      Also, $\rtassert{V'}{\heap}{\perms_E}{\set{\pair{p}{\multiple{t}}}}$ by assumptions and lemma \ref{lem:scheck-monotonicity}. Thus $\assertion{\heap}{\perms_E}{[\multiple{x \mapsto V'(t)}]}{\fpred(p)}$ by \refrule{CheckPred}. Therefore $\assertion{\heap}{\perms_E}{\env}{p(\multiple{e})}$ by \refrule{AssertPredicate}.

      By \refrule{SConsumePredicateImprecise} $\sstate' = \sstate[\sheap = \emptyset, \oheap = \emptyset]$, thus $\simheap{V'}{\sheap(\sstate')}{\heap}{\perms \setminus \vfoot{V'}{\heap}{\set{\pair{p}{\multiple{t}}}}}$ and $\simheap{V'}{\oheap(\sstate')}{\heap}{\perms \setminus \vfoot{V'}{\heap}{\set{\pair{p}{\multiple{t}}}}}$. Therefore $\simstate{V'}{\sstate'}{\heap}{\perms \setminus \vfoot{V}{\heap}{\set{\pair{p}{\multiple{t}}}}}{\env}$ since $\sstate'$ and $\sstate$ differ only in their $\sheap$ and $\oheap$ components.

      For each $e$, $\frm{\heap}{\perms_E}{\env}{e}$ by lemma \ref{lem:pc-eval-soundness}. Therefore $\ifrm{\heap}{\perms_E}{\env}{p(\multiple{e})}$ by \refrule{IFramePredicate}, which completes the proof of \eqref{eq:consume-prec-assert-e}.

      Also, $\pair{p}{\multiple{t}} \in \SPerm$ is in the resulting set of checks, which contradicts the premises of \eqref{eq:consume-prec-assert}. Therefore it is vacuously true.

    \case \refrule{SConsumePredicateFailure} -- $\sconsume{\sstate}{\sstate_E}{p(\multiple{e})}{\sstate}{\set{\bot}}{\set{\pair{p}{\multiple{t}}}}$:

      $\rtassert{V'}{\heap}{\perms_E}{\set{\bot}}$ is a contradiction, thus the lemma vacuously holds.

    \case\label{case:consume-prec-acc} \refrule{SConsumeAcc} -- $\sconsume{\sstate}{\sstate_E}{\kacc(e.f)}{\sstate[\sheap = \sheap', \oheap = \oheap']}{\scheck}{\set{\pair{t_e}{f}}}$:

      By \refrule{SConsumeAcc} $\spceval{\sstate_E}{e}{t_e}{\scheck}$. Let $V'$ be the corresponding valuation, thus $V'$ is the corresponding valuation for this case.
      
      Thus $\eval{\heap}{\env}{e}{V'(t_e)}$ by lemma \ref{lem:pc-eval-soundness}. Also $\pc(\sstate) \implies t_e' \keq t_e$ by \refrule{SConsumeAcc}, thus $V'(t_e') = V'(t_e)$ since $V'(\pc(\sstate)) = \ktrue$. Therefore, $\eval{\heap}{\env}{e}{V'(t_e')}$.

      Since $\triple{f}{t_e'}{t} \in \sheap(\sstate)$ by \refrule{SConsumeAcc} and $\simheap{V}{\sheap(\sstate)}{\heap}{\perms}$, $\pair{V'(t_e')}{f} \in \perms$. Thus $\assertion{\heap}{\perms}{\env}{\kacc(e.f)}$ by \refrule{AssertAcc}, which proves \eqref{eq:consume-prec-assert}. Therefore $\assertion{\heap}{\perms_E}{\env}{\kacc(e.f)}$ by lemma \ref{lem:assert-monotonicity} since $\perms \subseteq \perms_E$.

      Let $\sheap' = \fremfp(\sheap(\sstate), \sstate, t_e, f)$. Therefore $\simheap{V'}{\sheap'}{\heap}{\perms \setminus \set{\pair{V'(t_e)}{f}}}$ by lemma \ref{lem:sheap-remfp-prec}. Also, $\set{\pair{V'(t_e')}{f}} = \set{\pair{V'(t_e)}{f}} = \vfoot{V'}{\heap}{\set{\pair{t_e}{f}}}$. Thus $\simheap{V'}{\sheap'}{\heap}{\perms \setminus \vfoot{V'}{\heap}{\set{\pair{t_e}{f}}}}$.

      Likewise, let $\oheap' = \fremf(\oheap(\sstate), \sstate, t_e, f)$. Then similarly $\simheap{V'}{\oheap'}{\heap}{\perms \setminus \vfoot{V}{\heap}{\set{\pair{t_e}{f}}}}$ by lemma \ref{lem:oheap-remf}.

      By \refrule{SConsumeAcc}, $\sstate' = \sstate[\sheap = \sheap', \oheap = \oheap']$. Now, since $\sstate$ and $\sstate'$ differ only in their $\sheap$ and $\oheap$ components, using the properties of $\sheap'$ and $\oheap'$ shown above, $\simstate{V'}{\sstate'}{\heap}{\perms \setminus \vfoot{V}{\heap}{\set{\pair{t_e}{f}}}}{\env}$.

      Finally, $\frm{\heap}{\perms_E}{\env}{e}$ by lemma \ref{lem:pc-eval-soundness}. Therefore $\ifrm{\heap}{\perms_E}{\env}{\kacc(e.f)}$ by \refrule{IFrameAcc}, which completes the proof of \eqref{eq:consume-prec-assert-e}.

    \case\label{case:consume-prec-acc-optimistic} \refrule{SConsumeAccOptimistic} --  $\sconsume{\sstate}{\sstate_E}{\kacc(e.f)}{\sstate'}{\scheck}{\set{\pair{t_e}{f}}}$:
      Similar to case \ref{case:consume-prec-acc}, except to show that $\simheap{V}{\sheap'}{\heap}{\perms \setminus \set{\pair{V'(t_e')}{f}}}$.

      Since $\triple{f}{t_e'}{t} \in \oheap(\sstate)$, $\oheap(\sstate) \ne \emptyset$. Therefore, since $\sstate$ is well-formed, $\imp(\sstate) = \top$. Let $\sheap' = \fremf(\sheap(\sstate), \sstate, t_e', f)$. Therefore $\simheap{V}{\sheap'}{\heap}{\perms \setminus \set{\pair{V'(t_e')}{f}}}$ by lemma \ref{lem:sheap-remf-imp}.

      Continue as in case \ref{case:consume-prec-acc}.

    \case \refrule{SConsumeAccImprecise} -- $\sconsume{\sstate}{\sstate_E}{\kacc(e.f)}{\sstate'}{\scheck; \pair{t_e}{f}}{\set{\pair{t_e}{f}}}$:

      By \refrule{SConsumeAcc} $\spceval{\sstate_E}{e}{t_e}{\scheck}$. Let $V'$ be the corresponding valuation, thus $V'$ is the corresponding valuation for this case.

      Then $\rtassert{V'}{\heap}{\perms_E}{\scheck}$ by assumptions and lemma \ref{lem:scheck-monotonicity}. Thus $\eval{\heap}{\env}{e}{V'(t_e)}$ by lemma \ref{lem:pc-eval-soundness}.

      Also, $\rtassert{V'}{\heap}{\perms_E}{\pair{t_e}{f}}$ by assumptions and lemma \ref{lem:scheck-monotonicity}. Then $\pair{V'(t_e)}{f} \in \perms_E$ by \refrule{CheckAcc}. Therefore $\assertion{\heap}{\perms_E}{\env}{\kacc(e.f)}$ by \refrule{AssertAcc}.

      Let $\sheap' = \fremf(\sheap(\sstate), \sstate, t_e, f)$. By \refrule{SConsumeAccImprecise} $\imp(\sstate)$. Thus $\simheap{V'}{\sheap'}{\heap}{\perms \setminus \set{\pair{V'(t_e)}{f}}}$ by lemma \ref{lem:sheap-remf-imp}. Also, $\set{\pair{V'(t_e)}{f}} = \vfoot{V'}{\heap}{\set{\pair{t_e}{f}}}$. Thus $\simheap{V'}{\sheap'}{\heap}{\perms \setminus \vfoot{V'}{\heap}{\set{\pair{t_e}{f}}}}$.

      Likewise, let $\oheap' = \fremf(\oheap(\sstate), \sstate, t_e, f)$. Then $\simheap{V'}{\oheap'}{\heap}{\perms \setminus \vfoot{V}{\heap}{\pair{t_e}{f}}}$ by lemma \ref{lem:oheap-remf}.

      By \refrule{SConsumeAcc}, $\sstate' = \sstate[\sheap = \sheap', \oheap = \oheap']$. Now, since $\sstate$ and $\sstate'$ differ only in their $\sheap$ and $\oheap$ components, using the properties of $\sheap'$ and $\oheap'$ shown above, $\simstate{V'}{\sstate'}{\heap}{\perms \setminus \vfoot{V}{\heap}{\set{\pair{t_e}{f}}}}{\env}$.

      By lemma \ref{lem:pc-eval-soundness} $\frm{\heap}{\perms_E}{\env}{e}$, therefore $\ifrm{\heap}{\perms_E}{\kacc(e.f)}$ by \refrule{IFrameAcc}. Thus \eqref{eq:consume-prec-assert-e} holds.

      Also, $\pair{t_e}{f} \in \SPerm$ is in the resulting set of checks, which contradicts the premises of \eqref{eq:consume-prec-assert}, therefore it vacuously holds.

    \case \refrule{SConsumeAccFailure} -- $\sconsume{\sstate}{\sstate_E}{\kacc(e.f)}{\sstate}{\set{\bot}}{\set{\pair{t_e}{f}}}$:

    $\rtassert{V'}{\heap}{\perms_E}{\set{\bot}}$ is a contradiction, thus the lemma vacuously holds.

    \case\label{case:consume-prec-conj} \refrule{SConsumeConjunction} -- $\sconsume{\sstate}{\sstate_E}{\phi_1 * \phi_2}{\sstate''}{\scheck_1 \cup \scheck_2}{\sperms_1 \cup \sperms_2}$:

      By \refrule{SConsumeConjunction} $\sconsume{\sstate}{\sstate_E}{\phi_1}{\sstate'}{\scheck_1}{\sperms_1}$ and $\sconsume{\sstate'}{\sstate_E[\pc = \pc(\sstate')]}{\phi_2}{\sstate''}{\scheck_2}{\sperms_2}$. Let $V_1$ and $V'$ be the respective corresponding valuations, with initial valuations $V$ and $V_1$, respectively. Then $V'$ is the corresponding valuation for this case.

      By lemma \ref{lem:consume-subpath}, $\pc(\sstate'') \implies \pc(\sstate')$. Thus $V'(\pc(\sstate')) = V_1(\pc(\sstate')) = \ktrue$. Also, $\rtassert{V_1}{\heap}{\perms_E}{\scheck_1}$ by lemma \ref{lem:scheck-monotonicity}, since $V'$ extends $V_1$.

      By \refrule{SConsumeConjunction} $(\scheck_1 \cup \scheck_2) \cap \SPerm = \emptyset$, thus $\scheck_1 \cap \SPerm = \emptyset$.

      Let $\perms_1 = \vfoot{V}{\heap}{\sperms_1}$. By induction, using \eqref{eq:consume-prec-assert}, $\assertion{\heap}{\perms}{\env}{\phi_1}$. Thus $\assertion{\heap}{\perms_1}{\env}{\phi_1}$ and $\perms_1 \subseteq \perms$ by lemma \ref{lem:consume-assert-vfoot}.
      
      Also $\simstate{V_1}{\sstate'}{\heap}{\perms \setminus \perms_1}{\env}$ by induction, $\simstate{V_1}{\sstate_E[\pc = \pc(\sstate')]}{\heap}{\perms_E}{\env}$ since $V_1$ extends $V$ and $V_1(\pc(\sstate')) = \ktrue$, and $\rtassert{V}{\heap}{\perms_E}{\scheck_2}$ by lemma \ref{lem:scheck-monotonicity}. Finally, by assumptions $V'(\pc(\sstate'')) = \ktrue$, and $(\perms \setminus \perms_1) \subseteq \perms \subseteq \perms_E$. Thus by induction $\simstate{V'}{\sstate''}{\heap}{(\perms \setminus \perms_1) \setminus \vfoot{V}{\heap}{\sperms_2}}{\env}$.

      Also by induction, using \eqref{eq:consume-prec-assert}, $\assertion{\heap}{\perms \setminus \perms_1}{\env}{\phi_2}$.

      Now $(\perms \setminus \perms_1) \subseteq \perms$, $\perms_1 \subseteq \perms$, and $(\perms \setminus \perms_1) \cap \perms_1 = \emptyset$. Therefore $\assertion{\heap}{\perms}{\env}{\phi_1 * \phi_2}$ by \refrule{AssertConjunction}, which proves \eqref{eq:consume-prec-assert}. Then by lemma \ref{lem:assert-monotonicity} $\assertion{\heap}{\perms_E}{\env}{\phi_1 * \phi_2}$.

      As shown before, $\simstate{V'}{\sstate''}{\heap}{(\perms \setminus \perms_1) \setminus \vfoot{V}{\heap}{\sperms_2}}{\env}$, and $(\perms \setminus \perms_1) \setminus \vfoot{V}{\heap}{\sperms_2} = \perms \setminus (\vfoot{V}{\heap}{\sperms_1} \cup \vfoot{V}{\heap}{\sperms_2}) = \perms \setminus \vfoot{V}{\heap}{\sperms_1 \cup \sperms_2}$, therefore $\simstate{V'}{\sstate''}{\heap}{\perms \setminus \vfoot{V}{\heap}{\sperms_1 \cup \sperms_2}}{\env}$.

      By induction $\ifrm{\heap}{\perms_E}{\env}{\phi_1}$ and $\ifrm{\heap}{\perms_E}{\env}{\phi_2}$. Therefore $\ifrm{\heap}{\perms_E}{\env}{\phi_1 * \phi_2}$ by \refrule{IFrameConjunction}, which completes the proof of \eqref{eq:consume-prec-assert-e}.

    \case \refrule{SConsumeConjunctionImprecise} -- $\sconsume{\sstate}{\sstate_E}{\phi_1 * \phi_2}{\sstate''}{\scheck_1 \cup \scheck_2; \fsep(\sperms_1, \sperms_2)}{\sperms_1 \cup \sperms_2}$:
    
      Similar to case \ref{case:consume-prec-conj}, except when showing that $\assertion{\heap}{\perms_E}{\env}{\phi_1 * \phi_2}$ and when proving \eqref{eq:consume-prec-assert}:

      By induction $\assertion{\heap}{\perms_E}{\env}{\phi_1}$ and $\assertion{\heap}{\perms_E}{\env}{\phi_2}$. Thus by lemma \ref{lem:consume-assert-vfoot} $\assertion{\heap}{\vfoot{V'}{\heap}{\sperms_1}}{\env}{\phi_1}$, $\assertion{\heap}{\vfoot{V'}{\heap}{\sperms_2}}{\env}{\phi_2}$, and $\vfoot{V'}{\heap}{\sperms_1} \cup \vfoot{V'}{\heap}{\sperms_2} \subseteq \perms_E$.
      
      By assumptions $\rtassert{V'}{\heap}{\perms_E}{\fsep(\sperms_1, \sperms_2)}$. Then by \refrule{CheckSep} $\vfoot{V'}{\heap}{\sperms_1} \cap \vfoot{V'}{\heap}{\sperms_2} = \emptyset$. Therefore $\assertion{\heap}{\perms_E}{\env}{\phi_1 * \phi_2}$ by \refrule{AssertConjunction}.

      By \refrule{SConsumeConjunctionImprecise} $(\scheck_1 \cup \scheck_2) \cap \SPerm \ne \emptyset$. Therefore the premises of \eqref{eq:consume-prec-assert} do not hold, therefore it is vacuously true.

    \case\label{case:consume-prec-cond-a} \refrule{SConsumeConditionalA} -- $\sconsume{\sstate}{\sstate_E}{\sif{e}{\phi_1}{\phi_2}}{\sstate'}{\scheck \cup \scheck'}{\sperms}$:

      By \refrule{SConsumeConditionalA}, $\spceval{\sstate_E}{e}{t}{\scheck}$ and $\sconsume{\sstate[\pc = \pc']}{\sstate_E[\pc = \pc']}{\phi_1}{\sstate'}{\scheck'}{\sperms}$ where $\pc' = \pc(\sstate) \kand t$. Let $V_1$ and $V'$ be the respective corresponding valuations, with initial valuations $V$ and $V_1$, respectively. Then $V'$ is the corresponding valuation for this case.

      By lemma \ref{lem:consume-subpath} $\pc(\sstate') \implies \pc' = \pc(\sstate) \kand t$, thus $V'(\pc(\sstate) \kand t) = V_1(\pc(\sstate) \kand t) = \ktrue$. Therefore $\simstate{V_1}{\sstate[\pc = \pc(\sstate) \kand t]}{\heap}{\perms}{\env}$ and $\simstate{V_1}{\sstate_E[\pc = \pc(\sstate)]}{\kand t}{\heap}{\perms}{\env}$.

      By assumptions and lemma \ref{lem:scheck-monotonicity} $\rtassert{V_1}{\heap}{\perms_E}{\scheck}$. Thus $\eval{\heap}{\env}{e}{V_1(t)}$ by lemma \ref{lem:pc-eval-soundness}. Furthermore, $V_1(t) = \ktrue$ since $\pc(\sstate) \kand t \implies t$, thus $\eval{\heap}{\env}{e}{\ktrue}$. Finally, $\assertion{\heap}{\perms_E}{\env}{\phi_1}$ by induction. Therefore $\assertion{\heap}{\perms_E}{\env}{\sif{e}{\phi_1}{\phi_2}}$ by \refrule{AssertIfA}.

      Also by induction $\simstate{V'}{\sstate'}{\heap}{\perms \setminus \vfoot{V'}{\heap}{\sperms}}{\env}$.

      Finally, $\frm{\heap}{\perms_E}{\env}{\phi_1}$ by induction, and $\frm{\heap}{\perms_E}{\env}{e}$ by lemmas \ref{lem:pc-eval-soundness}. As shown before, $\eval{\heap}{\env}{e}{\ktrue}$. Therefore $\frm{\heap}{\perms_E}{\env}{\sif{e}{\phi_1}{\phi_2}}$ by \refrule{FrameIfA}, which completes the proof of \eqref{eq:consume-prec-assert-e}.

      Now suppose that $(\scheck \cup \scheck') \cap \SPerm = \emptyset$, thus $\scheck' \cap \SPerm = \emptyset$. Then by induction $\assertion{\heap}{\perms}{\env}{\phi}$, and as before $\eval{\heap}{\env}{e}{\ktrue}$, therefore $\assertion{\heap}{\perms}{\env}{\sif{e}{\phi_1}{\phi_2}}$ by \refrule{AssertIfA}, which completes the proof of \eqref{eq:consume-prec-assert}.

    \case \refrule{SConsumeConditionalB} -- $\sconsume{\sstate}{\sstate_E}{\sif{e}{\phi_1}{\phi_2}}{\sstate'}{\scheck \cup \scheck'}{\sperms}$: Similar to case \ref{case:consume-prec-cond-a}.

  \end{enumcases}
\end{proof}

\begin{lemma}[Soundness of consume (long form)]\label{lem:consume-soundness}
  Let $\gform$ be some specification, $\sstate$ and $\sstate_E$ some well-formed symbolic states such that $\pc(\sstate) \implies \pc(\sstate_E)$, and $\triple{\heap}{\perms}{\env}$ some evaluation state such that $\simstate{V}{\sstate}{\heap}{\perms}{\env}$ and $\simstate{V}{\sstate_E}{\heap}{\perms}{\env}$.
  
  If $\sconsume{\sstate}{\sstate_E}{\gform}{\sstate'}{\scheck}{\sperms}$ with corresponding valuation $V'$, $\rtassert{V'}{\heap}{\perms_E}{\scheck}$, and $V'(\pc(\sstate')) = \ktrue$, then
  $$\assertion{\heap}{\perms_E}{\env}{\gform} \quad\text{and}\quad
    \simstate{V'}{\sstate'}{\heap}{\perms \setminus \efoot{\heap}{\env}{\gform}}{\env}.
  $$
\end{lemma}
\begin{proof}
  Suppose that $\sconsume{\sstate}{\sstate_E}{\gform}{\sstate'}{\scheck}{\sperms}$ with corresponding valuation $V'$, $\rtassert{V'}{\heap}{\perms_E}{\scheck}$, and $V'(\pc(\sstate')) = \ktrue$. Then one of the following cases applies

  \begin{enumcases}
    \case $\gform$ is imprecise, i.e. $\gform = \simprecise{\phi}$ for some $\phi \in \Formula$:

      Then, since $\sconsume{\sstate}{\sstate_E}{\simprecise{\phi}}{\sstate'}{\scheck}{\sperms}$, by \refrule{SConsumeImprecision} $\sconsume{\sstate}{\sstate_E[\imp = \top]}{\phi}{\sstate_0}{\scheck}{\sperms}$ for some $\sstate_0$ where $\sstate' = \quintuple{\top}{\pc(\sstate_0)}{\senv(\sstate_0)}{\emptyset}{\emptyset}$. Let $V'$ be the corresponding valuation, therefore $V'$ is the corresponding valuation for the original derivation.

      Thus by lemma \ref{lem:consume-soundness-precise}, $\assertion{\heap}{\perms_E}{\env}{\phi}$, $\simstate{V'}{\sstate_0}{\heap}{\perms \setminus \vfoot{V'}{\heap}{\sperms}}{\env}$, and $\ifrm{\heap}{\perms_E}{\env}{\phi}$.

      Then $\efrm{\heap}{\perms_E}{\env}{\phi}$ by lemma \ref{lem:ifrm-implies-efrm}. Therefore $\assertion{\heap}{\perms_E}{\env}{\phi}$.

      Also, $\sheap(\sstate') = \oheap(\sstate') = \emptyset$, thus $\simheap{V'}{\sheap(\sstate')}{\heap}{\emptyset}$ and $\simheap{V'}{\oheap(\sstate')}{\heap}{\emptyset}$. Since $\simstate{V'}{\sstate_0}{\heap}{\perms}{\env}$ and $\senv(\sstate') = \senv(\sstate_0)$ and $\pc(\sstate') = \pc(\sstate_0)$, $\simenv{V'}{\senv(\sstate')}{\env}$ and $V'(\pc(\sstate')) = \ktrue$. Therefore $\simstate{V'}{\sstate'}{\heap}{\emptyset}{\env}$, and thus $\simstate{V'}{\sstate'}{\heap}{\perms \setminus \efoot{\heap}{\env}{\simprecise{\phi}}}{\env}$ by lemma \ref{lem:assert-monotonicity}.

    \case $\gform$ is precise, i.e. $\gform = \phi$:

      Then by lemma \ref{lem:consume-soundness-precise} $\simstate{V'}{\sstate'}{\heap}{\perms \setminus \vfoot{V'}{\heap}{\sperms}}{\env}$ and $\assertion{\heap}{\perms_E}{\env}{\phi}$.

      By lemma \ref{lem:consume-assert-vfoot} $\assertion{\heap}{\vfoot{V'}{\heap}{\sperms}}{\env}{\phi}$. Then by lemma \ref{lem:efoot-subset-spec} $\efoot{\heap}{\perms}{\phi} \subseteq \vfoot{V'}{\heap}{\sperms}$, since $\phi$ is a specification. Therefore $\perms \setminus \vfoot{V'}{\heap}{\sperms} \subseteq \perms \setminus \efoot{\heap}{\perms}{\phi}$, thus $\simstate{V'}{\sstate'}{\heap}{\perms \setminus \efoot{\heap}{\perms}{\phi}}{\env}$.
  \end{enumcases}
\end{proof}

\begin{lemma}[Soundness of consume (short form)]\label{lem:cons-soundness}
  Let $\gform$ be some specification, $\sstate$ be some well-formed symbolic state, $\triple{\heap}{\perms}{\env}$ some evaluation state, and $V$ be some valuation such that \\
  $\simstate{V}{\sstate}{\heap}{\perms}{\env}$.

  If $\scons{\sstate}{\gform}{\sstate'}{\scheck}$ with corresponding valuation $V'$, $\rtassert{V'}{\heap}{\perms}{\scheck}$, and $V'(\pc(\sstate')) = \ktrue$ then
  $$
    \assertion{\heap}{\perms}{\env}{\gform} \quad\text{and}\quad
    \simstate{V'}{\sstate'}{\heap}{\perms \setminus \efoot{\heap}{\env}{\gform}}{\env}.
  $$
\end{lemma}
\begin{proof}
  Suppose $\scons{\sstate}{\gform}{\sstate'}{\scheck}$ with corresponding valuation $V'$, $\rtassert{V'}{\heap}{\perms}{\scheck}$, and $V'(\pc(\sstate')) = \ktrue$. Then $\sconsume{\sstate}{\sstate}{\gform}{\sstate'}{\scheck}{\_}$ by \refrule{SConsume}. Let $V'$ be the corresponding valuation, thus $V'$ is the corresponding valuation for the original derivation.

  Trivially $\pc(\sstate) \implies \pc(\sstate)$, thus the conditions of lemma \ref{lem:consume-soundness} are satisfied, and thus $\assertion{\heap}{\perms}{\env}{\gform}$ and $\simstate{V'}{\sstate'}{\heap}{\perms \setminus \efoot{\heap}{\env}{\gform}}{\env}$.
\end{proof}

\begin{lemma}[Progress of consume (long form)]\label{lem:consume-progress}
  For any heap $\heap$, $\sstate$, $\sstate_E$, $\gform$, and valuation $V$, if $V(\pc(\sstate_E)) = \ktrue$ then $\sconsume{\sstate}{\sstate_E}{\gform}{\sstate'}{\_}{\_}$ for some $\sstate'$ such that $V'(\pc(\sstate')) = \ktrue$ where $V'$ is the corresponding valuation.
\end{lemma}
\begin{proof}
  By induction on the syntax forms of $\phi$:

  \begin{enumcases}
    \case $e \in \Expr$:
      By lemma \ref{lem:pc-eval-progress}, $\spceval{\sstate_E}{e}{t}{\_}$ for some $t$. Then one of the following cases applies to yield $\sconsume{\sstate}{\sstate_E}{e}{\sstate}{\_}{\_}$:

      \subcase $\pc(\sstate) \implies t$: Then \refrule{SConsumeValue} applies.

      \subcase $\imp(\sstate)$ and $\pc(\sstate) \notimplies t$: Then \refrule{SConsumeValueImprecise} applies.

      \subcase $\neg \imp(\sstate)$ and $\pc(\sstate) \notimplies t$: Then \refrule{SConsumeValueFailure} applies.

    \case $p(\multiple{e})$ -- $p \in \Predicate, \multiple{e \in \Expr}$:

      By lemma \ref{lem:pc-eval-progress}, for each $e$, $\spceval{\sstate_E}{e}{t}{\_}$ for some $t$. Then one of the following cases applies to yield $\sconsume{\sstate}{\sstate_E}{p(\multiple{e})}{\sstate'}{\_}{\_}$ where $\pc(\sstate') = \pc(\sstate)$:

      \subcase $p(\multiple{t'}) \in \sheap(\sstate)$ and $\pc(\sstate) \implies \multiple{t \keq t'}$ for some $\multiple{t'}$: Then \refrule{SConsumePredicate} applies.

      \subcase $\imp(\sstate)$ and $\nexistential{\pair{p}{\multiple{t'}} \in \sheap(\sstate)}{\bigwedge \multiple{\pc(\sstate) \implies t \keq t'}}$: Then \refrule{SConsumePredicateImprecise} applies.

      \subcase $\neg\imp(\sstate)$ and $\nexistential{\pair{p}{\multiple{t'}} \in \sheap(\sstate)}{\bigwedge \multiple{\pc(\sstate) \implies t \keq t'}}$: Then \refrule{SConsumePredicateFailure} applies.

    \case $\kacc(e.f)$ -- $e \in \Expr, f \in \Field$:

      By lemma \ref{lem:pc-eval-progress}, $\spceval{\sstate_E}{e}{t_e}{\_}$ for some $t_e$. Note that $\fremf$ and $\fremfp$ are defined for all inputs. Then one of the following cases applies to yield $\sconsume{\sstate}{\sstate_E}{\kacc(e.f)}{\sstate'}{\_}{\_}$ where $\pc(\sstate') = \pc(\sstate')$:

      \subcase $\triple{f}{t_e'}{t} \in \sheap(\sstate)$ and $\pc(\sstate) \implies t_e' \keq t_e$ for some $t_e'$ and $t$: Then \refrule{SConsumeAcc} applies.

      \subcase $\nexistential{t_e', t}{\triple{f}{t_e}{t} \in \sheap(\sstate) \wedge (\pc(\sstate) \implies t_e' \keq t_e)}$ and $\triple{f}{t_e'}{t} \in \sheap(\sstate)$ for some $t_e'$ and $t$ where $\pc(\sstate) \implies t_e' \keq t_e$: Then \\
      \refrule{SConsumeAccOptimistic} applies.

      \subcase $\nexistential{t_e', t}{\triple{f}{t_e}{t} \in \sheap(\sstate) \cup \oheap(\sstate) \wedge (\pc(\sstate) \implies t_e' \keq t_e)}$ and $\imp(\sstate)$: Then \\
      \refrule{SConsumeAccImprecise} applies.

      \subcase $\nexistential{t_e', t}{\triple{f}{t_e}{t} \in \sheap(\sstate) \cup \oheap(\sstate) \wedge (\pc(\sstate) \implies t_e' \keq t_e)}$ and $\neg \imp(\sstate)$: Then \\
      \refrule{SConsumeAccFailure} applies.

    \case $\phi_1 * \phi_2$ -- $\phi_1, \phi_2 \in \Formula$

      By induction, $\sconsume{\sstate}{\sstate_E}{\phi_1}{\sstate'}{\_}{\_}$ for some $\sstate'$ such that $V'(\pc(\sstate')) = \ktrue$ where $V'$ is the corresponding valuation.

      Then also by induction, $\sconsume{\sstate'}{\sstate_E[\pc = \pc(\sstate')]}{\phi_2}{\sstate''}{\_}{\_}$ for some $\sstate''$ such that $V''(\pc(\sstate'')) = \ktrue$ where $V''$ is the corresponding valuation, with initial valuation $V'$. Then one of the following cases applies to yield $\sconsume{\sstate}{\sstate_E}{\phi_1 * \phi_2}{\sstate''}{\_}{\_}$:

      \subcase $(\scheck_1 \cup \scheck_2) \cap \SPerm \ne \emptyset$: Then \refrule{SConsumeConjunctionImprecise} applies.

      \subcase $(\scheck_1 \cup \scheck_2) \cap \SPerm = \emptyset$: Then \refrule{SConsumeConjunction} applies.

    \case $\sif{e}{\phi_1}{\phi_2}$ -- $e \in \Expr, \phi_1, \phi_2 \in \Formula$:

      By lemma \ref{lem:pc-eval-progress}, $\spceval{\sstate_E}{e}{t}{\_}$ for some $t$. Let $V'$ be the valuation corresponding to this derivation. Then since this is a well-typed program, one of the following cases must apply:

      \subcase $V'(t) = \ktrue$: Let $\pc' = \pc(\sstate) \kand t$. Then $V(\pc') = \ktrue$ and by induction, $\sconsume{\sstate[\pc = \pc']}{\sstate_E[\pc = \pc']}{\phi_1}{\sstate'}{\_}{\_}$ for some $\sstate'$ where $V''(\pc(\sstate')) = \ktrue$ for the corresponding derivation $V''$ with initial valuation $V'$. Then by \refrule{SConsumeConditionalA}, $\sconsume{\sstate}{\sstate_E}{\sif{e}{\phi_1}{\phi_2}}{\sstate'}{\_}{\_}$, and $V''$ is the corresponding valuation for this derivation with initial valuation $V$.

      \subcase $V'(t) = \kfalse$: Let $\pc' = \pc(\sstate) \kand \kneg t$. Then $V(\pc') = \ktrue$ and by induction, $\sconsume{\sstate[\pc = \pc']}{\sstate_E[\pc = \pc']}{\phi_2}{\sstate'}{\_}{\_}$ for some $\sstate'$ where $V''(\pc(\sstate')) = \ktrue$ for the corresponding derivation $V''$ with initial valuation $V'$. Then by \refrule{SConsumeConditionalB}, $\sconsume{\sstate}{\sstate_E}{\sif{e}{\phi_1}{\phi_2}}{\sstate'}{\_}{\_}$, and $V''$ is the corresponding valuation for this derivation with initial valuation $V$.

  \end{enumcases}
\end{proof}

\begin{lemma}[Progress of consume (short form)]\label{lem:cons-progress}
  For any heap $\heap$, $\sstate$, $\gform$, and valuation $V$, if $V(\pc(\sstate)) = \ktrue$ then $\scons{\sstate}{\gform}{\sstate'}{\_}$ for some $\sstate'$ such that $V'(\pc(\sstate')) = \ktrue$ where $V'$ is the corresponding valuation.
\end{lemma}
\begin{proof}
  By lemma \ref{lem:consume-progress}, $\sconsume{\sstate}{\sstate}{\gform}{\sstate'}{\_}{\_}$ for some $\sstate'$ where $V'(\pc(\sstate')) = \ktrue$. Then by \refrule{SConsume}, $\scons{\sstate}{\gform}{\sstate'}{\_}$, and $V'$ is the corresponding valuation for this derivation.
\end{proof}

\subsection{Progress}


\begin{definition}\label{def:sguard-valuation}
  For a derivations $\sguard{\vstate}{\sstate}{\scheck}{\sperms}$, given an initial valuation $V$ and heap $\heap$, the \textbf{corresponding valuation} is denoted as
    $$V[\sguard{\vstate}{\sstate}{\scheck}{\sperms} \mid \heap].$$
  This function is defined as follows, depending on the rule that proves the derivation. Values are referenced using the respective name from the rule definition.
  \begin{itemize}
    \item \refrule{SGuardInit}:
      $$V[\sguard{\initsym}{\quintuple{\bot}{\emptyset}{\emptyset}{\emptyset}{\ktrue}}{\emptyset}{\emptyset} \mid \heap] := V$$
    \item \refrule{SGuardSeq}:
      $$V[\sguard{\triple{\sstate}{\sseq{\kskip}{s}}{\gform}}{\sstate}{\emptyset}{\emptyset} \mid \heap] := V$$
    \item \refrule{SGuardAssign}:
      $$V[\sguard{\triple{\sstate}{\sseq{x = e}{s}}{\gform}}{\sstate}{\scheck}{\emptyset} \mid \heap] := V[\seval{\sstate}{e}{\_}{\sstate'}{\scheck} \mid \heap]$$
    \item \refrule{SGuardAssignField}:
      \begin{align*}
        &V[\sguard{\triple{\sstate}{\sseq{x.f = e}{s}}{\gform}}{\sstate''}{\scheck' \cup \scheck''}{\emptyset} \mid \heap] := \\
          &\quad V[\seval{\sstate}{e}{\_}{\sstate'}{\scheck'} \mid \heap]
          [\scons{\sstate'}{\kacc(x.f)}{\sstate''}{\scheck''} \mid \heap]
      \end{align*}
    \item \refrule{SGuardAlloc}:
      $$V[\sguard{\triple{\sstate}{\sseq{x = \salloc{S}}{s}}{\gform}}{\sstate}{\emptyset}{\emptyset} \mid \heap] := V$$
    \item \refrule{SGuardCall}:
      \begin{align*}
        &V[\sguard{\triple{\sstate}{\sseq{y \kassign m(\multiple{e})}{s}}{\gform}}{\sstate''[\senv = \senv(\sstate)]}{\multiple{\scheck} \cup \scheck'}{\frem(\sstate'', \fpre(m))} \mid \heap] := \\
          &\quad V\multiple{[\seval{\sstate}{e}{t}{\sstate'}{\scheck} \mid \heap]}[\scons{\sstate'[\senv = [\multiple{x \mapsto t}]]}{\fpre(m)}{\sstate''}{\scheck'} \mid \heap]
      \end{align*}
    \item \refrule{SGuardAssert}:
      $$V[\sguard{\triple{\sstate}{\sseq{\sassert{\phi}}{s}}{\gform}}{\sstate'}{\scheck}{\emptyset} \mid \heap] := V[\scons{\sstate}{\simprecise{\phi}}{\sstate'}{\scheck} \mid \heap]$$
    \item \refrule{SGuardFold}:
      \begin{align*}
        &V[\sguard{\triple{\sstate}{\sseq{\sfold{p(\multiple{e})}}{s}}{\gform}}{\sstate''[\senv = \senv(\sstate)]}{\_}{\emptyset} \mid \heap] := \\
        &\quad V\multiple{[\seval{\sstate}{e}{t}{\sstate'}{\scheck} \mid \heap]}[\scons{\sstate'[\senv = [\multiple{x \mapsto t}]]}{\fpred(p)}{\sstate''}{\scheck'} \mid \heap]
      \end{align*}
    \item \refrule{SGuardUnfold}:
      \begin{align*}
        &V[\sguard{\triple{\sstate}{\sseq{\sunfold{p(\multiple{e})}}{s}}{\gform}}{\sstate''}{\scheck' \cup \bigcup \multiple{\scheck}}{\emptyset} \mid \heap] := \\
        &\quad V\multiple{[\seval{\sstate}{e}{t}{\sstate'}{\scheck} \mid \heap]}[\scons{\sstate'}{p(\multiple{e})}{\sstate''}{\scheck'} \mid \heap]
      \end{align*}
    \item \refrule{SGuardIf}:
      \begin{align*}
        &V[\sguard{\triple{\sstate}{\sseq{\sif{e}{s_1}{s_2}}{s}}{\gform}}{\sstate'}{\scheck}{\emptyset} \mid \heap] := \\
        &\quad V[\seval{\sstate}{e}{\_}{\sstate'}{\scheck} \mid \heap]
      \end{align*}
    \item \refrule{SGuardWhile}:
      \begin{align*}
        &V[\sguard{\triple{\sstate}{\sseq{\swhile{e}{\gform}{s}}{s'}}{\gform'} }{\sstate'[\pc = \pc(\sstate'')]}{\_}{\_} \mid \heap] := \\
        &\quad V_0[\multiple{t \mapsto V_0(\senv(\sstate')(x))}][\sproduce{\sstate'[\senv = \senv(\sstate')[\multiple{x \mapsto t}]]}{\gform}{\sstate''} \mid \heap]
      \end{align*}
      where $V_0 = V[\scons{\sstate}{\gform}{\sstate'}{\scheck'} \mid \heap]$.
    \item \refrule{SGuardFinish}:
      $$V[\sguard{\triple{\sstate}{\kskip}{\gform}}{\sstate'}{\scheck}{\emptyset} \mid \heap] := V[\scons{\sstate}{\gform}{\sstate'}{\scheck} \mid \heap]$$
  \end{itemize}
\end{definition}

\begin{lemma}\label{lem:assert-after-rem}
  If $\assertion{\heap}{\perms}{\env}{\gform}$ and $\simstate{V}{\sstate}{\heap}{\perms \setminus \efoot{\heap}{\perms}{\gform}}{\env}$, then $\assertion{\heap}{\perms \setminus \vfoot{V'}{\heap}{\frem(\sstate', \gform')}}{\env}{\gform}$.
\end{lemma}
\begin{proof}
  Let $\perms_V = \vfoot{V'}{\heap}{\frem(\sstate', \gform')}$. By lemma \ref{lem:assert-efoot-subset} it suffices to show that $(\perms \setminus \perms_V) \subseteq \efoot{\heap}{\env}{\gform}$. By lemma \ref{lem:efoot-subset-spec} $\efoot{\heap}{\env}{\gform} \subseteq \perms$, therefore it suffices to show that $\perms_V \cap \efoot{\heap}{\env}{\gform} = \emptyset$.

  If $\fpre(m)$ is completely precise, $\frem(\sstate', \fpre(m)) = \emptyset$, thus $\perms_V = \emptyset$.
  
  Otherwise,
  \begin{align*}
    \perms_V &= \vfoot{V'}{\heap}{\frem(\sstate', \gform')} \\
      &= \vfoot{V'}{\heap}{\set{ \pair{t}{f} : \triple{f}{t}{t'} \in \sheap(\sstate) \cup
        \oheap(\sstate)} \cup \set{ \pair{p}{\multiple{t}} : \pair{p}{\multiple{t}} \in \sheap(\sstate)}} \\
      &= \bigcup_{\triple{f}{t}{t'} \in \sheap(\sstate) \cup
      \oheap(\sstate)} \vfoot{V'}{\heap}{\triple{f}{t}{t'}} \cup \bigcup_{\pair{p}{\multiple{t}} \in \sheap(\sstate)} \vfoot{V}{\heap}{\pair{p}{\multiple{t}}} \\
      &= \bigcup_{h \in \sheap(\sstate) \cup \oheap(\sstate)} \vfoot{V}{\heap}{h}
  \end{align*}
  But since $\simstate{V'}{\sstate'}{\heap}{\perms \setminus \perms_V}{\env}$, for each $h \in \sheap(\sstate) \cup \oheap(\sstate)$, $\vfoot{V'}{\heap}{h} \cap \perms_E = \emptyset$ by lemma \ref{lem:sim-heap-disjoint}. Therefore $\perms_V \cap \perms_E = \emptyset$.
\end{proof}

\begin{theorem}[Progress, part 1]\label{thm:dtrans-progress}
  Let $\Gamma$ be some dynamic state validated by $\vstate$ and valuation $V$. If $\sguard{\vstate}{\sstate'}{\scheck}{\sperms}$ with corresponding valuation $V'$ extending $V$, $V'(\pc(\sstate')) = \ktrue$, and $\pair{\heap}{\perms(\Gamma)} \vdash_{V'} \scheck$ then
  $$\dtrans{\prog}{\vfoot{V'}{\heap(\Gamma)}{\sperms}}{\Gamma}{\Gamma'}$$
  for some $\Gamma'$.

  In other words, if the dynamic state satisfies the matching symbolic checks, then dynamic execution can proceed.
\end{theorem}

\begin{proof}
  We procede by cases on $\sguard{\vstate}{\sstate'}{\scheck}{\sperms}$.
  \begin{enumcases}
    \case \refrule{SGuardInit}: Result is trival by \refrule{ExecInit}.

    \case \refrule{SGuardSeq}: Then $\Gamma = \pair{\heap}{\triple{\perms}{\env}{\sseq{\kskip}{s}} \cdot \stack'}$ for some $\heap, \perms, \env, s, \stack'$, thus \\
    $\dexec{\heap}{\triple{\perms}{\env}{\sseq{\kskip}{s}} \cdot \stack'}{\vfoot{V'}{\heap}{\emptyset}}{\heap}{\triple{\perms}{\env}{s} \cdot \stack'}$ by \refrule{ExecSeq}, and the result is immediate from \refrule{ExecStep}.

    \case \refrule{SGuardAssign}: Then $\Gamma = \pair{\heap}{\triple{\perms}{\env}{\sseq{x = e}{s}} \cdot \stack'}$ for some $\heap, \perms, \env, \stack'$.

    By \refrule{SGuardAssign} $\seval{\sstate(\vstate)}{e}{t}{\sstate'}{\scheck}$ for some $t, \sstate', \scheck$. Also, $V' = V[\seval{\sstate}{e}{t}{\sstate'}{\scheck} \mid \heap]$. Since $\Gamma$ corresponds to $\vstate$, by definition \ref{def:vstate-corresponds} $\simstate{V}{\sstate(\vstate)}{\heap}{\perms}{\env}$. By assumptions $\rtassert{V'}{\heap}{\perms}{\scheck}$, and $V'(\pc(\sstate')) = \ktrue$. Then $\eval{\heap}{\env}{e}{V'(t)}$ and $\frm{\heap}{\perms}{\env}{e}$ by lemma \ref{lem:seval-soundness}.

    Therefore $\dexec{\heap}{\triple{\perms}{\env}{\sseq{x = e}{s}} \cdot \stack'}{\vfoot{V'}{\heap}{\emptyset}}{\heap}{\triple{\perms}{\env[x \mapsto V'(t)]}{s} \cdot \stack'}$ by \refrule{ExecAssign}, and the result is immediate from \refrule{ExecStep}.

    \case \refrule{SGuardAssignField}: Then $\Gamma = \pair{\heap}{\triple{\perms}{\env}{\sseq{x.f = e}{s}} \cdot \stack'}$ for some $\heap, \perms, \env, x, f, e, s, \stack'$.

    By \refrule{SGuardAssignField} $\seval{\sstate(\vstate)}{e}{t}{\sstate'}{\scheck_1}$ for some $t, \sstate', \scheck_1$, and $\scons{\sstate'}{\kacc(x.f)}{\sstate''}{\scheck_2}$ for some $\sstate'', \scheck_2$. Also, $V' = V[\seval{\sstate(\vstate)}{e}{t}{\sstate'}{\scheck_1} \mid \heap][\scons{\sstate'}{\kacc(x.f)}{\sstate''}{\scheck_2} \mid \heap]$.

    Since $\Gamma$ corresponds to $\vstate$, by definition \ref{def:vstate-corresponds} $\simstate{V}{\sstate(\vstate)}{\heap}{\perms}{\env}$. By assumptions $\rtassert{V'}{\heap}{\perms}{\scheck_1 \cup \scheck_2}$, thus $\rtassert{V'}{\heap}{\perms}{\scheck_1}$ and $\rtassert{V'}{\heap}{\perms}{\scheck_2}$ by lemma \ref{lem:scheck-monotonicity}, and $V'(\pc(\sstate'')) = \ktrue$, thus $V'(\pc(\sstate')) = \ktrue$ by lemma \ref{lem:cons-subpath}.
    
    Now $\eval{\heap}{\env}{e}{V'(t)}$ and $\frm{\heap}{\perms}{\env}{e}$ by lemma \ref{lem:seval-soundness}. Let $\ell = \env(x)$, then $\eval{\heap}{\env}{x}{\env(x)}$ by \refrule{EvalVar}. Finally, $\assertion{\heap}{\perms}{\env}{\kacc(x.f)}$ by lemma \ref{lem:cons-soundness}.

    Let $\heap' = \heap[\pair{\ell}{f} \mapsto V'(t)]$. Then $\dexec{\heap}{\triple{\perms}{\env}{\sseq{x.f = e}{s}} \cdot \stack'}{\vfoot{V'}{\heap}{\emptyset}}{\heap'}{\triple{\perms}{\env}{s} \cdot \stack'}$ by \refrule{ExecAssignField}, and the result is immediate from \refrule{ExecStep}.

    \case \refrule{SGuardAlloc}: Then $\Gamma = \pair{\heap}{\triple{\perms}{\env}{\sseq{x = \kalloc(S)}{s}} \cdot \stack'}$ for some $\heap, \perms, \env, x, S, s, \stack'$.

    Let $\ell = \ffresh$, $\multiple{T~f} = \fstruct(S)$, and $\heap' = \heap[\multiple{(\ell, f) \mapsto \fdefault(T)}]$. Then \\
    $\dexec{\heap}{\triple{\perms}{\env}{\sseq{x = \kalloc(S)}{s}} \cdot \stack'}{\vfoot{V'}{\heap}{\emptyset}}{\heap'}{\triple{\perms}{\env}{s} \cdot \stack'}$ by \refrule{ExecAlloc}, and the result is immediate from \refrule{ExecStep}.

    \case \refrule{SGuardCall}:
    Then $\Gamma = \pair{\heap}{\triple{\perms}{\env}{\sseq{y \kassign m(\multiple{e})}{s}}} \cdot \stack$ for some $\heap, \perms, \env, y, m, \multiple{e}, s, \stack$. Let $\multiple{x} = \fparams(m)$.

    By \refrule{SGuardCall} $\multiple{\seval{\sstate(\vstate)}{e}{t}{\sstate'}{\scheck}}$ for some $t, \sstate', \scheck$, and $\scons{\sstate'[\senv = [\multiple{x \mapsto t}]]}{\fpre(m)}{\sstate''}{\scheck'}$ for some $\sstate'', \scheck'$.

    Also, by definition $V' = V\multiple{[\seval{\sstate(\vstate)}{e}{t}{\sstate'}{\scheck} \mid \heap]}[\scons{\sstate'[\senv = \multiple{x \mapsto t}]}{\fpre(m)}{\sstate''}{\scheck'} \mid \heap]$.

    Since $\Gamma$ corresponds to $\vstate$, by definition \ref{def:vstate-corresponds} $\simstate{V}{\sstate(\vstate)}{\heap}{\perms}{\env}$. By assumptions, $V'(\pc(\sstate'')) = \ktrue$ thus $V'(\pc(\sstate')) = \ktrue$ by lemma \ref{lem:cons-subpath}, and $\rtassert{V'}{\heap}{\perms}{\scheck \cup \scheck'}$, thus $\rtassert{V'}{\heap}{\perms}{\scheck}$ and $\rtassert{V'}{\heap}{\perms}{\scheck'}$ by lemma \ref{lem:scheck-monotonicity}.

    Then $\multiple{\eval{\heap}{\env}{e}{V'(t)}}$, $\multiple{\frm{\heap}{\perms}{\env}{e}}$, and $\simstate{V'}{\sstate'}{\heap}{\perms}{\env}$ by lemma \ref{lem:seval-soundness}.

    Let $\senv' = [\multiple{x \mapsto t}]$ and $\env' = [\multiple{x \mapsto V'(t)}]$. Since $\simstate{V'}{\sstate'}{\heap}{\perms}{\env}$ and $\simenv{V'}{\senv'}{\env'}$ by construction, $\simstate{V'}{\sstate'[\senv = \senv']}{\heap}{\perms}{\env'}$.

    Let $\perms_V = \vfoot{V'}{\heap}{\frem(\sstate', \fpre(m))}$ and $\perms_E = \efoot{\heap}{\env}{\fpre(m)}$.

    As noted before, $\scons{\sstate'[\senv = \senv']}{\fpre(m)}{\sstate''}{\scheck'}$ by \refrule{SExecCall}. Also, $V'(\pc(\sstate'')) = \ktrue$ and $\rtassert{V'}{\heap}{\perms}{\scheck'}$. Therefore $\simstate{V'}{\sstate'}{\heap}{\perms \setminus \perms_E}{\env'}$ and $\assertion{\heap}{\perms}{\env}{\fpre(m)}$ by lemma \ref{lem:cons-soundness}.

    Then $\assertion{\heap}{\perms \setminus \perms_V}{\env}{\fpre(m)}$ by lemma \ref{lem:assert-after-rem}.

    Let $\perms' = \foot{\heap}{\perms \setminus \perms_V}{\env'}{\fpre(m)}$. Then $\dexec{\heap}{\triple{\perms}{\env}{\sseq{y \kassign m(\multiple{e})}{s}} \cdot \stack}{\perms_V}{\heap}{\triple{\perms'}{\env'}{\sseq{\fbody(m)}{\kskip}} \cdot \triple{\perms \setminus \perms'}{\env}{\sseq{y \kassign m(\multiple{e})}{s}} \cdot \stack}$ by \refrule{ExecCallEnter}, and the result is immediate from \refrule{ExecStep}.

    \case \refrule{SGuardAssert}:
    Then $\Gamma = \pair{\heap}{\triple{\perms}{\env}{\sseq{\sassert{\phi}}{s}} \cdot \stack}$ for some $\heap, \perms, \env, \phi, s, \stack$.

    By \refrule{SGuardAssert} $\scons{\sstate(\vstate)}{\simprecise{\phi}}{\sstate'}{\scheck}$, also $V' = V[\scons{\sstate(\vstate)}{\simprecise{\phi}}{\sstate'}{\scheck} \mid \heap]$.

    Since $\Gamma$ corresponds to $\vstate$, by definition \ref{def:vstate-corresponds} $\simstate{V}{\sstate(\vstate)}{\heap}{\perms}{\env}$. Also by assumptions, $\rtassert{V'}{\heap}{\perms}{\scheck}$, and $V'(\pc(\sstate')) = \ktrue$. Thus $\assertion{\heap}{\perms}{\env}{\phi}$ by lemma \ref{lem:cons-soundness} since $\simprecise{\phi}$ is a specification.
    
    Therefore $\dexec{\heap}{\triple{\perms}{\env}{\sseq{\sassert{\phi}}{s}} \cdot \stack}{\vfoot{V'}{\heap}{\emptyset}}{\heap}{\triple{\perms}{\env}{s} \cdot \stack}$ by \refrule{ExecAssert}, and the result is immediate from \refrule{ExecStep}.

    \case \refrule{SGuardFold}:
    Then $s(\vstate) = s(\Gamma) = \sfold{p(\multiple{e})}$ for some $p, \multiple{e}$. Thus \refrule{ExecFold} trivially applies, and the result is immediate from \refrule{ExecStep}.

    \case \refrule{SGuardUnfold}:
    Then $s(\vstate) = s(\Gamma) = \sunfold{p(\multiple{e})}$ for some $p, \multiple{e}$. Thus \refrule{ExecUnfold} trivially applies, and the result is immediate from \refrule{ExecStep}.

    \case \refrule{SGuardIf}:
    Then $s(\vstate) = s(\Gamma) = \sseq{\sif{e}{s_1}{s_2}}{s}$ for some $e, s_1, s_2$, and thus $\Gamma = \pair{\heap}{\triple{\perms}{\env}{\sseq{\sif{e}{s_1}{s_2}}{s}} \cdot \stack}$ for some $\heap, \perms, \env, \stack$.

    By \refrule{SGuardIf} $\seval{\sstate(\vstate)}{e}{t}{\sstate'}{\scheck}$ for some $t, \sstate', \scheck$, and also $V' = V[\seval{\sstate(\vstate)}{e}{t}{\sstate'}{\scheck} \mid \heap]$.

    Now by assumptions $V'(\pc(\sstate')) = \ktrue$ and $\rtassert{V'}{\heap}{\perms}{\scheck}$. Then $\eval{\heap}{\env}{e}{V'(t)}$ and $\frm{\heap}{\perms}{\env}{e}$ by lemma \ref{lem:seval-soundness}.

    Now, since we assume a well-typed program, $V'(t) = \ktrue$ or $\kfalse$. Then either \refrule{ExecIfA} or \refrule{ExecIfB} applies, and the result is immediate from \refrule{ExecStep}.

    \case \refrule{SGuardWhile}:
      Then $s(\vstate) = s(\Gamma) = \sseq{\swhile{e}{\gform}{s}}{s'}$ for some $e$, $\gform$, $s$, $s'$, and thus $\Gamma = \pair{\heap}{\triple{\perms}{\env}{\sseq{\swhile{e}{\gform}{s}}{s'}} \cdot \stack}$ for some $\heap$, $\perms$, $\env$, $\stack$.

      Let $\multiple{x} = \fmodified(s)$. By \refrule{SGuardWhile} $\scons{\sstate}{\gform}{\sstate'}{\scheck'}$, $\sproduce{\sstate'[\senv = \senv(\sstate')[\multiple{x \mapsto \ffresh}]]}{\gform}{\sstate''}$, and $\spceval{\sstate''}{e}{t}{\scheck''}$. Then by definition \ref{def:sguard-valuation} $V'$ extends the corresponding valuation for these judgements.

      By assumptions $V'(\pc(\sstate'')) = \ktrue$, thus $V'(\pc(\sstate')) = \ktrue$ by lemma \ref{lem:produce-subpath}.

      Also by assumptions $\rtassert{V'}{\heap}{\perms}{\scheck' \cup \scheck''}$, thus $\rtassert{V'}{\heap}{\perms}{\scheck'}$ and $\rtassert{V'}{\heap}{\perms}{\scheck''}$ by lemma \ref{lem:scheck-monotonicity}.

      Therefore $\simstate{V'}{\sstate'}{\heap}{\perms \setminus \efoot{\heap}{\perms}{\gform}}{\env}$ and $\assertion{\heap}{\perms}{\env}{\gform}$.

      Let $\xperms = \vfoot{V'}{\heap}{\frem(\sstate', \gform)}$. Then by lemma \ref{lem:assert-after-rem} $\assertion{\heap}{\perms \setminus \xperms}{\env}{\gform}$.

      Let $\multiple{t}$ be the list of fresh values used in $\sproduce{\sstate'[\senv = \senv(\sstate')[\multiple{x \mapsto \ffresh}]]}{\gform}{\sstate''}$ when applying \refrule{SGuardWhile}. Let $\senv' = \senv(\sstate')[\multiple{x \mapsto t}]$. Then by definition \ref{def:sguard-valuation}, and since $\simenv{V'}{\senv(\sstate')}{\env}$, for each pair of $x$ and $t$, $V'(\senv'(x)) = V'(t) = V'(\senv(\sstate')(x)) = \env(x)$.
      
      Therefore $\simenv{V'}{\senv'}{\env}$, and thus $\simstate{V'}{\sstate'[\senv = \senv']}{\heap}{\perms \setminus \efoot{\heap}{\perms}{\gform}}{\env}$.

      Now by lemma \ref{lem:produce-soundness} $\simstate{V'}{\sstate''}{\heap}{\perms}{\gform}$.

      Also, as noted before, $\rtassert{V'}{\heap}{\perms}{\scheck''}$. Therefore $\eval{\heap}{\env}{e}{V'(t)}$ by lemma \ref{lem:pc-eval-soundness}.

      Let $\perms' = \efoot{\heap}{\perms \setminus \xperms}{\env}$.

      Now, since we assume the program to be properly typed, one of the following subcases apply:

      \subcase $V'(t) = \ktrue$: Then by \refrule{ExecWhileSkip} $\dexec{\heap}{\triple{\perms}{\env}{\sseq{\swhile{e}{\gform}{s}}{s'}} \cdot \stack}{\xperms}{\heap}{\triple{\perms'}{\env}{\sseq{s}{\kskip}} \cdot \triple{\perms \setminus \perms'}{\env}{\sseq{\swhile{e}{\gform}{s}}{s'}} \cdot \stack}$ and the result is immediate from \refrule{ExecStep}.

      \subcase $V'(t) = \kfalse$: Then by \refrule{ExecWhileSkip} $\dexec{\prog}{\heap}{\triple{\perms}{\env}{\sseq{\swhile{e}{\gform}{s}}{s'}} \cdot \stack}{\xperms}{\heap}{\triple{\perms}{\env}{s} \cdot \stack}$ and the result is immediate from \refrule{ExecStep}.

    \case \refrule{SGuardFinish}:
    Then $s(\vstate) = s(\Gamma) = \kskip$ and thus $\Gamma = \pair{\heap}{\triple{\perms}{\env}{\kskip} \cdot \stack}$ for some $\heap, \perms, \env, \stack$.

    By \refrule{SGuardFinish} $\scons{\sstate(\vstate)}{\gform(\vstate)}{\sstate'}{\scheck}$, and $V' = V[\scons{\sstate(\vstate)}{\gform(\vstate)}{\sstate'}{\scheck} \mid \heap]$. By assumptions, $V'(\pc(\sstate')) = \ktrue$ and $\rtassert{V'}{\heap}{\perms}{\scheck}$. Therefore $\assertion{\heap}{\perms}{\env}{\gform(\vstate)}$.

    Since $\Gamma$ is a valid state, the partial state $\pair{\heap}{\stack}$ must be validated by $\vstate$ and $V$, thus one of the following subcases must apply:

    \subcase $\stack = \nilsym$ -- then $\Gamma = \pair{\heap}{\triple{\perms}{\env}{\kskip} \cdot \nilsym}$. Then \refrule{ExecFinal} trivially applies to yield the result.

    \subcase $\stack = \triple{\perms_0}{\env_0}{\sseq{m(\multiple{e})}{s_0}} \cdot \stack'$ for some $\perms_0, \env_0, m, \multiple{e}, s_0, \stack'$ and $\gform(\vstate) = \fpost(m)$.

    Since $\gform(\vstate) = \fpost(m)$, $\assertion{\heap}{\perms}{\env}{\fpost(m)}$. Then \refrule{ExecCallExit} applies, and the result is immediate from \refrule{ExecStep}.

    \subcase $\stack = \triple{\perms_0}{\env_0}{\sseq{\swhile{e}{\gform}{s_0}}{s_0'}} \cdot \stack'$ for some $\perms_0, \env_0, e, \gform, s_0, s_0', \stack'$ and $\gform(\vstate) = \gform$.

    Since $\gform(\vstate) = \fpost(m)$, $\assertion{\heap}{\perms}{\env}{\gform}$. Then \refrule{ExecWhileFinish} applies, and the result is immediate from \refrule{ExecStep}.

  \end{enumcases}
\end{proof}

\begin{theorem}[Progress, part 2]\label{thm:guard-progress}
  Let $\Gamma$ be some well-formed dynamic state validated by $\vstate$ and valuation $V$. Then if $\Gamma \ne \finalsym$,
  $$\vstate \rightharpoonup \sstate', \scheck, \sperms$$
  for some $\sstate'$, $\scheck$, $\sperms$ such that $V'(\pc(\sstate')) = \ktrue$ where $V'$ is the corresponding valuation extending $V'$.

  In other words, there is always some matching guard that computes the necessary checks.
\end{theorem}
\begin{proof}
  First, if $\Gamma = \initsym$, then \refrule{SGuardInit} applies to yield the desired result.

  Otherwise, $\Gamma = \pair{\heap}{\stack}$ for some non-empty stack. Therefore $\vstate = \triple{\sstate}{s}{\gform}$ for some $\sstate$, $s$, and $\gform$ such that $\simstate{V}{\sstate}{\heap}{\perms(\Gamma)}{\env(\Gamma)}$ and $s = s(\Gamma)$.
  
  Then one of the following cases apply since $\stack$ is a well-formed stack:

  \begin{enumcases}
    \case $s = \kskip$: By lemma \ref{lem:cons-progress} $\scons{\sstate}{\gform}{\sstate'}{\scheck}$ for some $\sstate'$, $\scheck$ such that $V'(\pc(\sstate')) = \ktrue$ where $V'$ is the corresponding valuation. Then by \refrule{SGuardFinish} $\sguard{\vstate}{\sstate'}{\scheck}{\emptyset}$ and $V'$ is the corresponding valuation for this judgement.

    \case $s = \sseq{s'}{s''}$ for some $s'$, $s''$: We complete the proof by proving the following statement by induction on the syntax form of $s'$:

    If $s = \sseq{s'}{s''}$ then $\vstate \rightharpoonup \sstate', \scheck, \sperms$ for some $\sstate'$, $\scheck$, $\sperms$ such that $V'(\pc(\sstate')) = \ktrue$ where $V'$ is the corresponding valuation extending $V'$.

    \subcase $s' = \sseq{s_1}{s_2}$:
      By lemma \ref{lem:stmt-rearrangement}, $s' = \sseq{s_1'}{s_2'}$ where $s_1'$ is not a sequence statement. Then $\sseq{s}{s'} = \sseq{s_1'}{\sseq{s_2'}{s'}}$.

      Then the inductive hypothesis applies, which completes the proof.

    \subcase $s' = \kskip$:
      Then $\sguard{\vstate}{\sstate}{\emptyset}{\emptyset}$ by \textsc{SGuardSeq}.

    \subcase $s' = x \kassign e$:
      By lemma \ref{lem:eval-progress}, $\seval{\sstate}{e}{\_}{\sstate'}{\scheck}$ for some $\sstate'$ and $\scheck$ such that $V_1(\pc(\sstate')) = \ktrue$ for the corresponding valuation $V_1$.

      Then $\sguard{\vstate}{\sstate'}{\scheck}{\emptyset}$ by \refrule{SGuardAssign}. By definition the corresponding valuation extends $V_1$, thus $V'(\pc(\sstate')) = \ktrue$.

    \subcase $x.f \kassign e$:

      By lemma \ref{lem:eval-progress} $\seval{\sstate}{e}{\_}{\sstate'}{\scheck'}$ for some $\sstate'$ and $\scheck'$ such that $V_1(\pc(\sstate')) = \ktrue$ where $V_1$ is the corresponding valuation extending $V$.

      By lemma \ref{lem:cons-progress} $\scons{\sstate'}{\kacc(x.f)}{\sstate''}{\scheck''}$ for some $\sstate''$ and $\scheck''$ such that $V_2(\pc(\sstate'')) = \ktrue$ for corresponding valuation $V_2$, with initial valuation $V_1$.

      Then $\sguard{\vstate}{\sstate''}{\scheck' \cup \scheck''}{\emptyset}$ by \refrule{SGuardAssignField}. By definition the corresponding valuation $V'$ extends $V_2$, therefore $V'(\pc(\sstate'')) = \ktrue$.

    \subcase $x = \salloc{S}$:

      Then $\sguard{\vstate}{\sstate}{\emptyset}{\emptyset}$ by \refrule{SExecGuardAlloc}.

    \subcase $y \kassign m(e_1, \cdots, e_n)$:

      Let $\sstate_0 = \sstate$ and $V_0 = V$, then for each $e_i$, $\seval{\sstate_{i-1}}{e}{t_i}{\sstate_i}{\_}$ for some $\sstate_i$ such that $V_i(\pc(\sstate_i)) = \ktrue$ for corresponding valuation $V_i$, with initial valuation $V_{i-1}$, by lemma \ref{lem:eval-progress}.

      Let $x_1, \cdots, x_n = \fparams(m)$. By lemma \ref{lem:cons-progress}, $\scons{\sstate_n[\senv = [\multiple{x_i \mapsto t_i}]]}{\fpre(m)}{\sstate'}{\_}$ for some $\sstate'$ such that $V'(\pc(\sstate')) = \ktrue$ for corresponding valuation $V'$, with initial valuation $V_n$.

      Then $\sguard{\vstate}{\sstate'[\senv = \senv(\sstate)]}{\scheck_1 \cup \cdots \cup \scheck_n \cup \scheck'}{\frem(\sstate'', \fpre(m))}$ by \refrule{SGuardCall} and $V'$ is the corresponding valuation

    \subcase $\sassert{\gform}$:

      By lemma \ref{lem:cons-progress} $\scons{\sstate}{\gform}{\sstate'}{\scheck}$ for some $\sstate'$ and $\scheck$ where $V'(\sstate') = \ktrue$ for the corresponding valuation $V'$.

      Then by \refrule{SGuardAssert} $\sguard{\vstate}{\sstate'}{\scheck}{\emptyset}$ and $V'$ is the corresponding valuation.

    \subcase $\sif{e}{s_1}{s_2}$:

      By lemma \ref{lem:eval-progress} $\seval{\sstate}{e}{t}{\sstate'}{\scheck}$ for some $t$, $\sstate'$ such that $V'(\pc(\sstate')) = \ktrue$ where $V'$ is the corresponding valuation.

      Then $\sguard{\vstate}{\sstate'}{\scheck}{\emptyset}$ by \refrule{SGuardIf} and $V'$ is the corresponding valuation.

    \subcase $\swhile{e}{\gform}{s}$ for some $e$, $\gform$, $s$:

      By lemma \ref{lem:cons-progress} $\scons{\sstate}{\gform}{\sstate'}{\scheck'}$ for some $\sstate'$ and $\scheck'$ such that $V_1(\pc(\sstate')) = \ktrue$ where $V'$ is the corresponding valuation.

      Let $\multiple{x} = \fmodified(s)$ and $\sstate'' = \sstate'[\senv = \senv(\sstate')[\multiple{x \mapsto \ffresh}]]$.
      
      Let $V_2 = V_1[\multiple{t \mapsto V_1(\senv(\sstate')(x))}]$. Then $V_2(\pc(\sstate'')) = V_1(\pc(\sstate')) = \ktrue$.
      
      Then by lemma \ref{lem:produce-progress} $\sproduce{\sstate''}{\gform}{\sstate'''}$ for some $\sstate'''$ such that $V_3(\pc(\sstate''')) = \ktrue$ where $V_3$ is the corresponding valuation extending $V_2$.
      
      By lemma \ref{lem:pc-eval-progress} $\spceval{\sstate'''}{e}{t}{\_}$ for some $t$. Let $V'$ be the corresponding valuation extending $V_3$, thus $V'(\pc(\sstate''')) = V_3(\pc(\sstate''')) = \ktrue$.

      Then $\sguard{\vstate}{\sstate'[\pc = \pc(\sstate''')]}{\scheck' \cup \scheck''}{\frem(\sstate', \gform)}$ by \textsc{SGuardWhile} and $V'$ is the corresponding valuation.

    \subcase $\sfold{p(e_1, \cdots, e_n)}$:

      Let $\sstate_0 = \sstate$ and $V_0 = V$. For each $e_i$, $\seval{\sstate_{i-1}}{e}{t_i}{\sstate_i}{\scheck_i}$ by lemma \ref{lem:eval-progress} for some $\sstate_i$ and $\scheck_i$ such that $V_i(\pc(\sstate_i)) = \ktrue$ where $V_i$ is the corresponding valuation.

      Let $x_1, \cdots, x_n = \fpredparams(p)$. By lemma \ref{lem:cons-progress} $\scons{\sstate_n[\senv = [\multiple{x_i \mapsto t_n}]]}{\fpred(p)}{\sstate'}{\scheck'}$ for some $\sstate'$ and $\scheck'$ such that $V'(\pc(\sstate')) = \ktrue$ where $V'$ is the corresponding valuation extending $V_n$.

      Then $\sguard{\vstate}{\sstate'[\senv = \senv(\sstate)]}{\scheck_1 \cup \cdots \cup \scheck_n \cup \scheck'}{\emptyset}$ by \refrule{SGuardFold} and $V'$ is the corresponding valuation.

    \subcase $\sunfold{p(e_1, \cdots, e_n)}$:

      Let $\sstate_0 = \sstate$ and $V_0 = V$. For each $e_i$, $\seval{\sstate_{i-1}}{e}{t_i}{\sstate_i}{\scheck_i}$ by lemma \ref{lem:eval-progress} for some $\sstate_i$ and $\scheck_i$ such that $V_i(\pc(\sstate_i)) = \ktrue$ where $V_i$ is the corresponding valuation.

      By lemma \ref{lem:cons-progress} $\scons{\sstate_n}{p(e_1, \cdots, e_n)}{\sstate'}{\scheck'}$ for some $\sstate'$ and $\scheck'$ such that $V'(\pc(\sstate')) = \ktrue$ where $V'$ is the corresponding valuation extending $V_n$.

      Then $\sguard{\vstate}{\sstate'}{\scheck_1 \cup \cdots \cup \scheck_n \cup \scheck'}{\emptyset}$ by \refrule{SGuardUnfold} and $V'$ is the corresponding valuation.

  \end{enumcases}

\end{proof}

\subsection{Preservation}


\begin{lemma}\label{lem:preservation-heap-env-unchanged}
  Suppose $\Gamma = \pair{\heap}{\triple{\perms}{\env}{s} \cdot \stack}$ and $\Gamma' = \pair{\heap}{\triple{\perms'}{\env}{s'} \cdot \stack}$.

  If $\Gamma$ is validated by $\vstate$ and $V$, $\dtrans{\prog}{\Gamma}{\_}{\Gamma'}$ with valuation $V'$, and $\strans{\prog}{\vstate}{\vstate'}$ for some $\vstate'$ such that $\vstate'$ corresponds to $\Gamma'$, $\gform(\vstate') = \gform(\vstate)$, and $\dom(\senv(\vstate')) \supseteq \dom(\senv(\vstate))$, then $\Gamma'$ is a valid state.
\end{lemma}
\begin{proof}
  By definition, it suffices to show that $\Gamma'$ is validated by $\vstate'$ and $V'$.

  \textit{Part \ref{def:state-valid-reachable}:}
  By assumptions, $\vstate$ is reachable and $\strans{\prog}{\vstate}{\vstate'}$, thus $\vstate'$ is reachable with valuation $V'$.

  \textit{Part \ref{def:state-valid-correspond}:}
  By assumptions $\vstate'$ corresponds to $\Gamma'$ with $V'$.

  \textit{Part \ref{def:state-valid-partial}:}
  Since $\Gamma$ validated by $\vstate$ and $V$, the partial state $\pair{\heap}{\stack}$ is validated by $\vstate$ and $V$. Therefore one of the following cases apply:

  \textit{Case \ref{def:partial-valid-nil}:}
  Then $\stack = \nilsym$ and trivially the partial state $\pair{\heap}{\nilsym}$ is validated by $\vstate'$ and $V'$.

  \textit{Case \ref{def:partial-valid-call}:}
  Then $\stack$ is of the form $\triple{\perms'}{\env'}{\sseq{y \kassign m(e_1, \cdots, e_k)}{s'}} \cdot \stack'$ and, for some $\vstate''$, $V''$, $x_1, \cdots, x_k$, $t_1, \cdots, t_k$, $\sstate_0, \cdots, \sstate_k$ and $\sstate'$,
  \begin{gather*}
    \text{The partial state $\pair{\heap}{\stack'}$ is validated by $\vstate''$ and $V''$}, \\
    \vstate'' \text{ is reachable from $\prog$ with valuation $V''$}, \quad s(\vstate'') = s(\stack), \\
    x_1, \cdots, x_k = \fparams(m), \\
    \sstate_0 = \sstate(\vstate''), \quad \seval{\sstate_0}{e_1}{t_1}{\sstate_1}{\_}, \quad\cdots,\quad \seval{\sstate_{k-1}}{e_k}{t_k}{\sstate_k}{\_}, \\
    \universal{1 \le i \le k}{V(\senv(\vstate)(x_i)) = V''(t_i)}, \\
    \scons{\sstate_k}{\fpre(m)}{\sstate'}{\_}, \quad \simstate{V''}{\sstate'[\senv = \senv(\sstate_0)]}{\heap}{\perms'}{\env'}, \quad\text{and} \\
    \gform(\vstate) = \fpost(m)
  \end{gather*}

  Now immediately from above,
  \begin{gather*}
    \text{The partial state $\pair{\heap}{\stack'}$ is validated by $\vstate''$ and $V''$}, \\
    \vstate'' \text{ is reachable from $\prog$ with valuation $V''$}, \quad s(\vstate'') = s(\stack), \\
    x_1, \cdots, x_k = \fparams(m), \\
    \sstate_0 = \sstate(\vstate''), \quad \seval{\sstate_0}{e_1}{t_1}{\sstate_1}{\_}, \quad\cdots,\quad \seval{\sstate_{k-1}}{e_k}{t_k}{\sstate_k}{\_}, \\
    \scons{\sstate_k}{\fpre(m)}{\sstate'}{\_}, \quad \simstate{V''}{\sstate'[\senv = \senv(\sstate_0)]}{\heap}{\perms}{\env}, \quad\text{and}
  \end{gather*}

  Also, $\simenv{V}{\senv(\vstate)}{\env}$, $\simenv{V'}{\senv(\vstate')}{\env}$, and $\dom(\senv(\vstate')) \supseteq \dom(\senv(\vstate))$, and thus
  $$\universal{1 \le i \le k}{V'(\senv(\vstate')(x_i)) = \env(x_i) = V(\senv(\vstate)(x_i)) = V''(t_i)}.$$

  Also, by assumptions $\gform(\vstate') = \gform(\vstate) = \fpost(m)$.

  Therefore the partial state $\pair{\heap}{\triple{\perms'}{\env'}{\sseq{y \kassign m(e_1, \cdots, e_k)}{s'}} \cdot \stack'}$ is validated by $\vstate'$ and $V'$.

  \textit{If case \ref{def:partial-valid-while} applies:}
  Then $\stack$ is of the form $\triple{\env'}{\perms'}{\sseq{\swhile{e}{\gform}{s}}{s'}} \cdot \stack^*$ and there exists some $\vstate''$, $V''$, and $\sstate'$ such that:
  \begin{gather*}
    \text{The partial state $\pair{\heap}{\stack^*}$ is validated by $\vstate''$ and $V''$} \\
    \vstate'' \text{ is reachable from $\prog$ with valuation $V''$}, \quad s(\vstate'') = s(\stack) \\
    \scons{\sstate(\vstate'')}{\gform}{\sstate'}{\_},
    \simstate{V'}{\sstate'}{\heap}{\perms'}{\env'}, \quad\text{and} \\
    \gform(\vstate) = \gform
  \end{gather*}

  Now $\gform(\vstate') = \gform(\vstate) = \gform$. Then, with the conditions listed above, the partial state \\
  $\pair{\heap}{\triple{\env'}{\perms'}{\sseq{\swhile{e}{\gform}{s}}{s'}} \cdot \stack^*}$ is validated by $\vstate'$ and $V'$.

  Therefore $\Gamma'$ is validated by $\vstate'$.
\end{proof}

\begin{lemma}\label{lem:pres-add-heap}
  If $\simstate{V}{\sstate}{\heap}{\perms \setminus \vfoot{V}{\heap}{\triple{t}{f}{t'}}}{\env}$, $\pair{V(t)}{f} \in \perms$, and $\heap(V(t), f) = V(t')$, then $\simstate{V}{\sstate[\sheap = \sheap(\sstate); \triple{t}{f}{t'}]}{\heap}{\perms}{\env}$.
\end{lemma}
\begin{proof}
  Let $\sstate' = \sstate[\sheap = \sheap(\sstate); \triple{t}{f}{t'}]$. We want to show $\simstate{V}{\sstate'}{\heap}{\perms}{\env}$.

  By assumptions and lemma \ref{lem:simstate-monotonicity}, it is immediate that $\simstate{V}{\sstate}{\heap}{\perms}{\env}$.

  Since the $\oheap$, $\senv$, and $\pc$ are unchanged in $\sstate'$ WRT $\sstate$, we have $\simheap{V}{\oheap(\sstate')}{\heap}{\perms}$, $\simenv{V}{\senv(\sstate')}{\env}$, and $V(\pc(\sstate')) = \ktrue$. Thus it suffices to show that $\simheap{V}{\sheap(\sstate')}{\heap}{\perms}$.

  Let $h = \triple{t}{f}{t'}$, thus $\sstate' = \sstate[\sheap = \sheap(\sstate'); h]$.

  We have $\simheap{V}{\sstate}{\heap}{\perms \setminus \vfoot{V}{\heap}{h}}$. Then by lemma \ref{lem:sim-heap-disjoint},
  \begin{equation}\label{eq:pres-add-heap-disjoint}
    \universal{h' \in \sheap(\sstate)}{\vfoot{V}{\heap}{h'} \cap \vfoot{V}{\heap}{h} = \emptyset}.
  \end{equation}

  Then for arbitrary $h_1, h_2 \in \sheap(\sstate')$, if $h_1 \ne h_2$ one of the following applies:
  \begin{itemize}
    \item $h_1 = h$ or $h_2 = h$: WLOG we can assume $h_1 = h$ and thus $h_2 \ne h$. Then by \eqref{eq:pres-add-heap-disjoint}, $\vfoot{V}{\heap}{h_1} \cap \vfoot{V}{\heap}{h_2} = \emptyset$.
    \item $h_1 \ne h$ and $h_1 \ne h$: Then $h_1 \in \sheap(\sstate)$, $h_2 \in \sheap(\sstate)$, and $\simheap{V}{\sheap(\sstate)}{\heap}{\perms}$, thus $\vfoot{V}{\heap}{h_1} \cap \vfoot{V}{\heap}{h_2} = \emptyset$.
  \end{itemize}

  Therefore
  $$\universal{h_1, h_2 \in \sheap(\sstate')}{h_1 \ne h_2 \implies \vfoot{V}{\heap}{h_1} \cap \vfoot{V}{\heap}{h_2} = \emptyset}.$$

  Since predicate instances are the same in $\sheap(\sstate)$ and $\sheap(\sstate')$, and $\simheap{V}{\sheap(\sstate)}{\heap}{\perms}$,
  $$\universal{\pair{p}{\multiple{t}}}{\assertion{\heap}{\perms}{[\multiple{x \mapsto V(t)}]}{\fpred(p)}}.$$

  Now let $\triple{f}{t}{t'}$ be an arbitrary field instance in $\sheap(\sstate')$. Then one of the following applies:
  \begin{itemize}
    \item $\triple{f}{t}{t'} = h$: Then by assumptions $\heap(V(t), f) = V(t')$ and $\pair{V(t)}{f} \in \perms$.
    \item $\triple{f}{t}{t'} \ne h$, and thus $\triple{f}{t}{t'} \in \sheap(\sstate)$: Then $\heap(V(t), f) = V(t')$ and $\pair{V(t)}{f} \in \perms$ since $\simheap{V}{\sheap(\sstate)}{\heap}{\perms}$.
  \end{itemize}
  
  Therefore
  $$\universal{\triple{f}{t}{t'} \in \sheap(\sstate)}{\heap(V(t), f) = V(t')} \quad\text{and}$$
  $$\universal{\triple{f}{t}{t'} \in \sheap(\sstate)}{\pair{V(t)}{f} \in \perms}.$$

  Finally, since all requirements have been satisfied, $\simheap{V}{\sheap(\sstate')}{\heap}{\perms}$, which completes the proof.
\end{proof}

\begin{lemma}\label{lem:eval-heap-efoot-unchanged}
  If $\universal{\pair{\ell}{f} \in \efoot{\heap}{\perms}{e}}{\heap'(\ell, f) = \heap(\ell, f)}$ and $\eval{\heap}{\env}{e}{v}$ then $\eval{\heap'}{\env}{e}{v}$.
\end{lemma}
\begin{proof}
  By induction on $\eval{\heap}{\env}{e}{v}$:
  \begin{enumcases}
    \case \refrule{EvalLiteral} -- $\eval{\heap}{\env}{l}{l}$: $\eval{\heap'}{\env}{l}{l}$ by \refrule{EvalLiteral}.

    \case \refrule{EvalVar} -- $\eval{\heap}{\env}{x}{\env(x)}$: $\eval{\heap'}{\env}{x}{\env(x)}$ by \refrule{EvalVar}.

    \case\label{case:eval-heap-efoot-unchanged-anda} \refrule{EvalAndA} -- $\eval{\heap}{\env}{e_1 \kand e_2}{\kfalse}$:

      By inversion of \refrule{EvalAndA} $\eval{\heap}{\env}{e_1}{\kfalse}$. Then $\efoot{\heap}{\env}{e_1 \kand e_2} = \efoot{\heap}{\env}{e_1}$ by definition. Thus $\universal{\pair{\ell}{f} \in \efoot{\heap}{\perms}{e_1}}{\heap'(\ell, f) = \heap(\ell, f)}$, and therefore $\eval{\heap'}{\env}{e_1}{\kfalse}$ by induction.

    \case\label{case:eval-heap-efoot-unchanged-andb} \refrule{EvalAndB} -- $\eval{\heap}{\env}{e_1 \kand e_2}{v_2}$:

      By inversion of \refrule{EvalAndB} $\eval{\heap}{\env}{e_1}{\ktrue}$ and $\eval{\heap}{\env}{e_2}{v_2}$. Now $\efoot{\heap}{\env}{e_1 \kand e_2} = \efoot{\heap}{\env}{e_1} \cup \efoot{\heap}{\env}{e_2}$.

      Thus $\universal{\pair{\ell}{f} \in \efoot{\heap}{\perms}{e_1}}{\heap'(\ell, f) = \heap(\ell, f)}$ since $\efoot{\heap}{\env}{e_1} \subseteq \efoot{\heap}{\env}{e_1 \kand e_2}$. Therefore $\eval{\heap'}{\env}{e_1}{\ktrue}$ by induction. Similarly, $\eval{\heap'}{\env}{e_2}{v_2}$.

      Therefore $\eval{\heap'}{\env}{e_1 \kand e_2}{v_2}$ by \refrule{EvalAndB}.

    \case \refrule{EvalOrA} -- $\eval{\heap}{\env}{e_1 \kor e_2}{\ktrue}$: Similar to case \ref{case:eval-heap-efoot-unchanged-anda}.

    \case \refrule{EvalOrB} -- $\eval{\heap}{\env}{e_1 \kor e_2}{v_2}$: Similar to case \ref{case:eval-heap-efoot-unchanged-andb}.

    \case \refrule{EvalOp} -- $\eval{\heap}{\env}{e_1 \oplus e_2}{v_1 \oplus v_2}$:

      By inversion of \refrule{EvalOp} $\eval{\heap}{\env}{e_1}{v_1}$ and $\eval{\heap}{\env}{e_2}{v_2}$. Now $\efoot{\heap}{\env}{e_1 \oplus e_2} = \efoot{\heap}{\env}{e_1} \cup \efoot{\heap}{\env}{e_2}$.

      Thus $\universal{\pair{\ell}{f} \in \efoot{\heap}{\perms}{e_1}}{\heap'(\ell, f) = \heap(\ell, f)}$ since $\efoot{\heap}{\env}{e_1} \subseteq \efoot{\heap}{\env}{e_1 \kand e_2}$. Therefore $\eval{\heap'}{\env}{e_1}{v_1}$ by induction. Similarly, $\eval{\heap'}{\env}{e_2}{v_2}$.

      Therefore $\eval{\heap'}{\env}{e_1 \oplus e_2}{e_1 \oplus e_2}$ by \refrule{EvalOp}.

    \case \refrule{EvalNeg} -- $\eval{\heap}{\env}{\neg v}$:

      By inversion of \refrule{EvalNeg} $\eval{\heap}{\env}{e}{v}$. Also, $\efoot{\heap}{\env}{\kneg e} = \efoot{\heap}{\env}{e}$ by definition, thus $\universal{\pair{\ell}{f} \in \efoot{\heap}{\perms}{e}}{\heap'(\ell, f) = \heap(\ell, f)}$. Therefore $\eval{\heap'}{\env}{e}{v}$ by induction.

      Therefore $\eval{\heap'}{\env}{\kneg e}{\neg v}$ by \refrule{EvalNeg}.

    \case \refrule{EvalField} -- $\eval{\heap}{\env}{e.f}{\heap(\ell, f)}$:

      By inversion of \refrule{EvalField} $\eval{\heap}{\env}{e}{\ell}$. Then $\efoot{\heap}{\env}{e.f} = \efoot{\heap}{\env}{e}; \pair{\ell}{f}$.

      Thus $\universal{\pair{\ell}{f} \in \efoot{\heap}{\perms}{e}}{\heap'(\ell, f) = \heap(\ell, f)}$ since $\efoot{\heap}{\env}{e} \subseteq \efoot{\heap}{\env}{e.f}$. Therefore $\eval{\heap'}{\env}{e}{\ell}$ by induction.

      Now $\eval{\heap'}{\env}{e.f}{\heap'(\ell, f)}$. Also, since $\pair{\ell}{f} \in \efoot{\heap}{\env}{e.f}$, $\heap'(\ell, f) = \heap(\ell, f)$. Therefore $\eval{\heap'}{\env}{e.f}{\heap(\ell, f)}$.
    
  \end{enumcases}
\end{proof}

\begin{lemma}\label{lem:frm-heap-efoot-unchanged}
  If $\universal{\pair{\ell}{f} \in \efoot{\heap}{\env}{e}}{\heap'(\ell, f) = \heap(\ell, f)}$ and $\frm{\heap}{\perms}{\env}{e}$ then $\frm{\heap'}{\perms}{\env}{e}$.
\end{lemma}
\begin{proof}
  By induction on $\frm{\heap'}{\perms}{\env}{e}$:

  \begin{enumcases}
    \case \refrule{FrameLiteral} -- $\frm{\heap}{\perms}{\env}{l}$:
      $\frm{\heap'}{\perms}{\env}{l}$ by \refrule{FrameLiteral}.

    \case \refrule{FrameVar} -- $\frm{\heap}{\perms}{\env}{x}$:
      $\frm{\heap'}{\perms}{\env}{x}$ by \refrule{FrameVar}.

    \case \refrule{FrameField} -- $\frm{\heap}{\perms}{\env}{e.f}$:

      By inversion of \refrule{FrameField} $\frm{\heap}{\perms}{\env}{e}$ and $\assertion{\heap}{\perms}{\env}{\kacc(e.f)}$.

      Now since $\assertion{\heap}{\perms}{\env}{\kacc(e.f)}$, by inversion of \refrule{AssertAcc}, $\eval{\heap}{\env}{e}{\ell}$ and $\pair{\ell}{f} \in \perms$. Now $\efoot{\heap}{\env}{e.f} = \efoot{\heap}{\env}{e}; \pair{\ell}{f}$ by definition.

      Thus $\universal{\pair{\ell}{f} \in \efoot{\heap}{\perms}{e}}{\heap'(\ell, f) = \heap(\ell, f)}$ since $\efoot{\heap}{\env}{e} \subseteq \efoot{\heap}{\env}{e.f}$. Therefore $\eval{\heap'}{\env}{e}{\ell}$ by lemma \ref{lem:eval-heap-efoot-unchanged}, and $\pair{\ell}{f} \in \perms$ as noted previously. Now $\assertion{\heap'}{\perms}{\env}{\kacc(e.f)}$ by \refrule{AssertAcc}.

      Also, $\frm{\heap'}{\perms}{\env}{e}$ by induction. Therefore $\frm{\heap}{\perms}{\env}{e.f}$ by \refrule{FrameField}.

    \case \refrule{FrameOp} -- $\frm{\heap}{\perms}{\env}{e_1 \oplus e_2}$:

      By inversion of \refrule{FrameOp} $\frm{\heap}{\perms}{\env}{e_1}$ and $\frm{\heap}{\perms}{\env}{e_2}$.

      Also, $\efoot{\heap}{\env}{e_1 \oplus e_2} = \efoot{\heap}{\env}{e_1} \cup \efoot{\heap}{\env}{e_2}$ by definition. Thus $\universal{\pair{\ell}{f} \in \efoot{\heap}{\perms}{e_1}}{\heap'(\ell, f) = \heap(\ell, f)}$ since $\efoot{\heap}{\env}{e_1} \subseteq \efoot{\heap}{\env}{e_1 \oplus e_2}$. Therefore $\frm{\heap'}{\perms}{\env}{e_1}$ by induction. Similarly, $\frm{\heap'}{\perms}{\env}{e_2}$.

      Therefore $\frm{\heap'}{\perms}{\env}{e_1 \oplus e_2}$ by induction.

    \case\label{case:frm-efoot-heap-unchanged-ora} \refrule{FrameOrA} -- $\frm{\heap}{\perms}{\env}{e_1 \kor e_2}$:

      By inversion of \refrule{FrameOrA} $\eval{\heap}{\env}{e_1}{\ktrue}$ and $\frm{\heap}{\perms}{\env}{e_1}$.

      Now $\efoot{\heap}{\env}{e_1 \kor e_2} = \efoot{\heap}{\env}{e_1}$. Thus $\universal{\pair{\ell}{f} \in \efoot{\heap}{\perms}{e_1}}{\heap'(\ell, f) = \heap(\ell, f)}$. Therefore $\eval{\heap'}{\env}{e_1}{\ktrue}$ by lemma \ref{lem:eval-heap-efoot-unchanged}, and $\frm{\heap'}{\perms}{\env}{e_1}$ by induction.

      Therefore $\frm{\heap'}{\perms}{\env}{e_1 \kor e_2}$ by \refrule{FrameOrA}.

    \case\label{case:frm-efoot-heap-unchanged-orb} \refrule{FrameOrB} -- $\frm{\heap}{\perms}{\env}{e_1 \kor e_2}$:

      By inversion of \refrule{FrameOrA} $\eval{\heap}{\env}{e_1}{\kfalse}$, $\frm{\heap}{\perms}{\env}{e_1}$, and $\frm{\heap}{\perms}{\env}{e_2}$.

      Now $\efoot{\heap}{\env}{e_1 \kor e_2} = \efoot{\heap}{\env}{e_1} \cup \efoot{\heap}{\env}{e_2}$. Thus $\universal{\pair{\ell}{f} \in \efoot{\heap}{\perms}{e_1}}{\heap'(\ell, f) = \heap(\ell, f)}$ since $\efoot{\heap}{\env}{e_1} \subseteq \efoot{\heap}{\env}{e_1 \kor e_2}$. Therefore $\eval{\heap'}{\env}{e_1}{\kfalse}$ by lemma \ref{lem:eval-heap-efoot-unchanged}, and $\frm{\heap'}{\perms}{\env}{e_1}$ by induction. Similarly, $\frm{\heap'}{\perms}{\env}{e_2}$.

      Therefore $\frm{\heap'}{\perms}{\env}{e_1 \kor e_2}$ by \refrule{FrameOrB}.

    \case \refrule{FrameAndA} -- $\frm{\heap}{\perms}{\env}{e_1 \kand e_2}$: Similar to case \ref{case:frm-efoot-heap-unchanged-ora}.

    \case \refrule{FrameAndB} -- $\frm{\heap}{\perms}{\env}{e_1 \kand e_2}$: Similar to case \ref{case:frm-efoot-heap-unchanged-orb}.

    \case \refrule{FrameNeg} -- $\frm{\heap}{\perms}{\env}{\kneg e}$:

      By inversion of \refrule{FrameNeg} $\frm{\heap}{\perms}{\env}{e}$. Also, $\efoot{\heap}{\env}{\kneg e} = \efoot{\heap}{\env}{e}$. Thus $\universal{\pair{\ell}{f} \in \efoot{\heap}{\perms}{e}}{\heap'(\ell, f) = \heap(\ell, f)}$. Therefore $\frm{\heap'}{\perms}{\env}{e}$ by induction.

      Therefore $\frm{\heap'}{\perms}{\env}{\kneg e}$ by \refrule{FrameNeg}.
  \end{enumcases}
\end{proof}

\begin{lemma}\label{lem:efrm-heap-efoot-unchanged}
  If $\universal{\pair{\ell}{f} \in \efoot{\heap}{\perms}{\gform}}{\heap'(\ell, f) = \heap(\ell, f)}$ and $\efrm{\heap}{\perms}{\env}{\gform}$ then $\efrm{\heap'}{\perms}{\env}{\gform}$.
\end{lemma}
\begin{proof}
  By induction on $\efrm{\heap'}{\perms}{\env}{\gform}$:

  \begin{enumcases}
    \case \refrule{EFrameExpression} -- $\efrm{\heap}{\perms}{\env}{e}$:

      By inversion of \refrule{EFrameExpression} $\frm{\heap}{\perms}{\env}{e}$ and then $\frm{\heap'}{\perms}{\env}{e}$ by assumptions and lemma \ref{lem:frm-heap-efoot-unchanged}. Therefore $\efrm{\heap}{\perms}{\env}{e}$ by \refrule{EFrameExpression}.

    \case \refrule{EFrameConjunction} -- $\efrm{\heap}{\perms}{\env}{\phi_1 * \phi_2}$:

      By inversion of \refrule{EFrameConjunction} $\efrm{\heap}{\perms}{\env}{\phi_1}$ and $\efrm{\heap}{\perms}{\env}{\phi_2}$. Also, $\efoot{\heap}{\env}{\phi_1 * \phi_2} = \efoot{\heap}{\env}{\phi_1} \cup \efoot{\heap}{\env}{\phi_2}$ by definition.

      Now $\universal{\pair{\ell}{f} \in \efoot{\heap}{\env}{\phi_1}}{\heap'(\ell, f) = \heap(\ell, f)}$ since $\efoot{\heap}{\perms}{\phi_1} \subseteq \efoot{\heap}{\perms}{\gform}$. Therefore \\
      $\efrm{\heap'}{\perms}{\env}{\phi_1}$ by induction. Similarly, $\efrm{\heap'}{\perms}{\env}{\phi_2}$.

      Therefore $\efrm{\heap'}{\perms}{\env}{\phi_1 * \phi_2}$ by \refrule{EFrameConjunction}.

    \case \refrule{EFramePredicate} -- $\efrm{\heap}{\perms}{\env}{p(\multiple{e})}$:

      By inversion of \refrule{EFramePredicate} $\multiple{\frm{\heap}{\perms}{\env}{e}}$, $\multiple{\eval{\heap}{\env}{e}{v}}$, and $\efrm{\heap}{\perms}{[\multiple{x \mapsto v}]}{\fpred(p)}$, where $\multiple{x} = \fpredparams(p)$.

      Now $\efoot{\heap}{\env}{p(\multiple{e})} = \efoot{\heap}{[\multiple{x \mapsto v}]}{\fpred(p)} \cup \bigcup \multiple{\efoot{\heap}{\env}{e}}$.

      Then, for each $e$ and corresponding $x$, $\universal{\pair{\ell}{f} \in \efoot{\heap}{\env}{e}}{\heap'(\ell, f) = \heap(\ell, f)}$ since $\efoot{\heap}{\env}{e} \subseteq \efoot{\heap}{\env}{p(\multiple{e})}$. Therefore $\eval{\heap'}{\env}{e}{v}$ by lemma \ref{lem:eval-heap-efoot-unchanged} and $\frm{\heap'}{\perms}{\env}{e}$ by \ref{lem:frm-heap-efoot-unchanged}.

      Also, $\universal{\pair{\ell}{f} \in \efoot{\heap}{[\multiple{x \mapsto v}]}{\fpred(p)}}{\heap'(\ell, f) = \heap(\ell, f)}$, thus by induction \\
      $\efrm{\heap'}{\perms}{[\multiple{x \mapsto v}]}{\fpred(p)}$.

      Therefore $\efrm{\heap'}{\perms}{\env}{p(\multiple{e})}$ by \refrule{EFramePredicate}.

    \case\label{case:efrm-heap-efoot-unchanged-ifa} \refrule{EFrameConditionalA} -- $\efrm{\heap}{\perms}{\env}{\sif{e}{\phi_1}{\phi_2}}$:

      By inversion of \refrule{EFrameConditionalA} $\eval{\heap}{\env}{e}{\ktrue}$, $\frm{\heap}{\perms}{\env}{e}$, and $\efrm{\heap}{\perms}{\env}{\phi_1}$.

      Now $\efoot{\heap}{\env}{\sif{e}{\phi_1}{\phi_2}} = \efoot{\heap}{\env}{e} \cup \efoot{\heap}{\env}{\phi_1}$.
      
      Then $\universal{\pair{\ell}{f} \in \efoot{\heap}{\env}{e}}{\heap'(\ell, f) = \heap(\ell, f)}$ since $\efoot{\heap}{\env}{e} \subseteq \efoot{\heap}{\env}{\sif{e}{\phi_1}{\phi_2}}$. Therefore $\eval{\heap'}{\env}{e}{\ktrue}$ by lemma \ref{lem:eval-heap-efoot-unchanged} and $\frm{\heap'}{\perms}{\env}{e}$ by lemma \ref{lem:frm-heap-efoot-unchanged}.

      Also, Then $\universal{\pair{\ell}{f} \in \efoot{\heap}{\env}{\phi_1}}{\heap'(\ell, f) = \heap(\ell, f)}$ since \\
      $\efoot{\heap}{\env}{\phi_1} \subseteq \efoot{\heap}{\env}{\sif{e}{\phi_1}{\phi_2}}$. Therefore $\efrm{\heap'}{\perms}{\env}{\phi_1}$ by induction.

      Therefore $\efrm{\heap'}{\perms}{\env}{\sif{e}{\phi_1}{\phi_2}}$ by \refrule{EFrameConditionalA}.

    \case \refrule{EFrameConditionalB} -- $\efrm{\heap}{\perms}{\env}{\sif{e}{\phi_1}{\phi_2}}$: Similar to case \ref{case:efrm-heap-efoot-unchanged-ifa}.

    \case \refrule{EFrameAcc} -- $\efrm{\heap}{\perms}{\env}{\kacc(e.f)}$:

      By inversion of \refrule{EFrameAcc} $\frm{\heap}{\perms}{\env}{e}$.

      Also, $\efoot{\heap}{\env}{e} \subseteq \efoot{\heap}{\env}{\kacc(e.f)}$ by definition, thus $\universal{\pair{\ell}{f} \in \efoot{\heap}{\env}{e}}{\heap'(\ell, f) = \heap(\ell, f)}$. Therefore $\frm{\heap'}{\perms}{\env}{e}$ by lemma \ref{lem:frm-heap-efoot-unchanged}.

      Thus $\efrm{\heap'}{\perms}{\env}{\kacc(e.f)}$ by \refrule{EFrameAcc}.

  \end{enumcases}
\end{proof}

\begin{lemma}\label{lem:assert-heap-efoot-unchanged}
  If $\assertion{\heap}{\perms}{\env}{\gform}$ and $\universal{\pair{\ell}{f} \in \efoot{\heap}{\env}{\gform}}{\heap'(\ell, f) = \heap(\ell, f)}$ then $\assertion{\heap'}{\perms}{\env}{\gform}$.
\end{lemma}
\begin{proof}
  By induction on $\assertion{\heap}{\perms}{\env}{\gform}$.

  \begin{enumcases}
    \case \refrule{AssertImprecise} -- $\assertion{\heap}{\perms}{\env}{\simprecise{\phi}}$:

      $\efoot{\heap}{\env}{\simprecise{\phi}} = \efoot{\heap}{\env}{\phi}$, thus $\universal{\pair{\ell}{f} \in \efoot{\heap}{\env}{\phi}}{\heap'(\ell, f) = \heap(\ell, f)}$.
      
      Also, $\efrm{\heap}{\perms}{\env}{\phi}$ by \refrule{AssertImprecise}, therefore $\efrm{\heap'}{\perms}{\env}{\phi}$ by lemma \ref{lem:efrm-heap-efoot-unchanged}.

      Finally, $\assertion{\heap}{\perms}{\env}{\phi}$ by \refrule{AssertImprecise}, therefore $\assertion{\heap'}{\perms}{\env}{\phi}$ by induction.

      Now $\assertion{\heap'}{\perms}{\env}{\simprecise{\phi}}$ by \refrule{AssertImprecise}.

    \case \refrule{AssertValue} -- $\assertion{\heap}{\perms}{\env}{e}$:
    
      By \refrule{AssertValue} $\eval{\heap}{\env}{e}{\ktrue}$ and $\universal{\pair{\ell}{f} \in \efoot{\heap}{\env}{e}}{\heap'(\ell, f) = \heap(\ell, f)}$ by assumptions, therefore $\eval{\heap'}{\env}{e}{\ktrue}$ by lemma \ref{lem:eval-heap-efoot-unchanged}. Therefore $\assertion{\heap'}{\perms}{\env}{e}$ by \refrule{AssertValue}.

    \case\label{case:assert-heap-efoot-unchanged-ifa} \refrule{AssertIfA} -- $\assertion{\heap}{\perms}{\env}{\sif{e}{\phi_1}{\phi_2}}$:

      By \refrule{AssertIfA} $\eval{\heap}{\env}{e}{\ktrue}$, thus $\efoot{\heap}{\env}{\sif{e}{\phi_1}{\phi_2}} = \efoot{\heap}{\env}{e} \cup \efoot{\heap}{\env}{\phi_1}$.

      Now $\universal{\pair{\ell}{f} \in \efoot{\heap}{\env}{e}}{\heap'(\ell, f) = \heap(\ell, f)}$ and $\eval{\heap}{\env}{e}{\ktrue}$ thus $\eval{\heap'}{\env}{e}{\ktrue}$ by lemma \ref{lem:eval-heap-efoot-unchanged}.

      Also $\universal{\pair{\ell}{f} \in \efoot{\heap}{\env}{\phi_1}}{\heap'(\ell, f) = \heap(\ell, f)}$, and $\assertion{\heap}{\perms}{\env}{\phi_1}$ by \refrule{AssertIfA}, therefore $\assertion{\heap'}{\perms}{\env}{\phi_1}$ by induction.

      Therefore $\assertion{\heap'}{\perms}{\env}{\sif{e}{\phi_1}{\phi_2}}$ by \refrule{AssertIfA}.

    \case \refrule{AssertIfB} -- $\assertion{\heap}{\perms}{\env}{\sif{e}{\phi_1}{\phi_2}}$: Similar to case \ref{case:assert-heap-efoot-unchanged-ifa}.

    \case \refrule{AssertAcc} -- $\assertion{\heap}{\perms}{\env}{\kacc(e.f)}$:

      By inversion of \refrule{AssertAcc} $\eval{\heap}{\env}{e}{\ell}$. Thus $\efoot{\heap}{\env}{\kacc(e.f)} = \efoot{\heap}{\env}{e}; \pair{\ell}{f}$.
      
      Now $\universal{\pair{\ell}{f} \in \efoot{\heap}{\env}{e}}{\heap'(\ell, f) = \heap(\ell, f)}$ thus $\eval{\heap'}{\env}{e}{\ell}$ by lemma \ref{lem:eval-heap-efoot-unchanged}.

      Also, $\pair{\ell}{f} \in \perms$ by inversion of \refrule{AssertAcc}. Therefore $\assertion{\heap'}{\perms}{\env}{\kacc(e.f)}$ by \refrule{AssertAcc}.

    \case \refrule{AssertConjunction} -- $\assertion{\heap}{\perms}{\env}{\phi_1 * \phi_2}$:

      By inversion of \refrule{AssertConjunction} $\assertion{\heap}{\perms_1}{\env}{\phi_1}$ and $\assertion{\heap}{\perms_2}{\env}{\phi_2}$ for some $\perms_1, \perms_2$ such that $\perms_1 \cup \perms_2 \subseteq \perms$ and $\perms_1 \cap \perms_2 = \emptyset$.

      Also, $\efoot{\heap}{\env}{\phi_1 * \phi_2} = \efoot{\heap}{\env}{\phi_1} \cup \efoot{\heap}{\env}{\phi_2}$. Therefore $\universal{\pair{\ell}{f} \in \efoot{\heap}{\env}{\phi_1}}{\heap'(\ell, f) = \heap(\ell, f)}$ and thus $\assertion{\heap'}{\perms_1}{\env}{\phi_1}$ by induction. Similarly, $\assertion{\heap'}{\perms_2}{\env}{\phi_2}$.

      Now $\assertion{\heap'}{\perms}{\env}{\phi_1 * \phi_2}$ by \refrule{AssertConjunction}.

    \case \refrule{AssertPredicate} -- $\assertion{\heap}{\perms}{\env}{p(\multiple{e})}$:

      By inversion of \refrule{AssertPredicate} $\multiple{\eval{\heap}{\env}{e}{v}}$ and $\assertion{\heap}{\perms}{[\multiple{x \mapsto v}]}{\fpred(p)}$ where $\multiple{x} = \fpredparams(p)$.

      Also, $\efoot{\heap}{\perms}{p(\multiple{e})} = \efoot{\heap}{[\multiple{x \mapsto v}]}{\fpred(p)} \cup \bigcup \multiple{\efoot{\heap}{\env}{e}}$.

      Now for each $e$ and corresponding $x$, $\universal{\pair{\ell}{f} \in \efoot{\heap}{\env}{e}}{\heap'(\ell, f) = \heap(\ell, f)}$. Therefore $\eval{\heap'}{\perms}{e}{v}$ by lemma \ref{lem:eval-heap-efoot-unchanged}.

      Now $\universal{\pair{\ell}{f} \in \efoot{\heap}{[\multiple{x \mapsto v}]}{\fpred(p)}}{\heap'(\ell, f) = \heap(\ell, f)}$. Therefore $\assertion{\heap'}{\perms}{[\multiple{x \mapsto v}]}{\fpred(p)}$ by induction.

      Therefore $\assertion{\heap'}{\perms}{\env}{p(\multiple{e})}$ by \refrule{AssertPredicate}.
  \end{enumcases}
\end{proof}

\begin{lemma}\label{lem:assert-heap-perms-unchanged}
  If $\assertion{\heap}{\perms}{\env}{\gform}$ for some specification $\gform$ and $\universal{\pair{\ell}{f} \in \perms}{\heap'(\ell, f) = \heap(\ell, f)}$ then $\assertion{\heap'}{\perms}{\env}{\gform}$.
\end{lemma}
\begin{proof}
  By assumptions and lemma \ref{lem:efoot-subset-spec} $\efoot{\heap}{\env}{\gform} \subseteq \perms$. Therefore $\universal{\pair{\ell}{f} \in \efoot{\heap}{\env}{\gform}}{\heap'(\ell, f) = \heap(\ell, f)}$ since $\pair{\ell}{f} \in \perms \implies \pair{\ell}{f} \in \efoot{\heap}{\env}{\gform}$.

  Therefore $\assertion{\heap'}{\efoot{\heap}{\env}{\gform}}{\env}{\gform}$ by lemma \ref{lem:assert-heap-efoot-unchanged}. Then $\assertion{\heap'}{\perms}{\env}{\gform}$ by lemma \ref{lem:assert-monotonicity} since $\efoot{\heap}{\env}{\gform} \subseteq \perms$.
\end{proof}

\begin{lemma}\label{lem:pres-modify-heap}
  If $\simstate{V}{\sstate}{\heap}{\perms}{\env}$ and $\universal{\pair{\ell}{f} \in \perms}{\heap'(\ell, f) = \heap(\ell, f)}$ then $\simstate{V}{\sstate}{\heap'}{\perms}{\env}$.
\end{lemma}
\begin{proof}
  From the assumptions, clearly $\simenv{V}{\senv(\sstate)}{\env}$, thus it suffices to show that $\simheap{V}{\sheap(\sstate)}{\heap'}{\perms}$ and $\simheap{V}{\oheap(\sstate)}{\heap'}{\perms}$.

  Let $\sheap = \sheap(\sstate)$. From assumptions, clearly $\simheap{V}{\sheap}{\heap}{\perms}$. Therefore
  \begin{gather}
    \universal{\triple{f}{t}{t'} \in \sheap}{\heap(V(t), f) = V(t')} \label{eq:pres-modify-heap-1}\\
    \universal{\triple{f}{t}{t'} \in \sheap}{\pair{V(t)}{f} \in \perms} \label{eq:pres-modify-heap-2}\\
    \universal{\pair{p}{\multiple{t}} \in \sheap}{\assertion{\heap}{\perms}{[\multiple{x \mapsto V(t)}]}{\fpred(p)}} \label{eq:pres-modify-heap-3} \\
    \universal{h_1, h_2 \in \sheap^2}{h_1 \ne h_2 \implies \vfoot{V}{\heap}{h_1} \cap \vfoot{V}{\heap}{h_2} = \emptyset} \label{eq:pres-modify-heap-4}
  \end{gather}

  Let $\triple{f}{t}{t'}$ be an arbitrary field value in $\sheap$. Then by \eqref{eq:pres-modify-heap-1} $\heap(V(t), f) = V(t')$. Also, by \eqref{eq:pres-modify-heap-2} $\pair{V(t)}{f} \in \perms$, thus by our initial assumptions $\heap'(V(t), f) = \heap(V(t), f) = V(t')$. Therefore
  \begin{equation}\label{eq:pres-modify-heap-5}
    \universal{\triple{f}{t}{t'} \in \sheap}{\heap'(V(t), f) = V(t')}.
  \end{equation}

  let $\pair{p}{\multiple{t}}$ be an arbitrary predicate instance in $\sheap$. Then by \eqref{eq:pres-modify-heap-3} $\assertion{\heap}{\perms}{[\multiple{x \mapsto V(t)}]}{\fpred(p)}$. Since $\fpred(p)$ is a specification, and $\universal{\pair{\ell}{f} \in \perms}{\heap'(\ell, f) = \heap(\ell, f)}$, then \\
  $\assertion{\heap'}{\perms}{[\multiple{x \mapsto V'(t)}]}{\fpred(p)}$ by lemma \ref{lem:assert-heap-perms-unchanged}. Therefore
  \begin{equation}\label{eq:pres-modify-heap-6}
    \universal{\pair{p}{\multiple{t}} \in \sheap}{\assertion{\heap'}{\perms}{[\multiple{x \mapsto V'(t)}]}{\fpred(p)}}.
  \end{equation}

  Therefore by \eqref{eq:pres-modify-heap-5}, \eqref{eq:pres-modify-heap-2}, \eqref{eq:pres-modify-heap-6}, and \eqref{eq:pres-modify-heap-4}, $\simheap{V}{\sheap}{\heap'}{\perms}$.

  Let $\oheap = \oheap(\sstate)$. From assumptions, clearly $\simheap{V}{\oheap}{\heap}{\perms}$. Therefore
  \begin{gather}
    \universal{\triple{f}{t}{t'} \in \oheap}{\heap(V(t), f) = V(t')} \label{eq:pres-modify-heap-7} \\
    \universal{\triple{f}{t}{t'} \in \oheap}{\pair{V(t)}{f} \in \perms} \label{eq:pres-modify-heap-8}
  \end{gather}

  Let $\triple{f}{t}{t'}$ be an arbitrary field value in $\oheap$. Then by \eqref{eq:pres-modify-heap-7} $\heap(V(t), f) = V(t')$. Also, by \eqref{eq:pres-modify-heap-8} $\pair{V(t)}{f} \in \perms$, thus by our initial assumptions $\heap'(V(t), f) = \heap(V(t), f) = V(t')$. Therefore
  \begin{equation}\label{eq:pres-modify-heap-9}
    \universal{\triple{f}{t}{t'} \in \oheap}{\heap'(V(t), f) = V(t')}.
  \end{equation}

  Therefore by \eqref{eq:pres-modify-heap-9} and \eqref{eq:pres-modify-heap-8}, $\simheap{V}{\oheap}{\heap'}{\perms}$.
\end{proof}

\begin{lemma}\label{lem:pres-heap-change-partial}
  If $\pair{\heap}{\triple{\perms}{\env}{s} \cdot \stack}$ is a well-formed state, the partial state $\pair{\heap}{\stack}$ is validated by $\vstate$ and $V$, and $\heap' = \heap[\pair{\ell}{f} \mapsto v]$ for some $\ell$, $f$, and $v$ such that $\pair{\ell}{f} \in \perms$ or $\ell$ is a fresh value unused in $\stack$, then the partial state $\pair{\heap'}{\stack}$ is validated by $\vstate$ and $V$.
\end{lemma}
\begin{proof}
  For some $n$, let $\stack = \triple{\perms_n}{\env_n}{s_n} \cdot \ldots \cdot \triple{\perms_1}{\env_1}{s_1} \cdot \nilsym$.
  
  If $\pair{\ell}{f} \in \perms$: $\pair{\heap}{\triple{\perms}{\env}{s} \cdot \stack}$ is a well-formed state, $\perms, \perms_n, \cdots, \perms_1$ are all disjoint. Thus, since $\pair{\ell}{f} \in \perms$, for all $1 \le i \le n$, $\pair{\ell}{f} \notin \perms_i$.

  If $\ell$ is fresh: Then for all $1 \le i \le n$, $\pair{\ell}{f} \notin \perms_i$ since $\ell$ is not referenced in $\stack$.

  Therefore, for all $1 \le i \le n$, $\pair{\ell}{f} \notin \perms_i$.

  Let $\stack_0 = \nilsym$, and for all $1 \le i \le n$ let $\stack_i = \triple{\perms_i}{\env_i}{s_i} \cdot \stack_{i-1}$.

  We prove by induction that, if $0 \le i \le n$ and $\pair{\heap}{\stack_i}$ is validated by some $\vstate$ and $V$, the partial state $\pair{\heap'}{\stack_i}$ is validated by $\vstate$ and $V$. This is sufficient to prove the main result, since $\stack_n = \stack$ and $\pair{\heap}{\stack}$ is validated by $\vstate$ and $V$.

  Suppose $0 \le i \le n$ and $\pair{\heap}{\stack_i}$ is validated by some $\vstate$ and $V$.

  \textit{Case \ref{def:partial-valid-nil}:} Then $\stack_i = \nilsym$ (and thus $i = 0$), and trivially $\pair{\heap'}{\nilsym}$ is validated by $\vstate$ and $V$.

  \textit{Case \ref{def:partial-valid-call}:} Then
  $\stack_i = \triple{\env_i}{\perms_i}{\sseq{y \kassign m(e_1, \cdots, e_k)}{s}} \cdot \stack_{i-1}$ for some $y$, $m$, $k$, $e_1, \cdots, e_k$, $s$, and there exists some $\vstate'$, $V'$, $x_1, \cdots, x_k$, $t_1, \cdots, t_k$, $\sstate_0, \cdots, \sstate_k$ and $\sstate'$ such that:
  \begin{gather*}
    \text{The partial state $\pair{\heap}{\stack_{i-1}}$ is validated by $\vstate'$ and $V'$}, \\
    \vstate' \text{ is reachable from $\prog$ with valuation $V'$}, \quad s(\vstate') = s(\stack_i) \\
    x_1, \cdots, x_k = \fparams(m), \\
    \sstate_0 = \sstate(\vstate'), \quad \seval{\sstate_0}{e_1}{t_1}{\sstate_1}{\_}, \quad\cdots,\quad \seval{\sstate_{k-1}}{e_k}{t_k}{\sstate_k}{\_}, \\
    \universal{1 \le i \le k}{V(\senv(\vstate_i)(x_i)) = V'(t_i)}, \\
    \scons{\sstate_k}{\fpre(m)}{\sstate'}{\_}, \quad \simstate{V'}{\sstate'[\senv = \senv(\sstate_0)]}{\heap}{\perms_i}{\env_i}, \quad\text{and} \\
    \gform(\vstate_i) = \fpost(m).
  \end{gather*}

  By induction we can assume that if $0 \le i - 1 \le n$ and $\pair{\heap}{\stack_{i-1}}$ is validated by some $\vstate$ and $V$, the partial state $\pair{\heap'}{\stack_{i-1}}$ is validated by $\vstate$ and $V$.

  Since $\stack_i \ne \nilsym$, $1 \le i \le n$, thus $0 \le i - 1 \le n$. Also, in this case the partial state $\pair{\heap}{\stack_{i-1}}$ is validated by $\vstate'$ and $V'$. Therefore we can conclude that
  $$\text{The partial state $\pair{\heap'}{\stack_{i-1}}$ is validated by $\vstate'$ and $V'$.}$$
  Also, immediately from above,
  \begin{gather*}
    \vstate' \text{ is reachable from $\prog$ with valuation $V'$}, \quad s(\vstate') = s(\stack_i) \\
    x_1, \cdots, x_k = \fparams(m), \\
    \sstate_0 = \sstate(\vstate'), \quad \seval{\sstate_0}{e_1}{t_1}{\sstate_1}{\_}, \quad\cdots,\quad \seval{\sstate_{k-1}}{e_k}{t_k}{\sstate_k}{\_}, \\
    \universal{1 \le i \le k}{V(\senv(\vstate_i)(x_i)) = V'(t_i)}, \\
    \scons{\sstate_k}{\fpre(m)}{\sstate'}{\_}, \quad\text{and} \\
    \gform(\vstate_i) = \fpost(m).
  \end{gather*}
  It remains to be shown that $\simstate{V'}{\sstate'[\senv = \senv(\sstate_0)]}{\heap'}{\perms_i}{\env_i}$. But as noted before, $\pair{\ell}{f} \notin \perms_i$ (since $1 \le i \le n$ in this case). Thus $\universal{\pair{\ell'}{f'} \in \perms_i}{\heap'(\ell', f') = \heap(\ell', f')}$. Therefore by lemma \ref{lem:pres-modify-heap},
  $$\simstate{V'}{\sstate'[\senv = \senv(\sstate_0)]}{\heap'}{\perms_i}{\env_i}.$$

  Therefore $\pair{\heap'}{\stack_i}$ is validated by $\vstate$ and $V$.

  \textit{Case \ref{def:partial-valid-while}:}
  Then $\stack_i = \triple{\env_i}{\perms_i}{\sseq{\swhile{e}{\gform}{s}}{s'}} \cdot \stack_{i-1}$ for some $e$, $\gform$, $s$, and $s'$, and there exists some $\vstate'$, $V'$, and $\sstate'$ such that:
  \begin{gather*}
    \text{The partial state $\pair{\heap}{\stack_{i-1}}$ is validated by $\vstate'$ and $V'$} \\
    \vstate' \text{ is reachable from $\prog$ with valuation $V'$} \quad s(\vstate') = s(\stack_i) \\
    \scons{\sstate}{\gform}{\sstate'}{\_}, \quad
    \simstate{V'}{\sstate'}{\heap}{\perms_i}{\env_i}, \quad\text{and}\quad
    \gform(\vstate_i) = \gform
  \end{gather*}

  By induction we can assume that if $0 \le i - 1 \le n$ and $\pair{\heap}{\stack_{i-1}}$ is validated by some $\vstate$ and $V$, the partial state $\pair{\heap'}{\stack_{i-1}}$ is validated by $\vstate$ and $V$.

  Since $\stack_i \ne \nilsym$, $1 \le i \le n$, thus $0 \le i - 1 \le n$. Also, in this case the partial state $\pair{\heap}{\stack_{i-1}}$ is validated by $\vstate'$ and $V'$. Therefore we can conclude that
  $$\text{The partial state $\pair{\heap'}{\stack_{i-1}}$ is validated by $\vstate'$ and $V'$.}$$
  Also, immediately from above,
  \begin{gather*}
    \vstate' \text{ is reachable from $\prog$ with valuation $V'$} \quad s(\vstate') = s(\stack_i) \\
    \scons{\sstate}{\gform}{\sstate'}{\_}, \quad\text{and}\quad
    \gform(\vstate_i) = \gform
  \end{gather*}
  It remains to be shown that $\simstate{V'}{\sstate'}{\heap'}{\perms_i}{\env_i}$. But as noted before, $\pair{\ell}{f} \notin \perms_i$ (since $1 \le i \le n$ in this case). Thus $\universal{\pair{\ell'}{f'} \in \perms_i}{\heap'(\ell', f') = \heap(\ell', f')}$. Therefore by lemma \ref{lem:pres-modify-heap},
  $$\simstate{V'}{\sstate'}{\heap'}{\perms_i}{\env_i}.$$

  Therefore $\pair{\heap'}{\stack_i}$ is validated by $\vstate$ and $V$.

\end{proof}

\begin{lemma}\label{lem:rem-subset}
  If $\simstate{V}{\sstate}{\heap}{\perms}{\env}$ then $\vfoot{V}{\heap}{\frem(\sstate, \gform)} \subseteq \perms$.
\end{lemma}
\begin{proof}
  If $\gform$ is completely precise, then
  $$\vfoot{V}{\heap}{\frem(\sstate, \gform)} = \vfoot{V}{\heap}{\emptyset} = \emptyset \subseteq \perms.$$
  
  Otherwise,
  \begin{align*}
    \vfoot{V}{\heap}{\frem(\sstate, \gform)}
      &= \vfoot{V}{\heap}{\set{ \pair{f}{t} : \triple{f}{t}{t'} \in \sheap(\sstate) \cup \oheap(\sstate) } ~\cup \\
        &\hspace{2em} \set{\pair{p}{\multiple{t}} : \pair{p}{\multiple{t}} \in \sheap(\sstate) }} \\
      &= \set{ \pair{V(t)}{f} : \triple{f}{t}{t'} \in \sheap(\sstate) \cup \oheap(\sstate) } ~\cup \\
        &\hspace{2em} \bigcup_{\pair{p}{\multiple{t}} \in \sheap(\sstate)}{\efoot{\heap}{[\multiple{x \mapsto V(t)}]}{\fpred(p)}}
  \end{align*}

  For any $\triple{f}{t}{t'} \in \sheap(\sstate) \cup \oheap(\sstate)$, $\pair{V(t)}{f} \in \perms$ since $\simheap{V}{\sheap(\sstate)}{\heap}{\perms}$ and $\simheap{V}{\oheap(\sstate)}{\heap}{\perms}$. Therefore
  $$\set{ \pair{V(t)}{f} : \triple{f}{t}{t'} \in \sheap(\sstate) \cup \oheap(\sstate) } \subseteq \perms.$$

  Also, for any $\pair{p}{\multiple{t}} \in \sheap(\sstate)$, $\assertion{\heap}{\perms}{[\multiple{x \mapsto t}]}{\fpred(p)}$ where $\multiple{x} = \fpredparams(p)$; therefore $\efoot{\heap}{[\multiple{x \mapsto V(t)}]}{\fpred(p)} \subseteq \perms$ by \ref{lem:efoot-subset-spec}. Thus
  $$\bigcup_{\pair{p}{\multiple{t}} \in \sheap(\sstate)}{\efoot{\heap}{[\multiple{x \mapsto V(t)}]}{\fpred(p)}} \subseteq \perms.$$

  Therefore, by the previous calculations, $\vfoot{V}{\heap}{\frem(\sstate, \gform)} \subseteq \perms$.
\end{proof}

\begin{lemma}\label{lem:rem-simstate}
  Let $\gform$ be some specification, $\xperms = \vfoot{V}{\heap}{\frem(\sstate, \gform)}$ and $\perms' = \foot{\heap}{\perms \setminus \xperms}{\env}{\gform}$.

  If $\simstate{V}{\sstate}{\heap}{\perms \setminus \efoot{\heap}{\env}{\gform}}{\env}$ and $\assertion{\heap}{\perms \setminus \xperms}{\env}{\gform}$, then $\simstate{V}{\sstate}{\heap}{\perms \setminus \perms'}{\env}$.
\end{lemma}
\begin{proof}
  If $\gform$ is completely precise, then
  \begin{align*}
    \xperms &= \vfoot{V}{\heap}{\frem(\sstate, \gform)} \\
      &= \vfoot{V}{\heap}{\emptyset} \\
      &= \emptyset \\
    \perms' &= \foot{\heap}{\perms \setminus \xperms}{\env}{\gform} \\
      &= \efoot{\heap}{\env}{\gform}.
  \end{align*}
  And thus
  $$\simstate{V}{\sstate}{\heap}{\perms \setminus \efoot{\heap}{\env}{\gform}}{\env} \implies \simstate{V}{\sstate}{\heap}{\perms \setminus \perms'}{\env}.$$
  
  Otherwise, $\gform$ is not completely precise. Now,
  \begin{align*}
    \xperms &= \vfoot{V}{\heap}{\frem(\sstate, \gform)} \\
      &= \vfoot{V}{\heap}{\set{ \pair{f}{t} : \triple{f}{t}{t'} \in \sheap(\sstate) \cup \oheap(\sstate) } \cup \set{\pair{p}{\multiple{t}} : \pair{p}{\multiple{t}} \in \sheap(\sstate) }} \\
      &= \set{ \pair{V(t)}{f} : \triple{f}{t}{t'} \in \sheap(\sstate) \cup \oheap(\sstate) } ~\cup \\
        &\hspace{2em} \bigcup_{\pair{p}{\multiple{t}} \in \sheap(\sstate)}{\efoot{\heap}{[\multiple{t \mapsto V(t)}]}{\fpred(p)}}
  \end{align*}

  Let $\perms^* = \perms \setminus \efoot{\heap}{\env}{\gform}$. Since $\simheap{V}{\sheap(\sstate)}{\heap}{\perms^*}{\env}$,
  \begin{gather}
    \universal{\triple{f}{t}{t'} \in \sheap(\sstate)}{\heap(V(t), f) = V(t')} \label{eq:rem-simstate-value-assump}\\
    \universal{\pair{p}{\multiple{t}} \in \sheap(\sstate)}{\assertion{\heap}{\perms^*}{[\multiple{x \mapsto V(t)}]}{\fpred(p)}} \label{eq:rem-simstate-pred-assump} \\
    \universal{h_1, h_2 \in \sheap}{h_1 \ne h_2 \implies \vfoot{V}{\heap}{h_1} \cap \vfoot{V}{\heap}{h_2} = \emptyset} \label{eq:rem-simstate-disjoint-assump}
  \end{gather}
  where $\multiple{x} = \fpredparams(p)$.

  Note that by our calculation of $\xperms$,
  \begin{equation}\label{eq:rem-simstate-perms}
    \universal{\triple{f}{t}{t'} \in \sheap(\sstate)}{\pair{V(t)}{f} \in \xperms}.
  \end{equation}
  Also, by definition $\efoot{\heap}{[\multiple{x \mapsto V(t)}]}{\fpred(p)} \subseteq \xperms$ for each $\pair{p}{\multiple{t}} \in \sheap(\sstate)$ when \\
  $\multiple{x} = \fpredparams(p)$, thus by \eqref{eq:rem-simstate-pred-assump} and lemma \ref{lem:assert-efoot-subset}
  \begin{equation}\label{eq:rem-simstate-pred}
    \universal{\pair{p}{\multiple{t}} \in \sheap(\sstate)}{\assertion{\heap}{\xperms}{[\multiple{x \mapsto V(t)}]}{\fpred(p)}}
  \end{equation}
  where $\multiple{x} = \fpredparams(p)$.

  Therefore, by \eqref{eq:rem-simstate-value-assump}, \eqref{eq:rem-simstate-perms}, \eqref{eq:rem-simstate-pred}, and \eqref{eq:rem-simstate-disjoint-assump},
  $$\simheap{V}{\sheap(\sstate)}{\heap}{\xperms}.$$

  Since $\simheap{V}{\oheap(\sstate)}{\heap}{\perms^*}{\env}$,
  $$\universal{\triple{f}{t}{t'} \in \oheap(\sstate)}{\heap(V(t), f) = V(t')}.$$
  By our calculation of $\xperms$,
  $$\universal{\triple{f}{t}{t'} \in \oheap(\sstate)}{\pair{V(t)}{f} \in \xperms}.$$
  Therefore
  $$\simheap{V}{\oheap(\sstate)}{\heap}{\xperms}.$$

  Also, since $\simstate{V}{\sstate}{\heap}{\perms \setminus \efoot{\heap}{\env}{\gform}}{\env}$,
  $$\simenv{V}{\senv(\sstate)}{\env} \quad\text{and}\quad V(\pc(\sstate)) = \ktrue.$$
  Therefore
  $$\simstate{V}{\sstate}{\heap}{\xperms}{\env}.$$

  Now, since $\assertion{\heap}{\perms \setminus \xperms}{\env}{\gform}$, by lemma \ref{lem:foot-subset-spec} $\foot{\heap}{\perms \setminus \xperms}{\env}{\gform} \subseteq \perms \setminus \xperms$. Thus
  \begin{gather*}
    \perms' = \foot{\heap}{\perms \setminus \xperms}{\env}{\gform}
      \subseteq \perms \setminus \xperms \\
    \perms \setminus \perms' \supseteq \perms \setminus (\perms \setminus \xperms)
      = \perms \cap \xperms
  \end{gather*}

  Finally, $\xperms \subseteq \perms \setminus \efoot{\heap}{\env}{\gform} \subseteq \perms$ by lemma \ref{lem:rem-subset}, thus $\perms \cap \xperms = \xperms$ and then $\perms \setminus \perms' \supseteq \xperms$.

  Therefore by lemma \ref{lem:simstate-monotonicity}
  $$\simstate{V}{\sstate}{\heap}{\perms \setminus \perms'}{\env}.$$
\end{proof}

\begin{lemma}\label{lem:dexec-preservation}
  Let $\pair{\heap}{\stack}$ be some dynamic state validated by $\vstate$ and valuation $V$ for some program $\prog$. If $\sguard{\vstate}{\sstate'}{\scheck}{\sperms}$ with $V' = V[\sguard{\vstate}{\sstate'}{\scheck}{\sperms} \mid \heap]$, $V'(\pc(\sstate')) = \ktrue$, $\rtassert{V'}{\heap}{\perms(\stack)}{\scheck}$, and $\dexec{\heap}{\stack}{\vfoot{V'}{\heap}{\sperms}}{\heap'}{\stack'}$
  then $\Gamma'$ is a valid state.

  Note: This is a simplification of theorem \ref{thm:dtrans-preservation}, restricted to the particular case of a normal program step (in contrast to a step from $\initsym$ or to $\finalsym$).
\end{lemma}

\begin{proof}
  We proceed by cases on $\dexec{\heap}{\stack}{\vfoot{V'}{\heap}{\sperms}}{\heap'}{\stack'}$.

  \begin{enumcases}
    \case \refrule{ExecSeq}:
      We have
      $$\dexec{\heap}{\triple{\perms}{\env}{\sseq{\kskip}{s}} \cdot \stack^*}{\vfoot{V'}{\heap}{\sperms}}{\heap}{\triple{\perms}{\env}{s} \cdot \stack^*}.$$

      Since the initial state is validated by $\vstate$ and $V$, $\vstate = \triple{\sstate}{\sseq{\kskip}{s}}{\gform}$ for some $\sstate, \gform$, where $\simstate{V}{\sstate}{\heap}{\perms}{\env}$.
      
      Now $\sexec{\sstate(\vstate)}{\sseq{\kskip}{s}}{s}{\sstate}$ by \refrule{SExecSeq}. Therefore $\strans{\prog}{\vstate}{\vstate'}$ by \refrule{SVerifyStep} where $\vstate' = \triple{\sstate}{s}{\gform}$.

      Now, by lemma \ref{lem:preservation-heap-env-unchanged}, it suffices to show that $\vstate'$ corresponds to $\pair{\heap}{\triple{\perms}{\env}{s} \cdot \stack^*}$ (with valuation $V$).

      Since $\sstate(\vstate') = \sstate$, we have $\simstate{V}{\sstate(\vstate')}{\heap}{\perms}{\env}$, and $s(\vstate') = s$ by definition. Therefore $\vstate'$ corresponds to $\pair{\heap}{\triple{\perms}{\env}{s} \cdot \stack^*}$, which completes the proof.

    \case \refrule{ExecAssign}:
      We have
      \begin{gather*}
        \dexec{\heap}{\triple{\perms}{\env}{\sseq{x = e}{s}} \cdot \stack^*}{\vfoot{V'}{\heap}{\sperms}}{\heap}{\triple{\perms}{\env[x \mapsto v]}{s} \cdot \stack^*} \\
        \text{where}\quad \eval{\heap}{\env}{e}{v}
      \end{gather*}

      Since the initial state is validated by $\vstate$, $\vstate = \triple{\sstate}{\sseq{x = e}{s}}{\gform}$ for some $\sstate, \gform$ where $\simstate{V}{\sstate}{\heap}{\perms}{\env}$.

      The only guard that applies is \refrule{SGuardAssign}, so we have, for some $\sstate', t$:
      \begin{gather*}
        \sguard{\triple{\sstate}{\sseq{x = e}{s}}{\gform}}{\sstate}{\scheck}{\sperms} \\
        \text{where }
        \seval{\sstate}{e}{t}{\sstate'}{\scheck} \quad\text{and (by assumptions) }\rtassert{V'}{\heap}{\env}{\scheck}.
      \end{gather*}
      where $V'$ is the corresponding valuation, extending $V$.

      Let $\sstate'' = \sstate[\senv = \senv(\sstate)[x \mapsto t]]$, then $\sexec{\sstate}{\sseq{x = e}{s}}{s}{\sstate''}$ by \refrule{SExecAssign}.
      
      Let $\vstate' = \triple{\sstate''}{s}{\gform}$. Then $\strans{\prog}{\vstate}{\vstate'}$ by \refrule{SVerifyStep}, thus $\vstate'$ is reachable from $\prog$.

      We want to show that $\pair{\heap}{\triple{\perms}{\env[x \mapsto v]}{s} \cdot \stack^*}$ is validated by $\vstate'$ with valuation $V'$.

      \textit{Part \ref{def:state-valid-reachable}:} As shown above, $\vstate'$ is reachable from $\prog$ with valuation $V'$.

      \textit{Part \ref{def:state-valid-correspond}:}
      By lemma \ref{lem:seval-soundness} $\eval{\heap}{\env}{e}{V'(t)}$, therefore $V'(t) = v$. Also, $\simenv{V}{\senv(\sstate)}{\env}$, thus $\simenv{V'}{\senv(\sstate)[x \mapsto V'(t)]}{\env[x \mapsto v]}$. Rewriting using definitions, we get $\simenv{V'}{\senv(\sstate'')}{\env[x \mapsto v]}$.

      Also by lemma \ref{lem:seval-soundness}, $\simstate{V'}{\sstate'}{\heap}{\perms}{\env}$. Thus, since $\senv$ is the only component changed from $\sstate'$ to $\sstate''$ and $\simenv{V'}{\senv(\sstate'')}{\env[x \mapsto v]}$, $\simstate{V'}{\sstate''}{\heap}{\perms}{\env[x \mapsto v]}$.

      Also, by definition $\vstate' = \triple{\sstate''}{s}{\gform}$.

      Therefore $\vstate'$ corresponds to $\pair{\heap}{\triple{\perms}{\env[x \mapsto v]}{s} \cdot \stack^*}$.

      \textit{Part \ref{def:state-valid-partial}:}
      Since the initial state is validated by $\vstate$ and $V$, the partial state $\pair{\heap}{\stack^*}$ is validated by $\vstate$ and valuation $V$. Therefore one of \ref{def:partial-valid-nil}, \ref{def:partial-valid-call}, \ref{def:partial-valid-while} applies. We want to show that the partial state $\pair{\heap}{\stack^*}$ is validated by $\vstate'$ and valuation $V'$.

      \begin{itemize}
        \item \textit{Case \ref{def:partial-valid-nil}:} Then $\stack^* = \nilsym$ and trivially $\pair{\heap}{\nilsym}$ is validated by $\vstate'$ and valuation $V'$.
        
        \item \textit{Case \ref{def:partial-valid-call}:} Then $\stack^* = \triple{\env_0}{\perms_0}{\sseq{y \kassign m(e_1, \cdots, e_k)}{s_0}} \cdot \stack_0$ for some $\env_0$, $\perms_0$, $y$, $m$, $k$, $e_1, \cdots, e_k$, $s_0$, $\stack_0$. Also, there is some $\vstate_0$, $V_0$, $x_1, \cdots, x_k$, $t_1, \cdots, t_k$ and $\sstate_0, \cdots, \sstate_k, \sstate'$ such that
        \begin{gather*}
          \text{The partial state $\pair{\heap}{\stack_0}$ is validated by $\vstate_0$ and $V_0$}, \\
          \vstate_0 \text{ is reachable from $\prog$ with valuation $V_0$}, \quad s(\vstate_0) = s(\stack^*) \\
          x_1, \cdots, x_k = \fparams(m), \\
          \sstate_0 = \sstate(\vstate_0), \quad \seval{\sstate_0}{e_1}{t_1}{\sstate_1}{\_}, \quad\cdots,\quad \seval{\sstate_{k-1}}{e_k}{t_k}{\sstate_k}{\_}, \\
          \scons{\sstate_k}{\fpre(m)}{\sstate'}{\_}, \quad \simstate{V_0}{\sstate'[\senv = \senv(\sstate_0)]}{\heap}{\perms_0}{\env_0}, \quad\text{and} \\
          \universal{1 \le i \le k}{V(\senv(\vstate)(x_i)) = V_0(t_i)}, \\
          \gform(\vstate) = \fpost(m).
        \end{gather*}
  
        We want to show that the partial state $\pair{\heap}{\triple{\env_0}{\perms_0}{\sseq{y \kassign m(e_1, \cdots, e_k)}{s}} \cdot \stack_0}$ is validated by $\vstate'$ and valuation $V'$. Immediately from above we can conclude that
        \begin{gather*}
          \text{The partial state $\pair{\heap}{\stack_0}$ is validated by $\vstate_0$ and $V_0$}, \\
          \vstate_0 \text{ is reachable from $\prog$ with valuation $V_0$}, \quad s(\vstate_0) = s(\stack^*) \\
          x_1, \cdots, x_k = \fparams(m), \\
          \sstate_0 = \sstate(\vstate_0), \quad \seval{\sstate_0}{e_1}{t_1}{\sstate_1}{\_}, \quad\cdots,\quad \seval{\sstate_{k-1}}{e_k}{t_k}{\sstate_k}{\_}, \quad\text{and} \\
          \scons{\sstate_k}{\fpre(m)}{\sstate'}{\_}, \quad \simstate{V_0}{\sstate'[\senv = \senv(\sstate_0)]}{\heap}{\perms_0}{\env_0}.
        \end{gather*}
        Also, the frame $\triple{\perms}{\env}{\sseq{x = e}{s}}$ must be executing the body of $m$, since it is in the stack immediately above the frame that contains $y \kassign m(e_1, \cdots, e_k)$. Therefore, since $x_1, \cdots, x_k$ are all parameters of $m$, $y$ must be distinct from all of $x_1, \cdots, x_k$, since we do not allow assignment to parameters in a well-formed program. Thus
        \begin{align*}
          \forall 1 \le i \le k : V'(\senv(\vstate')(x_i))
            &= V'((\senv(\sstate)[x \mapsto t])(x_i)) = V'(\senv(\sstate)(x_i)) \\
            &= V'(\senv(\vstate)(x_i)) = V(\senv(\vstate)(x_i)) \\
            &= V_0(t_i).
        \end{align*}
  
        Finally, $\gform(\vstate') = \gform(\vstate)$ by definition, thus
        $$\gform(\vstate') = \gform(\vstate) = \fpost(m).$$
  
        Therefore the partial state $\pair{\heap}{\stack^*}$ is validated by $\vstate'$ and $V'$ in this case.
        
        \item \textit{Case \ref{def:partial-valid-while}:}
        Then $\stack^* = \triple{\env_0}{\perms_0}{\sseq{\swhile{e_0}{\gform_0}{s_0}}{s_0'}} \cdot \stack_0$ for some $\env_0$, $\perms_0$, $e$, $\gform_0$, $s_0$, $s_0'$, $\stack_0$, and there exists some $\vstate_0$, $V_0$, and $\sstate_0'$ such that:
        \begin{gather*}
          \text{The partial state $\pair{\heap}{\stack_0}$ is validated by $\vstate_0$ and $V_0$} \\
          \vstate_0 \text{ is reachable from $\prog$ with valuation $V_0$} \quad s(\vstate_0) = s(\stack^*) \\
          \scons{\sstate_0}{\gform_0}{\sstate_0'}{\_}, \quad
          \simstate{V_0}{\sstate_0'}{\heap}{\perms_0}{\env_0} \quad\text{and}\\
          \gform(\vstate) = \gform_0.
        \end{gather*}
  
        Now, by definition of $\vstate''$, $\gform(\vstate') = \gform(\vstate) = \gform_0$. Therefore, using the other assumptions given above, the partial state $\pair{\heap}{\stack^*}$ is validated by $\vstate'$ and $V'$ in this case.
      \end{itemize}

      Therefore definition part \ref{def:state-valid-partial} is satisfied.

      Therefore all parts of definition \ref{def:state-valid} are satisfied. Thus $\pair{\heap}{\triple{\perms}{\env[x \mapsto v]}{s} \cdot \stack^*}$ is validated by $\vstate'$ with valuation $V'$.

    \case \refrule{ExecAssignField}:

      We have
      \begin{align}
        &\dexec{\heap}{\triple{\perms}{\env}{\sseq{x.f = e}{s}} \cdot \stack^*}{\vfoot{V'}{\heap}{\sperms}}{\heap'}{\triple{\perms}{\env}{s} \cdot \stack^*} \\
        &\text{where}\quad \eval{\heap}{\env}{x}{\ell}, \quad \eval{\heap}{\env}{e}{v}, \quad \assertion{\heap}{\perms}{\env}{\kacc(x.f)}, \label{eq:dexec-pres-assign-field-eval-assert} \\
        &\hspace{3.5em} \frm{\heap}{\perms}{\env}{e}, \quad\text{and}\quad \heap' = \heap[\pair{\ell}{f} \mapsto v].
      \end{align}
      Since the initial state is validated by $\vstate$, $\vstate = \triple{\sstate}{\sseq{x.f = e}{s}}{\gform}$ for some $\sstate, \gform$ where $\simstate{V}{\sstate}{\heap}{\perms}{\env}$.

      The only guard rule that applies is \refrule{SGuardAssign}, so we have
      \begin{align}
        &\sguard{\triple{\sstate}{\sseq{x.f = e}{s}}{\gform}}{\sstate}{\scheck' \cup \scheck''}{\emptyset} \\
        &\text{where}\quad \seval{\sstate}{e}{t}{\sstate'}{\scheck'}, \quad \scons{\sstate'}{\kacc(x.f)}{\sstate''}{\scheck''} \\
        &\text{and (by assumptions)}\quad \rtassert{V'}{\heap}{\perms}{\scheck' \cup \scheck''}, \quad V'(\pc(\sstate'')) = \ktrue
      \end{align}

      Furthermore, since $\pc(\sstate'') \implies \sstate'$ by lemma \ref{lem:cons-subpath}, and by lemma \ref{lem:scheck-monotonicity},
      $$V'(\pc(\sstate')) = \ktrue, \quad \rtassert{V'}{\heap}{\perms}{\scheck'}, \quad\text{and}\quad \rtassert{V'}{\heap}{\perms}{\scheck''}.$$

      Now, by \refrule{SExecAssignField},
      \begin{align*}
        \sexec{\sstate}{\sseq{x.f = e}{s}}{s}{\sstate'''}
        \quad\text{where}\quad
        &\sstate''' = \sstate''[\sheap = \sheap'], \quad\text{and}\\
        &\sheap' = \sheap(\sstate''); \triple{\senv(\sstate'')(x)}{f}{t}.
      \end{align*}

      Let $\vstate' = \triple{\sstate'''}{s}{\gform}$.  We want to show that $\pair{\heap'}{\triple{\perms}{\env}{s} \cdot \stack^*}$ is validated by $\vstate'$ and $V'$.

      \textit{Part \ref{def:state-valid-reachable}:}
      By \refrule{SVerifyStep}, $\strans{\prog}{\vstate}{\vstate'}$. Therefore $\vstate'$ is reachable from program $\prog$ with valuation $V'$.

      \textit{Part \ref{def:state-valid-correspond}:}
      We want to show that $\vstate'$ corresponds to $\pair{\heap'}{\triple{\perms}{\env}{s} \cdot \stack^*}$. Since $\vstate' = \triple{\sstate'''}{s}{\_}$ by construction, it suffices to show that $\simstate{V'}{\sstate'''}{\heap'}{\perms}{\env}$.

      By lemma \ref{lem:seval-soundness}, $\simstate{V'}{\sstate'}{\heap}{\perms}{\env}$. By lemma \ref{lem:cons-soundness}, $\simstate{V'}{\sstate''}{\heap}{\perms \setminus \efoot{\heap}{\env}{\kacc(x.f)}}{\env}$.

      Since $\simenv{V}{\senv(\sstate)}{\env}$, $\env(x) = V(\senv(\sstate)(x)) = V'(\senv(\sstate)(x))$. Also, $\eval{\heap}{\env}{x}{\env(x)}$ by \refrule{EvalVar} and $\eval{\heap}{\env}{x}{\ell}$ by \eqref{eq:dexec-pres-assign-field-eval-assert}, thus
      $$V'(\senv(\sstate)(x)) = \env(x) = \ell.$$
      Also, $\senv(\sstate'') = \senv(\sstate)$ by lemmas \ref{lem:eval-unchanged} and \ref{lem:cons-unchanged}. Thus $\eval{\heap}{\env}{x}{V'(\senv(\sstate'')(x))}$. Therefore \\
      $\efoot{\heap}{\env}{\kacc(x.f)} = \set{\pair{\ell}{f}} = \set{ \pair{V'(\senv(\sstate'')(x))}{f} } = \vfoot{V'}{\heap}{\triple{\senv(\sstate'')(x)}{f}{t'}}$. Now,
      $$\simstate{V'}{\sstate''}{\heap}{\perms \setminus \vfoot{V'}{\heap}{\triple{\senv(\sstate'')(x)}{f}{t'}}}{\env}.$$

      Also, $v = V'(t)$ by lemma \ref{lem:seval-soundness}, thus $\heap'(V'(\senv(\sstate'')(x)), f) = \heap'(\ell, f) = v = V'(t)$. Now by lemma \ref{lem:pres-add-heap},
      $$\simstate{V'}{\sstate'''}{\heap'}{\perms}{\env}$$
      which is sufficient to show that $\vstate'$ corresponds to $\pair{\heap'}{\triple{\perms}{\env}{s} \cdot \stack^*}$.

      \textit{Part \ref{def:state-valid-partial}:}

      By \eqref{eq:dexec-pres-assign-field-eval-assert} $\assertion{\heap}{\perms}{\env}{\kacc(x.f)}$. This assertion must be given by \refrule{AssertAcc}, therefore $\pair{\ell}{f} \in \perms$.

      Now we show by induction that all partial states are validated, in other words we want to show that the partial state $\pair{\heap'}{\stack^*}$ is validated by $\vstate'$. By assumptions, $\pair{\heap}{\stack^*}$ is validated by $\vstate$ with valuation $V$. Also, by lemma \ref{lem:pres-heap-change-partial}, $\pair{\heap'}{\stack^*}$ is validated by $\vstate$ with $V$.
      
      Thus, one of the following cases apply:

      \begin{itemize}
        \item \textit{Case \ref{def:partial-valid-nil}:} Then $\stack^* = \nilsym$, and trivially $\pair{\heap'}{\nilsym}$ is validated by $\vstate'$ and $V$.
        
        \item \textit{Case \ref{def:partial-valid-call}:}
        Then $\stack^* = \triple{\env}{\perms}{\sseq{y \kassign m(e_1, \cdots, e_k)}{s}} \cdot \stack_0$ for some $\env$, $\perms$, $y$, $m$, $e_1, \cdots, e_k$, and there exists some $\vstate_0$, $V_0$, $x_1, \cdots, x_k$, $t_1, \cdots, t_k$, $\sstate_0, \cdots, \sstate_k$, and $\sstate_0'$ such that:
        \begin{gather*}
          \text{The partial state $\pair{\heap'}{\stack_0}$ is validated by $\vstate_0$ and $V_0$},\\
          \vstate_0 \text{ is reachable from $\prog$ with valuation $V_0$}, \quad s(\vstate_0) = s(\stack^*)\\
          x_1, \cdots, x_k = \fparams(m), \\
          \sstate_0 = \sstate(\vstate_0), \quad \seval{\sstate_0}{e_1}{t_1}{\sstate_1}{\_}, \quad\cdots,\quad \seval{\sstate_{k-1}}{e_k}{t_k}{\sstate_k}{\_}, \\
          \universal{1 \le i \le k}{V(\senv(\vstate)(x_i)) = V_0(t_i)}, \\
          \scons{\sstate_k[\senv = [x_1 \mapsto t_1, \cdots, x_k \mapsto t_k]]}{\fpre(m)}{\sstate_0'}{\_}, \\
          \simstate{V_0}{\sstate_0'[\senv = \senv(\sstate_0)]}{\heap'}{\perms}{\env}, \quad\text{and} \\
          \gform(\vstate) = \fpost(m).
        \end{gather*}

        We want to show that the partial state $\pair{\heap'}{\triple{\env}{\perms}{\sseq{y \kassign m(e_1, \cdots, e_k)}{s}} \cdot \stack_0}$ is validated by $\vstate'$. From above,
        \begin{gather*}
          \text{The partial state $\pair{\heap'}{\stack_0}$ is validated by $\vstate_0$ and $V_0$}, \\
          \vstate_0 \text{ is reachable from $\prog$ with valuation $V_0$}, \quad s(\vstate_0) = s(\stack^*)\\
          x_1, \cdots, x_k = \fparams(m), \\
          \sstate_0 = \sstate(\vstate_0), \quad \seval{\sstate_0}{e_1}{t_1}{\sstate_1}{\_}, \quad\cdots,\quad \seval{\sstate_{k-1}}{e_k}{t_k}{\sstate_k}{\_}, \\
          \scons{\sstate_k[\senv = [x_1 \mapsto t_1, \cdots, x_k \mapsto t_k]]}{\fpre(m)}{\sstate_0'}{\_}, \quad\text{and}\\
          \simstate{V_0}{\sstate_0'[\senv = \senv(\sstate_0)]}{\heap'}{\perms}{\env}.
        \end{gather*}

        $\senv(\vstate') = \senv(\sstate''') = \senv(\sstate'')$ by definition. Also, $\senv(\sstate'') = \senv(\sstate') = \senv(\sstate')$ by lemmas \ref{lem:eval-unchanged} and \ref{lem:cons-unchanged}. Also, $V'$ extends $V$. Thus
        $$\universal{1 \le i \le k}{V'(\senv(\vstate')(x_i)) = V(\senv(\vstate)(x_i)) = V_0(t_i)}.$$
        Also, by definition
        $$\gform(\vstate') = \gform(\vstate) = \fpost(m).$$
        Therefore the partial state $\pair{\heap'}{\stack^*}$ is validated by $\vstate'$.

        \item \textit{Case \ref{def:partial-valid-while}:}
        Then $\stack^* = \triple{\env}{\perms}{\sseq{\swhile{e}{\gform}{s}}{s'}} \cdot \stack_0$ for some $\env$, $\perms$, $e$, $\gform$, $s$, $s'$, and $\stack_0$, and there exists some $\vstate_0$, $V_0$, and $\sstate_0'$ such that:
        \begin{gather*}
          \text{The partial state $\pair{\heap'}{\stack_0}$ is validated by $\vstate$ and $V$} \\
          \vstate_0 \text{ is reachable from $\prog$ with valuation $V_0$} \quad s(\vstate_0) = s(\stack^*) \\
          \scons{\sstate(\vstate_0)}{\gform}{\sstate_0'}{\_}, \quad\text{and}\quad
          \simstate{V_0}{\sstate_0'}{\heap'}{\perms}{\env} \\
          \gform(\vstate) = \gform
        \end{gather*}

        Now, $\gform(\vstate') = \gform(\vstate)$ by definition, thus $\gform(\vstate') = \gform$. Therefore the partial state $\pair{\heap'}{\stack^*}$ is validated by $\vstate'$.
      \end{itemize}
      Now we have shown that $\vstate'$ corresponds to the resulting state, and therefore $\gamma'$ is validated by $\vstate'$.

    \case \refrule{ExecAlloc}:
      We have
      \begin{align*}
        &\dexec{\heap}{\triple{\perms}{\env}{\sseq{x = \kalloc(S)}{s}} \cdot \stack^*}{\vfoot{V'}{\heap}{\sperms}}{\heap'}{\triple{\perms'}{\env[x \mapsto \ell]}{s} \cdot \stack^*} \\
        &\text{where} \quad \multiple{T~f} = \fstruct(S), \quad \ell = \ffresh, \quad \heap' = \heap[\multiple{\pair{\ell}{f} \mapsto \fdefault(T)}], \\
        &\hspace{3.5em} \perms' = \perms \cup \set{\multiple{\pair{\ell}{f}}}
      \end{align*}

      Since the initial state is validated by $\vstate$, $\vstate = \triple{\sstate}{\sseq{x = \kalloc}{s}}{\gform}$ for some $\sstate, \gform$ where $\simstate{V}{\sstate}{\heap}{\perms}{\env}$.

      By \refrule{SExecAlloc} \begin{align*}
        \sexec{\sstate}{\sseq{x = \kalloc(S)}{s}}{s}{\sstate'}, ~\text{where}~
        &\sheap(\sstate') = \sheap(\sstate); \multiple{\triple{t}{f}{\fdefault(T)}}, \\
        &\senv(\sstate') = \senv(\sstate)[x \mapsto t], \quad\text{and} \\
        &t = \ffresh.
      \end{align*}
      Let $\vstate' = \triple{\sstate'}{s}{\gform}$, and $V' = V[t \mapsto \ell]$. We want to show that $\pair{\heap'}{\triple{\perms'}{\env[x \mapsto \ell]}{s} \cdot \stack^*}$ is validated by $\vstate'$ with $V$.

      \textit{Part \ref{def:state-valid-reachable}:}
      By \refrule{SVerifyStep}, $\prog \vdash \vstate \to \vstate'$, and all fresh values added to $\vstate'$ are defined in $V'$. Therefore, $\vstate'$ is reachable from $\prog$ with valuation $V'$.

      \textit{Part \ref{def:state-valid-correspond}:}
      We want to show that $\vstate'$ corresponds to ${\heap'}{\triple{\perms'}{\env[x \mapsto \ell]}{s} \cdot \stack^*}$. By definition $s(\vstate') = s$ and $\sstate(\vstate') = \sstate'$, thus it suffices to show $\simstate{V'}{\sstate'}{\heap'}{\perms'}{\env[x \mapsto \ell]}$.

      By assumptions, $\simenv{V}{\senv(\sstate)}{\env}$. Also,
      $$V'(\senv(\sstate')(x)) = V'((\senv(\sstate)[x \mapsto t])(x)) = V'(t) = \ell = (\env[x \mapsto \ell])(x).$$
      Therefore
      $$\simenv{V'}{\senv(\sstate)[x \mapsto t]}{\env[x \mapsto \ell]}.$$

      By assumptions, $\simstate{V}{\sstate}{\heap}{\perms}{\env}$. Since $V \subseteq V'$, $\simstate{V'}{\sstate}{\heap}{\perms}{\env}$.
      Since $\ell$ is a fresh value, $\ell \notin \perms$. Thus
      $$\universal{\pair{\ell}{f} \in \perms}{\heap'(\ell, f) = \heap(\ell, f)}.$$
      Thus by lemma \ref{lem:pres-add-heap}, $\simstate{V'}{\sstate}{\heap'}{\perms}{\env}$. Also, since $\perms \subseteq \perms'$, by lemma \ref{lem:simstate-monotonicity} $\simstate{V'}{\sstate}{\heap'}{\perms'}{\env}$.

      Let $\triple{f'}{t'}{t''} \in \sheap(\sstate')$. If $\triple{f'}{t'}{t''} \in \sheap(\sstate)$, then since $\simstate{V'}{\sstate}{\heap'}{\perms'}{\env}$, $\heap'(V'(t'), f') = V'(t'')$ and $\pair{V'(t')}{f'} \in \perms'$.

      Otherwise, $t' = t$ and $f' = f$ for some $T~f \in \fstruct(S)$, and thus $t'' = \fdefault(T)$. Now $\heap(V'(t'), f') = \heap'(V'(t), f) = \heap'(\ell, f) = \fdefault(T) = t''.$
      
      Therefore $$\universal{\triple{f'}{t'}{t''} \in \sheap(\sstate')}{\heap'(V'(t'), f') = V'(t'') ~\text{and}~ \pair{V'(t')}{f'} \in \perms'}.$$

      Also, since $\simstate{V'}{\sstate}{\heap'}{\perms'}{\env}$ and $\universal{\pair{p}{\multiple{t}} \in \sheap(\sstate')}{\pair{p}{\multiple{t}} \in \sheap(\sstate)}$,
      $$\universal{\pair{p}{\multiple{t}} \in \sheap(\sstate')}{\assertion{\heap'}{\perms'}{[\multiple{x \mapsto V(t)}]}{\fpred(p)}}$$
      where $\multiple{x} = \fpredparams(p)$.

      Let $h_1, h_2 \in \sheap(\sstate')$ and suppose $h_1 \ne h_2$. If $h_1 \in \sheap(\sstate)$ and $h_2 \in \sheap(\sstate)$, then in this case $\vfoot{V'}{\heap'}{h_1} \cap \vfoot{V'}{\heap'}{h_2} = \emptyset$ since $\simstate{V'}{\sstate}{\heap'}{\perms'}{\env}$.

      Otherwise, WLOG $h_1 = \triple{t}{f_1}{\fdefault(T_1)}$ for some $T_1~f_1 \in \fstruct(S)$. Thus 
      \begin{align*}
        \vfoot{V'}{\heap'}{h_1} &= \vfoot{V'}{\heap'}{\triple{t}{f_1}{\fdefault(T_1)}} &h_1 = \triple{t}{f_1}{\fdefault(T_1)} \\
          &= \set{ \pair{V'(t)}{f_1} } &\text{defn.} \\
          &= \set{ \pair{\ell}{f_1} } &\text{defn. $V'$}
      \end{align*}
      If $h_2 \in \sheap(\sstate)$, then
      \begin{align*}
        \vfoot{V'}{\heap'}{h_2} &= \vfoot{V}{\heap}{h_2} & V \subseteq V', \heap \subseteq \heap' \\
          &\subseteq \perms &\text{Lemma \ref{lem:sim-heap-contains}}
      \end{align*}
      Now $\pair{\ell}{f_1} \notin \perms$ since $\ell$ is a fresh value. Therefore $\vfoot{V'}{\heap'}{h_1} \cap \vfoot{V'}{\heap'}{h_2} = \emptyset$ in this case.
      
      Otherwise, $h_2 = \triple{t}{f_2}{\fdefault(T_2)}$ for some $T_2~f_2 \in \fstruct(S)$. Then $f_1 \ne f_2$ since $h_1 \ne h_2$. Therefore
      \begin{align*}
        \vfoot{V'}{\heap'}{h_1} \cap \vfoot{V'}{\heap'}{h_2} &= \vfoot{V'}{\heap'}{\triple{t}{f_1}{\fdefault(T_1)}} \cap \vfoot{V'}{\heap'}{\triple{t}{f_2}{\fdefault(T_2)}} \\
          &= \set{\pair{V'(t)}{f_1}} \cap \set{\pair{V'(t)}{f_2}} \\
          &= \emptyset
      \end{align*}

      Therefore,
      $$\universal{h_1, h_2 \in \sheap(\sstate')}{\vfoot{V'}{\heap'}{h_1} \cap \vfoot{V'}{\heap'}{h_2} = \emptyset}.$$

      Now, since we have shown all requirements, we can conclude that $$\simheap{V'}{\sheap(\sstate')}{\heap'}{\perms'}.$$

      Finally, since $\sheap$ is the only component that differs between $\sstate'$ and $\sstate$, $\simstate{V'}{\sstate}{\heap'}{\perms'}{\env}$, $\simenv{V'}{\senv(\sstate)[x \mapsto t]}{\env[x \mapsto \ell]}$, and $\simheap{V'}{\sheap(\sstate')}{\heap'}{\perms'}$,
      $$\simstate{V'}{\sstate'[\senv = \senv(\sstate)[x \mapsto t]]}{\heap'}{\perms'}{\env[x \mapsto \ell]}$$
      which is what we wanted to show.

      \textit{Part \ref{def:state-valid-partial}:}
      First, we show that the partial state $\pair{\heap}{\stack^*}$ is validated by $\vstate'$ and $V'$.

      Since the initial state is validated by $\vstate$ and $V$, the partial state $\pair{\heap}{\stack^*}$ is validated by $\vstate$ and valuation $V$. Therefore one of \ref{def:partial-valid-nil}, \ref{def:partial-valid-call}, \ref{def:partial-valid-while} applies. We want to show that the partial state $\pair{\heap}{\stack^*}$ is validated by $\vstate'$ and valuation $V'$.

      \begin{itemize}
        \item \textit{Case \ref{def:partial-valid-nil}:} Then $\stack^* = \nilsym$ and trivially $\pair{\heap}{\nilsym}$ is validated by $\vstate'$ and valuation $V'$.
        
        \item \textit{Case \ref{def:partial-valid-call}:} Then $\stack^* = \triple{\env_0}{\perms_0}{\sseq{y \kassign m(e_1, \cdots, e_k)}{s_0}} \cdot \stack_0$ for some $\env_0$, $\perms_0$, $y$, $m$, $k$, $e_1, \cdots, e_k$, $s_0$, $\stack_0$. Also, there is some $\vstate_0$, $V_0$, $x_1, \cdots, x_k$, $t_1, \cdots, t_k$ and $\sstate_0, \cdots, \sstate_k, \sstate'$ such that
        \begin{gather*}
          \text{The partial state $\pair{\heap}{\stack_0}$ is validated by $\vstate_0$ and $V_0$}, \\
          \vstate_0 \text{ is reachable from $\prog$ with valuation $V_0$}, \quad s(\vstate_0) = s(\stack^*) \\
          x_1, \cdots, x_k = \fparams(m), \\
          \sstate_0 = \sstate(\vstate_0), \quad \seval{\sstate_0}{e_1}{t_1}{\sstate_1}{\_}, \quad\cdots,\quad \seval{\sstate_{k-1}}{e_k}{t_k}{\sstate_k}{\_}, \\
          \scons{\sstate_k}{\fpre(m)}{\sstate'}{\_}, \quad \simstate{V_0}{\sstate'[\senv = \senv(\sstate_0)]}{\heap}{\perms_0}{\env_0}, \quad\text{and} \\
          \universal{1 \le i \le k}{V(\senv(\vstate)(x_i)) = V_0(t_i)}, \\
          \gform(\vstate) = \fpost(m).
        \end{gather*}
  
        We want to show that the partial state $\pair{\heap}{\triple{\env_0}{\perms_0}{\sseq{y \kassign m(e_1, \cdots, e_k)}{s}} \cdot \stack_0}$ is validated by $\vstate'$ and valuation $V'$. Immediately from above we can conclude that
        \begin{gather*}
          \text{The partial state $\pair{\heap}{\stack_0}$ is validated by $\vstate_0$ and $V_0$}, \\
          \vstate_0 \text{ is reachable from $\prog$ with valuation $V_0$}, \quad s(\vstate_0) = s(\stack^*) \\
          x_1, \cdots, x_k = \fparams(m), \\
          \sstate_0 = \sstate(\vstate_0), \quad \seval{\sstate_0}{e_1}{t_1}{\sstate_1}{\_}, \quad\cdots,\quad \seval{\sstate_{k-1}}{e_k}{t_k}{\sstate_k}{\_}, \quad\text{and} \\
          \scons{\sstate_k}{\fpre(m)}{\sstate'}{\_}, \quad \simstate{V_0}{\sstate'[\senv = \senv(\sstate_0)]}{\heap}{\perms_0}{\env_0}.
        \end{gather*}
        Also, the frame $\triple{\perms}{\env}{\sseq{x = \kalloc(S)}{s}}$ must be executing the body of $m$, since it is in the stack immediately above the frame that contains $y \kassign m(e_1, \cdots, e_k)$. Therefore, since $x_1, \cdots, x_k$ are all parameters of $m$, $y$ must be distinct from all of $x_1, \cdots, x_k$, since we do not allow assignment to parameters in a well-formed program. Thus
        \begin{align*}
          \forall 1 \le i \le k : V'(\senv(\vstate')(x_i))
            &= V'((\senv(\sstate)[x \mapsto t])(x_i)) = V'(\senv(\sstate)(x_i)) \\
            &= V'(\senv(\vstate)(x_i)) = V(\senv(\vstate)(x_i)) \\
            &= V_0(t_i).
        \end{align*}
  
        Finally, $\gform(\vstate') = \gform(\vstate)$ by definition, thus
        $$\gform(\vstate') = \gform(\vstate) = \fpost(m).$$
  
        Therefore the partial state $\pair{\heap}{\stack^*}$ is validated by $\vstate'$ and $V'$ in this case.
        
        \item \textit{Case \ref{def:partial-valid-while}:}
        Then $\stack^* = \triple{\env_0}{\perms_0}{\sseq{\swhile{e_0}{\gform_0}{s_0}}{s_0'}} \cdot \stack_0$ for some $\env_0$, $\perms_0$, $e$, $\gform_0$, $s_0$, $s_0'$, $\stack_0$, and there exists some $\vstate_0$, $V_0$, and $\sstate_0'$ such that:
        \begin{gather*}
          \text{The partial state $\pair{\heap}{\stack_0}$ is validated by $\vstate_0$ and $V_0$} \\
          \vstate_0 \text{ is reachable from $\prog$ with valuation $V_0$}, \quad s(\vstate_0) = s(\stack^*) \\
          \scons{\sstate_0}{\gform_0}{\sstate_0'}{\_}, \quad
          \simstate{V_0}{\sstate_0'}{\heap}{\perms_0}{\env_0} \quad\text{and}\\
          \gform(\vstate) = \gform_0.
        \end{gather*}
  
        Now, by definition of $\vstate''$, $\gform(\vstate') = \gform(\vstate) = \gform_0$. Therefore, using the other assumptions given above, the partial state $\pair{\heap}{\stack^*}$ is validated by $\vstate'$ and $V'$ in this case.

      \end{itemize}

      Therefore the partial state $\pair{\heap}{\stack^*}$ is validated by $\vstate'$ and $V'$.

      Now, by lemma \ref{lem:pres-heap-change-partial}, the partial state $\pair{\heap'}{\stack^*}$ is validated by $\vstate'$ and $V'$, which is what we need to show for this part.

      Therefore $\pair{\heap'}{\triple{\perms'}{\env[x \mapsto \ell]}{s} \cdot \stack^*}$ is validated by $\vstate'$ with $V$, which completes the proof.

    \case \refrule{ExecCallEnter}: We have
      \begin{align*}
        &\pair{\heap}{\triple{\perms}{\env}{\sseq{y \kassign m(\multiple{e})}{s}} \cdot \stack^*}, \vfoot{V'}{\heap}{\sperms} \to \\
        &\quad \pair{\heap}{\triple{\perms'}{\env'}{\sseq{\fbody(m)}{\kskip}} \cdot \triple{\perms \setminus \perms'}{\env}{\sseq{y \kassign m(\multiple{e})}{s}} \cdot \stack^*} \\
        &\text{where}~ \multiple{x} = \fparams(m), \quad \multiple{\eval{\heap}{\env}{e}{v}}, \quad \multiple{\frm{\heap}{\perms}{\env}{e}}, \\
        &\hspace{2em} \env' = [\multiple{x \mapsto v}], \quad
        \assertion{\heap}{\perms \setminus \xperms}{\env'}{\fpre(m)}, \\
        &\hspace{2em} \xperms = \vfoot{V'}{\heap}{\sperms}, \quad\text{and}\quad
        \perms' = \foot{\heap}{\perms \setminus \xperms}{\env'}{\fpre(m)}.
      \end{align*}

      By assumptions, the initial state is validated by some $\vstate$ and valuation $V$, thus $\vstate = \triple{\sstate}{\sseq{y \kassign m(\multiple{e})}{s}}{\gform}$ for some $\sstate$, $\gform$ where $\simstate{V'}{\sstate}{\heap}{\perms}{\env}$.

      The only guard rule that applies is \refrule{SGuardCall}, so we have
      \begin{gather*}
        \multiple{\seval{\sstate}{e}{t}{\sstate'}{\scheck}}, \quad \scons{\sstate'[\senv = [\multiple{x \mapsto t}]]}{\fpre(m)}{\sstate''}{\scheck'}, \\
        \sperms = \frem(\sstate'', \fpre(m)), \quad
        \text{and by assumptions,}\quad \rtassert{V}{\heap}{\perms}{\multiple{\scheck} \cup \scheck'}
      \end{gather*}

      For some $k$, let $x_1, \cdots, x_k = \multiple{x}$, $e_1, \cdots, e_k = \multiple{e}$, $v_1, \cdots, v_k = \multiple{v}$, and $t_1 = \ffresh, \cdots, t_k = \ffresh$.

      Also, let $\sstate_0 = \quintuple{\bot}{\emptyset}{\emptyset}{[x_1 \mapsto t_1, \cdots, x_k \mapsto t_k]}{\ktrue}$, and let $V_0 = [t_1 \mapsto v_1, \cdots, t_k \mapsto v_k]$. Then $\simstate{V_0}{\sstate_0}{\heap}{\perms'}{\env'}$.

      $\efoot{\heap}{\env'}{\fpre(m)} \subseteq \perms'$: If $\fpre(m)$ is completely precise, then $\perms' = \foot{\heap}{\perms \setminus \xperms}{\env'}{\fpre(m)} = \\ \efoot{\heap}{\env'}{\fpre(m)}$. Otherwise, $\perms' = \foot{\heap}{\perms \setminus \xperms}{\env'}{\fpre(m)} = \perms \setminus \xperms$, but also $\efoot{\heap}{\env'}{\fpre(m)} \subseteq \perms' = \perms \setminus \xperms$ by lemma \ref{lem:efoot-subset-spec} since $\assertion{\heap}{\perms \setminus \xperms}{\env'}{\fpre(m)}$.

      Now $\assertion{\heap}{\perms'}{\env'}{\fpre(m)}$ by lemma \ref{lem:assert-efoot-subset}, since $\assertion{\heap}{\perms \setminus \xperms}{\env'}{\fpre(m)}$.

      Now $\assertion{\heap}{\perms'}{\env'}{\fpre(m)}$, $\simstate{V_0}{\sstate_0}{\heap}{\perms'}{\env'}$. Thus by lemma \ref{lem:produce-progress}, for some $\sstate_0'$,
      $$\sproduce{\sstate_0}{\fpre(m)}{\sstate_0'} ~\text{and}~V_0'(\pc(\sstate_0)) \quad\text{where}\quad V_0' = V_0[\sproduce{\sstate_0}{\fpre(m)}{\sstate_0'} \mid \heap].$$

      Let $\vstate_0' = \triple{\sstate_0'}{\sseq{\fbody(m)}{\kskip}}{\fpost(m)}$. We want to show that
      $$\Gamma' = \pair{\heap}{\triple{\perms'}{\env'}{\sseq{\fbody(m)}{\kskip}} \cdot \triple{\perms \setminus \perms'}{\env}{\sseq{y \kassign m(\multiple{e})}{s}} \cdot \stack^*}$$ is validated by $\vstate_0'$ with $V_0'$.

      \textit{Part \ref{def:state-valid-reachable}:}
      By \refrule{SVerifyMethod}, $\strans{\prog}{\initsym}{\vstate_0'}$. Therefore $\vstate_0'$ is reachable from $\prog$ with valuation $V_0'$.

      \textit{Part \ref{def:state-valid-correspond}:}
      We want to show that $\Gamma'$ corresponds to $\vstate_0'$.

      By definition $s(\vstate_0') = \sseq{\fbody(m)}{\kskip} = s(\Gamma')$. Therefore, since $\sstate(\vstate_0') = \sstate_0'$, it suffices to show $\simstate{V_0'}{\sstate_0'}{\heap}{\perms'}{\env'}$.

      Since $\sheap(\sstate_0) = \oheap(\sstate_0) = \emptyset$, $\simheap{V_0}{\sheap(\sstate_0)}{\heap}{\perms' \setminus \efoot{\heap}{\env'}{\fpre(m)}}$ and $\simheap{V_0}{\oheap(\sstate_0)}{\heap}{\perms' \setminus \efoot{\heap}{\env'}{\fpre(m)}}$. Also, for each $1 \le i \le k$, $V_0(\senv(\sstate_0)(x_i)) = v_i = \env'(x_i)$, thus $\simenv{V_0}{\senv(\sstate_0)}{\env'}$. Finally, $V_0(\pc(\sstate_0)) = V_0(\ktrue) = \ktrue$. Therefore $\simstate{V_0}{\sstate_0}{\heap}{\perms' \setminus \efoot{\heap}{\env'}{\fpre(m)}}{\env'}$.

      Also, as shown before, $\assertion{\heap}{\perms'}{\env'}{\fpre(m)}$. Therefore, by lemma \ref{lem:produce-soundness},
      $$\simstate{V_0'}{\sstate_0'}{\heap}{\perms'}{\env'}.$$

      \textit{Part \ref{def:state-valid-partial}:}
      We want to show that the partial state $\pair{\heap}{\triple{\perms \setminus \perms'}{\env}{\sseq{y \kassign m(\multiple{e})}{s}} \cdot \stack^*}$ is validated by $\vstate_0'$ with $V_0'$, thus it suffices to show that case \ref{def:partial-valid-call} is satisfied.

      Since $\triple{\perms}{\env}{\sseq{y \kassign m(\multiple{e})}{s}} \cdot \stack^*$ was validated by $\vstate$ and $V$, the partial state $\pair{\heap}{\stack^*}$ is validated by $\vstate$ and $V$.

      Also, by assumptions, $\vstate$ is reachable from $\prog$ with valuation $V$ and $s(\vstate) = \sseq{y \kassign m(\multiple{e})}{s}$ as shown before.

      Furthermore, by assumptions, $x_1, \cdots, x_k = \multiple{x} = \fparams(m)$

      Now let $\sstate_0 = \sstate = \sstate(\vstate)$, then $\multiple{\seval{\sstate}{e}{t}{\sstate'}{\scheck}}$, which was shown before, represents the series of judgements
      $$\seval{\sstate_0}{e_1}{t_1}{\sstate_1}{\scheck_1}, \quad\cdots\quad \seval{\sstate_{k-1}}{e_k}{t_k}{\sstate_k}{\stack_k}$$
      where $\sstate_k = \sstate'$. Also, as shown before,
      $$\scons{\sstate_k[\senv = [x_1 \mapsto t_1, \cdots, x_k \mapsto t_k]]}{\fpre(m)}{\sstate''}{\scheck'}.$$
      Note that by definition \ref{def:sguard-valuation} $V' = V[\sguard{\vstate}{\sstate'}{\scheck}{\sperms} \mid \heap]$ is  the valuation corresponding to the series of judgements above, extending $V$.

      By lemmas \ref{lem:eval-subpath} and \ref{lem:cons-subpath}, $\pc(\sstate'') \implies \pc(\sstate_k) \implies \cdots \implies \pc(\sstate_1)$. Thus, since $V'(\pc(\sstate'')) = \ktrue$ by assumption,
      $$V'(\pc(\sstate'')) = V'(\pc(\sstate_k)) = \cdots = V'(\pc(\sstate_1)) = \ktrue.$$
      Therefore, by lemmas \ref{lem:seval-soundness} and \ref{lem:cons-soundness},
      $$\simstate{V'}{\sstate_1}{\heap}{\perms}{\env}, \quad\cdots,\quad \simstate{V'}{\sstate_k}{\heap}{\perms}{\env}, \quad \simstate{V'}{\sstate''}{\heap}{\perms \setminus \efoot{\heap}{\env'}{\fpre(m)}}{\env'},$$
      Furthermore, since $\xperms = \vfoot{V'}{\heap}{\sperms} = \vfoot{V'}{\heap}{\frem(\sstate'', \fpre(m))}$ and $\assertion{\heap}{\perms \setminus \xperms}{\env'}{\fpre(m)}$, we can apply lemma \ref{lem:rem-simstate} to get
      $$\simstate{V'}{\sstate''[\senv = \senv(\sstate_0)]}{\heap}{\perms \setminus \perms'}{\env'}.$$
      Since $\simenv{V}{\senv(\sstate_0)}{\env}$ (by assumptions and since $\sstate_0 = \sstate$),
      $$\simstate{V'}{\sstate''[\senv = \senv(\sstate_0)]}{\heap}{\perms \setminus \perms'}{\env}.$$

      For each $1 \le i \le k$, $\eval{\heap}{\env}{e_i}{V'(t_i)}$ by lemma \ref{lem:seval-soundness}, and $\eval{\heap}{\env}{e_i}{v_i}$ as shown before, thus $V'(t_i) = v_i$. Thus
      \begin{align*}
        \universal{1 \le i \le k}{V_0'(\senv(\vstate_0')(x_i)) &= V_0'(\senv(\sstate_0')(x_i))} &\text{by defn.} \\
          &= V_0'(\senv(\sstate_0)(x_i)) &\text{Lemma \ref{lem:cons-unchanged}} \\
          &= V_0(\senv(\sstate_0)(x_i)) &V \subseteq V' \\
          &= v_i &\text{by def.}\\
          &= V'(t_i). &\text{shown above}
      \end{align*}

      Finally, by definition $\gform(\vstate) = \fpost(m)$.

      Therefore the partial state $\pair{\heap}{\triple{\perms'}{\env'}{\sseq{\fbody(m)}{\kskip}} \cdot \triple{\perms \setminus \perms'}{\env}{\sseq{y \kassign m(\multiple{e})}{s}} \cdot \stack^*}$ is validated by $\vstate_0'$ with $V_0'$.

      Therefore $\Gamma'$ is validated by $\vstate_0'$ with $V_0'$.

    \case \refrule{ExecCallExit}: We have
      \begin{gather}
        \dexec{\heap}{\triple{\perms}{\env}{\kskip} \cdot \triple{\perms'}{\env'}{\sseq{y \kassign m(\multiple{e})}{s}} \cdot \stack}{\vfoot{V'}{\heap}{\sperms}}{\heap}{\triple{\perms''}{\env''}{s} \cdot \stack^*} \label{eq:dexec-pres-return-exec}\\
        \text{where } \assertion{\heap}{\perms}{\env}{\fpost(m)}, \quad
        \env'' = \env'[y \mapsto \env(\kresult)], \label{eq:dexec-pres-return-1} \\
        \text{and }\perms'' = \perms' \cup \foot{\heap}{\perms}{\env}{\fpost(m)}.
      \end{gather}

      By assumptions, the initial state is validated by some $\vstate$ and valuation $V$, thus $\vstate = \triple{\sstate}{\kskip}{\gform}$ for some $\sstate, \gform$ where $\simstate{V'}{\sstate}{\heap}{\perms}{\env}$.

      Also, by \ref{def:state-valid-partial}, the partial state $\pair{\heap}{\triple{\perms'}{\env'}{\sseq{y \kassign m(\multiple{e})}{s}} \cdot \stack}$ is validated by $\vstate$ and $V$. Thus \ref{def:partial-valid-call} must apply, and thus there is some $\vstate'$ reachable from $\prog$ and valuation $V$ such that $\vstate' = \triple{\sstate_0}{\sseq{y \kassign m(\multiple{e})}{s}}{\gform'}$ for some $\sstate_0, \gform'$. Also, we can let $e_1, \cdots, e_k = \multiple{e}$ and then there are sequences $\sstate_1, \cdots, \sstate_k$, $x_1, \cdots, x_k$, and $t_1, \cdots, t_k$ where
      \begin{align}
        &\seval{\sstate_0}{e_1}{t_1}{\sstate_1}{\_}, \quad\cdots,\quad
        \seval{\sstate_{k-1}}{e_k}{t_k}{\sstate_k}{\_}, &\text{by \eqref{eq:partial-valid-call-eval}}\label{eq:dexec-pres-return-seval} \\
        &\scons{\sstate_k[\senv = [\multiple{x_i \mapsto t_i}]]}{\fpre(m)}{\sstate'}{\_}, \label{eq:dexec-pres-return-cons} \\
        &\text{and}\quad
        \simstate{V'}{\sstate'[\senv = \senv(\sstate_0)]}{\heap}{\perms'}{\env'} &\text{by \eqref{eq:partial-valid-call-sim}} \label{eq:dexec-pres-return-2}
      \end{align}
      where $V'$ is a valuation corresponding to this series of judgements.

      Let
      \begin{gather*}
        t = \ffresh, \quad
        \hat{V}' = V'[t \mapsto \env(\kresult)], \quad
        \hat{\env} = [x_1 \mapsto \env(x_1), \cdots, x_k \mapsto \env(x_k)], \\
        \text{and}\quad \hat{\senv} = [x_1 \mapsto t_1, \cdots, x_k \mapsto t_k, \kresult \mapsto t]
      \end{gather*}

      We have $\assertion{\heap}{\perms}{\env}{\fpost(m)}$ by \eqref{eq:dexec-pres-return-1}. Since $\hat{\env}$ is simply the restriction of $\env$ to $\fparams(m)$ and $\kresult$, and $\fpost(m)$ may only reference variables in $\fparams(m)$ as well as $\kresult$, $\assertion{\heap}{\perms}{\hat{\env}}{\fpost(m)}$ and $\efoot{\heap}{\hat{\env}}{\fpost(m)} = \efoot{\heap}{\env}{\fpost(m)}$.

      By lemma \ref{lem:efoot-subset-spec} $\efoot{\heap}{\env}{\fpost(m)} \subseteq \perms$. Recall that $\perms'' = \perms \cup \foot{\heap}{\perms}{\env}{\fpost(m)}$. If $\fpost(m)$ is completely precise, then $\foot{\heap}{\perms}{\env}{\fpost(m)} = \efoot{\heap}{\env}{\fpost(m)}$. Otherwise, $\foot{\heap}{\perms}{\env}{\fpost(m)} = \perms$, but $\efoot{\heap}{\env}{\fpost(m)} \subseteq \perms$ as shown before. In both cases, $\efoot{\heap}{\env}{\fpost(m)} = \efoot{\heap}{\hat{\env}}{\fpost(m)} \subseteq \perms''$.

      Therefore by lemma \ref{lem:assert-efoot-subset}
      \begin{equation}\label{eq:dexec-pres-return-3}
        \assertion{\heap}{\perms''}{\hat{\env}}{\fpost(m)}.
      \end{equation}

      Note that, for all $1 \le i \le k$,
      \begin{align*}
        \hat{V}'(t_i) &= V'(t_i) &\text{by definition} \\
          &= V(\senv(\vstate)(x_i)) &\text{by \eqref{eq:partial-valid-call-params}} \\
          &= \env(x_i) &\text{since $\simenv{V}{\senv(\vstate)}{\env}$, since initial valid by $\vstate$ and $V$} \\
          &= \hat{\env}(x_i) &\text{by definition}
      \end{align*}
      Thus $\simenv{\hat{V}'}{\senv'}{\hat{\env}}$. Also, $\simstate{V'}{\sstate'[\senv = \senv(\sstate_0)]}{\heap}{\perms'}{\env'}$ by \eqref{eq:dexec-pres-return-2}. Therefore $\simstate{\hat{V}'}{\sstate'[\senv = \hat{\senv}]}{\heap}{\perms'}{\hat{\env}}$. Finally, since $\perms'' \subseteq \perms'$, by lemma \ref{lem:simstate-monotonicity},
      \begin{equation}\label{eq:dexec-pres-return-4}
        \simstate{\hat{V}'}{\sstate'[\senv = \hat{\senv}]}{\heap}{\perms''}{\hat{\env}}.
      \end{equation}

      Now, by lemma \ref{lem:produce-progress}, \eqref{eq:dexec-pres-return-3}, and \eqref{eq:dexec-pres-return-4},
      \begin{equation}\label{eq:dexec-pres-return-prod}
        \sproduce{\sstate'[\senv = \hat{\senv}]}{\fpost(m)}{\sstate''}, \quad\text{and}\quad V''(\pc(\sstate''')) = \ktrue
      \end{equation}
      where $V''$ is the corresponding valuation extending $\hat{V}'$.

      Let $\sstate''' = \sstate''[\senv = \senv(\sstate_0)[y \mapsto t]]$.

      Now \eqref{eq:dexec-pres-return-seval}, \eqref{eq:dexec-pres-return-cons}, and \eqref{eq:dexec-pres-return-prod}, and the definition of $\sstate'''$ satisfy the antecedent for \refrule{SExecCall}, therefore
      $$\sexec{\sstate_0}{\sseq{y \kassign m(e_1, \cdots, e_k)}{s}}{s}{\sstate'''}.$$
      Let $\vstate'' = \triple{\sstate'''}{s}{\gform'}$, now by \refrule{SVerifyStep}
      $$\strans{\prog}{\vstate'}{\vstate''}.$$

      We want to show that $\pair{\heap}{\triple{\perms''}{\env''}{s} \cdot \stack^*}$ is validated by $\vstate''$.
      
      Part \ref{def:state-valid-reachable}: Since $\strans{\prog}{\vstate'}{\vstate''}$, $\vstate''$ is reachable from $\prog$. Let its corresponding valuation be $V'''$.

      Part \ref{def:state-valid-correspond}: By definition, $s(\vstate'') = s$.

      By \eqref{eq:dexec-pres-return-2} $\simstate{V'}{\sstate'[\senv = \senv(\sstate_0)]}{\heap}{\perms'}{\env'}$. Since the initial state must be well-formed, $\perms$ and $\perms'$ are disjoint, and as shown before, $\efoot{\heap}{\env}{\fpost(m)} = \efoot{\heap}{\hat{\env}}{\fpost(m)} \subseteq \perms$, therefore $\perms' \setminus \efoot{\heap}{\hat{\env}}{\fpost(m)} = \perms'$. Also, $\hat{V}' \subseteq V'$. Thus
      $$\simstate{\hat{V}'}{\sstate'[\senv = \senv(\sstate_0)]}{\heap}{\perms' \setminus \efoot{\heap}{\hat{\env}}{\fpost(m)}}{\env'}.$$

      Also, as shown before, $\simenv{\hat{V}'}{\hat{\senv}}{\hat{\env}}$, therefore
      $$\simstate{\hat{V}'}{\sstate'[\senv = \hat{\senv}]}{\heap}{\perms' \setminus \efoot{\heap}{\hat{\env}}{\fpost(m)}}{\hat{\env}}.$$

      Then, since $\perms' \subseteq \perms''$, $\perms' \setminus \efoot{\heap}{\hat{\env}}{\fpost(m)} \subseteq \perms'' \setminus \efoot{\heap}{\hat{\env}}{\fpost(m)}$, and thus by lemma \ref{lem:simstate-monotonicity}
      $$\simstate{\hat{V}'}{\sstate'[\senv = \hat{\senv}]}{\heap}{\perms'' \setminus \efoot{\heap}{\hat{\env}}{\fpost(m)}}{\hat{\env}}.$$

      Now, since it was shown in \eqref{eq:dexec-pres-return-3} that $\assertion{\heap}{\perms''}{\hat{\env}}{\fpost(m)}$ and in \eqref{eq:dexec-pres-return-prod} that $V''(\pc(\sstate'')) = \ktrue$, by lemma \ref{lem:produce-soundness}
      $$\simstate{V''}{\sstate''}{\heap}{\perms''}{\hat{\env}}.$$

      Now by \eqref{eq:dexec-pres-return-2} $\simenv{V'}{\senv(\sstate_0)}{\env'}$, then since $V' \subseteq V''$, $\simenv{V''}{\senv(\sstate_0)}{\env'}$. Now, since $\senv(\sstate''') = \senv(\sstate_0)[y \mapsto t]$, to show $\simenv{V''}{\senv(\sstate''')}{\env''}$ it suffices to show that $V''(\senv(\sstate''')(y)) = \env''(y)$.

      But now $V''(\senv(\sstate''')(y)) = V''(t) = \hat{V}'(t) = \env(\kresult) = \env''(y)$, which is what we needed to show. Therefore, since $\senv$ is the only component changed between $\sstate''$ and $\sstate'''$,
      $$\simstate{V''}{\sstate'''}{\heap}{\perms''}{\env''}.$$

      Therefore, since $\sstate(\vstate'') = \sstate'''$, we have shown that $\vstate''$ corresponds to $\pair{\heap}{\triple{\perms''}{\env''}{s} \cdot \stack^*}$ with valuation $V''$.

      Part \ref{def:state-valid-partial}: We need to show that the partial state $\pair{\heap}{\stack^*}$ is validated by $\vstate''$ and $V''$. We already have that $\pair{\heap}{\stack^*}$ is validated by $\vstate'$ and $V'$. Thus one of \ref{def:partial-valid-nil}, \ref{def:partial-valid-call}, or \ref{def:partial-valid-while} must apply.

      \textit{If \ref{def:partial-valid-nil} applies:} Then $\stack^* = \nilsym$, thus trivially the partial state $\pair{\heap}{\stack^*}$ is validated by $\vstate''$ and $V''$.

      \textit{If \ref{def:partial-valid-call} applies:}
      Then $\stack^* = \triple{\perms_0}{\env_0}{\sseq{y' \kassign m'(e_1', \cdots, e_{k'}')}{s'}} \cdot \stack^*_0$ for some $k', y', m', e_1', \cdots, e_{k'}', s', \stack^*_0$. Also, there exists some $\vstate_0', V_0', x_1', \cdots, x_{k'}', t_1', \cdots, t_{k'}', \sstate_0, \cdots, \sstate_{k'}, \sstate'$ such that
      \begin{gather*}
        \text{The partial state $\pair{\heap}{\stack_0^*}$ is validated by $\vstate_0'$ and $V_0'$},\\
        \vstate_0' \text{ is reachable from $\prog$ with valuation $V_0'$}, \quad s(\vstate_0') = s(\stack^*), \\
        x_1', \cdots, x_{k'}' = \fparams(m),\\
        \sstate_0 = \sstate(\vstate_0'), \quad \seval{\sstate_0}{e_1'}{t_1'}{\sstate_1}{\_}, \quad\cdots,\quad \seval{\sstate_{k'-1}}{e_{k'}'}{t_{k'}'}{\sstate_{k'}}{\_},\\
        \scons{\sstate_{k'}}{\fpre(m')}{\sstate'}{\_}, \quad \simstate{V_0'}{\sstate'[\senv = \senv(\sstate_0)]}{\heap}{\perms_0}{\env_0}, \\
        \universal{1 \le i \le k'}{V'(\senv(\vstate')(x_i)) = V_0'(t_i')}, \quad\text{and} \\
        \gform(\vstate') = \fpost(m').
      \end{gather*}

      We want to show that the partial state $\pair{\heap}{\stack^*}$ is validated by $\vstate''$ and $V''$. Immediately from above,
      \begin{gather*}
        \text{The partial state $\pair{\heap}{\stack_0^*}$ is validated by $\vstate_0'$ and $V_0'$},\\
        \vstate_0' \text{ is reachable from $\prog$ with valuation $V_0'$}, \quad s(\vstate_0') = s(\stack^*), \\
        x_1', \cdots, x_{k'}' = \fparams(m),\\
        \sstate_0 = \sstate(\vstate_0'), \quad \seval{\sstate_0}{e_1'}{t_1'}{\sstate_1}{\_}, \quad\cdots,\quad \seval{\sstate_{k'-1}}{e_{k'}'}{t_{k'}'}{\sstate_{k'}}{\_}, \\
        \scons{\sstate_{k'}}{\fpre(m')}{\sstate'}{\_}, \quad \simstate{V_0'}{\sstate'[\senv = \senv(\sstate_0)]}{\heap}{\perms_0}{\env_0}.
      \end{gather*}
      Also, the frame $\triple{\perms'}{\env'}{\sseq{y \kassign m(e_1, \cdots, e_k)}{s}}$ must be executing the body of $m'$, since it is in the stack immediately above the frame that contains $y' \kassign m'(e_1', \cdots, e_{k'}')$. Therefore, since $x_1', \cdots, x_{k'}$ are all parameters of $m'$, $y$ must be distinct from all of $x_1', \cdots, x_{k'}'$, since we do not allow assignment to parameters in a well-formed program. Thus
      $$\forall 1 \le i \le k' : V''(\senv(\vstate'')(x_i)) = V''(\senv(\vstate')(x_i)) = V'(\senv(\vstate')(x_i)) = V_0'(t_i).$$

      Finally, $\gform(\vstate'') = \gform(\vstate')$ by definition, thus
      $$\gform(\vstate'') = \gform(\vstate') = \fpost(m').$$

      Therefore the partial state $\pair{\heap}{\stack^*}$ is validated by $\vstate''$ and $V''$ in this case.

      \textit{If \ref{def:partial-valid-while} applies:}
      Then $\stack^* = \triple{\env_0}{\perms_0}{\sseq{\swhile{e}{\gform_0}{s_0}}{s_0'}} \cdot \stack_0^*$ for some $\env_0$, $\perms_0$, $e$, $\gform_0$, $s_0$, $s_0'$, $\stack_0^*$, and there exists some $\vstate_0'$, $V_0'$, and $\sstate_0'$ such that:
      \begin{gather*}
        \text{The partial state $\pair{\heap}{\stack_0^*}$ is validated by $\vstate_0'$ and $V_0'$} \\
        \vstate_0' \text{ is reachable from $\prog$ with valuation $V_0'$}, \quad s(\vstate_0') = s(\stack^*), \\
        \scons{\sstate_0}{\gform_0}{\sstate_0'}{\_}, \quad
        \simstate{V_0'}{\sstate_0'}{\heap}{\perms_0}{\env_0} \quad\text{and}\\
        \gform(\vstate') = \gform_0.
      \end{gather*}

      Now, by definition of $\vstate''$, $\gform(\vstate'') = \gform(\vstate') = \gform_0$. Therefore, using the other assumptions given above, the partial state $\pair{\heap}{\stack^*}$ is validated by $\vstate''$ and $V''$ in this case.

      Therefore definition part \ref{def:state-valid-partial} is satisfied.

      Therefore all parts of definition \ref{def:state-valid} are satisfied. Thus $\pair{\heap}{\triple{\perms''}{\env''}{s} \cdot \stack^*}$ is validated by $\vstate''$, as we wanted to show.

    \case \refrule{ExecAssert}: We have
      \begin{gather*}
        \dexec{\heap}{\triple{\perms}{\env}{\sseq{\sassert{\phi}}{s}} \cdot \stack^*}{\vfoot{V'}{\heap}{\sperms}}{\heap}{\triple{\perms}{\env}{s} \cdot \stack^*} \\
        \text{where}\quad \assertion{\heap}{\perms}{\env}{\simprecise{\phi}}.
      \end{gather*}

      Since the initial state is validated by $\vstate$, $\vstate = \triple{\sstate}{\sseq{\sassert{\phi}}{s}}{\gform}$ for some $\sstate$, $\gform$ where $\simstate{V}{\sstate}{\heap}{\perms}{\env}$.

      The only guard rule that applies is \refrule{SGuardAssert}, so we have, for some $\sstate'$,
      \begin{gather*}
        \scons{\sstate}{\simprecise{\phi}}{\sstate'}{\scheck} \\
        \text{and by assumptions } V'(\pc(\sstate')) = \ktrue \quad\text{and}\quad \rtassert{V'}{\heap}{\perms}{\scheck}.
      \end{gather*}
      Also, by definition $V' = V[\scons{\sstate}{\simprecise{\phi}}{\sstate'}{\scheck} \mid \heap]$.

      Thus by lemma \ref{lem:cons-soundness} $\simstate{V'}{\sstate'}{\heap}{\perms \setminus \efoot{\heap}{\env}{\simprecise{\phi}}}{\env}$.

      Thus by lemma \ref{lem:simstate-monotonicity}, $\simstate{V'}{\sstate'}{\heap}{\perms}{\env}$. Also, as noted before, $\assertion{\heap}{\perms}{\env}{\simprecise{\phi}}$. Therefore, by lemma \ref{lem:produce-soundness}, for some $\sstate''$,
      $$\sproduce{\sstate'}{\simprecise{\phi}}{\sstate''} \quad\text{and}\quad V''(\pc(\sstate'')) = \ktrue.$$

      Now, by \refrule{SExecAssert}, $\sexec{\sstate}{\sseq{\sassert{\phi}}{s}}{s}{\sstate[\pc = \pc(\sstate'')]}$.

      Let $\vstate' = \triple{\sstate[\pc = \pc(\sstate'')]}{s}{\gform}$. We want to show that $\pair{\heap}{\triple{\perms}{\env}{s} \cdot \stack^*}$ is validated by $\vstate'$ and $V''$.

      By \refrule{SVerifyStep}, $\vstate'$ is reachable from $\prog$ with valuation $V''$.

      By assumptions, $\simstate{V}{\sstate}{\heap}{\perms}{\env}$, and $V''(\pc(\sstate'')) = \ktrue$, thus $\simstate{V''}{\sstate[\pc = \pc(\sstate'')]}{\heap}{\perms}{\env}$. Also, by definition, $\vstate' = \triple{\sstate[\pc = \pc(\sstate'')]}{s}{\gform}$. Therefore $\vstate'$ corresponds to $\pair{\heap}{\triple{\perms}{\env}{s} \cdot \stack^*}$ with valuation $V''$.

      Finally, by definition $\senv(\vstate') = \senv(\vstate)$ and $\gform(\vstate') = \gform = \gform(\vstate)$.

      Therefore $\pair{\heap}{\triple{\perms}{\env}{s} \cdot \stack^*}$ is a valid state by lemma \ref{lem:preservation-heap-env-unchanged}.

    \case\label{case:pres-dexec-ifa} \refrule{ExecIfA}: We have
      \begin{gather*}
        \dexec{\heap}{\triple{\perms}{\env}{\sseq{\sif{e}{s_1}{s_2}}{s}} \cdot \stack^*}{\vfoot{V'}{\heap}{\sperms}}{\heap}{\triple{\perms}{\env}{\sseq{s_1}{s}} \cdot \stack^*} \\
        \text{where}\quad \eval{\heap}{\env}{e}{\ktrue} \quad\text{and}\quad \frm{\heap}{\perms}{\env}{e}
      \end{gather*}

      Since the initial state is validated by $\vstate$, $\vstate = \triple{\sstate}{\sseq{\sif{e}{s_1}{s_2}}{s}}{\gform}$ for some $\sstate$, $\gform$ where $\simstate{V}{\sstate}{\heap}{\perms}{\env}$.

      The only guard rule that applies is \refrule{SGuardIf}, so we have, for some $\sstate'$,
      \begin{gather*}
        \seval{\sstate}{e}{t}{\sstate'}{\scheck} \\
        \text{and by assumptions}\quad V'(\pc(\sstate')) = \ktrue \quad\text{and}\quad \rtassert{V'}{\heap}{\perms}{\scheck}
      \end{gather*}
      where $V' = V[\seval{\sstate}{e}{t}{\sstate'}{\scheck} \mid \heap]$.

      Now by \refrule{SExecIfA},
      $$\sexec{\sstate}{\sseq{\sif{e}{s_1}{s_2}}{s}}{\sseq{s_1}{s}}{\sstate'[\pc = \pc(\sstate') \kand t]}.$$
      
      Let $\vstate' = \triple{\sstate'[\pc = \pc(\sstate') \kand t]}{\sseq{s_1}{s}}{\gform}$. Then by \refrule{SVerifyStep}, $\vstate'$ is reachable from $\prog$ with valuation $V'$.

      Now $\eval{\heap}{\env}{e}{V'(t)}$ and $\eval{\heap}{\env}{e}{\ktrue}$, thus $V'(t) = \ktrue$. Also, $\simstate{V'}{\sstate'}{\heap}{\perms}{\env}$, thus $V'(\pc(\sstate')) = \ktrue$, and then $V'(\pc(\sstate') \kand t) = V'(\pc(\sstate')) \wedge V'(t) = \ktrue$. Therefore $\simstate{V'}{\sstate'[\pc = \pc(\sstate') \kand t]}{\heap}{\perms}{\env}$.

      Also, by definition, $\vstate' = \triple{\sstate[\pc = \pc(\sstate') \kand t]}{\sseq{s_1}{s}}{\gform}$. Therefore $\vstate'$ corresponds to \\
      $\pair{\heap}{\triple{\perms}{\env}{\sseq{s_1}{s}} \cdot \stack^*}$ with valuation $V''$.

      Finally, $\senv(\vstate') = \senv(\sstate') = senv(\sstate) = \senv(\vstate)$ by lemma \ref{lem:eval-unchanged} and $\gform(\vstate') = \gform = \gform(\vstate)$ by definition.

      Therefore ${\heap}{\triple{\perms}{\env}{\sseq{s_1}{s}} \cdot \stack^*}$ is a valid state by lemma \ref{lem:preservation-heap-env-unchanged}.

    \case \refrule{ExecIfB}: Similar to case \ref{case:pres-dexec-ifa}, but using \refrule{SExecIfB}.

    \case \refrule{ExecWhileEnter}: We have
      \begin{align*}
        &\pair{\heap}{\triple{\perms}{\env}{\sseq{\swhile{e}{\gform}{s}}{s'}} \cdot \stack^*} \to \\
        &\quad \pair{\heap}{\triple{\perms'}{\env}{\sseq{s}{\kskip}} \cdot \triple{\perms \setminus \perms'}{\env}{\sseq{\swhile{e}{\gform}{s}}{s'}} \cdot \stack^*} \\
        &\text{where}~ \eval{\heap}{\env}{e}{\ktrue}, \quad \assertion{\heap}{\perms \setminus \xperms}{\env}{\gform}, \\
        &\hspace{3em} \xperms = \vfoot{V}{\heap}{\sperms}, \quad\text{and}\quad \perms' = \foot{\heap}{\perms \setminus \xperms}{\env}{\gform}
      \end{align*}

      Since the initial state is validated by $\vstate$, $\vstate = \triple{\sstate}{\sseq{\swhile{e}{\gform}{s}}{s'}}{\gform_0}$ for some $\sstate$, $\gform$ where $\simstate{V}{\sstate}{\heap}{\perms}{\env}$.

      The only guard rule that applies is \refrule{SGuardWhile}, so we have, for some $\sstate'$, $\sstate''$, $k$, $x_1, \cdots, x_k$, $t_1, \cdots, t_k$, and $t$,
      \begin{gather}
        \scons{\sstate}{\gform}{\sstate'}{\scheck'}, \quad x_1, \cdots, x_k = \fmodified(s), \quad t_1 = \ffresh, \cdots, t_k = \ffresh, \label{eq:pres-dexec-while-enter-cons} \\
        \sproduce{\sstate'[\senv = \senv(\sstate')[x_1 \mapsto t_1, \cdots, x_k \mapsto t_k]]}{\gform}{\sstate''}, \\
        \spceval{\sstate''}{e}{t}{\scheck''}, \quad \sperms = \frem(\sstate', \gform), \\
        \text{and by assumptions}\quad V'(\pc(\sstate'')) = \ktrue \quad\text{and}\quad \rtassert{V'}{\heap}{\perms}{\scheck' \cup \scheck''}
      \end{gather}
      where $V'$ is the corresponding valuation for these judgements (see definition \ref{def:sguard-valuation}).

      By lemma \ref{lem:scheck-monotonicity}
      \begin{equation} \label{eq:pres-dexec-while-enter-rt}
        \rtassert{V'}{\heap}{\perms}{\scheck'} \quad\text{and}\quad \rtassert{V'}{\heap}{\perms}{\scheck''}.
      \end{equation}

      Let $t_1' = \ffresh, \cdots, t_k' = \ffresh$ and $\sstate_0 = \quintuple{\bot}{\emptyset}{\emptyset}{\senv(\sstate)[x_1 \mapsto t_1', \cdots, x_k \mapsto t_k']}{\pc(\sstate)}$.

      Let $V_0 = V'[t_1' \mapsto \env(x_1), \cdots, t_k' \mapsto \env(x_k)]$. Now, for any $x \in \dom(\senv(\sstate_0))$, if $x = x_i$ for some $i$, then $V_0(\senv(x)) = V_0(\senv(\sstate_0)(x_i)) = V_0(t_i) = \env(x_i) = \env(x)$. Otherwise, $x \in \dom(\senv(\sstate))$ and thus $V_0(\senv(\sstate)(x)) = V(\senv(\sstate)(x)) = \env(x)$ since $\simenv{V}{\senv(\sstate)}{\env}$. Therefore $\simenv{V_0}{\senv(\sstate_0)}{\env}$.

      Also, since $\sheap(\sstate_0) = \oheap(\sstate_0) = \emptyset$, $\simheap{V_0}{\sheap(\sstate_0)}{\heap}{\perms' \setminus \efoot{\heap}{\env}{\gform}}$ and $\simheap{V_0}{\oheap(\sstate_0)}{\heap}{\perms' \setminus \efoot{\heap}{\env}{\gform}}$. Finally, $V_0(\pc(\sstate_0)) = V(\pc(\sstate)) = \ktrue$ since $\simstate{V}{\sstate}{\heap}{\perms}{\env}$.

      Therefore
      $$\simstate{V_0}{\sstate_0}{\heap}{\perms' \setminus \efoot{\heap}{\env}{\gform}}{\env}$$
      and then also $\simstate{V_0}{\sstate_0}{\heap}{\perms'}{\env}$ by lemma \ref{lem:simstate-monotonicity}.

      Furthermore, by assumptions, $\assertion{\heap}{\perms \setminus \xperms}{\env}{\gform}$, thus $\assertion{\heap}{\perms'}{\env}{\gform}$ by lemma \ref{lem:foot-assert}, since $\perms' = \foot{\heap}{\perms \setminus \xperms}{\env}{\gform}$.

      Therefore, by lemma \ref{lem:produce-progress} $$\sproduce{\sstate_0}{\gform}{\sstate_0'} \quad\text{and}\quad V_0'(\pc(\sstate_0')) = \ktrue$$
      where $V_0' = V_0[\sproduce{\sstate_0}{\gform}{\sstate_0'} \mid \heap]$. Also, by lemma \ref{lem:produce-soundness},
      $$\simstate{V_0'}{\sstate_0'}{\heap}{\perms'}{\env}.$$

      Now by lemma \ref{lem:pc-eval-progress},
      $$\spceval{\sstate_0'}{e}{t_0}{\_}$$
      and let $V_0'' = V_0'[\spceval{\sstate_0'}{e}{t_0}{\_}]$.

      Let $\vstate_0 = \triple{\sstate_0'[\pc = \pc(\sstate_0') \kand t_0]}{\sseq{s'}{\kskip}}{\gform}$.
      
      We want to show that
      $$\Gamma' = \pair{\heap}{\triple{\perms'}{\env}{\sseq{s}{\kskip}} \cdot \triple{\perms \setminus \perms'}{\env}{\sseq{\swhile{e}{\gform}{s}}{s'}} \cdot \stack^*}$$
      is validated by $\vstate_0$ and $V_0''$.

      \textit{Part \ref{def:state-valid-reachable}:}
      By \refrule{SVerifyLoopBody} $\strans{\prog}{\vstate}{\vstate_0}$. Therefore $\vstate_0$ is reachable from $\prog$ with valuation $V_0''$.

      \textit{Part \ref{def:state-valid-correspond}:}
      Since $\eval{\heap}{\env}{e}{\ktrue}$, by lemma \ref{lem:spceval-correspondence}, $V_0''(t_0) = \ktrue$. Therefore $V_0''(\pc(\sstate_0') \kand t_0) = V'(\pc(\sstate_0')) \wedge V'(t_0') = \ktrue$. Thus, since we have already shown $\simstate{V_0'}{\sstate_0'}{\heap}{\perms'}{\env}$,
      $$\simstate{V_0''}{\sstate_0'[\pc = \pc(\sstate_0') \kand t_0]}{\heap}{\perms'}{\env}.$$
      
      Also, $s(\vstate_0) = \sseq{s}{\kskip}$ by definition. Therefore $\Gamma'$ corresponds to $\vstate_0$ with $V_0''$.

      \textit{Part \ref{def:state-valid-partial}:}
      We want to show that the partial state
      $$\Gamma^* = \pair{\heap}{\triple{\perms \setminus \perms'}{\env}{\sseq{\swhile{e}{\gform}{s}}{s'}} \cdot \stack^*}$$
      is validated by $\vstate_0$ and $V_0''$.

      By assumptions, $\pair{\heap}{\triple{\perms}{\env}{\sseq{\swhile{e}{\gform}{s}}{s'}} \cdot \stack^*}$ was validated by $\vstate$ and $V$. Therefore the partial state $\pair{\heap}{\stack^*}$ was validated by $\vstate$ and $V$, $\vstate$ is reachable from $\prog$ with $V$, $s(\vstate) = \sseq{\swhile{e}{\gform}{s}}{s'}$, and $\simstate{V}{\sstate}{\heap}{\perms}{\env}$.
      
      By \eqref{eq:pres-dexec-while-enter-cons} $\scons{\sstate}{\gform}{\sstate'}{\scheck}$ and by \eqref{eq:pres-dexec-while-enter-rt} $\rtassert{V'}{\heap}{\perms}{\scheck}$. Therefore $\simstate{V'}{\sstate'}{\heap}{\perms \setminus \efoot{\heap}{\env}{\gform}}{\env}$.

      Furthermore, since $\xperms = \vfoot{V'}{\heap}{\sperms} = \vfoot{V'}{\heap}{\frem(\sstate', \gform)}$, $\perms' = \foot{\heap}{\perms \setminus \xperms}{\env}{\gform}$, and $\assertion{\heap}{\perms \setminus \xperms}{\env}{\gform}$, we can apply lemma \ref{lem:rem-simstate} to get
      $$\simstate{V'}{\sstate'}{\heap}{\perms \setminus \perms'}{\env}.$$

      Finally, $\gform(\vstate_0) = \gform$ by definition.

      Therefore $\Gamma^*$ is validated by $\vstate_0$ and $V_0''$ since we have satisfied all requirements of case \ref{def:partial-valid-while}.

      Therefore $\Gamma'$ is validated by $\vstate_0$ and $V_0''$, which completes the proof.

    \case \refrule{ExecWhileSkip}: We have
      \begin{gather*}
        \dexec{\heap}{\triple{\perms}{\env}{\sseq{\swhile{e}{\gform}{s}}{s'}} \cdot \stack^*}{\xperms}{\heap}{\triple{\perms}{\env}{s} \cdot \stack^*} \\
        \text{where}\quad \xperms = \vfoot{V'}{\heap}{\sperms}, \quad \eval{\heap}{\env}{e}{\kfalse}, \quad
        \frm{\heap}{\perms}{\env}{e}, \quad\text{and} \\
        \assertion{\heap}{\perms \setminus \xperms}{\env}{\gform}
      \end{gather*}

      Since the initial state is validated by $\vstate$, $\vstate = \triple{\sstate}{\sseq{\swhile{e}{\gform}{s}}{s'}}{\gform'}$ for some $\sstate$, $\gform'$ where $\simstate{V}{\sstate}{\heap}{\perms}{\env}$.

      The only guard rule that applies is \refrule{SGuardWhile}, so we have, for some $\sstate'$, $\sstate''$, $k$, $x_1, \cdots, x_k$, $t_1, \cdots, t_k$, and $t$,
      \begin{gather}
        \scons{\sstate}{\gform}{\sstate'}{\scheck'}, \quad x_1, \cdots, x_k = \fmodified(s), \quad t_1 = \ffresh, \cdots, t_k = \ffresh, \label{eq:pres-dexec-while-skip-cons} \\
        \sproduce{\sstate'[\senv = \senv(\sstate')[x_1 \mapsto t_1, \cdots, x_k \mapsto t_k]]}{\gform}{\sstate''},\label{eq:pres-dexec-while-skip-prod} \\
        \spceval{\sstate''}{e}{t}{\scheck''}, \quad \sperms = \frem(\sstate', \gform), \\
        \text{and by assumptions}\quad V'(\pc(\sstate'')) = \ktrue \quad\text{and}\quad \rtassert{V'}{\heap}{\perms}{\scheck' \cup \scheck''} \label{eq:pres-dexec-while-skip-path}
      \end{gather}
      where $V'$ is the corresponding valuation for these judgements (see definition \ref{def:sguard-valuation}).

      Let $\vstate' = \triple{\sstate''[\pc = \pc(\sstate'') \kand \kneg t]}{s'}{\gform'}$.

      By \refrule{SExecWhileSkip}
      $$\sexec{\sstate}{\sseq{\swhile{e}{\gform}{s}}{s'}}{s'}{\sstate''[\pc = \pc(\sstate'') \kand \kneg t]}.$$
      Therefore by \refrule{SVerifyStep} $\strans{\prog}{\vstate}{\vstate'}$. Thus $\vstate'$ is reachable from $\prog$ with valuation $V'$.

      By lemma \ref{lem:scheck-monotonicity}
      \begin{equation} \label{eq:pres-dexec-while-skip-rt}
        \rtassert{V'}{\heap}{\perms}{\scheck'} \quad\text{and}\quad \rtassert{V'}{\heap}{\perms}{\scheck''}.
      \end{equation}

      By lemmas \ref{lem:cons-subpath} and \ref{lem:produce-subpath}, $\pc(\sstate'') \implies \pc(\sstate')$. Therefore $V'(\pc(\sstate')) = \ktrue$. Now, by lemma \ref{lem:cons-soundness},
      $$\simstate{V'}{\sstate'}{\heap}{\perms \setminus \efoot{\heap}{\perms}{\gform}}.$$

      By definition \ref{def:sguard-valuation}, for all $1 \le i \le k$, $V'(t_i) = V(\senv(\sstate)(x_i))$. By assumptions, $\simenv{V}{\senv(\sstate)}{\env}$, thus $V(\senv(\sstate)(x_i)) = \env(x_i)$. Also, as shown above, $\simenv{V'}{\senv(\sstate')}{\env}$.

      Let $\senv' = \senv(\sstate')[x_1 \mapsto t_1, \cdots, x_k \mapsto t_k]$. Now, for any $x \in \dom(\senv')$, if $x = x_i$ for some $i$ then $V'(\senv'(x)) = V'(t_i) = \env(x_i)$. Otherwise, $x \in \dom(\senv(\sstate'))$ and thus $V'(\senv'(x)) = V'(\senv(\sstate')(x)) = \env(x)$. Therefore $\simenv{V'}{\senv'}{\env}$.

      Therefore $\simstate{V'}{\sstate'[\senv = \senv']}{\heap}{\perms \setminus \efoot{\heap}{\perms}}{\env}$. Using the definition of $\senv'$ and \eqref{eq:pres-dexec-while-skip-prod}, $\sproduce{\sstate'[\senv = \senv']}{\gform}{\sstate''}$, and by \eqref{eq:pres-dexec-while-skip-path}, $V'(\pc(\sstate'')) = \ktrue$.

      In addition, as shown before, $\assertion{\heap}{\perms \setminus \xperms}{\env}{\gform}$, thus by lemma \ref{lem:assert-monotonicity} $\assertion{\heap}{\perms}{\env}{\gform}$.

      Now by lemma \ref{lem:produce-soundness}, $\simstate{V'}{\sstate''}{\heap}{\perms}{\env}$.

      Since $\spceval{\sstate''}{e}{t}{\_}$ and $\eval{\heap}{\env}{e}{\kfalse}$, by lemma \ref{lem:spceval-correspondence} $V'(t) = \kfalse$. Therefore $V'(\pc(\sstate'') \kand \kneg t) = V'(\pc(\sstate'')) \wedge \neg V'(t) = \ktrue$. Therefore
      $$\simstate{V'}{\sstate''[\pc = \pc(\sstate'') \kand \kneg t]}{\heap}{\perms}{\env}.$$

      Now, by definition, $s(\vstate') = s'$. Therefore $\pair{\heap}{\triple{\perms}{\env}{s} \cdot \stack^*}$ corresponds to $\vstate'$ with valuation $V'$.

      By definition $\gform(\vstate') = \gform' = \gform(\vstate)$. Also, $\senv(\vstate') = \senv(\sstate'') = \senv(\sstate')[x_1 \mapsto t_1, \cdots, x_k \mapsto t_k]$ by lemma \ref{lem:produce-unchanged} and $\senv(\sstate') = \senv(\sstate) = \senv(\vstate)$ by lemma \ref{lem:cons-unchanged}. Therefore $\dom(\senv(\vstate')) \supseteq \dom(\senv(\vstate))$.

      Thus by lemma \ref{lem:preservation-heap-env-unchanged} $\pair{\heap}{\triple{\perms}{\env}{s} \cdot \stack^*}$ is a valid state.

    \case \refrule{ExecWhileFinish}: We have
    \begin{align*}
      &\pair{\heap}{\triple{\perms'}{\env'}{\kskip} \cdot \triple{\perms}{\env}{\sseq{\swhile{e}{\gform}{s}}{s'}} \cdot \stack^*} \to \\
      &\quad \pair{\heap}{\triple{\perms''}{\env'}{\sseq{\swhile{e}{\gform}{s}}{s'}} \cdot \stack^*} \\
      &\text{where}\quad \assertion{\heap}{\perms'}{\env'}{\gform} \quad\text{and}\quad \perms'' = \perms \cup \foot{\heap}{\perms'}{\env'}{\gform}.
    \end{align*}

    By assumptions, the initial state is validated by some $\vstate$ and valuation $V'$, thus $\vstate' = \triple{\sstate'}{\kskip}{\gform'}$ for some $\sstate'$, $\gform$ where $\simstate{V}{\sstate'}{\heap}{\perms'}{\env'}$.

    Also, by \ref{def:state-valid-partial}, the partial state $\pair{\heap}{\triple{\perms}{\env}{\sseq{\swhile{e}{\gform}{s}}{s'}} \cdot \stack^*}$ is validated by $\vstate'$ and $V'$. Thus \ref{def:partial-valid-while} must apply, and thus there is some $\vstate_0$, $V_0$, $\sstate_0$, $\sstate_0'$, $\gform_0$ such that
    \begin{gather*}
      \vstate_0 = \triple{\sstate_0}{\sseq{\swhile{e}{\gform}{s}}{s'}}{\gform_0} \\
      \text{$\vstate_0$ is reachable from $\prog$ with valuation $V_0$} \\
      \scons{\sstate_0}{\gform}{\sstate_0'}{\_}, \quad \simstate{V_0'}{\sstate_0'}{\heap}{\perms}{\env}, \quad \gform' = \gform, \\
      \text{and}\quad V_0' = V_0[\scons{\sstate}{\gform}{\sstate'}{\_} \mid \heap].
    \end{gather*}

    We have $\assertion{\heap}{\perms'}{\env'}{\gform}$, thus by lemma \ref{lem:foot-assert}$\assertion{\heap}{\foot{\heap}{\perms'}{\env'}{\gform}}{\env'}{\gform}$, and then by lemma \ref{lem:assert-monotonicity} $\assertion{\heap}{\perms''}{\env'}{\gform}$.

    Also by lemma \ref{lem:simstate-monotonicity} $\simstate{V_0'}{\sstate_0'}{\heap}{\perms''}{\env}$.

    For some $k$, let $x_1, \cdots, x_k$ be the list of variables in $\fmodified(s)$. Let $t_1 = \ffresh, \cdots, t_k = \ffresh$, let $\hat{\senv} = \senv(\sstate_0)[x_1 \mapsto t_1, \cdots, x_k \mapsto t_k]$, and let $\hat{V}_0' = V_0'[t_1 \mapsto \env'(x_1), \cdots, t_k \mapsto \env'(x_k)]$.

    $\env'$ is contained in the stack frame executing the loop body, which is $s$, thus for all $x \in \dom(\env')$, either $\env(x) = \env'(x)$ or $x \in \fmodified(m)$.

    Also, since $\simenv{V_0'}{\senv(\sstate_0')}{\env}$ and $\senv(\sstate_0') = \senv(\sstate_0)$ by lemma \ref{lem:cons-unchanged}, $\simenv{V_0}{\senv(\sstate_0)}{\env}$.

    Now, for any $x \in \dom(\hat{\senv})$, if $x = x_i$ for some $i$, then $\hat{V}_0'(\hat{\senv}(x)) = \hat{V}_0'(\hat{\senv}(x_i)) = \env'(x_i) = \env'(x)$. Otherwise, $x \notin \fmodified(s)$, thus $\env'(x) = \env(x)$, and $x \in \dom(\senv(\sstate_0))$. Thus $\hat{V}_0'(\hat{\senv}(x)) = V_0(\senv(\sstate_0)(x)) = \env(x) = \env'(x_i)$ since $\simenv{V_0}{\senv(\sstate_0)}{\env}$. Therefore $\simenv{\hat{V}_0'}{\hat{\senv}}{\env'}$. Thus $\simstate{\hat{V}_0'}{\sstate_0'[\senv = \senv']}{\heap}{\perms''}{\env'}$.
    
    Therefore by lemma \ref{lem:produce-progress} $\sproduce{\sstate_0'[\senv = \senv']}{\gform}{\sstate_0''}$ for some $\sstate_0''$ such that $V_0''(\pc(\sstate_0'')) = \ktrue$ where $V_0'' = \hat{V}_0'[\sproduce{\sstate_0'}{\gform}{\sstate_0''} \mid \heap]$.

    Let $\vstate_0' = \triple{\sstate_0''}{\sseq{\swhile{e}{\gform}{s}}{s'}}{\gform_0}$. We want to show that
    $$\Gamma' = \pair{\heap}{\triple{\perms''}{\env'}{\sseq{\swhile{e}{\gform}{s}}{s'}} \cdot \stack^*}$$
    is validated by $\vstate_0'$ and $V_0''$.

    \textit{Part \ref{def:state-valid-reachable}:}
    By \refrule{SVerifyLoop}, and since $\vstate_0$ is reachable from $\prog$, $\prog \vdash \vstate_0 \to \vstate_0'$. Therefore $\vstate_0'$ is reachable from $\prog$ with valuation $V_0''$.

    \textit{Part \ref{def:state-valid-correspond}:}
    Since $\assertion{\heap}{\perms'}{\env'}{\gform}$, by lemma \ref{lem:efoot-subset-spec} $\efoot{\heap}{\env'}{\gform} \subseteq \perms'$. Also, since the stack is well-formed, $\perms'$ and $\perms$ are disjoint, thus $\perms \setminus \efoot{\heap}{\env'}{\gform} = \perms$. Therefore $\simstate{V_0'}{\sstate_0'}{\heap}{\perms \setminus \efoot{\heap}{\env'}{\gform}}{\env}$, and now since $\perms \subseteq \perms''$, by lemma lemma \ref{lem:simstate-monotonicity}, and since $\simenv{\hat{V}_0'}{\senv'}{\env'}$,
    $$\simstate{\hat{V}_0'}{\sstate_0'[\senv = \senv']}{\heap}{\perms'' \setminus \efoot{\heap}{\env'}{\gform}}{\env'}.$$

    Now, by lemma \ref{lem:produce-soundness},
    $$\simstate{V_0''}{\sstate_0''}{\heap}{\perms''}{\env'}.$$

    Also, $s(\vstate_0') = s(\Gamma')$ by construction. Therefore $\Gamma'$ corresponds to $\vstate_0'$.

    \textit{Part \ref{def:state-valid-partial}:} We need to show that the partial state $\pair{\heap}{\stack^*}$ is validated by $\vstate_0'$ and $V_0''$. We already have that $\pair{\heap}{\stack^*}$ is validated by $\vstate_0$ and $V_0$. Thus one of \ref{def:partial-valid-nil}, \ref{def:partial-valid-call}, or \ref{def:partial-valid-while} must apply.

    \textit{If \ref{def:partial-valid-nil} applies:} Then $\stack^* = \nilsym$, thus trivially the partial state $\pair{\heap}{\stack^*}$ is validated by $\vstate_0'$ and $V_0''$.

    \textit{If \ref{def:partial-valid-call} applies:}
    Then $\stack^* = \triple{\perms^*}{\env^*}{\sseq{y \kassign m(e_1, \cdots, e_k)}{s^*}} \cdot \stack^*_1$ for some $\perms^*$, $\env^*$, $k$, $y$, $m$, $e_1, \cdots, e_k$, $s^*$, $\stack^*_1$. Also, there exists some $\vstate^*$, $V^*$, $x_1, \cdots, x_k$, $t_1, \cdots, t_k$, $\sstate_0, \cdots, \sstate_k$, $\sstate'$ such that
    \begin{gather*}
      \text{The partial state $\pair{\heap}{\stack^*_1}$ is validated by $\vstate^*$ and $V^*$},\\
      \vstate^* \text{ is reachable from $\prog$ with valuation $V^*$}, \quad s(\vstate^*) = s(\stack^*), \\
      x_1, \cdots, x_k = \fparams(m),\\
      \sstate_0 = \sstate(\vstate^*), \quad \seval{\sstate_0}{e_1}{t_1}{\sstate_1}{\_}, \quad\cdots,\quad \seval{\sstate_{k-1}}{e_k}{t_k}{\sstate_k}{\_},\\
      \scons{\sstate_k}{\fpre(m)}{\sstate'}{\_}, \quad \simstate{V^*}{\sstate'[\senv = \senv(\sstate_0)]}{\heap}{\perms^*}{\env^*}, \\
      \universal{1 \le i \le k}{V_0(\senv(\vstate_0)(x_i)) = V^*(t_i)}, \quad\text{and} \\
      \gform(\vstate_0) = \fpost(m).
    \end{gather*}

    We want to show that the partial state $\pair{\heap}{\stack^*}$ is validated by $\vstate''$ and $V''$. Immediately from above,
    \begin{gather*}
      \text{The partial state $\pair{\heap}{\stack^*_1}$ is validated by $\vstate^*$ and $V^*$},\\
      \vstate^* \text{ is reachable from $\prog$ with valuation $V^*$}, \quad s(\vstate^*) = s(\stack^*), \\
      x_1, \cdots, x_k = \fparams(m),\\
      \sstate_0 = \sstate(\vstate^*), \quad \seval{\sstate_0}{e_1}{t_1}{\sstate_1}{\_}, \quad\cdots,\quad \seval{\sstate_{k-1}}{e_k}{t_k}{\sstate_k}{\_},\\
      \scons{\sstate_k}{\fpre(m)}{\sstate'}{\_}, \quad \simstate{V^*}{\sstate'[\senv = \senv(\sstate_0)]}{\heap}{\perms^*}{\env^*}.
    \end{gather*}
    Also, the frame $\triple{\perms}{\env}{\sseq{\swhile{e}{\gform}{s}}{s'}}$ must be executing the body of $m$, since it is in the stack immediately above the frame that contains $y \kassign m(e_1, \cdots, e_k)$. Therefore, since $x_1, \cdots, x_k$ are all parameters of $m$, $\fmodified(s)$ cannot contain any of $x_1, \cdots, x_k$, since we do not allow assignment to parameters in a well-formed program. Thus
    \begin{align*}
      \forall 1 \le i \le k : V_0''(\senv(\vstate_0')(x_i))
        &= V_0''(\senv(\sstate_0'')(x_i)) &\text{defn. $\vstate_0'$} \\
        &= V_0''(\hat{\senv}(x_i)) &\text{Lemma \ref{lem:produce-unchanged}} \\
        &= V_0''(\senv(\sstate_0')(x_i)) &x_i \notin \fmodified(s) \\
        &= V_0''(\senv(\sstate_0)(x_i)) &Lemma \ref{lem:cons-unchanged} \\
        &= V_0(\senv(\sstate_0)(x_i)) &V_0 \subseteq V_0'' \\
        &= V^*(t_i) &\text{prev. assump.}
    \end{align*}

    Finally, $\gform(\vstate_0') = \gform(\vstate_0)$ by definition, thus
    $$\gform(\vstate_0') = \gform(\vstate_0) = \fpost(m).$$

    Therefore the partial state $\pair{\heap}{\stack^*}$ is validated by $\vstate_0'$ and $V_0''$ in this case.

    \textit{If \ref{def:partial-valid-while} applies:}
    Then $\stack^* = \triple{\env^*}{\perms^*}{\sseq{\swhile{e^*}{\gform^*}{s^*}}{s^{*\prime}}} \cdot \stack_1^*$ for some $\env^*$, $\perms^*$, $e^*$, $\gform^*$, $s^*$, $s^{*\prime}$, $\stack_1^*$, and there exists some $\vstate^*$, $V^*$, and $\sstate'$ such that:
    \begin{gather*}
      \text{The partial state $\pair{\heap}{\stack_1^*}$ is validated by $\vstate^*$ and $V^*$} \\
      \vstate^* \text{ is reachable from $\prog$ with valuation $V^*$} \quad s(\vstate^*) = s(\stack^*), \\
      \scons{\sstate(\vstate^*)}{\gform^*}{\sstate'}{\_}, \quad
      \simstate{V^*}{\sstate'}{\heap}{\perms^*}{\env^*} \quad\text{and}\\
      \gform(\vstate^*) = \gform^*.
    \end{gather*}

    Now, by definition of $\vstate_0'$, $\gform(\vstate_0') = \gform(\vstate_0) = \gform^*$. Therefore, using the other assumptions given above, the partial state $\pair{\heap}{\stack^*}$ is validated by $\vstate_0'$ and $V_0''$ in this case.

    Therefore definition part \ref{def:state-valid-partial} is satisfied.

    Therefore all parts of definition \ref{def:state-valid} are satisfied. Thus $\Gamma'$ is validated by $\vstate_0'$ and $V_0''$, as we wanted to show.

    \case \refrule{ExecFold}: We have
      $$\dexec{\heap}{\triple{\perms}{\env}{\sseq{\sfold{p(\multiple{e})}}{s}} \cdot \stack^*}{\vfoot{V'}{\heap}{\sperms}}{\heap}{\triple{\perms}{\env}{s} \cdot \stack^*}$$

      By assumptions, the initial state is validated by some $\vstate$ and valuation $V'$, thus \\
      $\vstate = \triple{\sstate}{\sseq{\sfold{p(\multiple{e})}}{s}}{\gform}$ for some $\sstate$, $\gform$ where $\simstate{V}{\sstate}{\heap}{\perms}{\env}$.

      The only guard that applies is \refrule{SGuardFold}, thus we have
      \begin{gather*}
        \multiple{\seval{\sstate}{e}{t}{\sstate'}{\scheck}}, \quad \multiple{x} = \fpredparams(p), \\
        \scons{\sstate'[\senv = [x_i \mapsto t_i]]}{\fpred(p)}{\sstate''}{\scheck'}, \\
        \text{and by assumptions}\quad \rtassert{V'}{\heap}{\perms}{\multiple{\scheck} \cup \scheck'} \quad\text{and}\quad V'(\pc(\sstate'')) = \ktrue
      \end{gather*}
      where $V'$ is the valuation corresponding to this series of judgements, extending $V$ (see definition \ref{def:sguard-valuation}).

      Let $e_1, \cdots, e_n = \multiple{e}$, $t_1, \cdots, t_n = \multiple{t}$, and $\scheck_1, \cdots, \scheck_n = \multiple{\scheck}$. Let $\sstate_0$, then for some $\sstate_1, \cdots, \sstate_n$ we have
      $$\seval{\sstate_0}{e_1}{t_1}{\sstate_1}{\scheck_1}, \cdots, \seval{\sstate_{n-1}}{e_n}{t_n}{\sstate_n}{\scheck_n}$$
      where $\sstate_n = \sstate'$. By lemmas \ref{lem:eval-subpath} and \ref{lem:cons-subpath} $\pc(\sstate'') \implies \pc(\sstate_n) \implies \cdots \implies \pc(\sstate_1)$. Therefore $V'(\pc(\sstate'')) = V'(\pc(\sstate_n)) = \cdots = V'(\pc(\sstate_1)) = \ktrue$. Also, by lemma \ref{lem:scheck-monotonicity} we have $\rtassert{V'}{\heap}{\perms}{\scheck_i}$ for all $1 \le i \le n$. Thus, by lemma \ref{lem:seval-soundness}
      \begin{gather*}
        \eval{\heap}{\env}{e_1}{V'(t_1)}, \quad\cdots,\quad \eval{\heap}{\env}{e_n}{V'(t_n)} \\
        \text{and}\quad \simstate{V'}{\sstate_1}{\heap}{\perms}{\env}, \quad\cdots,\quad \simstate{V'}{\sstate_n}{\heap}{\perms}{\env}.
      \end{gather*}
      Therefore $\simstate{V'}{\sstate'}{\heap}{\perms}{\env}$.

      Let $\senv' = [\multiple{x \mapsto t}]$ and $\env' = [\multiple{x \mapsto V'(t)}]$. Then, by construction, $\simenv{V'}{\senv'}{\env'}$. Therefore $\simstate{V'}{\sstate'[\senv = \senv']}{\heap}{\perms}{\env}$.

      From above we have $\scons{\sstate'[\senv = \senv']}{\fpred(p)}{\sstate''}{\scheck'}$ and $\rtassert{\heap}{\perms}{\env}{\scheck''}$ by lemma \ref{lem:scheck-monotonicity}. Thus, by lemma \ref{lem:cons-soundness}
      $\simstate{V'}{\sstate''}{\heap}{\perms \setminus \efoot{\heap}{\env'}{\fpred(p)}}{\env'}$ and thus
      $$\simstate{V'}{\sstate''[\senv = \senv(\sstate)]}{\heap}{\perms \setminus \efoot{\heap}{\env'}{\fpred(p)}}{\env}.$$

      Let $\sheap' = \sheap(\sstate''); \pair{p}{\multiple{t}}$. Expanding definitions,
      $$\vfoot{V'}{\heap}{\pair{p}{\multiple{t}}} = \efoot{\heap}{\env'}{\fpred(p)}.$$
      Now $\simstate{V'}{\sstate''[\senv = \senv(\sstate)]}{\heap}{\perms \setminus \vfoot{V'}{\heap}{\pair{p}{\multiple{t}}}}{\env}$, thus
      $$\universal{h_1, h_2 \in \sheap(\sstate'')}{h_1 \ne h_2 \implies \vfoot{V'}{\heap}{h_1} \cap \vfoot{V'}{\heap}{h_2} = \emptyset}$$
      and by lemma \ref{lem:disjoint-sim-heap-subset},
      $$\universal{h \in \sheap(\sstate'')}{\vfoot{V'}{\heap}{h} \cap \vfoot{V'}{\heap}{\pair{p}{\multiple{t}}} = \emptyset}.$$
      From these we can deduce that
      $$\universal{h_1, h_2 \in \sheap'}{h_1 \ne h_2 \implies \vfoot{V'}{\heap}{h_1} \cap \vfoot{V'}{\heap}{h_2} = \emptyset}.$$

      Also, from lemma \ref{lem:cons-soundness}, $\assertion{\heap}{\perms}{[\multiple{x \mapsto V'(t)}]}{\fpred(p)}$ since $\env' = [\multiple{x \mapsto V'(t)}]$.

      Since $\simstate{V'}{\sstate''[\senv = \senv(\sstate)]}{\heap}{\perms \setminus \vfoot{V'}{\heap}{\pair{p}{\multiple{t}}}}{\env}$,
      $$\universal{\pair{p}{\multiple{t}} \in \sheap(\sstate'')}{\assertion{\heap}{\perms}{[\multiple{x \mapsto V'(t)}]}{\fpred(p)}}.$$
      From these we can deduce that
      $$\universal{\pair{p}{\multiple{t}} \in \sheap'}{\assertion{\heap}{\perms}{[\multiple{x \mapsto V'(t)}]}{\fpred(p)}}.$$

      Since field values are unchanged between $\sheap(\sstate'')$ and $\sheap'$,
      \begin{gather*}
        \universal{\triple{f}{t}{t'} \in \sheap'}{\pair{V(t)}{f} \in \perms} \quad\text{and} \\
        \universal{\triple{f}{t}{t'} \in \sheap'}{\heap(V(t), f) = V(t')}.
      \end{gather*}

      Therefore $\simheap{V'}{\sheap'}{\heap}{\perms}$, and thus
      $$\simstate{V'}{\sstate''[\senv = \senv(\sstate), \sheap = \sheap']}{\heap}{\perms}{\env}.$$

      Let $\vstate' = \triple{\sstate''[\senv = \senv(\sstate), \sheap = \sheap']}{s}{\gform}$. By \refrule{SExecFold} $\sexec{\sstate}{\sseq{\sfold{p(\multiple{e})}}{s}}{s}{\sstate(\vstate')}$ (after expanding definitions). Therefore $\strans{\prog}{\vstate}{\vstate'}$ by \refrule{SVerifyStep}. Therefore $\vstate'$ is reachable from $\prog$ with valuation $V'$.

      Also, as shown before, $\simstate{V'}{\sstate(\vstate')}{\heap}{\perms}{\env}$, and by definition $s(\vstate') = s$. Therefore  $\vstate'$ corresponds to $\pair{\heap}{\triple{\perms}{\env}{s} \cdot \stack^*}$.

      By definition $\senv(\vstate') = \senv(\sstate) = \senv(\vstate)$ and $\gform(\vstate') = \gform = \gform(\vstate)$. Therefore by lemma \ref{lem:preservation-heap-env-unchanged} $\pair{\heap}{\triple{\perms}{\env}{s} \cdot \stack^*}$ is a valid state.

    \case \refrule{ExecUnfold}: We have
      $$\dexec{\heap}{\triple{\perms}{\env}{\sseq{\sunfold{p(\multiple{e})}}{s}} \cdot \stack^*}{\vfoot{V'}{\heap}{\sperms}}{\heap}{\triple{\perms}{\env}{s} \cdot \stack^*}$$

      By assumptions, the initial state is validated by some $\vstate$ and valuation $V'$, thus \\
      $\vstate = \triple{\sstate}{\sseq{\sunfold{p(\multiple{e})}}{s}}{\gform}$ for some $\sstate$, $\gform$ where $\simstate{V}{\sstate}{\heap}{\perms}{\env}$.

      The only guard that applies is \refrule{SGuardUnfold}, thus we have
      The only guard that applies is \refrule{SGuardFold}, thus we have
      \begin{gather*}
        \multiple{\seval{\sstate}{e}{t}{\sstate'}{\scheck}}, \quad
        \scons{\sstate'}{p(\multiple{e})}{\sstate''}{\scheck'}, \\
        \text{and by assumptions}\quad \rtassert{V'}{\heap}{\perms}{\scheck' \cup \bigcup \multiple{\scheck}} \quad\text{and}\quad V'(\pc(\sstate'')) = \ktrue
      \end{gather*}

      Let $e_1, \cdots, e_n = \multiple{e}$, $t_1, \cdots, t_n = \multiple{t}$, and $\scheck_1, \cdots, \scheck_n = \multiple{\scheck}$. Let $\sstate_0$, then for some $\sstate_1, \cdots, \sstate_n$ we have
      $$\seval{\sstate_0}{e_1}{t_1}{\sstate_1}{\scheck_1}, \cdots, \seval{\sstate_{n-1}}{e_n}{t_n}{\sstate_n}{\scheck_n}$$
      where $\sstate_n = \sstate'$. By lemmas \ref{lem:eval-subpath} and \ref{lem:cons-subpath} $\pc(\sstate'') \implies \pc(\sstate_n) \implies \cdots \implies \pc(\sstate_1)$. Therefore $V'(\pc(\sstate'')) = V'(\pc(\sstate_n)) = \cdots = V'(\pc(\sstate_1)) = \ktrue$. Also, by lemma \ref{lem:scheck-monotonicity} we have $\rtassert{V'}{\heap}{\perms}{\scheck_i}$ for all $1 \le i \le n$. Thus, by lemma \ref{lem:seval-soundness}
      \begin{gather*}
        \eval{\heap}{\env}{e_1}{V'(t_1)}, \quad\cdots,\quad \eval{\heap}{\env}{e_n}{V'(t_n)} \\
        \text{and}\quad \simstate{V'}{\sstate_1}{\heap}{\perms}{\env}, \quad\cdots,\quad \simstate{V'}{\sstate_n}{\heap}{\perms}{\env}.
      \end{gather*}
      Therefore $\simstate{V'}{\sstate'}{\heap}{\perms}{\env}$.

      Thus by lemma \ref{lem:cons-soundness}
      $$\assertion{\heap}{\env}{\perms}{p(\multiple{e})} \quad\text{and}\quad \simstate{V'}{\sstate'}{\heap}{\perms \setminus \efoot{\heap}{\env}{p(\multiple{e})}}.$$

      Let $\multiple{x} = \fpredparams(p)$, $\senv' = [\multiple{x \mapsto t}]$, and $\env' = [\multiple{x \mapsto V'(t)}]$. Then, by construction, $\simenv{V'}{\senv'}{\env'}$. Therefore $\simstate{V'}{\sstate'[\senv = \senv']}{\heap}{\perms \setminus \efoot{\heap}{\env}{p(\multiple{e})}}{\env'}$.

      Now, by definition, $\efoot{\heap}{\env}{p(\multiple{e})} = \efoot{\heap}{\env'}{\fpred(p)} \cup \bigcup \multiple{\efoot{\heap}{\env}{e}}$.

      Therefore $\efoot{\heap}{\env'}{\fpred(p)} \subseteq \perms \setminus \efoot{\heap}{\env}{p(\multiple{e})}$, thus by lemma \ref{lem:simstate-monotonicity},
      $$\simstate{V'}{\sstate'[\senv = \senv']}{\heap}{\perms \setminus \efoot{\heap}{\env'}{\fpred(p)}}{\env'}$$
      and $\simstate{V'}{\sstate}{\sstate'[\senv = \senv']}{\heap}{\perms}{\env'}$.

      Since $\assertion{\heap}{\perms}{\env}{p(\multiple{e})}$, by \refrule{AssertPredicate} $\assertion{\heap}{\perms}{\env'}{\fpred(p)}$.

      Therefore by lemma \ref{lem:produce-progress}
      $$\sproduce{\sstate'[\senv = \senv']}{\fpred(p)}{\sstate''} \quad\text{and}\quad V''(\pc(\sstate'')) = \ktrue$$
      where $V'' = V'[\sproduce{\sstate'[\senv = \senv']}{\fpred(p)}{\sstate''} \mid \heap]$. Also, by lemma \ref{lem:produce-soundness} $\simstate{V''}{\sstate''}{\heap}{\perms}{\env'}$, and thus
      $$\simstate{V''}{\sstate''[\senv = \senv(\sstate)]}{\heap}{\perms}{\env}.$$

      Now, by \refrule{SExecUnfold}, $\sexec{\sstate}{\sseq{\sunfold{p(\multiple{e})}}{s}}{s}{\sstate''[\senv = \senv(\sstate)]}$.

      Let $\vstate' = \triple{\sstate''[\senv = \senv(\sstate)]}{s}{\gform}$. By \refrule{SVerifyStep} $\strans{\prog}{\vstate}{\vstate'}$. Therefore $\vstate'$ is reachable from $\prog$ with valuation $V''$.

      Also, as shown before, $\simstate{V''}{\sstate(\vstate')}{\heap}{\perms}{\env}$. Furthermore, $s(\vstate') = s$ by definition. Thus $\pair{\heap}{\triple{\perms}{\env}{s} \cdot \stack^*}$ corresponds to $\vstate'$ with valuation $V''$.

      By definition $\senv(\vstate') = \senv(\sstate) = \senv(\vstate)$ and $\gform(\vstate') = \gform = \gform(\vstate)$. Therefore, by lemma \ref{lem:preservation-heap-env-unchanged} $\pair{\heap}{\triple{\perms}{\env}{s} \cdot \stack^*}$ is a valid state.
  \end{enumcases}
\end{proof}

\begin{theorem}[Preservation]\label{thm:dtrans-preservation}
  Let $\Gamma$ be some dynamic state validated by the $\vstate$ and valuation $V$ for some program $\prog$. If $\sguard{\vstate}{\sstate'}{\scheck}{\sperms}$ with corresponding valuation $V$ extending $V'$, $V'(\pc(\sstate')) = \ktrue$, $\rtassert{V'}{\heap}{\perms(\Gamma)} {\scheck}$, and $\dtrans{\prog}{\Gamma}{\vfoot{V'}{\heap(\Gamma)}{\sperms}}{\Gamma'}$
  then $\Gamma'$ is a valid state.

  In other words, if the dynamic state satisfies the matching symbolic checks, and dynamic execution procedes, then the resulting state is valid.
\end{theorem}

\begin{proof}
  We proceed by cases on the judgement $\dtrans{\prog}{\Gamma}{\vfoot{V'}{\heap(\Gamma)}{\sperms}}{\Gamma'}$.

  \begin{enumcases}
    \case \refrule{ExecInit}: Then $\Gamma = \initsym$ and $\Gamma' = \pair{\emptyset}{\triple{\emptyset}{\emptyset}{s(\prog)} \cdot \nilsym}$. Since $\Gamma$ is validated by $\vstate$, then $\vstate = \initsym$.

    Let $\vstate' = \triple{\quintuple{\bot}{\emptyset}{\emptyset}{\emptyset}{\ktrue}}{s(\prog)}{\ktrue}$. Then by \refrule{SVerifyInit} $\strans{\prog}{\initsym}{\vstate'}$.

    Since $\Gamma'$ has a stack with a sole stack frame, and clearly $\vstate'$ is reachable from $\prog$, in order to show that $\Gamma'$ is validated by $\vstate'$ it suffices to show that $\Gamma'$ corresponds to $\vstate'$. In turn, since clearly $s(\Gamma') = s(\vstate')$, it suffices to show that $\simstate{V}{\sstate(\vstate')}{\heap(\Gamma')}{\perms(\Gamma')}{\env(\Gamma')}$. Since all the requisite sets or values are trivial, this is immediate from the definition.

    \case \refrule{ExecFinal}: Then $\Gamma' = \finalsym$ and $\Gamma = \pair{\heap}{\triple{\perms}{\env}{\kskip} \cdot \nilsym}$. Since $\Gamma$ is validated by $\vstate$, $s(\Gamma) = s(\vstate) = \kskip$. By lemma \ref{lem:cons-progress}, $\scons{\sstate(\vstate)}{\gform(\vstate)}{\sstate'}{\_}$ for some $\sstate'$. Thus $\strans{\prog}{\vstate}{\finalsym}$ by \refrule{SVerifyFinal}.

    \case \refrule{ExecStep}: Then $\Gamma = \pair{\heap}{\stack}$ and $\Gamma' = \pair{\heap'}{\stack'}$ for some $\heap, \stack, \heap', \stack'$ where \\
    $\dexec{\heap}{\stack}{\vfoot{V'}{\heap}{\sperms}}{\heap'}{\stack'}$. Therefore $\Gamma'$ is a valid state by lemma \ref{lem:dexec-preservation}.
  \end{enumcases}
\end{proof}

\end{document}
\endinput
